\documentclass{article}

\usepackage[utf8]{inputenc}
\usepackage{hyperref}
\usepackage{natbib}
\usepackage{fullpage}
\usepackage{amssymb,amsmath,amsthm,amsfonts}
\usepackage{complexity}
\usepackage{cleveref}
\usepackage{paralist}
\usepackage{color}
\usepackage{comment}
\usepackage{mathrsfs}
\usepackage{float}
\usepackage[dvipsnames]{xcolor}
\usepackage{indentfirst}
\usepackage{complexity}
\usepackage{thmtools}
\usepackage{thm-restate}
\usepackage{bm}
\usepackage[boxed,linesnumbered]{algorithm2e}
\usepackage[normalem]{ulem}
\usepackage{rotating}
\usepackage{mathtools}
\usepackage{mdframed}
\usepackage{multicol}

\newcounter{anonymous}
\makeatletter
\newcommand{\dummylabel}[2]{\def\@currentlabel{#2}\label{#1}}
\makeatother

\numberwithin{equation}{section}

\usepackage{tikz}
\usetikzlibrary{positioning,calc,arrows.meta, positioning, shapes.multipart}

\makeatletter
\providecommand*{\cupdot}{%
  \mathbin{%
    \mathpalette\@cupdot{}%
  }%
}
\newcommand*{\@cupdot}[2]{%
  \ooalign{%
    $\m@th#1\cup$\cr
    \hidewidth$\m@th#1\cdot$\hidewidth
  }%
}
\makeatother

\newtheorem{theorem}{Theorem}[section]
\newtheorem{proposition}[theorem]{Proposition}

\newtheorem{lemma}[theorem]{Lemma}
\newtheorem{corollary}[theorem]{Corollary}

\theoremstyle{definition}
\newtheorem{definition}[theorem]{Definition}

\newtheorem{open}{Open Problem}

\theoremstyle{remark}
\newtheorem{preremark}[theorem]{Remark}
\newenvironment{remark}%
  {\vspace{0.5em}\begin{mdframed}[innertopmargin=-0em,skipabove=0em]\small
	\begin{preremark}\upshape}
{\end{preremark}\end{mdframed}}

\newenvironment{remarkn}[1]%
  {\vspace{0.5em}\begin{mdframed}[innertopmargin=0em,skipabove=0em]\small
	\begin{preremark}[#1]\upshape}
{\end{preremark}\end{mdframed}}

\definecolor{corlinks}{RGB}{150,0,150}
\definecolor{cormenu}{RGB}{0,0,150}
\definecolor{corurl}{RGB}{0,0,150}

\hypersetup{
	colorlinks=true,
	urlcolor=corlinks,
	linkcolor=corlinks,
	menucolor=cormenu,
	citecolor=corlinks,
	pdfborder= 0 0 0
}

\newcommand{\eps}{\varepsilon}

\newcommand{\Alice}{\mathsf{A}}
\newcommand{\Bob}{\mathsf{B}}
\newcommand{\eqdef}{\triangleq}
\newcommand{\calP}{\mathcal{P}}
\newcommand{\bbN}{\mathbb{N}}
\newcommand{\zo}{\{0,1\}}
\DeclareMathOperator*{\Ex}{\mathbb{E}}

\newcommand{\F}{\mathbb{F}}

\newcommand{\PV}{\mathsf{PV}}
\newcommand{\dWPHP}{\mathsf{dWPHP}}
\newcommand{\rWPHP}{\mathsf{rWPHP}}
\newcommand{\WPHP}{\mathsf{WPHP}}

\newcommand{\Hard}{\mathsf{Hard}}
\newcommand{\APC}{\mathsf{APC}}
\newcommand{\T}{\mathsf{T}}
\newcommand{\Log}{\mathsf{Log}}

\newcommand{\HARD}{\mathsf{HARD}}

\newcommand{\pr}{\mathsf{pr}}
\newcommand{\LossyCode}{\mathsf{LossyCode}}

\newcommand{\calD}{\mathcal{D}}

\newcommand{\rLB}{\mathsf{rLB}}

\newcommand{\CAPP}{\mathsf{CAPP}}

\newcommand{\TR}{\mathsf{TR}}

\newcommand{\SetDisj}{\mathsf{SetDisj}}
\newcommand{\EQ}{\mathsf{EQ}}

\newcommand{\Jerabek}{\text{Je\v r\'abek}}

\newcommand{\Krajicek}{\text{Kraj\'{i}\v{c}ek}}

\newcommand{\injecto}{\hookrightarrow}

\newcommand{\Size}{\mathsf{Size}}

\newcommand{\emphax}[1]{\textsc{#1}}
\newcommand{\emphmeta}[1]{\emph{#1}}
\newcommand{\pgr}[1]{\paragraph{#1.}}

\newcommand{\bbE}{\mathbb{E}}
\newcommand{\Id}{\mathbb{I}}

\newcommand{\bbQ}{\mathbb{Q}}
\newcommand{\RefYao}{\mathsf{Refuter}(\mathsf{Yao})}
\newcommand{\TFZPP}{\mathsf{TFZPP}}

\newcommand{\rcWPHP}{\texttt{\#}\mathsf{rWPHP}}
\newcommand{\rrWPHP}{\mathsf{rrWPHP}}

\newcommand{\rsLB}{\protect\underrightarrow{\mathsf{rLB}}}
\newcommand{\de}{\text{-}}
\newcommand{\some}{\mathsf{some}}

\newcommand{\pub}{\mathsf{pub}}
\newcommand{\prv}{\mathsf{priv}}

\newcommand{\Apx}{\mathsf{APX}}
\newcommand{\Eval}{\mathsf{Eval}}

\newcommand{\ITR}{\mathsf{ITR}}

\newcommand{\Bool}{\mathsf{Bool}}
\newcommand{\Fix}{\mathsf{Fix}}
\newcommand{\bbM}{\mathbb{M}}

\DeclareMathOperator{\sm}{\#}

\title{A Theory for Probabilistic Polynomial-Time  Reasoning\vspace{0.4cm}}

\author{
Lijie Chen\footnote{University of California at Berkeley. Email: \texttt{\href{mailto:lijiechen@berkeley.edu}{lijiechen@berkeley.edu}}}\\
 \and
Jiatu Li\footnote{Massachusetts Institute of Technology. Email: \texttt{\href{mailto:jiatuli@mit.edu}{jiatuli@mit.edu}}. This work was supported by National Science Foundation under Grant No.~CCF 2420092.}\\
\and
Igor C. Oliveira\footnote{University of Warwick. Email: \texttt{\href{mailto:igor.oliveira@warwick.ac.uk}{igor.oliveira@warwick.ac.uk}}. This work was supported in part by the UKRI Frontier Research Guarantee Grant EP/Y007999/1 and the Centre for Discrete Mathematics and its Applications (DIMAP) at the University of Warwick.}\\
\and
Ryan Williams\footnote{Massachusetts Institute of Technology. Email: \texttt{\href{mailto:rrw@mit.edu}{rrw@mit.edu}}. Parts of this work were completed while the author was visiting the Institute for Advanced Study, Princeton, NJ. This material is based upon work supported by the National Science Foundation under grants DMS-2424441 (at IAS) and CCF-2420092 (at MIT).}\\
}

\begin{document}
\maketitle

\pagenumbering{gobble} 

\begin{abstract} 
In this work, we propose a new bounded arithmetic theory, denoted $\APX_1$, designed to formalize a broad class of probabilistic arguments commonly used in theoretical computer science. Under plausible assumptions, $\APX_1$ is strictly weaker than previously proposed frameworks, such as  the  theory $\APC_1$ introduced in the seminal work of Jeřábek (2007). From a computational standpoint, $\APX_1$ is closely tied to approximate counting and to the central question in derandomization, the $\pr \BPP$ versus $\pr \P$ problem, whereas $\APC_1$ is linked to the dual weak pigeonhole principle and to the existence of Boolean functions with exponential circuit complexity.

A key motivation for introducing $\APX_1$ is that its weaker axioms expose finer proof-theoretic structure, making it a natural setting for several lines of research, including unprovability of complexity conjectures and reverse mathematics of randomized lower bounds. In particular, the framework we develop for $\APX_1$ enables the formulation of precise questions concerning the provability of $\pr \BPP = \pr \P$ in \emph{deterministic} feasible mathematics. Since the (un)provability of $\P$ versus $\NP$ in bounded arithmetic has long served as a central theme in the field, we expect this line of investigation to be of particular interest. 

Our technical contributions include developing a comprehensive foundation for probabilistic reasoning from weaker axioms, formalizing non-trivial results from theoretical computer science in $\APX_1$, and establishing a tailored witnessing theorem for its provably total $\TFNP$ problems. As a byproduct of our analysis of the minimal proof-theoretic strength required to formalize statements arising in theoretical computer science, we resolve an open problem regarding the provability of $\AC^0$ lower bounds in $\PV_1$, which was considered in earlier works by Razborov (1995), Krajíček (1995), and Müller and Pich (2020).
\end{abstract}

\newpage

\pagenumbering{roman} 

\tableofcontents

\newpage

\pagenumbering{arabic} 

\section{Introduction}

\subsection{Overview}

Bounded arithmetic extends traditional complexity theory by capturing not only the computational resources (e.g., running time or circuit size) required by algorithms, but also the complexity of proving their correctness. By integrating computational and proof complexity within a unified framework, it opens new angles on foundational questions in theoretical computer science. The area has a long history (see \citep{hajekpudlak-book, Krajicek-book, cook_nguyen_2010, krajicek_2019} and references therein) and has seen renewed momentum through new formalizations \citep{DBLP:journals/apal/BussKKK20, gaysin2024proofcomplexityuniversalalgebra, DBLP:conf/focs/Khaniki24, AARK25}; unprovability results \citep{DBLP:conf/stoc/PichS21, LO23, DBLP:conf/stoc/AtseriasBM23, DBLP:journals/lmcs/ChenLO25, DBLP:conf/stoc/0001RT25, DBLP:journals/corr/abs-2504-03320}; connections  to $\mathsf{TFNP}$ \citep{LLR24}, complex analysis \citep{jerabek:elementary-analytic-vtc0}, reverse mathematics \citep{CLO24, AT25}, complexity lower bounds \citep{DBLP:conf/stoc/GrosserC25, DBLP:conf/innovations/CarmosinoKKOT25}, and propositional proof complexity \citep{krajicek-25}; and applications in cryptography \citep{DBLP:conf/focs/0002J22, DBLP:conf/stoc/JinKLV24, JKLM25, JJMP25}, among other developments. We refer to \citep{buss-survey, Oliveira25} for background and for connections to algorithms and complexity theory.

Two central and extensively studied theories are Cook’s $\PV_1$ \citep{Coo75, KrajicekPT91} (see also \citep{krajicek_2019,DBLP:journals/eccc/Li25}) and Je\v r\'abek’s $\APC_1$ \citep{Jerabek04, Jerabek-phd, Jerabek07}. The theory $\PV_1$ formalizes polynomial-time reasoning and captures many classical results in algorithms and complexity. Since it is unclear whether randomized algorithms can, in general, be derandomized, $\PV_1$ is not well-suited for reasoning about probabilities or analyzing randomized algorithms. The theory $\APC_1$ extends $\PV_1$ by adding the dual weak pigeonhole principle $\dWPHP(\PV)$, yielding a convenient framework for reasoning about probabilities and randomized constructions. In particular, $\APC_1$ (and its mild extensions) is sufficient to formalize several nontrivial results, including the correctness of randomized algorithms for graph problems \citep{DBLP:journals/corr/abs-1103-5215}, polynomial identity testing \citep{AT25}, and circuit lower bounds \citep{DBLP:journals/apal/MullerP20}. However, the axioms of $\APC_1$ may be stronger than necessary: many results of interest could plausibly be provable in a weaker theory closer to $\PV_1$.

There are concrete reasons to expect $\APC_1$ to exceed the minimal strength required for probabilistic polynomial-time reasoning. On the one hand, $\APC_1$ is tied to $\dWPHP(\PV)$ and to the \emph{existence} of functions of exponential circuit complexity; from a computational perspective, the \emph{explicit construction} of such functions (i.e., circuit lower bounds) is a widely used derandomization assumption that may be stronger than the derandomization of $\pr\BPP$ (see, e.g., \citep{DBLP:conf/coco/Fortnow01, Goldreich11g, DBLP:journals/sigact/0001T23}). On the other hand, even if $\pr\BPP=\pr\P$ with a ``feasible'' proof, $\APC_1$ need not collapse to $\PV_1$; indeed, under plausible cryptographic assumptions, $\APC_1$ is strictly stronger than $\PV_1$ \cite{ILW23}. 

The search for a weaker theory that still supports the broad class of probabilistic arguments used across theoretical computer science is motivated by several considerations:

\begin{itemize}
\item \emph{Unprovability of complexity-theoretic conjectures.} A central objective in this area is to identify frameworks that both formalize existing tools in complexity theory and remain amenable to unprovability results. $\APC_1$ is likely strictly stronger than $\PV_1$ by \cite{ILW23}, and its witnessing functions cannot in general be made deterministic even if $\pr\BPP=\pr\P$, which complicate unprovability arguments and pose significant challenges (see, e.g., \cite{LO23, CKKO21}). In particular, the introduction of the strong principle $\dWPHP(\PV)$ is the main obstacle to extending unprovability of complexity lower bounds in $\PV_1$ \cite{DBLP:conf/stoc/PichS21} to $\APC_1$ \cite{LO23}. 
\item \emph{Bounded reverse mathematics with probabilistic reasoning.} Following recent developments such as \cite{CLO24, AT25} (see also \citep{cook_nguyen_2010} for related background), one can hope to pursue a systematic reverse mathematics of algorithms and complexity theory that classifies ``probabilistic proofs'' by the axioms they use. Similarly, it suggests the possibility of classifying randomized algorithms by the complexity of their \emph{correctness proofs}, supplementing the standard classification via space (see, e.g., \cite{Nisan92}) or circuit complexity (see, e.g., \cite{AW89}). This perspective is potentially insightful for derandomization, namely, derandomization based on the proof complexity of correctness proofs. Because $\dWPHP(\PV)$ is itself strong, $\APC_1$ is an overly powerful base theory for fine-grained correspondences weaker than $\dWPHP(\PV)$.
\item \emph{Correctness proofs in cryptography.}  Jain and Jin \citep{DBLP:conf/focs/0002J22} and subsequent papers \citep{DBLP:conf/stoc/JinKLV24, JKLM25, JJMP25,MDS25}  explore $\PV_1$ and its connection to propositional proofs to help construct $i\mathcal{O}$ and other cryptographic primitives, highlighting that the logical complexity of proving certain statements  can play an important role in  cryptographic systems and their efficiency. In particular, \cite{JJMP25,JKLM25} rely heavily on cryptographic primitives with $\PV_1$ proofs of \emph{correctness}. However, existing work typically considers ``perfect correctness'' because $\PV_1$ cannot natively talk about approximate counting and randomness, whereas $\APC_1$ seems both too strong and inconvenient for this purpose.
\item \emph{Feasible provability of probabilistic statements.} It is natural to formulate precise, feasible notions of the provability of $\pr\BPP=\pr\P$ and related questions. Yet even formulating $\pr\BPP=\pr\P$ feasibly is nontrivial, as it seems to require defining probabilistic computation within the theory in the first place. Given that the (un)provability of $\P=\NP$ in $\PV_1$ has long been central to bounded arithmetic (see, e.g., \citep{DBLP:journals/jsyml/CookK07, Oliveira25}), this direction holds significant potential for advancing the study of the interplay between randomized computations and mathematical proofs.
\end{itemize}

These considerations point to a common objective: designing a \emph{minimal theory} for reasoning about probabilities and randomized constructions in feasible mathematics. 

\paragraph{Summary of contributions.} We propose a theory corresponding to ``probabilistic polynomial-time reasoning'' in a strong sense.  Our main conceptual and technical contributions are:
\begin{enumerate}
\item \textbf{Theory $\APX_1$ and its relative strength.} We introduce $\APX_1$, establish its basic properties, and develop core probabilistic tools. The theory extends $\PV_1$ and is contained in $\APC_1$, in the sense that all of its consequences in the language of $\APC_1$ are also provable in $\APC_1$. Moreover, under plausible assumptions, $\APX_1$ is strictly weaker than $\APC_1$. 
\item \textbf{Advanced formalizations.} We formalize in $\APX_1$ several nontrivial results from algorithms and complexity, including the Blum-Luby-Rubinfeld linearity testing, Schwartz-Zippel lemma\footnote{It is worth noting that the standard proof of Schwartz-Zippel lemma (see, e.g., \cite[Lemma A.36]{AB09}) is not known to be formalizable even in $\APC_1$. In this work, we formalize an alternative proof due to Atserias and Tzameret \cite{AT25} in $\APX_1$.}, and an average-case $\AC^0$ lower bound for Parity. Additionally, as a byproduct of our refined analysis of $\AC^0$ circuits in bounded arithmetic, we describe a matching worst-case lower bound in $\PV_1$. The latter formalization addresses a problem considered by Razborov \citep{Razborov-switching-l}, Krajíček \citep[Section 15.2]{Krajicek-book}, and Müller-Pich \citep{DBLP:journals/apal/MullerP20}, which was only known for stronger theories.
\item \textbf{Tailored witnessing theorem.} We show that the provably total $\NP$ relations of $\APX_1$ deterministically reduce to a natural $\TFZPP$ problem\footnote{A $\TFNP$ problem $R(x,y)$ is said to be in $\TFZPP$ if, for every input $x$, at least an inverse-polynomial fraction of strings $y$ are valid solutions, i.e., $R(x,y)=1$. It is clear that $\TFZPP$ problems admit simple zero-error randomized algorithms running in polynomial time.} we introduce, $\mathsf{Refuter}(\mathsf{Yao})$. Moreover, if $\pr\BPP=\pr\P$, then $\APX_1$ admits deterministic polynomial-time witnessing.
\item[4.] \textbf{Feasible derandomization.} Using the new framework, we put forward a natural formalization of the fundamental question: Is $\pr\P=\pr\BPP$ feasibly provable? In other words, is there a deterministic feasible proof of general derandomization? 
\item[5.] \textbf{Reverse mathematics of randomness.} Finally, we show that $\APX_1$ serves as a suitable base theory for developing the reverse mathematics of average-case and randomized lower bounds, illustrated here through the study of randomized communication protocols and their communication complexity.
\end{enumerate}

Before presenting our results in more detail, we provide additional context and background.  

\paragraph{Dual use of $\dWPHP(\PV)$ in $\APC_1$.} Why does $\APC_1$, until now the weakest known theory capable of formalizing probabilities and randomized algorithms, appear stronger than necessary? By looking into the construction of $\APC_1$ \cite{Jerabek04,Jerabek-phd,Jerabek07}, we observe  two different reasons for introducing  $\dWPHP(\PV)$. 
\begin{itemize}
    \item First, it is used to implement approximate counting. Je\v r\'abek \cite{Jerabek04} shows that $\dWPHP(\PV)$ proves the existence of an exponentially hard Boolean function, and by formalizing a form of correctness of the Nisan-Wigderson PRG \cite{NW} in the theory $\PV_1$, we can approximately compute the acceptance probability of circuits by instantiating the PRG  with the hard Boolean function.
    \item Second, it also serves as a counting principle to derive tools in combinatorics and probability theory, including the inclusion-exclusion principle, union bound, and Chernoff bound \cite[Section 2]{Jerabek07}.
\end{itemize}
The first role appears essential, as approximate counting is the foundation for the formalization of probabilistic polynomial-time algorithms. However, $\dWPHP(\PV)$, as a counting principle, appears to be overly powerful and not necessary for many applications. 

\medskip

\begin{remarkn}{Computational Aspects of $\dWPHP$}
To add more context, $\dWPHP(\PV)$ asserts that for any function $f$ implemented by circuits whose co-domain is much larger than its domain, say $f:\zo^n\to\zo^{n+1}$, there exists a string in the co-domain that does not have a pre-image. The computational aspect of the principle, namely the search problem of finding such a string given a function $f$, has recently drawn attention in computational complexity (see \cite{Korten-survey} for a survey). This problem, which is now called the \emph{Range Avoidance Problem} \cite{DBLP:conf/innovations/KleinbergKMP21,DBLP:conf/focs/Korten21,DBLP:journals/eccc/RenSW22}, is known to be hard even for nondeterministic search algorithms under plausible assumptions \cite{ILW23,CL24}. 
\end{remarkn}

\paragraph{Axiomatizing approximate counting in $\APX_1$.} Since $\dWPHP(\PV)$ fulfills two essential functions in $\APC_1$, devising a weaker theory is nontrivial — one must find a way to relax the counting principle without sacrificing the capacity to formalize approximate counting. 

Our approach, which in hindsight appears quite natural, is to put \emph{approximate counting} at the foundation, elevating it to a central primitive rather than deriving it from stronger principles such as $\dWPHP(\PV)$ \cite{Jerabek04,Jerabek-phd,Jerabek07}. Starting from $\PV_1$ as the base theory, we directly introduce an oracle that is intended to perform approximate counting, and govern it with appropriate axioms. Through this approach, we decouple \emph{the concept of approximate counting} from \emph{counting principles}. 

The main technical challenge is to select an appropriate set of axioms. These axioms should be sufficiently strong to carry out our advanced formalizations, reverse mathematics results, and potentially more results in theoretical computer science. At the same time, the set of axioms should be minimal. The contradictory objectives make it hard to select appropriate axioms; indeed, it is not even a priori clear whether a suitable \emph{finite} set of axioms exists without resorting to variants of $\dWPHP(\PV)$. 

Perhaps surprisingly, we distill four simple and intuitive axioms that suffice to implement all our results, among which the only nontrivial axiom captures the ``local'' behavior of the approximate counting oracles. Arguably, this makes $\APX_1$ a plausible candidate for the minimal theory of probabilistic polynomial-time reasoning.

\medskip

\begin{remarkn}{Minimal Assumption for Derandomization} The conjectured inclusion $\pr \BPP \subseteq \pr \P$ is a central question in derandomization. 
The celebrated results of Impagliazzo, Nisan, and Wigderson \cite{NW,IW97} give a positive answer under $\E\nsubseteq\text{i.o.-}\SIZE[2^{\Omega(n)}]$, a plausible worst-case circuit lower bound. Conversely, derandomization results are also known to imply weaker circuit lower bounds such as $\NTIME(n^{\omega(1)})\nsubseteq\P_{/\poly}$ (see, e.g., \cite{IKW02,Williams14,Tell19}). Yet it has been a longstanding open problem  whether the strong circuit lower bounds used in \cite{NW,IW97} are \emph{necessary} for proving $\pr\BPP=\pr\P$. Indeed, there has been significant progress indicating that derandomization may \emph{not} require strong circuit lower bounds, see, e.g., \cite{DBLP:conf/coco/Fortnow01,Goldreich11g,CT21}. Moreover, several \emph{characterizations} of $\pr\BPP=\pr\P$ have been recently discovered \cite{LP22,Korten22,DBLP:journals/sigact/0001T23,CTW23,LPT24}, motivated by the question of understanding the \emph{minimal} assumption required for derandomization.

In a sense, our results attempt to address a similar question in the context of proof complexity. We aim to propose a \emph{minimal theory} that is strong enough to carry out meaningful feasible proofs on probabilistic polynomial-time algorithms. In particular, we provide evidence that $\dWPHP(\PV)$ and the existence of hard Boolean functions, which are at the foundation of Je\v r\'abek's theory $\APC_1$ \citep{Jerabek04, Jerabek-phd, Jerabek07}, might not be necessary in a minimal theory for probabilistic polynomial-time reasoning.  
\end{remarkn}

\subsection{Main Contributions}

We now describe our contributions and their implications in detail.

\subsubsection{Theory \texorpdfstring{$\APX_1$}{APX1}}\label{sec:intro_theory}

As alluded to above, rather than deriving probabilities from stronger combinatorial principles (as in $\APC_1$ via $\dWPHP(\PV)$), we axiomatize approximate counting directly. Our aim is a weaker theory in which the probability of any \emph{feasibly definable event}\footnote{In other words, an event $E \subseteq \{0,1\}^{m}$ for which there is a polynomial-size Boolean circuit $C$ such that $C(x) = 1$ if and only if $x \in E$.} can be named and  reasoned about with additive slack, while keeping proof-theoretic strength low. 

To achieve this, we introduce a first-order bounded arithmetic theory, $\APX_1$, whose central primitive is an \emph{approximate counting} function $\P$. Intuitively, given a Boolean circuit $C$ on $n$ input bits and a \emph{precision parameter} $\Delta$, the term $\P(C,\Delta)$ returns a rational number in $[0,1]$ that approximates the acceptance probability of $C$ within additive error $1/|\Delta|$, where $|\Delta|$ denotes the bitlength of the input parameter $\Delta$. For convenience, we often write $\P_\delta(C)$ instead of $\P(C,\Delta)$, where $\delta=1/|\Delta|$.

The equational core of $\APX_1$, called $\APX$, is obtained by extending Cook’s equational theory $\PV$ with the oracle symbol $\P$ (the new language is denoted $\PV(\P)$) and its governing axioms. $\APX_1$ is then the usual first–order closure of $\APX$, i.e., universal closures of $\APX$-equations together with the standard $\PV$-style induction on notation. 

\medskip

\begin{remarkn}{$\PV$ and $\PV_1$}
$\PV$ \citep{Coo75} is an equational theory whose intended model is $\mathbb{N}$ with the usual interpretation of basic symbols such as $0$, $+$, and $\times$. Its language contains a function symbol for every polynomial-time algorithm $f\colon \mathbb{N}^k \to \mathbb{N}$ (for any fixed $k$); these symbols and their defining axioms are given via Cobham’s characterization of the polynomial-time functions. The theory includes an induction scheme formalizing binary search and, in particular, proves induction for quantifier-free formulas (i.e., polynomial-time predicates). A standard first-order strengthening is $\PV_1$ \citep{KrajicekPT91}. While the formal definition of $\PV_1$ is fairly technical, the theory is robust: distinct presentations yield the same theorems. For example, $\PV_1$ has an equivalent axiomatization that avoids Cobham’s theorem \citep{jerabek:sharply-bounded}; alternatively, it can be presented as the set of all $\forall\Sigma^b_1$-sentences provable in Buss’s theory $\S^1_2$ \citep{Buss}. We refer to \citep{Oliveira25} for a brief overview and to \citep{DBLP:journals/eccc/Li25} for a detailed introduction.
\end{remarkn}

\vspace{0.2cm}

A key aspect of the definition of $\APX_1$ is to employ ``local'' constraints governing the behavior of $\P_\delta$, which together enforce the ``global'' desired behavior, i.e., that $\P_\delta$ approximates the acceptance probability of any input circuit up to an additive error term $\delta$. The entire probabilistic machinery of $\APX_1$ (random variables, expectation, tail bounds, etc.) is built on top of the axioms below. 

\paragraph{$\APX_1$ axioms governing $\P$.}
All axioms are universal $\PV(\P)$-equations; below $\beta^{-1}\in\Log$ is a freely available ``slack'' parameter used to absorb routine finite-precision effects.\footnote{The expression $\beta^{-1}\in\Log$ is standard notation in bounded arithmetic used to denote that $\beta = 1/|y|$ for some variable $y$, where $|y|$ is the bitlength of $y$.} We sketch the statements at an informal level; the formal version appears in \Cref{sec:sec_theory}.

\begin{itemize}
  \item \textbf{Basic Axiom.} For every Boolean circuit $C$ and $\Delta$, the value $\P(C,\Delta)$ is a rational in $[0,1]$ (encoded in $\PV$) and all $\PV(\P)$-provable equations hold. Together with an output-length bound for $\P(C,\Delta)$, this forces feasibility of approximate counting at any requested precision.
  \item \textbf{Boundary Axiom.} If $C$ is syntactically constant (reads no inputs), then $\P_\delta(C) \in\{0,1\}$ agrees with the output bit of $C$. Thus $\P$ is \emph{exact} on trivial cases.

  \item \textbf{Precision Consistency.} For any two precisions $\delta_1,\delta_2$ and circuit $C$,
  \[
  \big|\P_{\delta_1}(C)-\P_{\delta_2}(C)\big|\ \le\ \delta_1+\delta_2+\beta.
  \]
  Hence asking for finer precision can only move the reported probability by the sum of the specified error parameters (up to $\beta$).

  \item \textbf{Local Consistency.} If $C$ has at least one input bit, and $\Fix_b(C)$ denotes the circuit obtained by fixing the rightmost input bit to $b\in\{0,1\}$, then
  \[
  \Big|\,\P_\delta(C)\ -\ \tfrac12\big(\P_\delta(\Fix_0(C))+\P_\delta(\Fix_1(C))\big)\,\Big|\ \le\ 2\delta+\beta.
  \]
  Thus the reported acceptance probability of $C$ is (up to additive slack of $\beta$) the average of the reported probabilities after fixing a fresh random bit. This aims to capture the intended semantics of counting over the uniform hypercube.\footnote{In other words, for every (feasibly definable) set $X \subseteq \{0,1\}^n$, as $X$ is the disjoint union of $X_0$ and $X_1$, where $X_b = \{ x \in X \mid x_n = b\}$, we expect $|X|\approx |X_0|+|X_1|$.}
\end{itemize}

In practice, one can think of $\APX_1$ as the theory extending $\PV_1$ with the symbol $\P$ and its governing axioms, together with induction over quantifier-free formulas in the language $\PV(\P)$. Everything else -- such as random variables, expectation, union bound, etc. -- will be introduced and derived from the language and axioms inside $\APX_1$.

\paragraph{Soundness of approximate counting in $\APX_1$.}
We say that a $\PV$-standard model (i.e., $\mathbb{N}$) where the function symbol $\P$ is interpreted by any correct approximate counting function (returning a rational within $\pm 1/|\Delta|$ and exact on syntactically constant circuits) is a \emph{standard model} of $\APX_1$. A simple but central  result shows that these are \emph{exactly} the models satisfying the axioms (``admissible models''), yielding semantic soundness for the intended interpretation and a correct axiomatization of approximate counting when the underlying model is $\mathbb{N}$ (see \Cref{sec:models}).

\paragraph{Minimality of $\APX_1$.} We believe that $\APX_1$ is a good candidate of the \emph{minimal theory} for probabilistic polynomial-time reasoning. This is not a formal assertion from a mathematical perspective. However, the axioms above appear close to the weakest workable base theory that can consistently \emph{define} and \emph{operate on} approximate probabilities of feasibly described events. Specifically: 
\begin{itemize}
\item Any theory that reasons about probabilistic polynomial time algorithms should be able to \emph{define} the acceptance probability of the algorithms. This requires the capability of approximate counting with an additive error, i.e., the symbol $\P$. 
\item Because the \emph{Basic Axiom} and the \emph{Boundary Axiom} are rather syntactic promises of the oracle $\P$, we expect them to be available. Arguably, the \emph{Precision Consistency Axiom}, which asserts the consistency of $\P$ on different precision parameters, should also be available. Note that these three axioms are not sufficient, as one can easily specify a trivial and incorrect polynomial-time function such that these axioms are provable in $\PV_1$.
\item Therefore, we use the \emph{Local Consistency Axiom} to capture the correctness of $\P$ --- it shows that the approximate counting oracle withstands a simple statistical test with three queries made throughout the proof. It seems unlikely that one can make nontrivial use of an oracle for the purpose of approximate counting that may fail this test; subsequently, the axiom also seems necessary. 
\end{itemize}

Another evidence of the minimality of $\APX_1$ is that, computationally, the function symbol $\P$ aligns with the \emph{Circuit Acceptance Probability Problem} ($\mathsf{CAPP}$), which is complete for $\pr \BPP$ (see, e.g., \citep{DBLP:journals/fttcs/Vadhan12}). In contrast, the \emph{Range Avoidance Problem}, which corresponds to $\dWPHP(\PV)$ and is relevant for $\APC_1$, is likely hard even against nondeterministic algorithms \cite{ILW23,CL24}.  

\subsubsection{Probabilistic Reasoning in \texorpdfstring{$\APX_1$}{APX1}}\label{sec:intro_prob_reasoning}

We develop a self-contained ``probabilistic calculus'' inside $\APX_1$ using the approximate counting function $\P_\delta$. As a preliminary step, we show that $\P_\delta$ behaves in the expected way on feasibly described events. Concretely, $\APX_1$ establishes the following properties (each up to an arbitrarily small additive slack $\beta^{-1} \in \Log$):

\begin{itemize}
  \item \textbf{Semantic invariance.} $\P_\delta$ respects semantic equivalence, i.e., if $\APX_1$ proves that circuits $C$ and $D$ compute the same function, then $\bigl|\P_\delta(C)-\P_\delta(D)\bigr| \le 2 \cdot \delta + \beta$ (\Cref{lmm: global consistency}).

  \item \textbf{Permutation invariance.}  Permuting input bits does not noticeably change the value of $\P_\delta$. In other words, for any circuit $C$ and permutation $\pi$ of input bits, 
  $
  \bigl|\P_\delta(C\circ \pi)-\P_\delta(C)\bigr| \le 2 \cdot \delta+\beta
  $
  (\Cref{lmm: permutation approx}).

  \item \textbf{Existence via the probabilistic method.} Suppose that strings accepted by a circuit $C$ are considered \emph{good}. Then if good strings are abundant, i.e., more than $(\delta+\beta)$-fraction with respect to precision parameter $\delta$, there must exist a good string (\Cref{lem:prob_method}). A bit more formally, 
  $$\Apx_1 \vdash \forall n,\delta^{-1},\beta^{-1}\in\Log~\forall C ~(\beta > 0\land \P_\delta(C) > \delta + \beta \to \exists x\in\zo^n~C(x)=1).$$
  
  \item \textbf{Consistency on concrete circuits.} $\P_\delta$ agrees with simple tests. For instance, for a naturally defined threshold circuit $C_{<t}(x)$ on $n$-bit inputs that accepts if and only if $x$ (viewed as an integer) is less than $t$, $\P_\delta(C_{<t}) \approx t/2^n$ (see \Cref{sec:concrete_circuits}).
\end{itemize}

These meta-properties ensure that the definitions and inequalities developed in $\APX_1$ inherit the intended probabilistic behavior with only small, explicitly controlled additive losses.

\medskip
With these guarantees in place, we now introduce  \emph{feasible random variables}. A random variable $X$ is specified by an explicit support \(V\subseteq\mathbb{Q}\), a seed length \(n\), and a multi-output sampler circuit \(C:\{0,1\}^n\to V\). Its \emph{approximate expectation} is defined by querying $\P_\delta$ on the indicator Boolean circuits \(\{C_v\}_{v\in V}\), where $C_v(z)$ accepts $z$ if and only if $C(z) = v$. In other words:
$$
\bbE_{\delta} [X]\;\eqdef \; \sum_{v\in V} v\cdot \P_\delta(C_v).
$$
We observe that there exists a $\PV(\P)$ function $\E(V,n,C,\Delta)$ that computes $\Ex_{|\Delta|^{-1}}[X]$ for the random variable $X$ defined by $(V,n,C)$. Specifically, $\E$ enumerates all $v \in V$, constructs the corresponding circuit $C_v$, queries the oracle to obtain $p_v \gets \P(C_v,\Delta)$, and outputs the sum $\sum_{v \in V} v \cdot p_v$.

\medskip

We introduce a central technical tool that provides a general version of the \emph{averaging argument for expectation} (\Cref{sec: avg for expectation}): Given random variables $X_1,\ldots,X_m$ on the same seed and coefficients $\lambda_1,\ldots,\lambda_m$, $\APX_1$ can \emph{search} for a suffix $z$ of the seed such that a lower bound on the value  $\sum_i \lambda_i \cdot \bbE_\delta[X_i]$ is approximately preserved after fixing that suffix. This is used repeatedly to move between \emph{global} and \emph{pointwise} statements and underlies the proof of several results. We explain this technique in more detail in \Cref{sec:intro_techniques}.

Using the tool described above, together with some additional ideas, $\APX_1$ derives approximate formulations of several standard probability inequalities. In particular, it establishes the \emph{linearity of expectation} for linear combinations of feasible random variables, the \emph{union bound} for polynomially many events, and \emph{Markov's inequality} for non-negative variables with the usual $1/k$ decay.\footnote{We note that some results suffer an approximation loss that depends on $\|V\| \eqdef \sum_{v \in V} |v|$ and on the magnitude of the involved coefficients, where here $|\cdot|$ denotes absolute value. In applications where these quantities are polynomially bounded, this can be mitigated  by taking sufficiently small parameters $\delta$ and $\beta$ in the application of $\bbE_{\delta} [X]$.}

\medskip

\begin{remarkn}{Example: Union Bound in $\APX_1$; see \Cref{thm: union bound}} $\Apx_1$ proves the following statement. Let $n,m,\delta^{-1},\beta^{-1}\in\Log$, $C_1,\dots,C_m$ be single-output circuits, and $V=\{0,1\}$. Suppose that $\forall x\in\zo^n$ and $i\in[m]$, $C_i(x)\in V$, and let $Y,X_1,\dots, X_m$ be  random variables defined as follows.
\begin{compactitem}
\item For each $i\in[m]$, $X_i$ is defined by $(V,n,C_i)$. 
\item $Y$ is defined by $(V,n,S)$, where $S(x)\in\zo$ is a circuit such that $S(x) \le C_1(x) \lor \dots \lor C_m(x)$. 
\end{compactitem}
Then we have $\bbE_\delta[Y] \le \bbE_\delta[X_1] + \dots + \bbE_\delta[X_m] + (2\delta+\beta)\cdot m$.
\end{remarkn}

 Additionally, $\APX_1$ defines an \emph{approximate variance}
$
\mathsf{Var}_\delta[X]\eqdef \bbE_\delta\!\left[(X-\mu)^2\right],
$
with $\mu\eqdef \bbE_\delta[X]$,
and shows an identity of the form  \(\mathsf{Var}_\delta[X]\approx \bbE_\delta[X^2]-\mu^2\), which leads to a natural formulation of \emph{Chebyshev's inequality}. $\APX_1$ also formalizes \emph{(almost) pairwise independence} via approximate covariance, and proves that the variance of a sum is (approximately) the sum of variances for (almost) pairwise independent variables.  

\medskip
Finally, we address \emph{independence and concentration}. We work in $\APX_1$ with \emph{explicit} independence: variables are sampled by disjoint parts of the seed. Under this notion, the theory proves a \emph{multiplication principle}
\[
\bbE_\delta[XY]\;\approx\;\bbE_\delta[X]\cdot \bbE_\delta[Y],
\]
and, for Bernoulli variables, a convenient product bound
\[
\Bigl|\bbE_\delta\!\Bigl[\textstyle\prod_{i=1}^m X_i\Bigr]-\prod_{i=1}^m \bbE_\delta[X_i]\Bigr|\;\le\; 8\delta \cdot m.
\]
These yield \emph{one-sided error reduction}. Moreover, $\APX_1$ proves a \emph{Chernoff bound} for sums of \(m=O(\log n)\) i.i.d. Bernoulli random variables; the bound has the standard exponential tail with controlled additive slack.

\medskip

\begin{remarkn}{Example: One-Sided Error Reduction in $\APX_1$; see \Cref{thm:one-sided}} 
    For a Boolean circuit $C:\zo^n\to\zo$, let $C^{\lor k}:\zo^{nk}\to\zo$ be the circuit defined as $C^{\lor k}(x_1,\dots,x_k)\eqdef \bigvee_{i\in[k]} C(x_i)$. The following statement is provable in $\Apx_1$. For any $n,k,\delta^{-1},\beta^{-1}\in\Log$ and $C:\zo^n\to\zo$, if $\P_\delta(\lnot C) \le \eps$  then $\P_{\delta}(\lnot C^{\lor k})\le (\delta+\beta+\eps)^k+\delta + \beta$.
\end{remarkn}

We refer to \Cref{sec: basic properties} for a detailed description of how these different notions and results are implemented in $\APX_1$.

\subsubsection{Theoretical Computer Science in \texorpdfstring{$\APX_1$}{APX1}}

As explained above,  approximate counting -- as axiomatized in $\APX_1$ -- suffices to build the typical probabilistic toolkit (such as existence arguments, linearity of expectation, averaging argument, union bound, Markov, Chebyshev, limited independence, error reduction, and a version of Chernoff for logarithmically many samples). This lightweight yet robust framework can be exploited to formalize several nontrivial results. We illustrate this point through a set of detailed formalizations of influential results from different areas of theoretical computer science:

\begin{itemize}
\item Yao’s distinguisher-to-predictor transformation via the hybrid argument, a central tool in computational pseudorandomness (see \Cref{thm: Yao D2P});
\item the Schwartz-Zippel Lemma (as stated in \citep{AT25}), an algebraic result for polynomial identity testing with broad applications in randomness and complexity (see \Cref{thm: SW});
\item the classical lower bound for the parity function against bounded-depth polynomial-size circuits in circuit complexity (see \Cref{sec:AC0_LB});
\item the correctness of the Blum-Luby-Rubinfeld linearity test from sublinear time algorithms and property testing (see \Cref{sec: BLR}).
\end{itemize}

For concreteness and in order to contrast our results with previous work, we focus here on the formalization of circuit lower bounds for the $n$-bit parity function, denoted $\oplus_n$. In fact, we show that a stronger \emph{average-case} lower bound against depth-$d$ Boolean circuits ($\AC^0_d$) can be proved in $\APX_1$.

\dummylabel{avg parity}{Average-Case $\AC^0$ Lower Bound for $\oplus_n$}
\begin{restatable}[Average-Case $\AC^0$ Lower Bound for $\oplus_n$ in $\APX_1$]{theorem}{AvgParity}\label{thm:intro_ave_parity}
    For all constants $k,d\ge 1$, there exists a constant $n_0\ge 1$ such that $\APX_1$ proves the following statement. Let $n,\delta^{-1},\beta^{-1}\in\Log$, $n>n_0$, and $C:\zo^n\to\zo$ be an $\AC^0_d$ circuit of size at most $n^k$. Let $T_C:\zo^n\to\zo$ be the circuit that, given $x\in\zo^n$, outputs $1$ if and only if $C(x)=\oplus_n(x)$. Then 
    \begin{equation}
    \P_\delta(T_C)\le \frac{1}{2} + \frac{1}{n^{k}} + \delta + \beta.\label{equ: prob T small AC0}  
    \end{equation}
\end{restatable}

The main technical challenge is to avoid ``encoding-based counting arguments'' (pigeonhole-principle variants) unavailable in $\APX_1$, such as those used in Razborov’s proof of the switching lemma \cite{Razborov-switching-l}. Instead, our proof builds on a technique of Furst, Saxe, and Sipser \cite{DBLP:journals/mst/FurstSS84}. The approach was refined by \cite{AAIPR01} (see also \cite{DBLP:conf/coco/Agrawal01}), who gave a deterministic polynomial-time algorithm that outputs an appropriate restriction supplied by the switching lemma. One of our contributions is to show that the correctness of the algorithm in \cite{AAIPR01} can be established within $\APX_1$. Combined with the probabilistic tools above and other ideas, this yields the \emph{average-case} lower bound in $\APX_1$.

As a consequence of our refined proof-theoretic framework, and with some additional effort, we can extract from the above formalization a \emph{worst-case} lower bound within the weaker theory $\PV_1$.

\dummylabel{worst parity}{$\oplus_n\notin\AC^0$} 
\begin{restatable}[Worst-Case $\AC^0$ Lower Bound for $\oplus_n$ in $\PV_1$]{theorem}{WorstParity}\label{thm:intro_worst_parity}
For all constants $k,d\ge 1$, there exists a constant $n_0\ge 1$ such that $\PV_1$ proves the following statement. For every $n\in\Log$, $n>n_0$, and $\AC^0_d$ circuit $C:\zo^n\to\zo$ of size at most $n^k$, there exists a string $x\in\zo^n$ such that $C(x)\ne \oplus_n(x)$. 
\end{restatable}

Earlier formalizations of the worst-case parity lower bound for bounded-depth circuits required stronger theories. In particular, \citep{DBLP:journals/apal/MullerP20} and \citep{Krajicek-book} formalize different proofs in $\APC_1=\PV_1+\dWPHP(\PV)$, while \citep{Razborov-switching-l} works in $\PV_1$ but in the $\Log\Log$ regime -- i.e., with $n$ of doubly logarithmic order -- so the proof can manipulate exponentially large objects (see \citep{DBLP:journals/apal/MullerP20} for details).

These formalizations reinforce the intuition that a substantial portion of results in algorithms and complexity theory are already captured within $\PV_1$ or its mild extensions, and that establishing unprovability results would therefore be of considerable significance (see \citep{Oliveira25} for related discussions).

In \Cref{sec:intro_techniques} below, we elaborate on the proofs of \Cref{thm:intro_ave_parity} and \Cref{thm:intro_worst_parity}. For further details about these and other formalizations, see \Cref{sec:formalizations}.

\subsubsection{Witnessing, Relative Strength of  \texorpdfstring{$\APX_1$}{APX1}, and Provability of \texorpdfstring{$\pr \BPP = \pr \P$}{prBPP = prP}}

We now discuss relations between theories $\PV_1$, $\APX_1$, and $\APC_1$, and connections to the $\pr\BPP$ versus $\pr\P$ problem. We also introduce a new computational problem called $\RefYao$, and provide a tailored witnessing theorem for the $\forall\Sigma_1^b(\PV)$-consequences of $\APX_1$ (i.e.~provably total $\TFNP$ problems in $\APX_1$). 

\vspace{-0.2cm}

\paragraph{$\APX_1$ versus $\APC_1$.} By construction, every sentence provable in $\PV_1$ is also a theorem of $\APX_1$. It is also possible to show that if $\varphi$ is a sentence in the language of $\PV_1$ (i.e., without the approximate counting symbol $\P$) provable in $\APX_1$, then it is provable in $\APC_1$ (see \Cref{cor: APX subset APC1}). This means that, modulo the difference in languages (i.e.~$\APC_1$ does not have the symbol $\P$), $\APX_1$ is a sub-theory of $\APC_1$.\footnote{Indeed, there is a conservative extension of $\APC_1$ known as $\mathsf{HARD}^\mathsf{A}$ \cite{Jerabek07} that contains $\APX_1$ in a stronger sense --- the symbol $\P$ can be simulated by a term in $\mathsf{HARD}^{\sf A}$ such that all axioms governing $\P$ are provable (see \Cref{thm: ub APC}).}

On the other hand, under plausible computational assumptions, there are sentences provable in $\APC_1$ that are not provable in $\APX_1$ (see \Cref{cor: APC strict extension}).\footnote{More formally, there is a  $\forall \Sigma^b_2(\PV)$-sentence provable in $\APC_1$ that is not provable in $\APX_1$, under the existence of indistinguishability obfuscation and $\coNP$ not contained infinitely often in $\NP_{/\poly}$ (see \Cref{cor: APC strict extension}).} This is obtained by adapting a technique from \cite{ILW23}. In other words, $\PV_1$ is contained in $\APX_1$, while $\APX_1$ is likely strictly weaker than $\APC_1$. 

Subsequently, a fundamental research direction is to determine whether $\APX_1$ is stronger than $\PV_1$, a question closely connected to the $\pr\BPP$ versus $\pr\P$ problem and to understanding the role of randomness in feasible proofs.

\paragraph{$\PV_1$ versus $\APX_1$ and feasible derandomization.} From a meta-mathematical standpoint, it is natural to ask whether $\pr\BPP=\pr\P$ is \emph{\emph{(}un\emph{)}provable} in a weak arithmetic theory such as $\PV_1$. A key obstacle is formalization: the language of $\PV_1$ is tailored to \emph{deterministic} polynomial-time functions, whereas the statement $\pr\BPP = \pr\P$ quantifies over acceptance probabilities of circuits on an exponentially large space. We  propose the following  question.

\begin{open}\label{intro:open BPP = P}
Is there a $\PV$ function symbol $\widetilde{\P}(C,\Delta)$ for which the \emph{basic}, \emph{boundary}, \emph{precision consistency}, and \emph{local consistency} axioms (\Cref{sec:intro_theory}) are provable in $\PV_1$?
\end{open}

An unconditional positive answer seems out of reach at present, as it would immediately imply $\pr\BPP=\pr\P$ by the soundness of the approximate counting axioms and the polynomial running time of $\widetilde{\P}$ (see \Cref{thm: standard equiv admissible}). Intuitively, this would amount to a \emph{deterministic polynomial-time proof} of the collapse. At the moment, it is unclear whether a positive or a negative answer is more plausible.

A weaker possibility is that, even if no such $\PV$ function symbol $\widetilde{\P}(C,\Delta)$ exists with the axioms provable in $\PV_1$, adding the approximate counting oracle $\P$ might nonetheless be conservative for \emph{deterministic} statements in the base language. Formally:

\begin{open}\label{open:intro conservative PV}
Is $\APX_1$ conservative over $\PV_1$? Equivalently, does every first-order sentence in the language of $\PV$ that is provable in $\APX_1$ already have a proof in $\PV_1$?
\end{open}

A positive answer to \Cref{intro:open BPP = P} would imply a positive answer here.  The relationship between \Cref{open:intro conservative PV} and $\pr\BPP=\pr\P$ appears incomparable. If $\pr\BPP=\pr\P$ holds but only via a non-feasible proof, $\APX_1$ need not be conservative over $\PV_1$. Conversely, even if $\APX_1$ is conservative over $\PV_1$, it is not clear to us whether $\pr \BPP = \pr \P$ follows. At a high level, we are interested in the relationship between \emph{derandomization of computations} and \emph{derandomization of proofs}. While we are currently unable to provide definite answers, we believe these questions are fundamental and merit further study. We refer to \citep{krajicek-25} and references therein for related questions in the context of $\APC_1$ versus $\PV_1$.

\paragraph{A Witnessing Theorem for $\APX_1$: Reductions to $\RefYao$.} A key characteristic of bounded theories is to have a suitable \emph{witnessing theorem} corresponding to certain computational problems (see, e.g., \cite{Buss,KrajicekPT91}). We isolate a  
certain (total) search problem as a key computational task for producing witnesses for the  $\forall \Sigma^b_1(\PV)$-consequences of $\APX_1$.

\medskip

\noindent\emph{$\RefYao$}:
An instance fixes the following parameters: input length \(n\), multiset size \(m\), predictor circuit description size \(s\), and advantage \(\delta>0\). The input is a \emph{predictor generator}, i.e., a circuit
\[
G \colon \{0,1\}^{nm}\to [n]\times\{0,1\}^s
\]
which, on a flat distribution \(\mathcal{D}\in(\{0,1\}^n)^m\) (i.e., an \(m\)-tuple of $n$-bit strings), returns an index \(i\in[n]\) and the description of a predictor circuit \(P\colon \{0,1\}^{i-1}\to\{0,1\}\) of size $s$.
A \emph{solution} is any flat distribution \(\mathcal{D}\) such that, writing \((i,P)=G(\mathcal{D})\),
\[
\Pr_{x\leftarrow \mathcal{D}}\!\big[P(x_{<i})=x_i\big] < \tfrac{1}{2}+\delta.
\]

\medskip

Thus a solution $\mathcal{D}$ \emph{refutes} that $G$ can produce predictors for any distribution with advantage $\delta$. When parameters satisfy $(\delta^2/10)\cdot m \ge s + \lceil \log n\rceil +1$, $\RefYao$
lies in $\TFZPP$; in other words, uniformly random distribution is likely a solution. The problem is called $\RefYao$, as $G$ is intended to output a predictor like the standard ``distinguisher\(\to\)predictor'' transformation of Yao \cite{Yao82}.

\medskip

\begin{remarkn}{$\RefYao$ and  Derandomization}
    Note that $\RefYao$ requires generating a distribution $\mathcal{D}$ that is unpredictable with respect to a \emph{given} predictor generator $G$ -- a deterministic procedure that attempts to produce a predictor $P$ for $\mathcal{D}$. The distribution $\mathcal{D}$ need not be pseudorandom (or equivalently, unpredictable) against \emph{all small circuits}; it only needs to fool the specific generator $G$. This can be viewed as a special case of constructing \emph{targeted pseudorandom generators}, a task known to be $\pr\BPP$-complete (see  \cite{Goldreich11g, CT21, LPT24}).
\end{remarkn}

\medskip

We establish the following result for the provably total $\TFNP$ problems of $\APX_1$.

\begin{restatable}[Witnessing for $\APX_1$]{theorem}{WitnessingAPX}\label{thm: witness TFNP}
Let $\varphi(x,y)$ be a quantifier-free formula in the language of $\PV_1$. If $\APX_1\vdash \forall x~\exists y~\varphi(x,y)$, there exists a deterministic polynomial-time Turing reduction from the search problem defined by $\varphi$ to $\RefYao$ with parameters satisfying $(\delta^2/10)\cdot m \ge s + \lceil \log n\rceil +1$. 
\end{restatable}

In \Cref{sec:intro_techniques} below, we provide an overview of the proof of \Cref{thm: witness TFNP}.

\paragraph{Relation to $\mathsf{LossyCode}$ and $\APC_1$.}
Recall the definition of the search problem $\mathsf{LossyCode}$ \cite{Korten22}: given a compressor circuit \(C\colon\{0,1\}^n\to\{0,1\}^{n-1}\) and a decompressor circuit \(D\colon\{0,1\}^{n-1}\to\{0,1\}^n\), output \(x\) with \(D(C(x))\neq x\). Similarly to $\RefYao$, this problem is total and in $\TFZPP$. 

We observe the existence of a deterministic polynomial-time mapping reduction from $\RefYao$ to $\mathsf{LossyCode}$ whenever the input instances of $\RefYao$ satisfy
\[
(\delta^2/10)\cdot m \;\ge\; s+\lceil\log n\rceil+1.
\]
Therefore, in the stated regime, derandomizing $\mathsf{LossyCode}$ subsumes derandomizing $\RefYao$. Since every $\APX_1$-provably total $\TFNP$ problem reduces to $\RefYao$,  under the above parameter condition it further reduces to $\mathsf{LossyCode}$.

Recall that Wilkie (unpublished) and Thapen \citep{thapen2002weak} (see \citep[Proposition 1.14]{Jerabek04} and \citep[Theorem D.1]{LPT24}) proved that $\mathsf{LossyCode}$ captures the $\forall\Sigma_1^b$-fragment of  $\APC_1$. Consequently, these results organize the $\TFNP$ landscapes of the two theories: $\mathsf{LossyCode}$ witnesses the  $\forall \Sigma^b_1$-consequences of $\APC_1$, while $\RefYao$ witnesses those of $\APX_1$.  

\medskip

 We return to these topics in \Cref{sec: witnessing}, providing detailed proofs of all results mentioned above and further discussions.

\subsubsection{Reverse Mathematics of Randomized and Average-Case Lower Bounds}

The \emph{retraction weak pigeonhole principle} ($\rWPHP(\PV)$) \cite{Jerabek07-dWPHP,LLR24,CLO24} is one of the most important combinatorial principles known to be provable in 
$\APC_1$, but whose provability in $\APX_1$
 remains unclear. Recall that $\rWPHP(\PV)$ asserts that for every $n,m \in \Log$ with $m < n$ and for all (deterministic) circuits $C \colon \{0,1\}^n \to \{0,1\}^{m}$ (``compressor'') and $D \colon \{0,1\}^{m} \to \{0,1\}^n$ (``decompressor''), there is $x \in \{0,1\}^n$ such that $D(C(x)) \neq x$. 
 
 In other words, $\rWPHP(\PV)$ captures the combinatorial principle underlying the total search problem $\mathsf{LossyCode}$ discussed above. Its provability in $\APX_1$ would mean that $\APX_1$ and $\APC_1$ prove the same $\forall \Sigma^b_1(\PV)$ sentences, and by the witnessing theorem (see \Cref{thm: witness TFNP}), this would further imply that $\LossyCode$ and $\RefYao$ are equivalent with respect to deterministic polynomial-time Turing reductions.
 
 \medskip
 
 We study \emph{counting variants} of the retraction weak pigeonhole principle and characterize their \emph{equivalence class} with respect to provability in $\APX_1$. We show that this class encompasses certain \emph{communication complexity lower bounds against randomized protocols}, establishing that these results are all equivalent (over the base theory $\APX_1$) to suitable variants of the retraction pigeonhole principle.

\paragraph{Counting Variants of $\rWPHP(\PV)$.} We consider the following statements:
\begin{itemize}
  \item \underline{Approximate Counting $\rWPHP$}: $\rcWPHP[m,\eps]$.\vspace{0.2cm}\\For any \emph{deterministic} compressor-decompressor pair with encoding length $m<n$, an $\varepsilon$-fraction of inputs cannot be correctly decompressed.
  \item \underline{Randomized Compression $\rWPHP$}: $\rrWPHP[m,\eps]$.\vspace{0.2cm}\\ For a \emph{randomized} compressor and a deterministic decompressor with encoding length $m < n$, there is some input on which the pair has error probability at least $\varepsilon$.
\end{itemize}
These principles are formalized in a natural way using the probabilistic framework provided by $\APX_1$.\\

\paragraph{One-Way Communication Complexity.} We prove an equivalence result involving  communication complexity (CC) lower bounds  against randomized \emph{one-way protocols} with either \emph{public randomness} or \emph{private randomness}. Recall that the Set Disjointness function $\SetDisj(x,y)$ outputs $1$ if and only if for every index $i\in[n]$, either $x_i=0$ or $y_i=0$, i.e., $x$ and $y$ have no common $1$-index. The following statements, presented informally for clarity, are relevant to our result:

\vspace{0.05cm}

\begin{itemize}
    \item  \underline{Public Randomized CC Lower Bound for Set Disjointness}: $\pub\de\rsLB^\SetDisj[m,\varepsilon]$.\vspace{0.2cm}\\
    Every public-coin one-way protocol computing $\SetDisj$ with communication complexity $m$ must have error probability at least $\eps$ on some input pair $(x,y)$.
   \item  \underline{Private Randomized  CC Lower Bound for Set Disjointness}: $\prv\de\rsLB^\SetDisj[m,\varepsilon]$.\vspace{0.2cm}\\
    Every private-coin one-way protocol computing $\SetDisj$ with communication complexity $m$ must have error probability at least $\eps$ on some input pair $(x,y)$.
      \item  \underline{Public Randomized  CC Lower Bound for Some Function}: $\pub\de\rsLB^\some[m,\varepsilon]$.\vspace{0.05cm}\\
    For every $n\in\Log$, there exists  $f:\zo^n\times \zo^n\to\zo$ such that $\pub\de\rsLB^{f}[m, \varepsilon]$ holds. 
      \item  \underline{Private Randomized  CC Lower Bound for Some Function}: $\prv\de\rsLB^\some[m,\varepsilon]$.\vspace{0.05cm}\\
   For every $n\in\Log$, there exists  $f:\zo^n\times \zo^n\to\zo$ such that $\prv\de\rsLB^{f}[m,\varepsilon]$ holds. 
\end{itemize}

\noindent We leave the details about the formalization of the corresponding lower bound sentences to \Cref{sec:rev_math}. We note that $\APX_1$ is able to show that some concrete functions admit low-cost communication protocols. For instance, using linear hashing, it proves that $\mathsf{Equality}$ admits public-randomness one-way communication protocols of cost $O(\log n)$.\\

We can state an informal  version of our equivalence result as follows.\footnote{In particular, the simplified formulation given here omits considerations about the number of random bits employed in the randomized protocols, which plays a role in the parameters of some statements.}

\begin{theorem}[Main Equivalence Result (Informal); see \Cref{thm:rev_equivalences}]\label{thm:intro_equiv}
    The following statements are equivalent over $\Apx_1$, for suitable relations between the constants $k\ge 1$ and $0<\varepsilon<1$, quantified outside the theory:
  \vspace{-0.1cm}
    \begin{multicols}{3}
    \begin{compactenum}[\rm(1)]
        \item $\rcWPHP[n-1,n^{-k}]$
        \vspace{0.085cm}
        \item $\rcWPHP[n^{\eps},n^{-k}]$
        \vspace{0.085cm}
        \item $\rrWPHP[n-1,n^{-k}]$
        \vspace{0.085cm}
        \item $\rrWPHP[n^{\eps},n^{-k}]$
        \columnbreak
        \item $\pub\de\rsLB^\SetDisj[n-1,n^{-k}]$
        \item $\pub\de\rsLB^\SetDisj[n^{\eps},n^{-k}]$
        \item $\prv\de\rsLB^\SetDisj[n-1,n^{-k}]$
        \item $\prv\de\rsLB^\SetDisj[n^{\eps},n^{-k}]$
        \columnbreak
        \item \,\,\,$\pub\de\rsLB^\some[n-1,n^{-k}]$
        \item $\pub\de\rsLB^\some[n^{\eps},n^{-k}]$
        \item $\prv\de\rsLB^\some[n-1,n^{-k}]$
        \item $\prv\de\rsLB^\some[n^{\eps},n^{-k}]$
    \end{compactenum}
    \end{multicols}
\end{theorem}

As a consequence, one of these statements is provable in $\APX_1$ if and only if every statement in \Cref{thm:intro_equiv} is provable in $\APX_1$. This result provides evidence that $\APX_1$ can serve as a suitable base theory for developing the reverse mathematics of average-case and randomized lower bounds. 

For more details and additional discussion, we refer to \Cref{sec:rev_math}.

\subsection{Techniques}\label{sec:intro_techniques}

We next outline some of the main techniques used in our proofs, starting with a recurring argument that establishes basic probabilistic inequalities in $\APX_1$.

\subsubsection{Probabilistic Reasoning in \texorpdfstring{$\APX_1$}{APX1}: The  ``Pointwise to Global'' Technique \texorpdfstring{(\Cref{sec:intro_prob_reasoning})}{}}

At the core of our probabilistic reasoning is a simple but powerful seed-fixing lemma that lets us pass from \emph{global} inequalities to \emph{pointwise} statements about suitably chosen restrictions of the randomness. Recall that a feasible random variable \(X\) is specified by an explicit support \(V\subseteq\mathbb{Q}\), a seed length \(n\), and a multi-output circuit \(C:\{0,1\}^n\to V\). Its approximate expectation is
\[
\bbE_\delta[X] \eqdef \sum_{v\in V} v\cdot \P_\delta(C_v),
\]
where \(C_v\) is the indicator circuit for the event \(C(x)=v\). This is a $\PV(\P)$-computable quantity.

\medskip 
\medskip 
\noindent\textbf{A general averaging argument for expectation (\Cref{thm: averaging on exp}).} The following holds in $\APX_1$. Let \(X_1,\ldots,X_m\) be random variables over the same seed with support \(V\), and fix coefficients \(\lambda_1,\ldots,\lambda_m\in\mathbb{Q}\). Write
\[
\mu \eqdef \sum_{i=1}^m \lambda_i\,\bbE_\delta[X_i],
\qquad
\mu\!\upharpoonright z \eqdef \sum_{i=1}^m \lambda_i\,\bbE_\delta[X_i \mid z],
\]
where \(X_i\mid z\) denotes \(X_i\) after fixing a suffix of the seed to \(z\). Then, for every desired suffix length \(k\), there exists \(z\in\{0,1\}^k\) such that
\[
\mu\!\upharpoonright z \;\ge\; \mu \;-\; (2\delta+\beta)\cdot\|\lambda\|_1\cdot \|V\|_1.
\tag{3.7}
\]
Thus a lower bound on \(\mu\) can be \emph{witnessed} (up to controlled additive slack) by conditioning on a  partial assignment of the seed.

\medskip 
\medskip 

The proof iteratively fixes one seed bit at a time. By a form of Local Consistency for expectation, the average of \(\mu\!\upharpoonright (0\circ z)\) and \(\mu\!\upharpoonright (1\circ z)\) is close to \(\mu\!\upharpoonright z\); hence it is possible to prove that one of the two extensions preserves the current value up to an additive loss \(O(\eta\cdot \|\lambda\|_1\cdot \|V\|_1)\), where $\eta$ is an auxiliary parameter in the proof. Greedily repeating this for \(k\) steps yields a \(k\)-bit suffix with total loss \(O(k \cdot \eta\cdot \|\lambda\|_1\|V\|_1)\). A  precision-smoothing argument (switching from \(\eta\) to \(\delta\) via precision consistency) then gives the stated \((2\delta+\beta)\)-type bound, independent of \(k\). The greedy construction is formally captured by a $\PV(\P)$-procedure $\mathsf{AvgSampler}$. Conceptually, $\mathsf{AvgSampler}$ searches for a good suffix using calls to the approximate counting oracle $\P$, and its correctness is established using (polynomial) induction on $k$ over a quantifier-free $\PV(\P)$-formula, which is available in $\APX_1$.

\medskip

\begin{remarkn}{Example: Consistency of Complementation (\Cref{cor:complementation}).}  Let \(E_1,E_2:\{0,1\}^n\to\{0,1\}\) be complementary predicates, i.e., $E_1 = \neg E_2$ as Boolean circuits. Let \(X_1,X_2\) be their indicator variables. We argue in $\APX_1$. Given an arbitrary $\beta^{-1} \in \Log$, we set $\eta \eqdef \beta/C$, for a large enough constant $C$. Pointwise, for \emph{every full assignment} \(\rho\) to the seed, using the relation between $E_1$ and $E_2$ we have
\[
\bbE_\eta[X_1\mid \rho] + \bbE_\eta[X_2\mid \rho] - 1 \;=\; 0.
\]
Set \(\lambda \eqdef (1,1,-1)\) and consider \(\mu\eqdef\bbE_\eta[X_1]+\bbE_\eta[X_2]-1\). Applying \Cref{thm: averaging on exp} with \(k=n\), we obtain a \(\rho\) such that \(\mu \le \mu\!\upharpoonright \rho + O(\eta)\). But $\mu\!\upharpoonright \rho$ is exactly $0$ by the pointwise identity above, yielding \(\bbE_\eta[X_1]+\bbE_\eta[X_2]-1 = O(\eta)\). Similarly, one can show that \(-\bbE_\eta[X_1]-\bbE_\eta[X_2]+1 = O(\eta)\). Translating expectations back to probabilities via the indicator correspondence, and applying a standard precision-smoothing argument, one can  conclude that $\P_\delta$ is consistent with complementation, i.e.,
\[
\big|\P_\delta(E_1) + \P_\delta(E_2) - 1\big| \;\le\; 2\delta + \beta.
\]
\end{remarkn}

\medskip

To summarize, \Cref{thm: averaging on exp} provides a  way to fix randomness while preserving lower bounds on linear combinations of expectations. It turns equalities or inequalities that hold \emph{for each seed} into quantitative global bounds in $\APX_1$ with explicit additive slack depending only on \(\delta,\beta\) and the natural \(\ell_1\) norms of the supports and coefficients. This mechanism is the engine that drives many of our probability inequalities and applications in \Cref{sec: basic properties}.

\medskip

\begin{remark}
The bit-by-bit fixing trick is a standard technique in computational complexity theory. For instance, it is used in the search-to-decision reduction for $\SAT$ \cite[Section 2.5]{AB09} and to derive circuit lower bounds from derandomization \cite{DBLP:journals/toc/AaronsonM11}, among other results. In particular, our approach is inspired by a new proof of $\BPP\subseteq\MA\subseteq\Sigma_2^p$ via a bit-by-bit ``dueling argument'' \cite[Lemma A.10]{LPT24}. 
\end{remark}

\subsubsection{Provability of Circuit Lower Bounds \texorpdfstring{(\Cref{thm:intro_ave_parity} and \ref{thm:intro_worst_parity})}{}}

\Cref{thm:intro_ave_parity} gives an \emph{average-case} lower bound for the parity function against depth-$d$ $\AC^0$ circuits, formalized in $\APX_1$, while \Cref{thm:intro_worst_parity} gives a corresponding \emph{worst-case} lower bound, formalized in the weaker theory $\PV_1$. As alluded to above, the challenge is to avoid encoding-based arguments that rely on pigeonhole principles (or frameworks that build on them), since they are unavailable in these theories. The first result showcases how $\APX_1$ approximate-probability calculus supports average-case arguments in a more  sophisticated setting, while the second shows that, with additional derandomization work, we can carry a corresponding worst-case lower bound argument entirely within $\PV_1$. 

At a high level, the formalizations implement a simplification and derandomization \cite{AAIPR01} of the Furst–Saxe–Sipser \citep{DBLP:journals/mst/FurstSS84} \emph{random restriction approach} to $\AC^0$ lower bounds. Recall that the argument proceeds in stages, where at each stage we fix a suitable partial restriction $\rho \colon [n] \to \{0,1,*\}$ that sets all input variables in $T \eqdef \rho^{-1}(\{0,1\}) \subseteq [n]$. The crucial point is that a depth-$d$ circuit $C$ \emph{simplifies} when restricted by $\rho$, leading to a not much larger circuit $C \!\upharpoonright \rho$ of depth $d -1$, while the parity function \emph{retains its hardness}. 

Two central lemmas employed in the specification of $\rho$ drive the proofs of \Cref{thm:intro_ave_parity} and \Cref{thm:intro_worst_parity}.

\medskip \medskip 
\noindent\textbf{Subset Selection Lemma (\Cref{lem:subset_selection}).} The first step is to algorithmically choose the set \(T \subseteq [n]\) of variables with a “\emph{narrow-or-wide}” guarantee: for every bounded-width CNF/DNF $F$ at the bottom of the circuit $C$, after fixing the variables in \(T\), $F$ either already depends on few literals (narrow) or contains many disjoint subclauses supported on \(T\) (wide). Crucially, this subset selection is \emph{constructed and proved correct} in $\PV_1$ using a delicate potential-function  argument that simulates the method of derandomization via conditional expectations. This makes the selection of the subset $T \subseteq [n]$ feasible in our theories, not merely existential.

\medskip \medskip
\noindent\textbf{Restriction Selection Lemmas (\Cref{lem:random_restriction_lemma} and \ref{lem:derand_restriction}).} Given the narrow-or-wide structure exposed by the subset selection step, we then \emph{choose values for variables in \(T\)} so that all relevant gates in the circuit $C$ simultaneously simplify after applying the resulting restriction $\rho$. There are two versions, matching our two theorems:

\begin{itemize}
  \item In the $\APX_1$ setting of  \Cref{thm:intro_ave_parity}, we consider a \emph{random assignment of bits} (\Cref{lem:random_restriction_lemma}) and then fix a good selection of the values via $\APX_1$'s ``pointwise-to-global'' \emph{averaging argument for expectation} explained above. This lets us form a partial restriction $\rho$ and obtain a corresponding circuit $C \!\upharpoonright \rho$ that approximately retains the relative advantage of $C$ when computing parity.
  \item In the more constrained $\PV_1$ setting of  \Cref{thm:intro_worst_parity}, we \emph{derandomize} the same choice (\Cref{lem:derand_restriction}). Again, this is implemented by a potential argument that feasibly simulates the method of conditional expectations within $\PV_1$. A crucial aspect of the proof that facilitates the formalization is that the relevant expectations depend on at most $O(\log n)$ input coordinates and thus can be efficiently computed by $\PV_1$ terms.
\end{itemize}

In both settings, the circuit lower bound is obtained by an inductive application of the restriction technique, as in the standard proof of the result. The details appear in \Cref{sec:AC0_LB}.

\medskip 

The novelty is not in the combinatorics of $\AC^0$ versus Parity but in the underlying \emph{proof-theoretic framework}. $\APX_1$ supplies a minimal yet sufficient probabilistic infrastructure that lets us carry out the average-case lower bound argument internally. In our proofs, this lets us define and reason about the agreement tester $T_C$ in the statement of \Cref{thm:intro_ave_parity}, define and analyze appropriate events, quantify advantage and pass from randomized restrictions to concrete choices, and keep track of the small additive losses accumulated across iterations — all within $\APX_1$.

Moreover, this streamlined setup and the perspective it provides clarify the boundary with the weaker theory $\PV_1$. The $\APX_1$ framework and our formalization isolate exactly where probabilistic reasoning is used  and where the argument is purely combinatorial. This separation indicates which components can be replaced by deterministic potential-based arguments available in $\PV_1$, thereby guiding the adaptation that yields our worst-case formalization in $\PV_1$.

\subsubsection{The Witnessing Theorem \texorpdfstring{(\Cref{thm: witness TFNP})}{}}

\Cref{thm: witness TFNP} states that every \(\forall\Sigma^b_1(\PV)\)-sentence provable in $\APX_1$ admits a deterministic polynomial-time Turing reduction to the total search problem $\RefYao$, with parameters obeying \((\delta^2/10)\cdot m \ge s+\lceil\log n\rceil+1\). Recall that an instance of $\RefYao$ gives a \emph{predictor generator} \(G\) and asks for a flat distribution \(\mathcal{D}\) such that the predictor \((i,P)=G(\mathcal{D})\) fails to predict the \(i\)-th bit of \(\mathcal{D}\) with advantage \(\delta\). (For the stated parameter range, $\RefYao$ is in $\TFZPP$ and  map-reduces to $\mathsf{LossyCode}$.)

\medskip

Suppose that $\APX_1 \vdash \forall x\,\exists y\,\varphi(x,y)$, where $\varphi$ is a quantifier-free $\PV$-formula. Given $x$ of length $n$, we describe a predictor generator $G_x$ such that a solution $\mathcal{D}$ to $\RefYao$ over $G_x$ allows us to compute $y$ such that $\varphi(x,y)$ holds. 

Starting from an $\APX_1$-proof  of \(\forall x\,\exists y\,\varphi(x,y)\), we first apply Herbrand’s theorem over the universal axiomatization of $\APX_1$ to obtain finitely many $\PV(\P)$-terms \(t_1,\dots,t_c\) such that \(\bigvee_i \varphi(x,t_i(x))\) holds, and encode this disjunction by a single term \(t_\varphi\) with the equational core $\APX$ proving \(t_\varphi(x,t_1(x),\ldots,t_c(x))=1\). In standard models (\Cref{sec:intro_theory}), these terms are polynomial-time oracle algorithms for the approximate-counting oracle \(\P\). 

Next, we describe the construction of $G_x$. Given a candidate flat distribution \(\mathcal{D}\) (say, over \(n'\)-bit strings and of support size \(m'\), where $n',m' = \poly(n)$ are large enough), we \emph{simulate} each oracle call $\P(C, \Delta)$ inside any \(t_i\) by \emph{empirical counting on \(\mathcal{D}\)}:
\[
\P(C,\Delta) \eqdef \Pr_{u\leftarrow \mathcal{D}}\big[C(u_{\le \ell})\big],
\]
where $C \colon \{0,1\}^{\ell} \to \{0,1\}$ and $\ell \leq n'$. Let \(t_i^{\mathcal{D}}(x)\) and \(t_\varphi^{\mathcal{D}}(x,\cdot)\) denote the resulting outputs. Note that these can be computed in \emph{deterministic} polynomial time, since $\mathcal{D}$ is explicitly given as a collection of $\poly(n)$ strings of length $\poly(n)$, and  $t_\varphi$, $t_1$, $\ldots$, $t_c$ run in polynomial time.

\medskip \medskip

\noindent \textbf{Predictor Extraction Lemma (\Cref{lmm: predictor extraction}).} The key technical step says: if under this simulation the $\APX$-provable equation fails, i.e., \(t_\varphi^{\mathcal{D}}(x,\cdot)\neq 1\), then we can algorithmically extract a small predictor \(P\) of size $s'$ that achieves advantage at least \(\delta'\) for an explicitly computed bit position of \(\mathcal{D}\), where \((\delta'^2/10)\cdot m' \ge s'+\lceil\log n'\rceil+1\). 

Conceptually, an $\APX$ proof asserts: for \emph{every} interpretation of \(\P\), either the equation holds or one of the approximate-counting axioms (\emph{Basic}, \emph{Boundary}, \emph{Precision Consistency}, \emph{Local Consistency}) is violated. Under our empirical interpretation, the first three axioms continue to hold, so any failure must exhibit a \emph{Local Consistency} violation of the form
\[
\Big|\Pr_{u\leftarrow \mathcal{D}}[C(u_{\le \ell})]-\tfrac12\!\left(\Pr_{u\leftarrow \mathcal{D}}[C(u_{<\ell}0)] + \Pr_{u\leftarrow \mathcal{D}}[C(u_{<\ell}1)]\right)\Big|
\;>\; \tfrac{2}{|\Delta|}+\tfrac{1}{|B|},
\]
for a circuit \(C\) and strings \(\Delta,B\) produced in the $\APX$ proof. 

Similarly to the analysis of Yao’s distinguisher-to-predictor lemma, such a gap yields a predictor for the next bit via a deterministic transformation of \(C\); here the ``signal'' comes not from distinguishing \(\mathcal{D}\) from uniform, but from detecting a \emph{local inconsistency} of empirical counts across bit-fixings. Thus predictors arise not only from distinguishers, but also from the ability to spot local inconsistencies when \(\mathcal{D}\) is used as a random source for approximate counting -- a viewpoint that might be of independent interest.\footnote{In particular, Yao's distinguisher to predictor transformation requires \emph{randomness} (unless we have $\pr\BPP=\pr\P$) \cite{LPT24}. The construction of predictors from local inconsistency, however, is deterministic.} 

Formally, the lemma is established by a \emph{proof-theoretic analysis}, proceeding by induction on the steps of the $\APX$ proof and tracking the parameters \((n',m',s',\delta')\) through the final rule used.

\medskip \medskip

Wrapping up the argument, for an input \(x\), the reduction outputs the predictor generator $G_x$ obtained from \Cref{lmm: predictor extraction}. Note  that any solution \(\mathcal{D}\) to the resulting $\RefYao$ instance \emph{cannot} trigger a successful predictor extraction, since by definition $G_x$ fails to produce a predictor on \(\mathcal{D}\). Hence given a solution $\mathcal{D}$ to this instance of $\RefYao$, it must make the simulated identity true, i.e., 
$$
t_\varphi(x,t^{\mathcal{D}}_1(x),\ldots,t^{\mathcal{D}}_c(x))=1.
$$
In other words, some \(t_i^D(x)\) is a valid witness \(y\) for \(\varphi(x,y)\). Finally, as observed above, because \(\mathcal{D}\) is explicit, all simulated calls and \(t_i^{\mathcal{D}}(x)\) are computable in deterministic polynomial time. 

This completes the sketch of the proof of \Cref{thm: witness TFNP}. For the details, see \Cref{sec: witnessing}.

\subsection{Related Work}

Below we provide a representative, though not exhaustive, list of related developments and references.

\paragraph{Probabilistic arguments in bounded arithmetic.} Paris, Wilkie, and Woods \citep{PWW88} (see also Pudlák \citep{DBLP:conf/csl/Pudlak90}) observed that many probabilistic arguments can be formalized using variants of the weak pigeonhole principle rather than exact counting. An early explicit link between the weak pigeonhole principle and randomized algorithms is due to Wilkie (cf.~\citep{Krajicek-book}), who showed that randomized polynomial-time algorithms witness all $\forall\Sigma^b_1$-consequences of $\S^1_2+\dWPHP(\PV)$.

Ojakian \citep{thesis_KO} undertakes a general study of how probabilistic methods from combinatorics can be formalized in bounded arithmetic. While such proofs can often be recast as purely counting-based arguments, the naive translation still leaves exponentially many objects to count. The central idea, again, is to use the weak pigeonhole principle to simulate the probabilistic counting argument and thereby avoid this blow-up. The  formalizations are carried out in $\S^1_2$ augmented with suitable variants of the pigeonhole principle.

Je\v r\'abek \citep{Jerabek04} showed that within $\PV_1$ one can compare the sizes of two bounded $\P/\poly$-definable sets by constructing a surjection from one onto the other; he used this to formalize descriptions of algorithms in $\ZPP$ and $\RP$. He further showed \citep{Jerabek04, Jerabek-phd} that $\APC_1=\PV_1+\dWPHP(\PV)$ is strong enough to formalize sophisticated derandomization results. In \citep{Jerabek07}, Je\v r\'abek developed a more systematic framework, showing in particular that for any bounded $\P/\poly$-definable set, $\APC_1$ proves that a suitable pair of surjective counting functions exists that approximates its cardinality up to a polynomially small error. (The notation $\APC_1$ follows the terminology of \citep{DBLP:journals/jsyml/BussKT14}.) 

Built on Je\v r\'abek's framework, Lê \cite{Le14} formalizes more results in $\APC_1$ and its extensions, including randomized matching algorithms, the Lov\'{a}sz Local Lemma, and the Goldreich-Levin theorem. Throughout these formalizations, Lê provides formulations of concepts in $\APC_1$ such as \emph{expectation}, \emph{Markov inequalities}, and \emph{pairwise-independence}, which we also consider in this project. The formalization of random variables and expectation in \cite{Le14} heavily relies on the machinery of $\APC_1$ and is thus inadequate for our purposes. 

\begin{remark}
A concrete open problem is to state and prove a stronger form of the Chernoff bound in $\APX_1$. In this work, we show that the Chernoff bound with $O(\log n)$ variables (i.e., strings of length $n$ are considered feasible) can be formalized in $\APX_1$. It is unclear whether we could state a clean and meaningful Chernoff bound with $O(n)$ variables: the error probability will be exponentially small, which could be much smaller than the approximate counting error of the function $\P$. Moreover, even if a meaningful formalization exists, it is unclear whether existing proofs of the Chernoff bound can be formalized in $\APX_1$. Note that a strong form of Chernoff bound with $O(n)$ variables can be formalized in \Jerabek's theory $\APC_1$ (see \cite[Proposition 2.18]{Jerabek07}). 
\end{remark}

\paragraph{The proof complexity of $\dWPHP(\PV)$.} There is evidence that the pigeonhole-based axioms used throughout these frameworks exceed what can be proved in purely polynomial-time theories: while $\dWPHP(\PV)$ is available in $\T^2_2$, relativized variants are unprovable already in $\S^2_2$ \citep{Riis1993}. As noted above, under cryptographic assumptions, $\PV_1$ does not prove $\dWPHP(\PV)$ \citep{ILW23}. This supports the common stance that $\PV_1$ is too weak to derive the $\WPHP$-style principles exploited in the above formalizations. For a comprehensive investigation of $\dWPHP(\PV)$ and its provability in bounded arithmetic, see \citep{krajicek-25}. 

Je\v r\'abek's approximate-counting toolbox includes general principles such as inclusion-exclusion and strong Chernoff-type estimates, all formalized inside $\APC_1$. By contrast, the development of $\APX_1$ deliberately starts from weaker primitives: while we recover Markov/Chebyshev-style reasoning, error reduction, and other basic probabilistic tools, we do not reprove all of Je\v r\'abek’s strongest concentration bounds here. It remains an interesting direction to test the limits of $\APX_1$: which stronger probabilistic inequalities (e.g.,~full-strength Chernoff) are intrinsically beyond its axioms?

\paragraph{Beyond approximate counting with additive error.} Certain combinatorial proofs—e.g., of Ramsey’s theorem—typically require counting \emph{sparse sets}, which is unavailable both in our framework and in Je\v r\'abek’s theory $\APC_1$. In our setting, for $X \subseteq \{0,1\}^n$ we can estimate $|X|$ only to within an additive error that is a polynomial fraction of $2^n$, whereas these arguments require accuracy within a polynomial fraction of $|X|$. Such counting becomes possible in theories stronger than $\APC_1$, as developed in  \citep{DBLP:journals/jsyml/Jerabek09}.

In a concurrent work, Thapen \cite{Thapen25} introduces a framework to formulate stronger complexity classes (such as $\oplus\P$ and $\#\P$) in the theory of $\TFNP$ in a way that is similar in spirit to our  axiomatization of approximate counting. For instance, given an oracle that is intended to compute $\oplus\P$, Thapen considered the relativized $\TFNP$ problem that searches for a ``local inconsistency'' of the oracle. Note that the $\TFNP$ framework in \cite{Thapen25} considers only  \emph{query complexity}, while we additionally consider proofs in bounded arithmetic. Nevertheless, it is conceivable that results in these two directions may have analogues in each framework given the similarity in the setup. 

\paragraph{Theories with explicit counting.} In \citep[Chapter 6]{Jerabek-phd},  Je\v r\'abek studies bounded theories with explicit counting, revisiting the Impagliazzo-Kapron \citep{IK06} second-order logic for formalizing cryptographic reasoning.  The logic is multi-sorted: first-order variables range over strings, while
second-order variables of sort $k>0$ range over $k$-ary (intended polynomial-time) functions. In particular, functions are second-order objects rather than function symbols in the language. The theory includes recursive counting constructs for expressing the sizes of definable bounded sets.

The same chapter also introduces a feasible theory of approximate counting using a 3-valued semantics based on Kleene's logic, equipped with an LPF-style implication to support induction-like reasoning. Counting is approximate: the semantics distinguishes between having many versus few solutions (while using 3-valued logic to allow an explicit indeterminate region). The resulting ``$\Sigma^c_1$-consequences''  admit probabilistic polynomial-time witnessing (see \citep[Theorem 6.2.20]{Jerabek-phd}).

The counting framework in these theories is considerably different from ours, relying on exact
counting terms or approximate counting quantifiers in different logical settings. We refer to
these references for details.

\paragraph{Bounded reverse mathematics.} Cook and Nguyen \citep{cook_nguyen_2010} provide a thorough exposition of the bounded reverse mathematics program, systematically developing theories of bounded arithmetic and presenting formalizations of key combinatorial and algorithmic results, with the goal of identifying the weakest axioms sufficient to prove them. 

Finally, we refer to \citep{Le14, thesis_Pich, DBLP:journals/apal/MullerP20, AT25} and references therein for numerous  examples of results from theoretical computer science that can be formalized in bounded arithmetic. It would be interesting to further investigate which of these formalizations can be carried out in $\APX_1$.\\

\noindent \textbf{Acknowledgements.} We would like to thank Jan Kraj{\'{i}}{\v{c}}ek for discussions related to the $\forall \Sigma^b_1$-conservativity of $\APC_1$ over $\PV_1$ and for bringing some references to our attention. We also thank Dimitrios Tsintsilidas for comments on an earlier version of the paper. We would also like to thank Surya Mathialagan, Shuo Pang, and Hanlin Ren for helpful discussions. Finally, we thank the anonymous STOC reviewers for useful comments about the presentation.

\section{Formal Definition of the Theory}\label{sec:sec_theory}

In this section, we formally define the equational theory $\Apx$ and its first-order counterpart $\Apx_1$. We assume basic familiarity with Cook's Theory $\PV$ \cite{Coo75}. The necessary background can be found in \cite[Chapter 12]{Krajicek-book}, \cite[Chapter 12]{krajicek_2019}, and \citep{DBLP:journals/eccc/Li25}. 

\subsection{Notation}
\label{sec: def: notation}

\paragraph{Base Theory.} Let $\PV(\P)$ be the theory $\PV$ relative to a fresh function symbol $\P$ with the axiom 
\[
\ITR(\P(C,\Delta),C\sm \Delta\sm \Delta) = \eps 
\]
that bounds the output length of the function symbol $\P$. Intuitively, the axiom means that the output length of $\P(C,\Delta)$ given strings $C$ and $\Delta$ as its input is at most $|C|\cdot |\Delta|^2$. This axiom ensures that $\PV(\P)$-terms are feasible functions in the standard model. Interested readers are referred to \cite{Jerabek04,Jerabek07,Jerabek07-dWPHP} for more examples of relativized $\PV$.

Slightly different from Cook's original notation, we will define $\PV(\P)$ with constant symbol $\eps$ (rather than $0$) and replace the initial functions $s_1(x)$ and $s_2(x)$ by $s_1(x)$ and $s_0(x)$, respectively. Other functions $\TR$, $\ITR$, $\circ$, and $\#$ are defined as in Cook's original definition. Let $\hat\P$ be a function over Boolean strings. The standard model of $\PV(\P)$ with respect to $\hat\P$, denoted by $\bbM(\hat\P)$, is defined as follows: 
\begin{compactitem}
\item The universe consists of all Boolean strings of finite length. 
\item The constant symbol $\eps$ is interpreted as the empty string. 
\item $s_b(x)$ (for $b\in\zo$) is interpreted as the function that appends $b$ to the right of the string $x$. 
\item $\TR(x)$ is interpreted as the function that trims the rightmost bit of $x$; $\ITR(x,y)$ is interpreted as the function that trims $x$ for $|y|$ times. 
\item $\circ$ is interpreted as string concatenation, while $\#(x,y)$ is interpreted as the function that concatenates $|y|$ copies of $x$.  
\item The function symbol $\P$ is interpreted as $\hat\P$.  
\end{compactitem}
A function introduced by one of the rules in $\PV$ (i.e.~introduction by terms or introduction by limited recursion on notation) is interpreted as the unique function over the universe that satisfies its introduction rule.

\newcommand{\IsCkt}{\mathsf{IsCkt}}
\newcommand{\IsConst}{\mathsf{IsConst}}
\newcommand{\Null}{\mathsf{Null}}
\newcommand{\True}{\mathsf{True}}

\paragraph{Circuits.} We define a few $\PV$ functions that manipulate Boolean circuits. Let $\IsCkt(C,z)$ be the $\PV$ function that outputs $1$ if $C$ is a circuit with input length $|z|$, and outputs $0$ otherwise; $\IsConst(C)$ be the $\PV$ function that outputs $1$ if $C$ is a circuit that does not read its input (i.e., there is no path from the output gate to an input variable); $\Bool(C)$ outputs $1$ if $\IsConst(C)$ and $C$ outputs $1$ and outputs $0$ if $\IsConst(C)$ and $C$ outputs $0$ (otherwise, outputs, e.g., $\eps$); $\Fix_b(C)$ be the function that, given a circuit $C$, output the circuit obtained from $C$ by fixing the rightmost input bit to be $b\in\zo$; $\Eval(C,x)$ be the function that evaluate the circuit $C$ on the input $x$. One may think of any straightforward implementations of these functions in $\PV$ as $\PV$ is a robust theory. 

For simplicity, we use the following abbreviations: 
\begin{compactitem}
\item For $n\in\Log$, $C\in B_n$ denotes $\IsCkt(C,1^n)$, i.e., $C$ is a circuit with $n$ input bits. Moreover, $\forall C\in B_n~\varphi(C)$ denotes $\forall C~(\IsCkt(C,1^n)\to \varphi(C))$ and $\exists C\in B_n~\varphi(C)$ denotes $\exists C~(\IsCkt(C,1^n)\land \varphi(C))$. 
\item For a circuit $C$, $C(x)$ denotes $\Eval(C,x)$. 
\item We use $\Null_n$ to denote the circuit with $n$ input bits that does not read its input bits and outputs $0$, and $\True_n$ to denote the circuit with $n$ input bits that does not read its input bits and outputs $1$. 
\end{compactitem}

\paragraph{$\PV$-Terms and Functions.} We say that a $\PV(\P)$-term is a $\PV$-term if its construction indicates that it does not call the $\P$-oracle. Formally, the set of $\PV$-terms is the minimum set that contains all base functions and is close under composition and the function formulation rules in $\PV$, that is: 
\begin{compactitem}
\item Base functions $s_i(x)$, $\TR(x)$, $\ITR(x)$, $\circ(x,y)$ are $\PV$-terms.
\item If $t$ is a $\PV$-term, the function $f_t$ introduced with the defining axiom $f_t = t$ is also a $\PV$ term. 
\item A term formulated from $\PV$-terms by composition is a $\PV$ term. 
\item If $g,h_0,h_1,k_0,k_1$ are $\PV$-terms, the function $f_\Pi$ constructed by limited recursion on notation from $\Pi=(g,h_0,h_1,k_0,k_1)$ is a $\PV$ term. 
\end{compactitem}
We say that a function symbol $f$ is a $\PV$-function if it is a $\PV$-term.

By the Cook-Levin theorem (see \cite{DBLP:journals/corr/Pich14} for a formalization in $\PV$), $\PV$ terms can be converted into polynomial-size Boolean circuits on any given input length $n\in\Log$, and the correctness can be proved in $\PV$. Similarly, $\PV(\P)$ terms can be converted into polynomial-size $\P$-oracle circuits on any given input length $n\in\Log$ with $\PV(\P)$-provable correctness. 

\paragraph{Encoding Conventions and Arithmetic Operations.} For functions and multi-output circuits, we will treat $\eps$ as \texttt{false} and any other value as \texttt{true} when we define the acceptance probability of the circuit. Let $\Bool(x)$ be the $\PV$ function that outputs $0$ if $x$ is $\eps$ and outputs $1$ otherwise. 

\newcommand{\IsNumber}{\mathsf{IsNumber}}
\newcommand{\IsRational}{\mathsf{IsRational}}

We assume that natural numbers are encoded in binary in a straightforward way. For instance, one can encode a natural number in dyadic notation as in \cite{Coo75} so that basic arithmetic operations such as addition, multiplication, and comparison can be defined naturally. We assume that the encoding can be verified efficiently, i.e., there is a $\PV$ function symbol $\IsNumber(x)$ that outputs $1$ if $x$ is the encoding of a natural number, and outputs $0$ otherwise, and use $[x\in\bbN]$ as the shorthand of $\IsNumber(x)$. We use $[x]_\bbN$ to denote the natural number encoded by $x$ when we want to be explicit about the interpretation of $x$ as a natural number. 

Elementary arithmetic operations, such as addition, multiplication, and comparison, can be defined naturally. Moreover, basic properties of the $\PV$ function symbols representing these operations can be established in $\PV$ (whenever the operations involve a feasible number of elements). 

We specify a standard encoding of rational numbers in $\PV$: We use the pair $(x,y)$ to denote the rational number
\[
\sum_{i=1}^{|x|} x_i 2^{i-1} + \sum_{j=1}^{|y|} y_j 2^{-j}.  
\]
Similarly to the encoding of natural numbers, we assume a $\PV$ function symbol $\IsRational(x)$ that tests whether $x$ encodes a rational number, and we use $[x \in \bbQ]$ as shorthand for $\IsRational(x)$. We write $[x]_\bbQ$ to denote the rational number encoded by $x$. We might directly treat $x$ as a rational number if this is clear from the context.

\paragraph{Data Structures and Explicit Sets.} We assume a straightforward encoding of \emph{explicit} sets (and multisets), i.e., sets of feasible size, that supports operations such as selection, union, intersection, and membership query. An explicit set $S$ may be encoded as a list containing all the elements in it. Note that this is different from the \emph{feasibly definable} sets in \cite{Jerabek04}, which may be of infeasible size. When discussing explicit sets,  we use $|S|$ to denote the size of $S$, i.e., the number of elements contained in $S$. When $S = \{q_1, \ldots, q_\ell\}$ is a set of rational numbers, we use $||S|| = \sum_{i=1}^{\ell} |q_i|$ to denote its $\ell_1$-norm, i.e., the sum of the absolute values of the elements in $S$.\footnote{While we abuse notation and employ $|\cdot|$ to denote both length and absolute value, the meaning will be clear in each context.} 

Moreover, for an explicit set $S$ and a quantifier-free formula $\varphi(x)$ in the language of $\PV_1$, we can define the universal quantification over $S$, denoted by $\forall x\in S:\varphi(x)$, as a \emph{quantifier-free} formula in $\PV_1$ that is true if and only if every element $x\in S$ satisfies $\varphi(x)$ (in the standard model). This is possible as $S$ is explicitly encoded, and thus there is a straightforward feasible algorithm that given the encoding of $S$, enumerates $S$ and checks whether there is an $x\in S$ such that $\varphi(x)$ is false. Similarly, we can define the existential quantification over $S$, denoted by $\exists x\in S:\varphi(x)$. All relevant deduction rules about quantification over sets should be admissible in $\PV$ assuming standard encoding, e.g., 
\begin{align}
& (\exists_i):\quad \frac{\Gamma\vdash \varphi[y/t]\quad \Gamma\vdash t\in S}{\Gamma\vdash \exists y\in S:\varphi}\\  & (\exists_e):\quad\frac{\Gamma\vdash\exists y\in S:\varphi\quad \Gamma,z\in S, \varphi[y/z]\vdash\psi}{\Gamma\vdash \psi}
\end{align} 
where in $(\exists_i)$ $t$ is an arbitrary term, and in $(\exists_e)$ $z$ must be a fresh variable that has no occurrence in $\Gamma,\varphi,\psi,y$. This ensures that most natural mathematical proofs regarding explicit sets can be easily formalized in $\PV$; see \cite[Chapter 4]{DBLP:journals/eccc/Li25} for more discussions. In the rest of the paper, we will only informally describe the proof and pinpoint the key idea to formalize it in $\PV$ if it is unclear.  

\subsection{Theory \texorpdfstring{$\Apx$}{}}

Intuitively, we will define the theory as $\PV(\P)$ together with additional axioms intended to formalize that $\P$ approximately computes the acceptance probability of a given circuit up to a specified precision. In other words, for every deterministic circuit $C$ and any $\Delta$, $\P(C,\Delta)$ outputs the encoding of a rational number in $[0,1]$ guaranteed to lie within the interval $[p - 1/|\Delta| , p + 1/|\Delta|]$, where $p$ denotes the acceptance probability of $C$. For simplicity, we will also denote $\P(C,\Delta)$ by $\P_\delta(C)$, where $\delta^{-1}\eqdef |\Delta|\in\Log$ is the precision of counting. 

\paragraph{Language of $\Apx$.} $\Apx$ is an equational theory whose language extends that of $\PV$ by including the new function symbol $\P$ and every additional function symbol that can be introduced through the usual function symbol introduction rules of $\PV$ (including composition and limited recursion on notation).\\

Although $\PV$ is an equational theory operating over strings, propositional connectives, arithmetic operations, and arithmetic relations (e.g., comparison between rational numbers) can be encoded by appropriate equations with desired properties (see, e.g., \citep{DBLP:journals/eccc/Li25}). This allows us to formulate the following axioms. 

\paragraph{Axioms of $\Apx$.} The axioms involve only universally quantified variables and can therefore be expressed as $\PV(\P)$ equations:

\dummylabel{basic}{Basic Axiom}
\dummylabel{boundary}{Boundary Axiom}
\dummylabel{precision}{Precision Consistency Axiom}
\dummylabel{local}{Local Consistency Axiom}
\dummylabel{logical}{Logical Rules}
\dummylabel{induction}{Induction Rule}
\newcommand{\refaxiom}[1]{\emphax{\ref{#1}}}

\begin{compactitem}
\item (\emph{Basic Axiom}). Any provable equation in $\PV(\P)$ is an axiom of $\Apx$. Moreover, $[\P(C,\Delta)\in \bbQ]=1$, $\P(C,\Delta)\le 1$, $0\le \P(C,\Delta)$ are axioms of $\Apx$, where ``$x\le y$'' is formalized by an appropriate $\PV$ equation that is valid if and only if $[x]_\bbQ\le[y]_\bbQ$.  
\item (\emph{Boundary Axiom}). For any $C\in B_n$, $\IsConst(C)\to \P_\delta(C) = \Bool(C)$. This axiom indicates that the acceptance probability of a syntactically constant circuit\footnote{In other words, there is no path from the output gate to an input variable, i.e., the relevant part of the circuit consists of Boolean operations applied to constant input bits.} that always outputs $b\in\zo$ is equal to $b$. 
\item (\emph{Precision Consistency Axiom}). For every $n,\delta_1^{-1},\delta_2^{-1},\beta^{-1}\in\Log$ and every $C\in B_n$,
\begin{equation}
\left|\P_{\delta_1}(C) - \P_{\delta_2}(C)\right|\le \delta_1+\delta_2 + \beta.
\end{equation}
Intuitively, this axiom states that the approximate counting function $\P$ should be consistent with different precision parameters.  
\item (\emph{Local Consistency Axiom}). For every $n,\delta^{-1},\beta^{-1}\in\Log$ and every $C\in B_n$, 
\begin{equation} 
\left|\P_\delta(C)-\frac{\P_\delta(\Fix_0(C))+\P_\delta(\Fix_1(C))}{2}\right| \le 2\cdot \delta + \beta.
\end{equation}
Intuitively, this axiom states that the approximate counting function $\P$ should be self-consistent in the sense that the acceptance probability of a circuit $C$ is close to the average acceptance probability of the circuit obtained by randomly fixing the rightmost input bit of $C$.
\end{compactitem}

\paragraph{Rules of $\Apx$.} Finally, the theory contains the following derivation rules:

\begin{compactitem}
\item (\emph{Logical Rules}). We include the logical rules of $\PV$: 
\begin{compactenum}
\item $t_1=t_2\vdash t_2=t_1$ 
\item $t_1=t_2,t_2=t_3\vdash t_1=t_3$ 
\item $t_1=t_2\vdash t_1(x/t) = t_2(x/t)$ 
\item $u=v\vdash t(x/u) = t(x/v)$ 
\end{compactenum}
\item (\emph{Structural Induction Rule}). Let $f_1(x,\vec y)$ and $f_2(x,\vec y)$ be $\PV(\P)$ functions. For $\PV(\P)$ functions $g(\vec y)$, $h_0(x,\vec y,z)$, and $h_1(x,\vec y,z)$, if the following equations are provable for $j\in\{1,2\}$ and $i\in\zo$
\begin{align}
& f_j(\eps,\vec y) = g(\vec y) \\ 
& f_j(s_i(x),\vec y) = h_i(x, \vec y, f_j(x,\vec y))
\end{align}
then we can deduce the equation $f_1(x,\vec y) = f_2(x,\vec y)$. This rule is analogous to the original induction rule in $\PV$. Intuitively, it means that if $f_1$ and $f_2$ are both identical to the function recursively defined from $g,h_0,h_1$, they are the same function.  
\end{compactitem}

\vspace{0.2cm}

\begin{remarkn}{Nested Probability Symbols}
    We stress that the function symbol $\P$ does not take $\P$-oracle circuits as input, and therefore sentences involving nested probability symbols such as 
    \[
    \Pr_x[\,\Pr_y[\varphi(x,y)]>\eps\,] > \delta 
    \]
    cannot be formalized directly in $\APX$. Nevertheless, for predicates $\varphi(x,p)$ and $\psi(y)$ that do not share inputs, the nested probability 
    \[
    \Pr_x\left[\psi\!\left(x,\Pr_{y}[\varphi(y)]\right)\right]
    \]
    can be expressed in $\APX$, as we can define a $\P$-oracle algorithm that first calculates $p=\Pr_{y}[\varphi(y)]$ by calling the $\P$-oracle, and then calculates $\Pr_x[\psi(x,p)]$ by calling the $\P$-oracle again.  
\end{remarkn}

\vspace{0.2cm}

\begin{remarkn}{Elementary Functions and Precision Issues}
In our formalizations, we sometimes employ elementary functions over the reals, such as $\sqrt{x}$, $\ln(x)$, or $\exp(x)$ (typically involving constants or for an $x$ of the form $a/b$ with $a,b \in \Log$). As the output of these functions may be an  irrational number, to implement them in $\PV$, we need to define each function by taking an additional parameter that determines the number of digits of precision. When the functions are defined appropriately, appropriate formulations of basic inequalities  (e.g.,~$\exp(x)\ge 1+x$) can be proved in $\PV$ by directly formalizing a standard mathematical proof. In this paper, in the context of the use of such values and inequalities, we always have a margin to tolerate any potential precision issue (e.g.,~the error term $\beta\in\Log^{-1}$ in axioms). For this reason, and following standard practice, we will not elaborate on the actual implementation of such functions and their basic properties. 
\end{remarkn}

\subsection{Models of \texorpdfstring{$\Apx$}{}}\label{sec:models}

Let $\bbM$ be the standard model of $\PV$, i.e., the universe is $\zo^*$, $\eps$ is interpreted as the empty string, and $s_0(x)$, $s_1(x)$ are interpreted as the functions that append $0$ and $1$ to $x$, respectively. For every function $\hat\P:\zo^*\times \zo^*\to\zo^*$, $\bbM(\hat\P)$ is the model of $\PV(\P)$ where the function symbol $\P$ is interpreted as the function $\hat\P$. 

\begin{definition}[Standard Models]
    Let $\hat\P \colon \zo^* \times \zo^* \to \zo^*$ be any \emph{correct approximate counting function}, i.e., 
    \begin{enumerate}
        \item $\hat\P(C,\Delta)$ outputs (the encoding of) a rational number $q\in [0,1]$ within the interval $[p-1/|\Delta|, p+1/|\Delta|]$ for every circuit $C:\zo^*\to\zo$ and $\Delta\in\zo^*$, where $p$ is the acceptance probability of $C$ and $q$ is of length at most $|C|\cdot |\Delta|^2$; and 
        \item $\hat\P(C,\Delta)$ outputs the correct value in $\{0,1\}$ whenever the input circuit $C$ satisfies $\IsConst(C)$. 
    \end{enumerate}
    We say that $\bbM(\hat\P)$ is \emph{a standard model} of $\Apx$. 
\end{definition}

\begin{definition}[Admissible Models]
    Let $\hat\P\colon \zo^*\times \zo^*\to\zo$ be a function. We say that $\bbM(\hat\P)$ is \emph{an admissible model} of $\APX$ if it satisfies all axioms and rules of $\APX$.
\end{definition}

The crucial observation is that a model is standard if and only if it is admissible. The proof of the theorem is highly constructive; indeed, similar induction arguments occur multiple times in the development of basic probability theory in $\APX_1$ (see \Cref{sec: basic properties}). 

\begin{theorem}\label{thm: standard equiv admissible}
Let $\hat \P:\zo^*\times\zo^*\to\zo^*$ be any function. Then $\bbM(\hat\P)$ is a standard model if and only if it is an admissible model.  
\end{theorem}
\begin{proof}
    We will only prove the $(\Leftarrow)$ direction, as the converse is straightforward. Suppose, towards a contradiction, that $\bbM(\hat\P)$ is admissible but is not standard. Then there is a circuit $C:\zo^n\to\zo$ and $\Delta\in\zo^*$ such that 
    \begin{equation}
    \hat \P(C,\Delta) \notin p(C) \pm \frac{1}{|\Delta|},\label{equ: hat P is not standard} 
    \end{equation}
    where $p(C) \eqdef \Pr_{x\in\zo^n}[C(x)=1]$ is the acceptance probability of $C$.
    Note that $C$ and $\Delta$ are encoded by finite strings (in the standard model $\bbM$ of $\PV$), and $n\in\bbN$. 
    
    As $\bbM(\hat\P)$ is admissible, it must satisfy the \refaxiom{boundary}. Subsequently, \Cref{equ: hat P is not standard} does not hold when $n=0$. It suffices to consider the case that $n>0$. Suppose for contradiction we have  
    \begin{equation}
    \left|\hat \P(C,\Delta) - p(C)\right| > \frac{1}{|\Delta|}+\eps,
    \end{equation}
    where $\eps>0$. Let $\Xi \eqdef 1^{10(n+1)/\eps}$. As $\hat\P$ satisfies the \refaxiom{precision}, we have 
    \begin{equation}
    \left|\hat \P(C,\Xi)-p(C)\right| \ge \left|\hat \P(C,\Delta)-p(C)\right| - \left(\frac{1}{|\Delta|}+\frac{2}{|\Xi|}\right) > \eps - \frac{2}{|\Xi|}. \label{equ: P hat C far}
    \end{equation}

    Let $C_0,C_1$ be the circuits obtained by fixing the rightmost input bit of $C$ to be $0$ and $1$, respectively. As $\hat\P$ must satisfy the \refaxiom{local}, we have 
    \begin{equation}
    \left|\hat \P(C,\Xi) - \frac{\hat \P(C_0,\Xi)+\hat\P(C_1,\Xi)}{2}\right|\le \frac{3}{|\Xi|}, 
    \end{equation}
    Also, $p(C)=(p(C_0) + p(C_1))/2$ by its definition. Subsequently, there exists $\sigma\in\zo$ such that 
    \begin{equation}
    \left|\hat \P(C_\sigma,\Xi) - p(C_\sigma)\right|\ge \left|\hat \P(C,\Xi) -  p(C)\right|-\frac{3}{|\Xi|}. 
    \end{equation}
    
    Recall that $n\in\bbN$ is a standard integer. Let $C^{(0)}\eqdef C$, and $C^{(1)}\eqdef C_\sigma$. By the procedure defined above, for every $1 \leq i\le n$, we can define $C^{(i)}$ as the circuit obtained from $C^{(i-1)}$ by fixing the rightmost input bit such that 
    \[
    \left|\hat \P(C^{(i)},\Xi) - p(C^{(i)})\right| \ge  \left|\hat \P(C^{(i-1)},\Xi) - p(C^{(i-1)})\right| - \frac{3}{|\Xi|},
    \]
    and therefore by \Cref{equ: P hat C far} we  eventually have 
    \[
    \left|\hat \P(C^{(n)},\Xi)-p(C^{(n)})\right| \ge   \eps - \frac{6\cdot (n+1)}{|\Xi|} > 0. 
    \]
    Note that the circuit $C^{(n)}$ has input length $0$ and as a consequence computes a constant function. The value $p(C^{(n)})\in\zo$ is its acceptance probability. Since $\bbM(\hat\P)$ is admissible, it satisfies the \refaxiom{boundary}, and consequently $\hat \P(C^{(n)},\Xi) = p(C^{(n)})$. This  contradicts the above inequality.
\end{proof}

\begin{definition}
Among the standard models of $\Apx$, the one that interprets $\P$ by the exact counting function is called the \emph{exact standard model} of $\Apx$, denoted by $\bbM^*$.    
\end{definition}

\begin{proposition}[Soundness of $\APX$]\label{prop: soundness APX}
Provable equations in $\APX$ are true in any standard model of $\APX$. 
\end{proposition}

\begin{proof}
    This can be verified by induction on the proof. 
\end{proof}

\subsection{First-Order Theory \texorpdfstring{$\Apx_1$}{Apx1}} 

In analogy with the first-order theory $\PV_1$, we will introduce a first-order theory $\Apx_1$ that includes all $\Apx$ provable equations as well as convenient deduction rules. 

\paragraph{Language of $\Apx_1$.} The language of the first-order theory $\Apx_1$ includes all $\PV(\P)$ symbols.

\paragraph{Axioms of $\Apx_1$.} The theory is axiomatized by the standard first-order logic with equality together with the following non-logical axiom schemes: 
\begin{compactitem}
\item For any provable equation $s(\vec x) = t(\vec x)$ of $\Apx$, $\forall\vec x~s(\vec x) = t(\vec x)$ is an axiom of $\Apx_1$. 
\item $\forall x~\forall y~(x=y\leftrightarrow s_i(x)=s_i(y))$, $i\in\zo$, is an axiom of $\Apx_1$.  
\item $\forall x~\eps\ne s_i(x)$ is an axiom of $\Apx_1$. 
\item $\forall x~s_0(x) \ne s_1(x)$ is an axiom of $\Apx_1$. 
\item ($n$-induction). Let $\varphi$ be a quantifier-free formula and $x_1,\dots,x_n,\vec y$ be variables. Suppose that $\varphi$ does not contain free variables other than $x_1,\dots,x_n$ and $\vec y$. Then 
\[
\forall \vec y~\left(\bigwedge_{j\in[n]} \varphi(x_j/\eps)\land\forall \vec x\left(\bigwedge_{\vec \sigma\in\zo^n}(\varphi\to\varphi_{\vec \sigma})\right) \to \forall\vec x~\varphi\right),
\]
where $\varphi_{\vec \sigma}$ denotes the formula $\varphi$ with all free occurrences of $x_i$ substituted by $s_{\sigma_i}(x_i)$ for each $i\in[n]$, is an axiom of $\Apx_1$.\\
\end{compactitem}
\dummylabel{ind}{$n$-Induction Axiom}
\dummylabel{ind 1}{$1$-Induction Axiom}
\newcommand{\refind}[1]{(#1-)\ref{ind}}

\newcommand{\Leq}{\mathsf{Leq}}
\newcommand{\Less}{\mathsf{Less}}
\newcommand{\calM}{\mathcal{M}}

We observe that $\Apx_1$ satisfies the following properties. 

\begin{proposition}\label{prop: apx1 universal}
    $\Apx_1$ admits a universal axiomatization. 
\end{proposition}

\begin{proof}[Proof Sketch]
    This is essentially the same as the proof that $\PV$ admits a universal axiomatization (see \cite{Coo75,Krajicek-book}). Note that all axioms of $\Apx_1$ are universal sentences except for the $n$-induction axiom scheme, which is a $\forall\exists$-sentence. In more detail, the $n$-induction axiom scheme is logically equivalent to the following sentence: For every $\vec y$ and $\vec x=(x_1,\dots,x_n)$ satisfying that 
    \begin{compactitem}
    \item $\varphi(x_j/\eps)$ for every $j\in[n]$, and 
    \item $\lnot \varphi$, 
    \end{compactitem}
    there exists an $\vec x'=(x_1',\dots,x_n')$ and $\vec \sigma=(\sigma_1,\dots,\sigma_n)\in\zo^n$ such that $\varphi(\vec x/\vec x')$ is true but $\varphi(\vec x/\vec x'_\sigma)$ is false, where $\vec x'_\sigma\eqdef (s_{\sigma_1}(x_1'),\dots,s_{\sigma_n}(x'_n))$. Nevertheless, there is a straightforward polynomial-time algorithm that outputs such $\vec x'$ given $\vec x$ and $\vec y$ by considering prefixes of $\vec x$, and the correctness of the algorithm can be proved in $\PV(\P)$. This can be used to show that the $n$-induction axiom scheme can be derived from other axiom schemes and $\Apx_1$ admits a universal axiomatization. 
\end{proof}

\begin{proposition}\label{prop: APX1 conservative}
    $\Apx_1$ is conservative over $\Apx$. 
\end{proposition}

\begin{proof}[Proof Sketch]
    The proof is essentially the same as the proof that $\PV_1$ is conservative over $\PV$. We refer interested readers to \cite{Coo75,Buss,Krajicek-book} for more details. 
\end{proof}

Similar to  $\PV_1$ (see \cite{KrajicekPT91}), we can show that a form of induction principle on quantifier-free formulas is provable in $\Apx_1$. In order to state this result, we need to make some remarks about notation.

Recall that we assume a straightforward encoding of natural numbers, such as the dyadic encoding in \cite{Coo75}, and use $[x]_\bbN$ to denote the natural number encoded by $x$. Let $\Less_\bbN(x,y)$ be the $\PV$ function that outputs $1$ if (in the standard model) $[x]_\bbN < [y]_\bbN$ and outputs $0$ otherwise. We use $[x<y]$ as a shorthand for $\Less_\bbN(x,y)$. We use $\forall x < y~\varphi(x)$ as a shorthand for $\forall x~([x < y] = 1\to \varphi(x))$, and $\exists x < y~\varphi(x)$ as a shorthand for $\exists x~([x<y] = 1\land \varphi(x))$. 

\begin{theorem}\label{thm: ind via binary search}
Let $\varphi(x,\vec y)$ be a quantifier-free formula. Then $\Apx_1$ proves  
\[
\forall \vec y~\forall b~\left(\varphi(0,\vec y) \land \forall x<b~(\varphi(x,\vec y)\to \varphi(x+1,\vec{y}))\to \varphi(b,\vec y)\right),
\]
where $0$ is the $\PV$-term encoding $0\in\bbN$ and $+$ is the $\PV$-function for addition of natural numbers. 
\end{theorem}

\begin{proof}[Proof Sketch]
    The proof is essentially the same as the admissibility proof of such induction scheme in $\PV_1$, following a binary search argument. We refer interested readers to \cite{Coo75,Buss,KrajicekPT91,Krajicek-book} for more details. 
\end{proof}

\paragraph{Models of $\APX_1$.} Any model $\calM(\calP)$ of $\APX$ induces a model of $\APX_1$ with the same universe and interpretation for $\PV(\P)$ terms. In particular, a model of $\APX_1$ is said to be a \emph{standard} model if it is derived from any standard model of $\APX$. A first-order sentence $\varphi$ in the language of $\APX_1$ is said to be a \emph{true sentence} if it is true in \emph{any} standard model. We provide two examples:  
\begin{compactitem}
\item Let $C_{n,m}\equiv 0$ be a constant circuit that takes $(x,y)\in\zo^{n}\times \zo^m$. The sentence 
\[ 
\forall n\in\Log~\forall m\in\Log~\forall x\in\zo^n~\forall\delta^{-1}\in\Log~\P_\delta(C(x,\cdot)) \le 2\delta
\] 
in suitable formalization, is a true sentence as it holds when $\P_\delta(\cdot)$ is interpreted as any valid approximate counting oracle with additive error $\delta$. 
\item For the same circuit $C_{n,m}$, the sentence 
\[
\forall n\in\Log~\forall m\in\Log~\forall x\in\zo^n~\forall\delta^{-1}\in\Log~\P_\delta(C(x,\cdot)) = 0
\]
is true in the exact standard model, but is not true in the standard model where $\P_\delta(C(x,\cdot))\eqdef \delta$ and $\P_\delta(D)\eqdef 0$ when $D\ne C(x,\cdot)$. Therefore it is not a true sentence. 
\end{compactitem}

\begin{proposition}[Soundness of $\APX_1$]\label{prop: soundness APX1}
Any provable sentence $\varphi$ in $\APX_1$ is a true sentence. 
\end{proposition}

\begin{proof}
    This can be verified by induction on the proof. 
\end{proof}

\newcommand{\refmeta}[1]{\emphmeta{\ref{#1}}}

\section{Probabilistic Reasoning in \texorpdfstring{$\APX_1$}{APX1}}
\label{sec: basic properties}

In this section, we prove meta-theorems that exhibit the robustness of the approximate counting function in $\Apx_1$ and develop basic concepts such as (approximate) expectation and variance for feasibly defined random variables. 

\subsection{Consistency of Approximate Counting}

We now state a couple of meta-theorems indicating that the approximate counting functionality provided in $\Apx$ is consistent in a strong sense.

\subsubsection{Global Consistency of Approximate Counting}

\paragraph{Monotonicity.} Suppose that there are two circuits $C_1,C_2:\zo^n\to\zo$ satisfying that $C_1(x) \le C_2(x)$ for every $x\in\zo^n$. Then the acceptance probability of $C_1(x)$ is at most that of $C_2(x)$. Therefore, if $\P$ is a function for approximate counting, the acceptance probability of $C_1$ reported by $\P$ should be no larger than the reported acceptance probability of $C_2$ plus twice the precision of counting. Formally:

\dummylabel{monotone}{Monotonicity of Approximate Counting}
\begin{restatable}{lemma}{ApxMono}\label{lmm: apx monotone}
    $\Apx_1$ proves that 
    \begin{align*}
    & \forall n\in\Log~\forall C_1,C_2\in B_n~\forall\delta^{-1}\in\Log~\forall \beta^{-1}\in\Log \\ 
    & \quad \big((\forall x\in\zo^n~C_1(x)\le C_2(x)\to \P_\delta(C_1)\le \P_{\delta}(C_2)+2\cdot \delta+\beta\big).
    \end{align*}
\end{restatable}

\newcommand{\pre}{\mathsf{pre}}
\newcommand{\suf}{\mathsf{suf}}

\begin{proof}
We argue in $\Apx_1$. Fix $\vec y$ and $n,\delta^{-1},\beta^{-1}\in\Log$ and circuits $C_1,C_2\in B_n$. Suppose that $\forall x~C_1(x)\le C_2(x)$. We will  prove that 
\[
\P_\delta(C_1)\le \P_{\delta}(C_2)+2\cdot \delta+\beta.
\]
Let $\eta^{-1}\in\Log$ be a parameter to be determined later, and $C_1^{k,x},C_2^{k,x}:\zo^{n-k}\to\zo$ be the circuits obtained by fixing the rightmost $k$ bits of $C_1,C_2$ by $x\in\zo^k$, respectively.  

We will prove that $\P_\eta(C_1)\le P_{\eta}(C_2)+6\cdot n\cdot \eta$. This suffices as we can pick $\eta=1/(100\cdot n\cdot \beta)$ and apply the \refaxiom{precision}. 

Towards a contradiction, assume that $\P_\eta(C_1)> P_{\eta}(C_2)+6\cdot n\cdot \eta$. We will design a $\P$-oracle algorithm that, for any such $C_1,C_2$ and $k\le n$, outputs a string $x$ of length $k$ that satisfies the invariant
\[ 
\P_{\eta}(C_1^{k,x}) > \P_{\eta}(C_2^{k,x}) + 6\cdot (n-k)\cdot \eta.
\] 
Moreover, the correctness of the algorithm can be proved in $\Apx_1$. The algorithm is an iterative algorithm that considers $k=0,1,\dots,n$:
\begin{compactitem}
\item For $k=0$, the algorithm outputs $\eps$. This is correct as $C_1^{0,\eps}=C_1$, $C_2=C_2^{0,\eps}$, and as a consequence the required statement follows from the assumption that $\P_\eta(C_1)> P_{\eta}(C_2)+6\cdot n\cdot \eta$. 
\item Now suppose the algorithm could output a string $x\in\zo^k$ such that 
$$\P_{\eta}(C_1^{k,x}) > \P_{\eta}(C_2^{k,x}) + 6\cdot (n-k)\cdot \eta.$$
Our goal is to output a string $x'\in\zo^{k+1}$ such that 
\[ 
\P_{\eta}(C_1^{k+1,x'}) > \P_{\eta}(C_2^{k+1,x'}) + 6\cdot (n-k-1)\cdot \eta. 
\] 
Note that by the \refaxiom{local}\footnote{Here, the parameter $\beta$ in the Local Consistency Axiom is set to $\eta$.}, we know that 
\begin{align*}
& \P_{\eta}(C_1^{k,x}) \ge (1/2)\cdot (\P_{\eta}(\Fix(C_1^{k,x},0)) + \P_{\eta}(\Fix(C_1^{k,x},1))) - 3\cdot \eta; \\ 
& \P_{\eta}(C_2^{k,x}) \le (1/2)\cdot (\P_{\eta}(\Fix(C_2^{k,x},0)) + \P_{\eta}(\Fix(C_2^{k,x},1))) + 3\cdot \eta.
\end{align*}
Subsequently, there must be $\sigma\in\zo$ such that $\P_\eta(\Fix(C_1^{k,x},\sigma)) > \P_\eta(\Fix(C_2^{k,x},\sigma))+6\cdot (n-k-1)\cdot \eta$. The algorithm queries the $\P$-oracle, finds such $\sigma\in\zo$, and outputs $x'\eqdef \sigma\circ x$. This satisfies the invariant, as $\Fix(C_1^{k,x},\sigma)=C_1^{k,\sigma\circ x}$ and $\Fix(C_2^{k,x},\sigma)=C_2^{k,\sigma\circ x}$. 
\end{compactitem}

\vspace{0.2cm}

It is clear that the correctness of the algorithm can be proved in $\Apx_1$ using induction for open formulas.\footnote{Formally, we use $n$-induction for $n=1$, and rely on the fact that open formulas are expressive enough, i.e., the language of $\Apx_1$ includes the function symbol $\P$ and oracle polynomial-time functions with access to $\P$.} It follows that as we fix $k=n$, the algorithm provably outputs $x^\star\in\zo^n$ such that $\P_{\eta}(C_1^{n,x^\star}) > \P_{\eta}(C_2^{n,x^\star})$, under the assumption that $\P_\eta(C_1)> P_{\eta}(C_2)+6\cdot n\cdot \eta$. Note that $C_1^{n,x}$ and $C_2^{n,x}$ are circuits that do not read their inputs and output $C_1(x^\star)$ and $C_2(x^\star)$, respectively. This violates the \refaxiom{boundary}, since $C_1(x)\le C_2(x)$ for every $x \in \{0,1\}^n$. 
\end{proof}

\paragraph{Global Consistency.} A corollary of monotonicity is that if two circuits are provably identical, their acceptance probabilities given by the oracle $\P$ should not differ significantly. Formally: 

\dummylabel{global}{Global Consistency}

\begin{lemma}\label{lmm: global consistency}
    $\Apx_1$ proves that 
    \begin{align*}
    & \forall n\in\Log~\forall C_1,C_2\in B_n~\forall\delta^{-1},\beta^{-1}\in\Log~\\ 
    & \quad \big((\forall x\in\zo^n~C_1(x)=C_2(x))\to |\P_\delta(C_1)-\P_{\delta}(C_2)|\le 2\cdot \delta+\beta\big).
    \end{align*}
\end{lemma}

\begin{proof}
    We can prove by \Cref{lmm: apx monotone} that if $C_1(x)=C_2(x)$ for any $x\in\zo^n$, it follows that $\P_\delta(C_1)-\P_{\delta}(C_2)\le 2\cdot\delta+\beta$ and $\P_\delta(C_2)-\P_{\delta}(C_1)\le 2\cdot\delta+\beta$. Subsequently, we will have $|\P_\delta(C_1)-\P_{\delta}(C_2)|\le 2\cdot \delta+\beta$ as long as the absolute value function is properly defined. 
\end{proof}

\begin{remark}
    The standard way to formalize approximate counting for a polynomial-time decidable property in $\Apx_1$ is to first translate it to a circuit $C$ then query $\P_\delta(C)$. The global consistency property shows that the approximate counting oracle $\P$ is robust with respect to the translation of $\PV$ functions into circuits, provided that we can prove in $\Apx_1$ that the translation is functionally correct. The latter can be done already in $\PV_1$ (see, e.g., \citep[Section 2.4]{thesis_Pich}).
\end{remark}

\subsubsection{Permutational Symmetry of Approximate Counting}

Next we show that the approximate counting oracle $\P$ is \emph{permutational symmetric}, in the sense that a permutation of input variables does not change the acceptance probability given by $\P$ significantly. 

\newcommand{\Swap}{\mathsf{Swap}}

\paragraph{Local Symmetry.} As a first step, we show that swapping two adjacent input bits of a circuit does not change the acceptance probability significantly. Concretely: 
\begin{restatable}{lemma}{LocalSym}\label{lmm: local symm}
    $\Apx_1$ proves that
    \[
    \forall n\in\Log~\forall i\in[n-1]~\forall C\in B_n~\forall\delta^{-1},\beta^{-1}\in\Log~|\P_\delta(C) - \P_\delta(\Swap(C,i))| \le 2\cdot \delta + \beta, 
    \]
    where $\mathsf{Swap}(C,i)$ is a $\PV$-function that outputs a circuit $C'$ obtained by swapping the $i$-th and the $(i+1)$-th input bits (from the rightmost bit) of $C$.  
\end{restatable}

\begin{proof}
We argue in $\Apx_1$. Fix $n\in\Log$, $i\in[n-1]$, $C\in B_n$, $\delta^{-1}\in\Log$, and $\beta^{-1}\in\Log$. Let $\eta^{-1}\in\Log$ be a parameter to be determined later. 

For $C\in B_n$ and $|x|<n$, we define $\Fix(C,x)$ as the $\PV$ function that outputs a circuit obtained by fixing the last $|x|$ input bits of $C$ to $x$, i.e., it outputs $C_x\in B_{n-|x|}$ such that $C_x(u) = C(u\circ x)$. Note that for properly defined $\PV$ functions $\Swap$ and $\Fix$, we can prove in $\PV$ that if $|x| = i-1$, $\sigma=(\sigma_1,\sigma_2)\in\zo^2$, $C_{x,\sigma}^{12} \eqdef \Fix(C,\sigma_1\circ\sigma_2\circ x)$ is functionally equivalent to $C_{x,\sigma}^{21}\eqdef \Fix(\Swap(C,i),\sigma_2\circ\sigma_1\circ x)$. Moreover, we may assume that it is provable in $\PV$ that for any circuit $C\in B_n$, $x\in\zo^k$, and $z\in\zo^{n-k-1}$, $\Eval(\Fix(C,x), s_i(z)) = \Eval(\Fix(C,s_i(x)),z)$, and that $\Fix(C,\eps) = C$. 

Therefore, by the \refmeta{global} of approximate counting, we have that 
\begin{equation}
\forall x\in\zo^{i-1}~\forall\sigma\in\zo^2~|\P_{\eta}(C_{x,\sigma}^{12}) - \P_{\eta}(C_{x,\sigma}^{21})| \le 3\cdot\eta.\label{equ: fixing sigma 12} 
\end{equation}
That is, for $x\in\zo^{i-1}$, if we arbitrarily fix the \emph{rightmost} two bits of $\Fix(C,x)$ and $\Fix(\Swap(C,i),x)$, their acceptance probabilities are close. By applying \refaxiom{local} twice and subsequently the \refmeta{global} of approximate counting, we can prove that
\begin{align*} 
& \forall x\in\zo^{i-1}~\left|\P_{\eta}(\Fix(C,x))-(1/4)\sum_{\sigma\in\zo^2} \P_{\eta}(C^{12}_{x,\sigma})\right| \le 10\cdot\eta; \\ 
& \forall x\in\zo^{i-1}~\left|\P_{\eta}(\Fix(\Swap(C,i),x))-(1/4)\sum_{\sigma\in\zo^2} P_{\eta}(C^{21}_{x,\sigma})\right| \le 10\cdot\hat \eta.
\end{align*}
Subsequently, we know from \Cref{equ: fixing sigma 12} that
\begin{equation}
\forall x\in\zo^{i-1}~|\P_{\eta}(\Fix(C,x)) - \P_{\eta}(\Fix(\Swap(C,i),x))| \le 20\cdot\eta.\label{equ: swapping expanding twice} 
\end{equation}
This shows that the acceptance probabilities of $C$ and $\Swap(C,i)$ are close when the rightmost $i-1$ bits of them are both fixed by $x\in\zo^{i-1}$. 

We will now prove that 
\begin{equation}
|\P_\eta(C) - \P_\eta(\Swap(C,i))| \le 20\cdot (n+1)\cdot \eta.\label{equ: swap close eta}
\end{equation}
This suffices as we can pick $\eta\eqdef \beta/(100\cdot(n+1))$ and apply the \refaxiom{precision}. 

Suppose, towards a contradiction, that \Cref{equ: swap close eta} does not hold. We design an iterative algorithm that given $C$, $i$, and $k\le i-1$, outputs a string $x$ of length $k$ such that 
\[
|\P_\eta(\Fix(C,x)) - \P_\eta(\Fix(\Swap(C,i),x))| > 20\cdot (n+1-k)\cdot \eta.
\]
The algorithm is essentially the same as the algorithm in the proof of \Cref{lmm: apx monotone}, i.e., it extends the string by one bit in each iteration by querying the approximate counting oracle. In particular, the base case $k=0$ is satisfied, as \Cref{equ: swap close eta} does not hold. Therefore, for $k=i-1$, the algorithm outputs a string $x\in\zo^{i-1}$ such that 
\[
|\P_\eta(\Fix(C,x)) - \P_\eta(\Fix(\Swap(C,i),x))| > 20\cdot (n+1-i)\cdot \eta \ge 20\cdot \eta. 
\]
This violates \Cref{equ: swapping expanding twice} and thus completes the proof. 
\end{proof}

\paragraph{Permutational Symmetry.} We can then state and prove the \emph{permutational symmetry} of approximate counting by decomposing a permutation into a sequence of transformations $C\mapsto \Swap(C,i)$. 

\newcommand{\Permute}{\mathsf{Permute}}

We assume a straightforward encoding of permutations of $[n]$ for $n\in\Log$, and write $\pi\in S_n$ as an abbreviation of ``$\pi$ is a permutation of $[n]$'' encoded by a straightforward $\PV$ function. Let $C\in B_n$ and $\pi\in S_n$ be a permutation of $[n]$. We define $\Permute(C,\pi)$ be the $\PV$ function that outputs a circuit $C\circ\pi\in B_n$ defined as $(C\circ\pi)(x) = C(x_{\pi_n}\circ\dots\circ x_{\pi_1})$. Then we have that: 

\dummylabel{permutation}{Permutational Symmetry}
\begin{lemma}\label{lmm: permutation approx}
    $\Apx_1\vdash \forall n\in\Log~\forall \pi\in S_n~\forall C\in B_n~\forall \delta^{-1},\beta^{-1}\in\Log~|\P_\delta(C)-\P_\delta(C\circ\pi)|\le 2\cdot\delta+\beta$. 
\end{lemma}

\begin{proof}[Proof Sketch]
    Under a straightforward encoding of permutations of $[n]$, we can prove in $\PV$ that there is a list $\ell=(i_1,\dots,i_k)$ for some $k\in\Log$ such that $C\circ \pi$ is functionally equivalent to $C_k$ defined as 
    \[
    C_0 \eqdef C, \quad C_j \eqdef \Swap(C_{j-1}, i_j) \quad (j\in[k]). 
    \]
    By induction on $j$, we can prove by applying \Cref{lmm: local symm} that for any $\eta\in\Log$, $|\P_{\eta}(C)-\P_{\eta}(C_j)|\le 3\cdot j\cdot \eta$. This, together with the \refmeta{global} of approximate counting, implies that 
    \[
    |\P_{\eta}(C)-\P_{\eta}(C\circ\pi)|\le 3\cdot (k+1)\cdot \eta. 
    \]
    We then prove the lemma by taking $\eta\eqdef \beta/(10\cdot (k+1))$ and applying the \refaxiom{precision}. 
\end{proof}

\subsubsection{Existence Lemma for Approximate Counting}

An important counting principle is that if a mathematical object can be sampled with non-zero probability, then it must \emph{exist}. This simple result is the bedrock of the celebrated probabilistic method in combinatorics (see, e.g., \cite{AlonSpencer}). The following lemma formalizes the principle in the context of approximate counting: 

\dummylabel{avg}{Existence Lemma for Approximate Counting}
\begin{lemma}\label{lem:prob_method}
    $\Apx_1 \vdash \forall n,\delta^{-1},\beta^{-1}\in\Log~\forall C\in B_n ~(\beta > 0\land \P_\delta(C) > \delta + \beta \to \exists x\in\zo^n~C(x)=1$.
\end{lemma}

\begin{proof}
    We argue in $\Apx_1$. Fix any $n,\delta^{-1},\beta^{-1}\in\Log$ and $n$-input circuit $C\in B_n$. Suppose that $\beta > 0$ and $\P_\delta(C) > \delta + \beta$, and let $\eta^{-1}\in\Log$ be determined later. 
    
    By the \refaxiom{precision}, we can see that $\P_\eta(C)\ge \beta - \eta$. Suppose, towards a contradiction, that for every $x\in\zo^n$, $C(x)=0$. In such case, $C(\vec x)$ is equivalent to the circuit $\Null_n$ that always outputs $0$. By the \refmeta{global} of approximate counting, we have that 
    \[
    \P_\eta(\Null_n) \ge \P_\eta(C) - 3\eta \ge \beta - 4\eta. 
    \]
    Let $\eta = \beta /10$. We have that $\P_\eta(\Null_n) > 3\eta$, which leads to a contradiction with the \refaxiom{boundary}.
\end{proof}

We note that a more general version of the principle will be proved in \Cref{sec: avg for expectation} following a similar but more complicated argument, which will be later used to prove the \emph{linearity of approximate expectation}.  

\subsubsection{Approximate Counting for Concrete Circuits}\label{sec:concrete_circuits}

In this subsection, we consider the behavior of the approximate counting oracle on concrete circuits: the ``less-than-$t$'' circuit that parses its input as a number and outputs $1$ if it is less than a fixed threshold, and circuits with a small number of inputs.   

\paragraph{``Less-than-$t$'' Circuits.} Let $t\in\{0,1,\dots,2^n\}$, $C_{< t}:\zo^n\to\zo$ be the circuit that parses its input as the binary encoding of a number $x\in \{0,1,\dots,2^n-1\}$ and accepts if and only if $x< t$. The following lemma shows in $\Apx_1$ that the acceptance probability of $C_{< t}$ is approximately $t/2^n$, as expected. 

\dummylabel{less than t}{Less-than-$t$ Circuits}
\begin{lemma}[Less-than-$t$ Circuits]\label{lmm: less than t}
$\Apx_1\vdash \forall n,\delta^{-1},\beta^{-1}\in\Log~\forall t\in\{0,1,\dots,2^n\}~\left|\P_{\delta}(C_{<t})-t/2^n\right|\le \delta+\beta$. 
\end{lemma}

\begin{proof}
    We argue in $\Apx_1$. Fix $n,\delta^{-1},\beta^{-1}, t\in\{0,1,\dots,2^n\}$. Let $\eta^{-1}\in\Log$ be a parameter to be determined later. Note that when $t=2^n$, $C_{<t}$ is functionally equivalent to $\True_n$, and thus the lemma immediately follows from the \refmeta{global} and \refaxiom{boundary}.\footnote{For this step to hold, we need for $C_{<t}$ to be provably  equivalent to $\True_n$. This will hold in $\PV_1$ for a natural implementation of the circuits $C_{< t}$.} In the rest of the proof, we assume $t<2^n$.  
    
    Suppose, towards a contradiction, that $\left|\P_{\delta}(C_{\le t})-t/2^n\right|>\delta+\beta$. By the \refaxiom{precision}, we have that $|\P_\eta(C_{<t})-\P_{\delta}(C_{<t})|\le \delta+2\eta$, and subsequently 
    \[ 
    \left|\P_{\eta}(C_{<t})-t/2^n\right|>\beta-2\eta. 
    \] 
    We may assume that the rightmost bit of $t$ is the most significant bit; this is without loss of generality by the \refmeta{permutation} of approximate counting. 
    
    We will design an $\P$-oracle iterative algorithm that, in the $i$-th iteration, outputs $t_i\in\{0,1,\dots,2^{n-i}-1\}$ satisfying the following condition: 
    \begin{compactitem}
    \item Let $C_i:\zo^{n-i}\to\zo$ be the circuit that parses its input as a number $x\in \{0,1,\dots,2^{n-i}-1\}$ and outputs $1$ if $x<t_i$. Then $\left|\P_{\eta}(C_i)-t_i/2^{n-i}\right|>\beta-20\eta\cdot (i+1)$.
    \end{compactitem}
    The algorithm starts with $t_0\eqdef t$ (and thus $C_0\eqdef C_{<t}$). In the $i$-th iteration, the algorithm considers the rightmost bit (i.e.~the most significant bit) of $t_i$. Recall that $\Fix(C,b)$ outputs the circuit obtained from $C$ by fixing the rightmost input bit to be $b$. 
    
    If the rightmost bit of $t_i$ is $0$, the algorithm outputs $t_{i+1}\eqdef t_i$. It is clear that $\Fix(C_i,1)$ is functionally equivalent to the $\Null_{n-i-1}$, and thus by the \refmeta{global} and \refaxiom{boundary}, $\P_{\eta}(\Fix(C_i,1))\le 3\eta$. Let $C_{i+1}$ be the circuit that parses its input as a number $x\in \{0,1,\dots,2^n-1\}$ and outputs $1$ if $x<t_{i+1}$. It follows that $\Fix(C_i,0)$ is functionally equivalent to $C_{i+1}$, and thus by the \refmeta{global} of approximate counting, $|\P_\eta(\Fix(C_i,0))-\P_\eta(C_{i+1})|\le 3\eta$. Subsequently: 
    \begin{align*}
    &~\left|\P_{\eta}(C_{i+1})-t_{i+1}/2^{n-i-1}\right| \\
    \ge &~\left|\P_\eta(\Fix(C_{i},0))-t_{i+1}/2^{n-i-1}\right| -3\eta \\ 
    \ge &~\left|\P_\eta(\Fix(C_{i},0)) + \P_\eta(\Fix(C_{i},1))-t_{i+1}/2^{n-i-1}\right| - 6\eta \\
    = &~2\cdot \left|\frac{\P_\eta(\Fix(C_{i},0)) + \P_\eta(\Fix(C_{i},1))}{2}-\frac{t_{i+1}}{2^{n-i}}\right| - 6\eta  \\ 
    \ge &~2\cdot \left|\P_\eta(C_i)-t_{i}/2^{n-i}\right| - 12\eta \tag{\refaxiom{local}} \\ 
    > &~\beta - 20\eta\cdot(i+1) - 12\eta \ge \beta - 20\eta\cdot (i+2). 
    \end{align*}

    If the rightmost bit of $t_i$ is $1$, the algorithm outputs $t_{i+1}=t_i-2^{n-i-1}$. It is clear that $\Fix(C_i,0)$ is functionally equivalent to $\True_{n-i+1}$, and thus by the \refmeta{global} and \refaxiom{boundary}, $|\P_\eta(\Fix(C_i,0))-1|\le 3\eta$. Let $C_{i+1}$ be the circuit that parses its input as a number $x\in \{0,1,\dots,2^n-1\}$ and outputs $1$ if $x<t_{i+1}$. It follows that $\Fix(C_i,1)$ is functionally equivalent to $C_{i+1}$, and thus by the \refmeta{global} of approximate counting, $|\P_\eta(\Fix(C_i,1))-\P_\eta(C_{i+1})|\le 3\eta$. Subsequently: 
    \begin{align*}
    &~\left|\P_{\eta}(C_{i+1})-t_{i+1}/2^{n-i-1}\right| \\
    \ge &~\left|\P_\eta(\Fix(C_{i},1))-t_{i+1}/2^{n-i-1}\right| -3\eta \\ 
    \ge &~\left|\P_\eta(\Fix(C_{i},0))-1 + \P_\eta(\Fix(C_{i},1))-t_{i+1}/2^{n-i-1}\right| - 6\eta \\
    = &~2\cdot \left|\frac{\P_\eta(\Fix(C_{i},0)) + \P_\eta(\Fix(C_{i},1))}{2}-\frac{t_{i+1}+2^{n-i-1}}{2^{n-i}}\right| - 6\eta  \\ 
    = & ~2\cdot \left|\frac{\P_\eta(\Fix(C_{i},0)) + \P_\eta(\Fix(C_{i},1))}{2}-\frac{t_i}{2^{n-i}}\right| - 6\eta\\
    \ge &~2\cdot \left|\P_\eta(C_i)-t_{i}/2^{n-i}\right| - 12\eta \tag{\refaxiom{local}} \\ 
    > &~\beta - 20\eta\cdot(i+1) - 12\eta \ge \beta - 20\eta\cdot (i+2). 
    \end{align*}
    
    It is clear that the correctness of the algorithm follows from the induction principle for polynomial-time verifiable properties allowed by the \refaxiom{ind} of $\Apx_1$. Therefore, after $n$ iterations, the algorithm will output $t_n=0$ such that the acceptance probability of the circuit $C_n\equiv \Null_0$ is at least $\beta-20\eta\cdot (n+1)$. By setting $\eta\eqdef \beta/(40n+40)$, we can conclude a contradiction using the \refaxiom{boundary}. 
\end{proof}

\paragraph{Circuits with Short Inputs.} For circuits $C:\zo^n\to\zo$ such that $n\in\Log\Log$ it is feasible to enumerate all inputs of $C$. The following lemma shows that $\P_\delta(C)$ is consistent with its acceptance probability computed via the brute-force algorithm. 

\dummylabel{brute force}{Brute Force Counting Lemma}
\begin{lemma}[Brute Force Counting Lemma]\label{lmm: short input}
$\Apx_1$ proves the following statement. For every $n\in\Log\Log$, circuit $C\in B_n$, and $\delta^{-1},\beta^{-1}\in\Log$, let $t$ be the number of accepting inputs of $C$. Then $|\P_\delta(C)-t/2^n|\le \delta+\beta$. In particular, if $\delta\le 2^{-n-1}$, $t$ is the nearest integer to $\P_\delta(C)\cdot 2^n$.    
\end{lemma}

The proof of the lemma employs the \refaxiom{local} and the \refaxiom{precision}. We will formalize the argument using the following general tool: If there is a sequence of circuits serving as an approximating counting algorithm for a circuit $C$, in the sense that it satisfies the boundary condition and is locally consistent, then the acceptance probability estimated by the algorithm is necessarily close to $\P_\delta(C)$. Formally:  

\dummylabel{dueling}{Dueling Lemma}
\begin{lemma}[Dueling Lemma]
$\Apx_1$ proves the following statement. Let $n\in\Log$, $C\in B_n$ be a circuit, and $P_0,P_1,\dots,P_n$ be circuits that output rational numbers such that $P_i$ is of input length $i$. Let $\delta^{-1},\eta^{-1}\in\Log$. Suppose that for every $i<n$ and every $x\in\zo^i$,  
\[ 
\left|P_i(x) - \frac{P_{i+1}(x\circ 0)+P_{i+1}(x\circ 1)}{2}\right|\le \eta; 
\]
and that for every $x\in\zo^n$, $P_n(x)=C(x)$. Then $|\P_\delta(C)-P_0|\le \delta + 4\eta\cdot (n+1)$. 
\end{lemma}

\begin{proof}
    We argue in $\Apx_1$. Fix $n\in\Log$, $C\in B_n$, the circuits $P_0,P_1,\dots,P_n$, and $\delta^{-1},\eta^{-1}\in\Log$. Suppose that it satisfies the two conditions in the lemma. Assume for contradiction that $|\P_\delta(C)-P_0|>\delta + 3\eta\cdot (n+1)$. By the \refaxiom{precision}, we have that $|\P_\eta(C)-P_0|> 4\eta\cdot n+3\eta$. 
    
    Let $C^x$ be the circuit obtained by fixing the rightmost $|x|$ bits of $C$ to be $x$. Consider the following algorithm that, in the $i$-th iteration, outputs $x_i\in\zo^i$ such that $|\P_\eta(C^{x_i})-P_i(x_i)| > 4\eta \cdot (n-i)+3\eta$. The algorithm starts with $x_0\eqdef \eps$. In the $i$-th iteration, we know by the \refaxiom{local} that 
    \begin{align*}
    |\P_\eta(C^{x_i})-P_i(x_i)| &\le \left|\frac{\P_\eta(C^{0\circ x_i})+\P_\eta(C^{1\circ x_i})}{2}-P_i(x_i)\right|+3\eta \\ 
    &\le \left|\frac{\P_\eta(C^{0\circ x_i})+\P_\eta(C^{1\circ x_i})}{2}-\frac{P_{i+1}(0\circ x_i) + P_{i+1}(1\circ x_i)}{2}\right| + 4\eta \\ 
    &\le \frac{1}{2}\Big(|\P_\eta(C^{0\circ x_i})-P_{i+1}(0\circ x_i)| + |\P_\eta(C^{1\circ x_i})-P_{i+1}(1\circ x_i)|\Big) + 4\eta. 
    \end{align*}
    This means that for some $\sigma\in\zo$, $|\P_\eta(C^{\sigma\circ x_i})-P_{i+1}(\sigma\circ x_i)|> 4\eta\cdot (n-i-1)+3\eta$. The algorithm then proceeds by setting $x_{i+1}=\sigma\circ x_i$. 
    
    It is clear that the correctness of the algorithm can be proved by induction on a quantifier-free formula, which is available in $\Apx_1$. Therefore, in the $n$-th iteration, the algorithm outputs a string $x_n\in\zo^n$ such that $|\P_\eta(C^{x_n})-P_n(x_n)| > 3\eta$. However, this violates the \refaxiom{boundary} as $P_n(x_n)=C(x_n)$, and the circuit $C^{x_n}$ is a constant circuit that outputs $C(x_n)$. 
\end{proof}

\begin{proof}[Proof of \Cref{lmm: short input}]
    We argue in $\Apx_1$. Fix $n\in\Log\Log$ and $C\in B_n$, $\delta^{-1},\beta^{-1}\in\Log$, and let $t$ be the number of accepting inputs of $C$. Let $P_0,\dots,P_n$ be the circuits such that $P_i(x)$ takes an $i$-bit input $x$ and outputs the acceptance probability of $C^x:\zo^{n-i}\to\zo$ defined as $C^x(z)=C(z\circ x)$. In particular, $P_0=t/2^n$. Let $\eta^{-1}\in\Log$ be determined later. 
    
    Since $n \in \Log\Log$, it is provable in $\PV_1$ that $P_0,\dots,P_n$ satisfy the conditions in the \refmeta{dueling}. Then $|\P_\delta(C)-P_0|=|\P_\delta(C)-t/2^n| \le \delta + 4\eta\cdot (n+1)$. The lemma follows by setting $\eta\eqdef \beta/(4n+4)$. 
\end{proof}

\subsection{Approximate Expectation and its Basic Theory}

\newcommand{\disU}{\mathcal{U}}
\newcommand{\disD}{\mathcal{D}}
\newcommand{\calF}{\mathcal{F}}
\newcommand{\Sampler}{\mathsf{Sampler}}

We now develop a theory of feasible random variables and their approximate expectation. 

\subsubsection{Definition of Random Variables and Approximate Expectation}

We first define the approximate expectation of a discrete random variable taking values in $\mathbb{Q}$. Recall that explicit sets are sets encoded by an explicit list and, in particular, the size of explicit sets are always feasible. To have the expectation being feasibly computable (approximately), we restrict to the setting where the support of random variables are given as an \emph{explicit set}. 

Let $n\in\Log$, $V$ be an explicit set of rational numbers, and $C$ be a multi-output circuit such that $\Apx_1$ proves that
\[
\forall x\in\zo^n~C(x)\in V. 
\]
We  say that $(V,n,C)$ defines a random variable $X\eqdef C(\disU_n)$, and define the expectation of $X$ as 
\[
\sum_{v\in V} v\cdot\Pr_x[C(x)=v],
\]
where the probability can be implemented by the approximate counting quantifier $\P$ in $\Apx$. This leads to the following formal definition of random variables and expectation. 

\begin{definition}[Random Variable]
Let $V$ be an explicit set of rational numbers, $n\in\Log$, and $C$ be a multi-output circuit. We say that $(V,n,C)$ defines a \emph{random variable} $X$ over $V$ if $\forall x\in\zo^n~C(x)\in V$. The set $V$ is called the \emph{support} of $X$, $C$ is called the \emph{sampler} of $X$, and $n$ is called the \emph{seed length} of $X$. 
\end{definition}

\begin{definition}[Approximate Expectation]
Let $(V,n,C)$ be a tuple defining a random variable $X$ over $V$, and $\delta^{-1}\in\Log$. We define the \emph{approximate expectation} of $X$, denoted by $\Ex_\delta [X]$, as 
\[
\sum_{v\in V} v\cdot \P_\delta(C_v), 
\]
where $C_v$ is the $n$-input circuit that given $x\in\zo^n$, output $1$ (resp.~$0$) if $C(x)=v$ (resp.~$C(x)\ne v$), ``$\cdot$'' denotes the multiplication of rational numbers, and $\sum$ denotes the summation of rational numbers.
\end{definition}

We note that there is a $\PV(\P)$ function $\E(V,n,C,\Delta)$ computing $\bbE_{|\Delta|^{-1}}[X]$ for the random variable $X$ defined from $(V,n,C)$; it enumerates over $v\in V$, constructs the circuit $C_v$, calls the oracle $p_v\gets \P(C_v,\Delta)$, and sums over $v\cdot p_v$ for all $v\in V$. To see that this algorithm is feasible, notice that $V$ is an explicit set of feasible size, and under the encoding specified in \Cref{sec: def: notation}, the total length of all rational numbers in $V$ is $\PV$-provably feasible. 

For simplicity, we will use the notation $C:\zo^n\to\bbQ$ to denote that $C$ is a multi-output circuit whose output is parsed as a rational number.

Moreover, one may think of the acceptance probability of a circuit $C:\zo^n\to\zo$ as the expectation of the indicating random variable $I_C \eqdef (\{0,1\}, n, C)$ up to a small additive error, as shown by the following proposition. Therefore, the properties of expectation we will prove next also translate to properties of approximate counting.

\begin{proposition}\label{prop: indicator variable}
$\Apx_1$ proves the following statement. Let $n,\delta^{-1}\in\Log$ and $C:\zo^n\to\zo$ be a Boolean circuit. Let $I_C\eqdef (\zo, n,C)$ be the indicator random variable for $C(x) = 1$. Then for any $\beta^{-1}\in\Log$, $|\P_\delta(C) - \bbE_\delta[I_C]| \le 2\delta + \beta$.
\end{proposition}
\begin{proof}
We argue in $\Apx_1$. Fix $n,\delta^{-1},\beta^{-1}\in\Log$ and a circuit $C:\zo^n\to\zo$. By the definition of approximate expectation, we know that $\bbE_\delta[I_C] = 1 \cdot \P_\delta[C_1]$, where $C_1(x), C(x)$ are functionally equivalent circuits. By the \refaxiom{global} of approximate counting, we have that 
\[
|\P_\delta(C) - \bbE_\delta[I_C]| = |\P_\delta(C) - \P_\delta[C_1]| \le 2\delta + \beta. \qedhere
\]
\end{proof}

\subsubsection{Basic Properties of Approximate Expectation}

\paragraph{Precision Consistency.} Similar to the \refaxiom{precision} for approximate counting, the definition of approximate expectation is consistent with respect to different precisions as shown in the proposition below.

\dummylabel{precision exp}{Precision Consistency of Expectation}
\begin{proposition}\label{prop: precision consistency exp}
    $\Apx_1$ proves the following statement. Let $n,\delta_1^{-1},\delta_2^{-1}\in\Log$, $C:\{0,1\}^n\to \bbQ$, $V\subseteq\bbQ$ be an explicit set such that $\forall x\in\zo^n~C(x)\in V$. Let $X$ be the random variable defined by $(V,n,C)$. Then for every $\beta^{-1}\in\Log$, 
    \[
    \left|\bbE_{\delta_1}[X] - \bbE_{\delta_2}[X]\right| \le (\delta_1+\delta_2+\beta)\cdot \|V\|,
    \]
    where $\|V\|\eqdef \sum_{v\in V}|v|$ is the $\ell_1$-norm of $V$ and $|v|$ denotes the absolute value of the rational $v$. 
\end{proposition}

\begin{proof}
We argue in $\Apx_1$. Recall that $C_v$ is the circuit that outputs $1$ if $C(x)=v$, and outputs $0$ otherwise. By the definition of approximate counting, we can see that 
\begin{align*}
\left|\bbE_{\delta_1}[X] - \bbE_{\delta_2}[X]\right| &= 
\left|\sum_{v\in V}v\cdot \P_{\delta_1}(C_v) - \sum_{v\in V}v\cdot \P_{\delta_2}(C_v)\right| \\ 
&\le \left|\sum_{v\in V} v\cdot (\P_{\delta_1}(C_v) - \P_{\delta_2}(C_v)\right| \\ 
&\le \left|\sum_{v\in V} v\cdot (\delta_1+\delta_2+\beta)\right| \le (\delta_1+\delta_2+\beta)\cdot \|V\|,
\end{align*}
where the second last inequality follows from the \refaxiom{precision}. 
\end{proof}

\paragraph{Local Consistency.} Similarly, we can prove that approximate expectation is locally consistent by fixing the rightmost bit of the random seed to be $0$ or $1$ randomly. Formally: 

\dummylabel{local exp}{Local Consistency of Expectation}
\begin{proposition}\label{prop: local consistency exp}
    $\Apx_1$ proves the following statement. Let $n,\delta^{-1}\in\Log$, $C:\zo^n\to\bbQ$, $V\subseteq\bbQ$ be an explicit set such that $\forall x\in\zo^n~C(x)\in V$. Let $X$ be the random variable defined by $(V,n,C)$. Then for every $\beta^{-1}\in\Log$, 
    \[
    \left|\bbE_{\delta}[X] - \frac{\bbE_{\delta}[X|_0] + \bbE_{\delta}[X|_1]}{2}\right| \le (2\delta+\beta)\cdot \|V\|,
    \]
    where $\|V\|\eqdef \sum_{v\in V}|v|$ is the $\ell_1$ norm of $V$, and for $b\in\zo$, $X|_b$ denotes the random variable defined by $(V,n-1,\Fix_b(C))$. 
\end{proposition}

\begin{proof}
We argue in $\Apx_1$. Let $\eta^{-1}\in\Log$ be determined later. Recall that $\Fix_b(C)$ outputs the circuit obtained by fixing the rightmost input bit of $C$ to be $b$, where $b\in\zo$. By the \refaxiom{local}, we know that for every $v \in V$, 
\[
\left|\P_\eta(C_v) - \frac{\P_\eta(\Fix(C_v,0))+\P_\eta(\Fix(C_v,1))}{2}\right| \le 3\eta. 
\]
By the definition of approximate counting, we can calculate that
\begin{align*}
& ~\left|\bbE_{\eta}[X] - \frac{\bbE_{\eta}[X|_0] + \bbE_{\eta}[X|_1]}{2}\right| \\ 
=& ~\left|\sum_{v\in V}v\cdot \P_\eta(C_v)- \frac{1}{2}\left(\sum_{v\in V}v\cdot \P_\eta((\Fix_0(C))_v) + \sum_{v\in V}v\cdot \P_\eta((\Fix_1(C))_v)\right)\right| \\ 
\le & ~\left|\sum_{v\in V}v\cdot \P_\eta(C_v)- \frac{1}{2}\left(\sum_{v\in V}v\cdot \P_\eta(\Fix_0(C_v)) + \sum_{v\in V}v\cdot \P_\eta(\Fix_1(C_v))\right)\right| + 3\eta \cdot \|V\|\tag{\refmeta{global}}\\ 
=& ~\sum_{v\in V}|v|\cdot \left|\P_\eta(C_v) - \frac{\P_\eta(\Fix(C_v,0))+\P_\eta(\Fix(C_v,1))}{2}\right| + 3\eta\cdot \|V\| \\ 
\le& ~6\eta\cdot \|V\|.
\end{align*}
The result follows from the \refmeta{precision exp} by taking $\eta\eqdef \beta/10$. 
\end{proof}

\paragraph{Consistency in Support Extension.} Suppose that $X$ is a random variable defined by the tuple $(V,n,C)$. Consider an explicit set $V'$ such that $V\subseteq V'$. We can define another random variable $X'$ that is essentially the same as $X$ by considering the tuple $(V',n,C)$. The following proposition shows that the expectation of $X'$ and $X$ are nearly the same, i.e., a support extension does not affect the expectation of a random variable significantly.

\dummylabel{support ext}{Consistency in Support Extension}
\begin{proposition}\label{prop: support extension exp}
$\Apx_1$ proves the following statement. Let $n,\delta^{-1},\beta^{-1}\in\Log$, $C:\zo^n\to\bbQ$, $V,V'\subseteq\bbQ$ be explicit sets such that $\forall x\in\zo^n~C(x)\in V\subseteq V'$. Let $X$ be the random variable defined by $(V,n,C)$, and $X'$ be the random variable defined by $(V', n,C)$. Then:
\[
\left|\bbE_\delta[X] - \bbE_\delta[X']\right| \le (2\delta+\beta) \cdot \|V'\setminus V\|\le (2\delta+\beta)\cdot \|V'\|. 
\]
where $\|V'\setminus V\| = \sum_{v\in V'\setminus V}|v|$ is the $\ell_1$-norm of $V'\setminus V$.
\end{proposition}

\begin{proof}
We argue in $\Apx_1$. By the definition of approximate expectation, we know that 
\begin{align}
\left|\bbE_\delta[X] - \bbE_\delta[X']\right| = \left|\sum_{v\in V'\setminus V} v\cdot \P_\delta(C_v)\right| \le \sum_{v\in V'\setminus V} |v| \cdot \P_\delta(C_v), 
\end{align}
where $C_v(x)$ outputs $1$ if $C(x)=v$, and $0$ otherwise. Therefore it suffices to prove that $\P_\delta(C_v) \le 2\delta+\beta$ for $v\in V'\setminus V$. Fix any $v\in V'\setminus V$. Note that since $C(x)\in V$ for $x\in\zo^n$, we know that $C(x)\ne v$ and thus $C_v$ is (provably) functionally equivalent to $\Null_n$. The desired bound then follows from the \refmeta{global} of approximate counting using that $\P_\delta(\Null_n)=0$ by the \refaxiom{boundary}. 
\end{proof}

\paragraph{Permutational Symmetry.} Suppose that $X,X'$ are random variables defined by the tuples $(V,n,C)$ and $(V,n,C\circ\pi)$, where $\pi\in S_n$ denotes a permutation of the input bits. Similar to the \refmeta{permutation} of approximate counting, we will show that $\bbE[X]\approx \bbE[X']$. 

\dummylabel{permutation exp}{Permutational Symmetry of Expectation}
\begin{proposition}\label{prop: permutation exp}
$\Apx_1$ proves the following statement. Let $n,\delta^{-1},\beta^{-1}\in\Log$, $C:\zo^n\to\bbQ$, $V\subseteq\bbQ$ be an explicit set such that $\forall x\in\zo^n~C(x)\in V$. Let $\pi\in S_n$ be a permutation of the input bits. Let $X,X'$ be the random variables defined by $(V,n,C)$ and $(V,n,C\circ\pi)$, respectively. Then: 
\[
|\bbE_\delta[X] - \bbE_\delta[X']|\le (2\delta+\beta)\cdot \|V\|,
\]
where $\|V\|$ is the $\ell_1$-norm of $V$. 
\end{proposition}

\begin{proof}
    We argue in $\Apx_1$. Let $\eta^{-1}\in\Log$ be determined later. We can calculate that  
    \begin{align*}
    |\bbE_\eta[X]-\bbE_\eta[X']| &\leq \sum_{v\in V}|v|\cdot\left|\P_\eta(C_v)-\P_\eta((C\circ\pi)_v)\right|\\
    &\le \sum_{v\in V}|v|\cdot \left|\P_\eta(C_v)-\P_\eta(C_v\circ \pi)\right|+3\eta\cdot \|V\| \\ 
    &\le 6\eta\cdot \|V\|.
    \end{align*}
    Here, the second line follows from \refmeta{global} of approximate counting, and the third line follows from the \refmeta{permutation} of approximate counting. Subsequently, the proposition follows from the \refmeta{precision exp} by taking $\eta\eqdef \beta/10$.  
\end{proof}

\subsubsection{Averaging Argument for Expectation}
\label{sec: avg for expectation}

We will prove a general version of the \emph{averaging argument} that allows us to search for a suffix of the seed such that the given linear combination of expectations of random variables $X_1,\dots,X_m$ is approximately preserved after fixing part of the seed.   

Suppose that $X_1,\dots,X_m$ are random variables supported on $V$ defined by a sequence of circuits $C_1,\dots,C_m$, each with seed length $n$, i.e., for each $i\in[m]$ and every $x\in\zo^n$, $C_i(x)\in V$. Let $\delta^{-1}\in\Log$. Let $\lambda_1,\dots,\lambda_m\in\bbQ$ be coefficients, and consider the quantity 
\begin{equation}
\mu_{n,m,\delta,\vec\lambda}\eqdef\lambda_1\cdot \bbE_\delta[X_1] + \lambda_2\cdot \bbE_\delta[X_2] + \dots + \lambda_m \cdot \bbE_\delta[X_m]. \label{equ: def mu}
\end{equation}
Let $z\in\zo^k$ for $k\in[n]$. We can define the random variable $X_i|_z$ for each $i\in[m]$ from $(V,n-k,\Fix(C_i,z))$, where $\Fix(C_i,z)$ outputs the circuit obtained from $C_i$ by fixing the rightmost $k$ bits to $z$. That is, $X_i|_z$ is the random variable obtained by fixing the last $k$ input bits of $C_i$ to be $z$.\footnote{This is without loss of generality by the \refmeta{permutation exp}.} Let $\mu_{n,m,\delta,\vec\lambda}|_z$ be the quantity
\begin{equation}
\mu_{n,m,\delta, \vec\lambda}|_z\eqdef \lambda_1\cdot \bbE_\delta[X_1|_z] + \lambda_2\cdot \bbE_\delta[X_2|_z] + \dots + \lambda_m \cdot \bbE_\delta[X_m|_z]. \label{equ: def mu z}
\end{equation}

\begin{theorem}[Averaging Argument for Expectation]\label{thm: averaging on exp}\dummylabel{avg on exp}{Averaging Argument for Expectation}
$\Apx_1$ proves the following statement. Let $n,m,\delta^{-1}\in\Log$, $C_1,\dots,C_m:\zo^n\to\bbQ$ be circuits, $V\subseteq\bbQ$ be an explicit set such that $\forall x\in\zo^n~\forall i\in[m]~C_i(x)\in V$, and $\vec\lambda=(\lambda_1,\dots,\lambda_m)$ be a list of length $m$ such that $\lambda_i\in\bbQ$ for $i\in[m]$.  

Then for every $k\in[n]$ and $\beta^{-1}\in\Log$, there is a $z\in\zo^k$ such that 
\begin{equation}\label{equ: inequal def averaging}
\mu_{n,m,\delta,\vec\lambda}|_z \ge \mu_{n,m,\delta,\vec\lambda} - (2\delta + \beta) \cdot \|V\|\cdot\|\vec \lambda\|,
\end{equation}
where $\|V\|\eqdef \sum_{v\in V}|v|$ and $\|\vec \lambda\|\eqdef\sum_{i\in[m]}|\lambda_i|$ are the $\ell_1$-norm of $V$ and $\lambda$, respectively, and $\mu_{n,m,\delta,\vec\lambda}$ and $\mu_{n,m,\delta,\vec\lambda}|_z$ are defined by \Cref{equ: def mu} and \Cref{equ: def mu z}, respectively. 
\end{theorem}

\newcommand{\AvgSampler}{\mathsf{AvgSampler}}

\begin{proof}
We argue in $\Apx_1$. Fix $n,m,\delta^{-1}\in\Log$, $C_1,\dots,C_m:\zo^n\to\bbQ$, let $V\subseteq \bbQ$ be an explicit set, $\vec \lambda=(\lambda_1,\dots,\lambda_m)\in\bbQ$, $k\in[n]$, $\beta^{-1}\in\Log$. Recall that by definition, we have that for each $b\in\zo$,
\begin{align}
& \mu_{n,m,\eta,\vec\lambda}|_z = \lambda_1 \cdot \bbE_\delta[X_1|_z] + \dots +\lambda_m \cdot \bbE_\delta[X_m|_z]; \\ 
& \mu_{n,m,\eta,\vec\lambda}|_{z\circ b} = \lambda_1 \cdot \bbE_\delta[X_1|_{z\circ b}] + \dots +\lambda_m \cdot \bbE_\delta[X_m|_{z\circ b}].
\end{align}

We will design a $\P$-oracle polynomial-time algorithm $\AvgSampler(\pi)$ that takes 
\[ 
\pi\eqdef (1^n,1^m,1^{\delta^{-1}},C_1,\dots,C_m,V,\vec\lambda,1^k, 1^{\beta^{-1}})
\] 
as its input and outputs $z\in\zo^k$ such that \Cref{equ: inequal def averaging} holds. \Cref{thm: averaging on exp} then follows if the correctness of $\AvgSampler(\pi)$ can be proved in $\Apx_1$. 

\newcommand{\cst}{3}

$\AvgSampler(\pi)$ is an iterative algorithm on $k$ (i.e.~the length of $z$). We will prove the invariant that for any $k\in\{0,1,\dots,n\}$, the algorithm $\AvgSampler(\pi)$ outputs a string $z\in\zo^k$ such that 
\begin{equation}\label{equ: diff exp ind}
\mu_{n,m,\eta,\vec\lambda}|_z \ge \mu_{n,m,\eta,\vec \lambda} - \cst k\cdot \eta\cdot \|V\|\cdot\|\vec \lambda\|
\end{equation}
where $\eta^{-1} \in \Log$. We note that if this is possible, we can set $\eta \eqdef \beta / (\cst n+3)$ so that \Cref{equ: inequal def averaging} follows from the \refmeta{precision exp}. Specifically, we can see that 
\begin{align*}
& \left|\mu_{n,m,\eta,\vec\lambda}|_z - \mu_{n,m,\delta,\vec \lambda}|_z\right| \\
= & \left|\lambda_1\cdot(\bbE_\eta[X_1|_z]-\bbE_{\delta}[X_1|_z]) + \dots + \lambda_m\cdot (\bbE_\eta[X_m|_z]-\bbE_{\delta}[X_m|_z]) \right| \\ 
\le & \left|\lambda_1 \cdot (\delta+2\eta)\cdot \|V\| + \dots + \lambda_m \cdot (\delta+2\eta)\cdot \|V\|\right| \le (\delta+2\eta)\cdot\|\vec \lambda\|\cdot \|V\|
\end{align*}
and similarly 
\[
\left|\mu_{n,m,\eta,\vec\lambda} - \mu_{n,m,\delta,\vec \lambda}\right| \le (\delta + 2\eta)\cdot \|\vec\lambda\|\cdot\|V\|. 
\]
\Cref{equ: inequal def averaging} then follows from the triangle inequality.  

For $k=0$, $\AvgSampler(\pi)$ outputs $\eps$, and \Cref{equ: diff exp ind} holds as $\mu_{n,m,\vec y,\eta,\vec \lambda}|_z = \mu_{n,m,\vec y,\eta,\vec\lambda}$ by definition. Suppose that it has already obtained a string $z\in\zo^{k}$ such that \Cref{equ: diff exp ind} holds. Our goal is to choose a bit $b\in\zo$ such that 
\[
\mu_{n,m,\eta,\vec\lambda}|_{b\circ z} \ge \mu_{n,m,\eta,\vec \lambda} - \cst (k+1)\cdot \eta \cdot\|V\|\cdot\|\vec \lambda\|.   
\]
For each $i\in [m]$, we know by the \refmeta{local exp} that 
\[ 
\left|\bbE_\eta[X_i|_z] - \frac{\bbE_\eta[X_i|_{0\circ z}] + \bbE_\eta[X_i|_{1\circ z}]}{2}\right| \le 3\eta\cdot \|V\|
\] 
It then follows that 
\begin{align*}
  & \left|\mu_{n,m,\eta,\vec\lambda}|_z - \frac{\mu_{n,m,\eta,\vec\lambda}|_{0\circ z} + \mu_{n,m,\eta,\vec\lambda}|_{1\circ z}}{2}\right| \\ 
= & \left|\sum_{i\in[m]} \lambda_i\cdot \left(\bbE_\eta[X_i|_z] - \frac{\bbE_\eta[X_i|_{0\circ z}] + \bbE_\eta[X_i|_{1\circ z}]}{2}\right)\right| \le 3\eta \cdot\|\vec\lambda\|\cdot\|V\|. 
\end{align*}
Therefore, for some $b\in\zo$, we will have that $\mu_{n,m,\eta,\vec \lambda}|_{b\circ z} \ge \mu_{n,m,\eta,\vec \lambda}|_{z} - 3\eta \cdot\|\vec\lambda\|\cdot\|V\|$, which subsequently implies that 
\[
\mu_{n,m,\eta,\vec\lambda}|_{b\circ z} \ge \mu_{n,m,\eta,\vec \lambda}|_z - 3\eta\cdot\|\vec\lambda\|\cdot\|V\| \ge \mu_{n,m,\eta,\vec \lambda} - \cst(k+1)\cdot\eta \cdot \|\vec\lambda\|\cdot\|V\|.
\]
The algorithm $\AvgSampler$ can use the $\P$-oracle to determine $b$ and output $b\circ z$. This completes the proof.
\end{proof}

\subsubsection{Complementation}

An easy corollary of the \refmeta{avg on exp} is \emph{complementary counting}. That is, if $X$ is a random variable over $\zo$ and $\overline X\eqdef 1-X$, then $\bbE[\overline X]=1-\bbE[X]$. Formally: 

\dummylabel{complementation}{Complementary Counting}
\begin{corollary}[Complementary Counting]\label{cor:complementation}
    $\APX_1$ proves the following statement. Let $n,\delta^{-1},\beta^{-1}\in\Log$, $C_1,C_2\in B_n$ such that for every $x\in\zo^n$, $C_1(x)\ne C_2(x)$. Then $|\P_\delta(C_1)+\P_\delta(C_2)-1|\le 2\delta + \beta$. Moreover, let $X_1,X_2$ be the indicator random variables of $C_1,C_2$ over $\zo$, respectively. Then $|\bbE_\delta[X_1]+\bbE_\delta[X_2]-1|\le 2\delta + \beta$. 
\end{corollary}
\begin{proof}
    We argue in $\APX_1$. Fix $n,\delta^{-1},\beta^{-1}\in\Log$, $C_1,C_2\in B_n$. Let $\eta^{-1}\in\Log$ be determined later, and $X_1,X_2$ be the indicator random variables of $C_1$ and $C_2$, respectively. It is clear that for any total assignment $\rho$ to the seed, $\bbE_\eta[X_1|_\rho] + \bbE_\eta[X_2|_\rho] - 1 = 0$. Therefore, by \refmeta{avg on exp} with $k = n$, $\bbE_\eta[X_1] + \bbE_\eta[X_2]-1 \le 6\eta$. Similarly, we can show that $-\bbE_\eta[X_1]-\bbE_\eta[X_2]+1\le 6\eta$. This implies that 
    \begin{equation}
     \left|\bbE_\eta[X_1]+\bbE_\eta[X_2]-1\right| \le 6\eta.\label{equ: X1 and X2 close to 1}   
    \end{equation}
    By \Cref{prop: indicator variable}, we have 
    \[
    |\P_\eta(C_1)+\P_\eta(C_2)-1|\le 12\eta. 
    \]
    Subsequently, by the \refaxiom{precision}, we have $|\P_\delta(C_1)+\P_\delta(C_2)-1|\le 2\delta + 16\eta$. The desired bound then follows by setting $\eta \eqdef \beta/30$. The ``Moreover'' part follows from \Cref{equ: X1 and X2 close to 1} by the \refmeta{precision exp}.  
\end{proof}

\subsubsection{Linearity of Expectation} 

We are now ready to prove the (approximate) linearity of expectation, one of the most useful results in probability theory. Let $X_1,\dots,X_m,Y$ be random variables over an explicit set $V$. For a random seed $z$ of the random variables, we use $X_i|_z$ and $Y|_z$ to denote the value that $X_i$ and $Y$ evaluate to, respectively. Suppose that for each random seed $z$, we have that 
\[
Y|_z = \gamma + \lambda_1\cdot X_1|_z + \lambda_2\cdot X_2|_z + \dots + \lambda_m\cdot X_m|_z,
\]
for some $\lambda_1,\dots,\lambda_m\in\bbQ$. Then we should be able to obtain that $\bbE[Y]$ is close to 
\[
\gamma + \lambda_1\cdot \bbE[X_1] + \lambda_2\cdot \bbE[X_2] + \dots + \lambda_m\cdot \bbE[X_m].
\]

Formally, we have that: 

\dummylabel{linearity}{Linearity of Expectation}
\begin{theorem}[Linearity of Expectation]\label{thm: linearity exp}
$\Apx_1$ proves the following: Let $n,m,\delta^{-1},\beta^{-1}\in\Log$, $C_1,\dots,C_m:\zo^n\to\bbQ$ be a list of circuits, $\vec\lambda = (\lambda_1,\dots,\lambda_m)$ be a list of length $m$ such that $\lambda_i\in\bbQ$ for $i\in[m]$, $\gamma\in\bbQ$, and $V\subseteq\bbQ$ be an explicit set such that: 
\begin{compactitem}
\item For any $x\in\zo^n$ and $i\in[m]$, $C_i(x)\in V$.
\item For any $x\in\zo^n$, $\gamma + \lambda_1\cdot C_1(x) + \lambda_2\cdot C_2(x) + \dots + \lambda_m\cdot C_m(x)\in V$.  
\end{compactitem}

Let $X_i$ be the random variable defined by $(V,n,C_i)$ for $i\in[m]$, and $Y$ be the random variable defined by $(V,n,S)$, where $S:\zo^n\to\bbQ$ is a circuit such that  
\[
S(x)=\gamma+\lambda_1\cdot C_1(x) + \lambda_2\cdot C_2(x)+\dots + \lambda_m\cdot C_m(x). 
\]
Then: 
\begin{equation}
\left|\bbE_\delta[Y] - (\gamma + \lambda_1\cdot \bbE_\delta[X_1] + \dots + \lambda_m\cdot \bbE_\delta[X_m])\right| \le (2\delta + \beta)\cdot \|V\|\cdot \|\vec \lambda\|,  
\end{equation}
where $\|V\|$ and $\|\vec \lambda\|$ are the $\ell_1$-norm of $V$ and $\lambda$, respectively. 
\end{theorem}

\begin{proof}
We argue in $\Apx_1$. We first prove that $\bbE_\delta[Y] - (\gamma+\lambda_1\cdot \bbE_\delta[X_1] + \dots + \lambda_m\cdot \bbE_\delta[X_m]) \le (2\delta + \beta)\cdot \|V\|\cdot \|\vec \lambda\|$. Fix $n,m,\delta^{-1},\beta^{-1}\in\Log$, $C_1,\dots,C_m:\zo^n\to\bbQ$, $\vec\lambda$, $\gamma$, and $V$. Let 
\[
\mu\eqdef 1\cdot \bbE_{\delta}[Y] + (-\lambda_1)\cdot \bbE_\delta[X_1] + \dots + (-\lambda_m)\cdot \bbE_{\delta}[X_m], 
\]
and for each $z\in\zo^n$, we define $Y|_z$ and $X_i|_z$ to be random variable with seed length $0$ as 
\begin{equation}
Y|_z \eqdef S(z),\quad X_i|_z\eqdef C_i(z), \quad (i\in[m]), \\ 
\end{equation}
and $\mu|_z$ as 
\begin{equation}
\mu|_z \eqdef 1\cdot \bbE_{\delta}[Y|_z]+ (-\lambda_1)\cdot \bbE_{\delta}[X_1|_z] + \dots + (-\lambda_m) \cdot \bbE_{\delta}[X_m|_z].
\end{equation}
By the \refmeta{avg on exp}, there is a string $z\in\zo^n$ such that $\mu|_z \ge \mu - (2\delta+\beta)\cdot\|V\|\cdot\|\vec\lambda\|$, which implies that $\mu\le \mu|_z + (2\delta+\beta)\cdot\|V\|\cdot\|\vec\lambda\|$. 

Notice that $\bbE_\delta[Y|_z]$ is defined as 
\[
\bbE_\delta[Y|_z] = \sum_{v\in V}v\cdot \P_{\delta}(S_v) = \sum_{v\in V}v\cdot \Bool(S_v) = S(z),
\]
where $S_v$ is the circuit with no input that outputs $1$ if and only if $S(z)=1$. The second equality follows from the \refaxiom{boundary}. Similarly, we can prove that for each $i\in [m]$, $\bbE_\delta[X_i|_z] = C_i(z)$. Subsequently, $\mu|_z = S(z) - (\lambda_1\cdot C_1(z) + \dots + \lambda_m\cdot C_m(z)) = \gamma$, which further implies that $\mu\le \gamma + (2\delta + \beta)\cdot \|V\|\cdot\|\vec\lambda\|$, i.e., 
\[
\bbE_\delta[Y] - (\gamma+\lambda_1\cdot \bbE_\delta[X_1] + \dots + \lambda_m\cdot \bbE_\delta[X_m]) \le (2\delta + \beta)\cdot \|V\|\cdot \|\vec \lambda\|. 
\]

Finally, we can apply the same argument to $\mu'$ and $\mu'|_z$ defined by 
\begin{align*}
\mu' & \eqdef (-1)\cdot \bbE_\delta[Y] + \lambda_1\cdot \bbE_\delta[X_1] + \dots + \lambda_m\cdot\bbE_\delta[X_m] \\ 
\mu'|_z & \eqdef (-1)\cdot \bbE[Y|_z]+ \lambda_1\cdot \bbE[X_1|_z] + \dots + \lambda_m \cdot \bbE[X_m|_z]
\end{align*}
to conclude that $\mu=-\mu'\ge \gamma - (2\delta+\beta)\cdot\|V\|\cdot\|\vec\lambda\|$. This completes the proof.
\end{proof}

\subsection{Probability Inequalities}

We now develop several standard inequalities related to (approximate) probability and expectation, including the union bound, Markov's inequality, and Chebyshev's inequality.

\subsubsection{Union Bound}

Another application of the averaging argument for approximate expectation (see \Cref{thm: averaging on exp}) is the union bound. Recall that the acceptance probability of a circuit $C$ can be formalized as the expectation of its indicating random variable $I_C\in\zo$. Therefore the union bound can be derived from the following principle: Let $X_1,\dots,X_m, Y$ be Boolean-valued random variables such that for any random seed $z$, $Y|_z = X_1|_z\lor X_2|_z\lor\dots\lor X_m|_z$. Then $\bbE[Y]$ should not be much larger than $\bbE[X_1] + \bbE[X_2] + \dots + \bbE[X_m]$. Formally: 

\dummylabel{union bound}{Union Bound}
\begin{theorem}[Union Bound]\label{thm: union bound}
$\Apx_1$ proves the following statement. Let $n,m,\delta^{-1},\beta^{-1}\in\Log$, $C_1,\dots,C_m\in B_n$ be single-output circuits, $V=\{0,1\}$. Suppose that $\forall x\in\zo^n$ and $i\in[m]$, $C_i(x)\in V$, and let $Y,X_1,\dots, X_m$ be  random variables defined as follows.
\begin{compactitem}
\item For each $i\in[m]$, $X_i$ is defined by $(V,n,C_i)$. 
\item $Y$ is defined by $(V,n,S)$, where $S(x)\in\zo$ is a circuit such that $S(x) \le C_1(x) \lor \dots \lor C_m(x)$. 
\end{compactitem}
Then we have $\bbE_\delta[Y] \le \bbE_\delta[X_1] + \dots + \bbE_\delta[X_m] + (2\delta+\beta)\cdot m$.
\end{theorem}

\begin{proof}
We argue in $\Apx_1$. Fix $n,m,\delta^{-1},\beta^{-1}\in\Log$, $C_1,\dots,C_m\in B_n$, and $V=\{0,1\}$. For each $z\in\zo^n$, we define $Y|_z$ as the random variable with seed length $0$ that outputs $S(z)$, and $X_i$ as the random variable with seed length $0$ that outputs $C_i(z)$ for $i\in[m]$. Let $\mu$ and $\mu|_z$ be defined as 
\begin{align}
\mu & \eqdef 1\cdot \bbE_\delta[Y] + (-1)\cdot \bbE_\delta[X_1] + \dots + (-1)\cdot \bbE_{\delta}[X_m],\\ 
\mu|_z & \eqdef 1\cdot \bbE_\delta[Y|_z] + (-1)\cdot \bbE_\delta[X_1|_z] + \dots + (-1)\cdot \bbE_{\delta}[X_m|_z]. 
\end{align}
By the \refmeta{avg on exp}, we can conclude that $\mu\le \mu|_z + (2\delta + \beta)\cdot m$ for some string $z$. 

It then suffices to show that $\mu|_{z}\le 0$. Similarly to the proof of \Cref{thm: linearity exp}, we can prove by the \refaxiom{boundary} that $\bbE_{\delta}[Y|_z] = S(z)$ and $\bbE_\delta[X_i|_z] = C_i(z)$ for $i\in[m]$. Subsequently, we know by the definition of $\mu|_z$ and the assumption on $S$  that 
\[
\mu|_z = S(z) - (C_1(z) + \dots + C_m(z)) \le 0, 
\]
is provable in $\APX_1$. This completes the proof. 
\end{proof}

\subsubsection{Markov's Inequality}

Next, we consider Markov's inequality. For a random variable $X$ over an explicit set $V$, we should be able to prove that the probability that $X \ge k\cdot \bbE[X]$ cannot be much larger than $1/k$. This can be naturally formalized as follows: 

\dummylabel{markov}{Markov's Inequality}
\begin{theorem}[Markov's Inequality]\label{thm: markov}
The following statement is provable in $\Apx_1$. Let $X$ be a random variable defined by $(V,n,C)$, where $V$ is an explicit set of non-negative rational numbers, $n\in\Log$, and $C:\zo^n\to\bbQ$ is a circuit. Let $\delta^{-1},\beta^{-1}\in\Log$, $\mu\in\bbQ$ with $\mu \ge \bbE_\delta[X]$ and $\mu > 0$, $k\in\bbQ$ with $k > 0$, and $T(x)$ be the circuit that outputs $1$ if $C(x) \ge k\cdot \mu$, and outputs $0$ otherwise. Then
\[
\P_\delta(T) \le \delta +  k^{-1}\cdot (1+\delta\cdot \mu^{-1}\cdot \|V\|) + \beta \cdot (\mu^{-1}\cdot \|V\|+1),
\]
where $\|V\|$ is the $\ell_1$-norm of $V$.   
\end{theorem}

\begin{proof}
We argue in $\Apx_1$. Fix $n,\delta^{-1},\beta^{-1}\in\Log$, $C:\zo^n\to\bbQ$, $\mu,k\in\bbQ$, and $V$. Let $T(x)$ be the circuit as defined above, and $m = |V|$. We define the following random variables: 
\begin{compactitem}
\item $Y$ is the indicator variable of $T(x)$, i.e., it is defined by $(\zo, n, T)$.  
\item For each $v\in V$, $X_v$ is the indicator variable of $\EQ(C(x),v)$. Formally, let $C_v(x)$ be the circuit that outputs $1$ if and only if $C(x)=v$, and outputs $0$ otherwise, $X_v$ is the random variable defined by $(\zo, n, C_v)$.
\end{compactitem}

Let $\eta^{-1}\in\Log$ be a precision parameter to be determined later, and for $b\in\zo$, let $C^{(b)}_v(x)\in\zo$ be the circuit that outputs $1$ if $C(x)=b$. By the definition of approximate expectation, we have
\begin{align*}
\sum_{v\in V}v\cdot \bbE_\eta[X_v] &= \sum_{v\in V}\sum_{b\in\zo} v\cdot b\cdot \P_\eta(C_v^{(b)}) = \sum_{v\in V} v\cdot \P_\eta(C_v^{(1)})
\end{align*}
 By the definition of $C_v(x)$ and $C^{(1)}_v(x)$, we know that $C^{(1)}_v(x)=C_v(x)$, and therefore by the \refaxiom{global} of approximate counting, we have that
\[ 
|\P_\eta(C^{(1)}_v) - \P_\eta(C_v)| \le 3\cdot \eta. 
\]
Subsequently, we can see from the \refmeta{precision exp} that 
\begin{align}
\mu&\ge \bbE_{\delta}[X] \ge \bbE_\eta[X]-(\delta+2\eta)\cdot \|V\| \\ 
&= \sum_{v\in V} v\cdot \P_\eta(C_v) - (\delta+2\eta)\cdot \|V\| \\ 
& \ge \sum_{v\in V} v\cdot \bbE_\eta[X_v] - (\delta+5 \eta)\cdot \|V\| \\ 
& \ge k\mu \sum_{v\in V,v\ge k\mu}\bbE_\eta[X_v] - (\delta+5\eta)\cdot \|V\|, 
\end{align}
where the last inequality uses that $k$, $\mu$, and $V$ are all nonnegative.
This implies that 
\begin{equation}
\sum_{v\in V,v\ge k\cdot \mu}\bbE_\eta [X_v] \le k^{-1}\cdot (1 + (\delta+5\eta)\cdot \mu^{-1}\cdot  \|V\|). \label{equ: upper bound exp}
\end{equation} 

Below we also rely on the following inequality, which   follows from \Cref{prop: indicator variable}: 
\begin{equation}
|\bbE_\eta[Y] - \P_{\eta}(T)| \le 3\cdot \eta. \label{equ: exp Y}
\end{equation}

It is clear from the definition of $T(x)$ and $C_v(x)$ that
\[
T(x) = \bigvee_{v\in V,v\ge k\cdot \mu} C_v(x). 
\]  
Therefore, by the \refmeta{union bound}, we can conclude that 
\begin{align*}
\P_{\eta}(T) & \le \bbE_\eta [Y] + 3\cdot\eta  \tag{\Cref{equ: exp Y}}\\ 
&\le \sum_{v\in V,v\ge k\cdot \mu}\bbE_\eta[X_v] + 3\cdot \eta\cdot (m+1) \tag{Union Bound}\\ 
&\le k^{-1}\cdot (1+(\delta+5\eta)\cdot\mu^{-1} \cdot \|V\|) + 3\cdot \eta \cdot (|V| + 1), \tag{\Cref{equ: upper bound exp}} 
\end{align*}
Finally, we take $\eta = \min\{\beta / (50(|V|+1)), \beta k/(50(|V|+1))\}$ and apply the \refaxiom{precision}, so  
\[
\P_{\delta}(T) \le \delta + \P_{\eta}(T)  + 2\eta \le \delta +  k^{-1}\cdot (1+\delta\cdot \mu^{-1}\cdot \|V\|) + \beta \cdot (\mu^{-1}\|V\|+1).
\]
This completes the proof. 
\end{proof}

\subsubsection{Variance and Chebyshev's  Inequality} 

\newcommand{\Var}{\mathsf{Var}}

Next, we develop the basic theory of (approximate) variance and prove a form of Chebyshev's Inequality. 

\paragraph{Definition and Basic Properties.} Let $X$ be a random variable defined by the tuple $(V,n,C)$, where $n\in\Log$, $V$ is an explicit set, and $C:\zo^n\to\bbQ$. We can define a random variable $X^2$ by the tuple $(V^2,n,C^2)$, where $V^2\eqdef \{v^2\mid v\in V\}$ and $C^2(x) \eqdef (C(x))^2$. Similarly, we can define a random variable $X-\mu$ for any $\mu\in\bbQ$ by the tuple $(V-\mu, n, C_{-\mu})$, where $V-\mu\eqdef \{v-\mu\mid v\in V\}$ and $C_{-\mu}(x) \eqdef C(x) - \mu$. We can then define: 

\begin{definition}[Approximate Variance]
Let $X$ be a random variable defined by the tuple $(V,n,C)$, $\delta^{-1}\in\Log$. The approximate variance of $X$ with precision parameter $\delta$, denoted by $\Var_\delta[X]$, is defined as 
\[
\Var_\delta[X] \eqdef \bbE_\delta[(X-\mu)^2]
\]
where $\mu \eqdef \bbE_\delta[X]\in\bbQ$. 
\end{definition}

As an example, we prove an analogy of the equality $\Var[X] = \bbE[X^2]-\bbE[X]^2$ for approximate variance by directly formalizing the standard proof in $\Apx_1$. 

\begin{proposition}\label{prop: variance alternative def}
$\Apx_1$ proves the following statement. Let $X$ be a random variable defined by the tuple $(V,n,C)$, where $V\subseteq\bbQ$ is an explicit set, $n\in\Log$, and $C:\zo^n\to\bbQ$ satisfies $\forall x\in\zo^n~C(x)\in V$. Let $\mu\eqdef \bbE_\delta[X]\in\bbQ$. Then for any $\delta^{-1},\beta^{-1}\in\Log$,
\[ 
|\Var_\delta[X] - (\bbE_\delta[X^2] - \mu^2)| \le (2\delta +\beta)\cdot(1+|\mu|)\cdot \|\hat V\|,
\] 
where $\hat V = V\cup (V-\mu)^2\cup V^2\cup \{1\}$. 
\end{proposition}

\begin{proof}
We argue in $\Apx_1$. Fix $V\subseteq \bbQ$, $n,\delta^{-1},\beta^{-1}\in\Log$, and $C:\zo^n\to\bbQ$. Let $\hat V = V\cup (V-\mu)^2\cup V^2\cup \{1\}$ be an explicit set, and $\mu\eqdef \bbE_\delta[X]$. We define random variables $\hat Y$, $\hat X$, $\hat X^2$ over $\hat V$ as follows: 
\begin{compactitem}
\item $\hat Y$ is the random variable that outputs $(X-\mu)^2$. Formally, let $Y(x) \eqdef (C(x)-\mu)^2$ be a circuit. We can prove that $Y(x)\in \hat V$ from the assumption $C(x)\in V$ for any $x\in\zo^n$. We then define $\hat Y$ by the tuple $(\hat V, n, Y)$. 
\item $\hat X$ be the random variable that outputs $X$. Formally, it is defined by the tuple $(\hat V,n,C)$. 
\item $\hat X^2$ be the random variable that outputs $X^2$. Formally, let $C^2(x)\eqdef (C(x))^2$ be a circuit. We can prove that $C^2(x)\in \hat V$ from the assumption $C(x)\in V$ for any $x\in\zo^n$. We then define $\hat X^2$ by the tuple $(\hat V,n,C^2)$. 
\end{compactitem}

Let $\eta^{-1}\in\Log$ be determined later. By the definition of the circuits $Y$, $C$, $C^2$, it is clear that 
\[
Y(x) = \mu^2 + 1\cdot C^2(x) + (-2\mu)\cdot C(x). 
\]
By the \refmeta{linearity}, we can see that 
\begin{equation}
\left|\bbE_\eta [\hat Y] - (\mu^2 + \bbE_\eta[\hat X^2] - 2\mu \cdot \bbE_\eta[\hat X])\right| \le 6\eta \cdot \|\hat V\|\cdot (|\mu| + 1).\label{equ: upper bound linearity}
\end{equation}
By the \refmeta{support ext}, we can also conclude that
\begin{align}
\left|\bbE_\eta[\hat Y] - \bbE_\eta[(X-\mu)^2]\right|,\left|\bbE_\eta[\hat X^2] - \bbE_\eta[X^2]\right|, \left|\bbE_\eta[\hat X] - \bbE_\eta[X]\right|\le 3\eta \cdot \|\hat V\|.\label{equ: upper bound support ext} 
\end{align}

By the triangle inequality and the \refmeta{precision exp}, we have 
\begin{align*}
& |\Var_\delta[X] - (\bbE_\delta[X^2] - \mu^2)| \\ 
=~&|\bbE_\delta[(X-\mu)^2] - (\bbE_\delta[X^2] - \mu^2)| \\ 
\le~& |\bbE_\eta[(X-\mu)^2] - (\bbE_\eta[X^2] - \mu^2)| + 2\cdot (\delta+2\eta)\cdot\|\hat V\| \tag{\Cref{prop: precision consistency exp}}\\ 
\le~& |\bbE_\eta[\hat Y] - \bbE_\eta[\hat X^2] + \mu^2| + (2\delta+10\eta )\cdot\|\hat V\| \tag{\Cref{equ: upper bound support ext}} \\ 
\le~& |\mu^2 - (2\mu\cdot\bbE_\eta[\hat X] - \mu^2)| + (2\delta+10\eta )\cdot\|\hat V\| + 6\eta\cdot \|\hat V\|\cdot (|\mu|+1) \tag{\Cref{equ: upper bound linearity}} \\ 
\le~& |\mu^2 - (2\mu\cdot\bbE_\eta[X] - \mu^2)| + (2\delta+10\eta )\cdot\|\hat V\| + 6\eta\cdot \|\hat V\|\cdot (|\mu|+1) + 6\eta\cdot |\mu|\cdot \|\hat V\| \tag{\Cref{equ: upper bound support ext}
}\\ 
\le~& |\mu^2 - (2\mu\cdot\bbE_\delta[X] - \mu^2)| + (2\delta+10\eta )\cdot\|\hat V\| + 6\eta\cdot \|\hat V\|\cdot (|\mu|+1)\\ 
&\quad\quad  + 6\eta\cdot |\mu|\cdot \|\hat V\| + 2|\mu| \cdot (\delta+2\eta) \cdot\|\hat V\| \tag{\Cref{prop: precision consistency exp}}\\ 
\le~& (2\delta+10\eta )\cdot\|\hat V\| + 6\eta\cdot \|\hat V\|\cdot (|\mu|+1) + 6\eta\cdot|\mu|\cdot \|\hat V\| +  2|\mu|\cdot (\delta+2\eta)\cdot \|\hat V\|.\tag{$\mu \eqdef\bbE_\delta[X]$}
\end{align*}
The theorem then follows by taking $\eta = \beta / 40$. 
\end{proof}

\paragraph{Chebyshev's Inequality.} We now prove a form of Chebyshev's inequality that provides a tail bound for random variables with known (approximate) variance. Formally: 

\begin{theorem}[Chebyshev's Inequality]
$\Apx_1$ proves the following statement. Let $X$ be a random variable defined by the tuple $(V,n,C)$, $\delta^{-1}\in\Log$, where $V\subseteq\bbQ$ is an explicit set, $n\in\Log$, and $C:\zo^n\to\bbQ$ is a circuit such that $\forall x\in\zo^n~C(x)\in V$. Let   $\mu\eqdef \bbE_\delta[X]\in \bbQ$, $\sigma^2\eqdef\Var_\delta[X]$, and $T(x)$ be the circuit that outputs $1$ if $(C(x)-\mu)^2\ge k\cdot \sigma^2$, and outputs $0$ otherwise. 

Then for any $\beta^{-1}\in\Log$ and $k\in\bbQ$, where $k > 0$, we have that 
\[
\P_{\delta}(T) \le  \delta +  k^{-1}\cdot (1 + \delta\cdot \sigma^{-2}\cdot \|\hat V\|) + \beta\cdot (\sigma^{-2}\cdot \|\hat V\| + 1), 
\]
where $\hat V = (V-\mu)^2 = \{(v-\mu)^2 \mid v\in V\}$. 
\end{theorem}

\begin{proof}
We argue in $\Apx_1$. Fix $V,n,\delta^{-1},\beta^{-1},C:\zo^n\to\bbQ, \mu\eqdef \bbE_\delta[X], \sigma^2 \eqdef \Var_\delta[X], \beta^{-1}, k$. Let $Y$ be the random variable defined by the tuple $(\hat V, n, S)$, where $S(x) \eqdef (C(x)-\mu)^2$. By the definition of approximate variance, we know that 
\[
\sigma^2 = \Var_\delta[X] = \bbE_\delta[(X-\mu)^2] = \bbE_\delta[Y]. 
\]
By applying \refmeta{markov} to the random variable $Y$, we can see that 
\[
\P_\delta(T) \le \delta + k^{-1}\cdot (1 + \delta\cdot \sigma^{-2}\cdot \|\hat V\|) + \beta\cdot (\sigma^{-2}\cdot \|\hat V\| + 1).
\]
This completes the proof.
\end{proof}

\subsubsection{Pairwise Independence and Variance}

\newcommand{\Cov}{\mathsf{Cov}}

Now we develop the notion of (almost) pairwise independence, and prove a form of the equality $\Var[X_1 + \dots + X_m] = \Var[X_1] + \dots + \Var[X_m]$ for pairwise independent random variables $X_1,\dots,X_m$. This combined with  Chebyshev's inequality serve as a standard technique to reduce the error probability of randomized algorithms. 

\paragraph{Definition of (Almost) Independence.} We start with the definition of (almost) independence of random variables. Let $X_1,X_2$ be random variables over $V_1,V_2$, respectively. Recall that the covariance of $X$ and $Y$, denoted by $\Cov(X,Y)$, is defined as the quantity $\bbE[X\cdot Y] - \bbE[X]\cdot \bbE[Y]$, where $X\cdot Y$ is a random variable over $V_1V_2\eqdef \{v_1\cdot v_2\mid v_1\in V_1,v_2\in V_2\}$. Formally: 
\begin{definition}[Covariance]
Let $X_1$ and $X_2$ be the random variables defined by the tuples $(V_1,n,C_1)$ and $(V_2,n,C_2)$, respectively, where $V_1,V_2\subseteq \bbQ$ are explicit sets, $n\in\Log$, and $C_1,C_2:\zo^n\to \bbQ$ are circuits. Let $\delta^{-1}\in\Log$, $V_1V_2\eqdef\{v_1\cdot v_2\mid v_1\in V_1,v_2\in V_2\}$, $Y(x)$ be the circuit computing $C_1(x)\cdot C_2(x)$, and $Y$ be the random variable defined by the tuple $(V_1V_2,n,Y)$. The $\delta$-approximate covariance of $X_1$ and $X_2$, denoted by $\Cov_\delta(X_1,X_2)$, is defined as 
\[
\Cov_{\delta}(X_1,X_2) \eqdef |\bbE_{\delta}[Y] - \bbE_{\delta}[X_1]\cdot \bbE_{\delta}[X_2]|. 
\]
\end{definition}

\begin{definition}[Almost Independence]
Let $\delta^{-1}\in\Log$, $\eps\in\bbQ$. Random variables $X_1$ and $X_2$ are said to be $\eps$-almost $\delta$-approximately independent  if $\Cov_{\delta}(X_1,X_2)\le \eps$.  
\end{definition}

We can then define the pairwise independence of a sequence of random variables. 

\begin{definition}[Pairwise Independence]
Let $n,m\in\Log$, $C_1,\dots,C_m:\zo^n\to\bbQ$ be circuits, and $V\subseteq\bbQ$ be an explicit set such that $\forall i\in[m]~\forall x\in\zo^n~C_i(x)\in V$. Let $X_1,\dots, X_m$ be random variables, where for each $i\in[m]$, $X_i$ is defined by the tuple $(V,n,C_i)$, and let $\delta^{-1}\in\Log$ and $\eps\in\bbQ$. The sequence $X_1,\dots,X_m$ of random variables is said to be $\eps$-almost $\delta$-approximately pairwise independent if for every pair $(i,j)$ with $i,j\in [m]$ and $i\ne j$, $X_i$ and $X_j$ are $\eps$-almost $\delta$-approximately independent.  
\end{definition}

We may drop the parameter $\eps$ and simply say $\delta$-approximately independent if $\eps=0$. Note that since the approximate expectation of random variables may incur an error, random variables $X_1$ and $X_2$ may not be perfectly independent even if $\eps=0$.

\paragraph{Sum of Pairwise Independent Variables.} Now we are ready to prove the following result: Suppose that $X_1,\dots,X_m$ are almost pairwise independent, and $Y=X_1 + \dots + X_m$. Then the variance of $Y$ is close to the sum of the variances of $X_1,X_2,\dots,X_m$. Formally: 
\begin{theorem}
The following statement is provable in $\Apx_1$. Let $n,m\in\Log$, $C_1,\dots,C_m:\zo^n\to\bbQ$ be circuits, and $V\subseteq\bbQ$ be an explicit set such that the following holds: 
\begin{compactitem}
\item $\forall i\in[m]~\forall x\in\zo^n~C_i(x)\in V$;
\item $\forall x\in\zo^n~C_1(x) + C_2(x) + \dots + C_m(x)\in V$. 
\end{compactitem}

Let $X_1,\dots, X_m$ be random variables, where for each $i\in[m]$, $X_i$ is defined by the tuple $(V,n,C_i)$. Let $Y$ be the random variable defined by the tuple $(V,n,S)$, where $S(x)$ is the circuit computing $C_1(x) + C_2(x) + \dots + C_m(x)$. Let $\delta^{-1},\beta^{-1}\in\Log$ and $\eps\in\bbQ$. Suppose that $X_1,\dots,X_m$ are $\eps$-almost $\delta$-approximately pairwise independent. Then  
\[
\left|\Var_\delta[Y] - (\Var_\delta[X_1] + \dots + \Var_\delta[X_m])\right| \le (\eps +3\delta\cdot \|V\|^2) \cdot m^2 + \beta\cdot (\|V\|+1)^3,
\]
where $\|V\|=\sum_{v\in V}|v|$ is the $\ell_1$-norm of $V$.
\end{theorem}

\begin{proof}
    We argue in $\Apx_1$. Fix $n,m$, circuits $C_1,\dots,C_m\in\zo^n\to\bbQ$, $V,\delta^{-1},\eps^{-1},\beta^{-1}$. Let $\eta^{-1}\in\Log$ be a precision parameter to be determined later, $\mu \eqdef \bbE_\eta[Y]$, and $\mu_i \eqdef \bbE_\eta[X_i]$ for $i\in[m]$. 
    
    \paragraph{Overview of the proof.} Recall that by \Cref{prop: variance alternative def}, we have that 
    \begin{equation}
    \left|\Var_\eta[Y] - (\bbE_\eta[Y^2] - \mu^2)\right| \le 3\eta\cdot (1+|\mu|)\cdot \|\hat V_Y\|,\label{equ: tech var Y}
    \end{equation}
    where $\hat V_Y = V\cup (V-\mu)^2\cup V^2 \cup \{1\}$. Similarly, for each $i\in[m]$, we have that 
    \begin{equation}
    \left|\Var_\eta[X_i] - (\bbE_\eta[X_i^2] - \mu_i^2)\right| \le 3\eta \cdot (1+|\mu_i|)\cdot \|\hat V_i\|, \label{equ: tech var Xi}
    \end{equation}
    where $\hat V_i = V\cup(V-\mu_i)^2\cup V^2\cup\{1\}$. Therefore, it suffices to bound 
    \begin{equation}
    \Delta \eqdef |\bbE_\eta[Y^2] - (\bbE_\eta[X_1^2] + \dots + \bbE_\eta[X_m^2]) - \mu^2 + (\mu_1^2 + \dots \mu_m^2)|  \label{equ: tech def delta}  
    \end{equation}
    and combine it with \Cref{equ: tech var Y} and \eqref{equ: tech var Xi}. At a high level, our plan is to apply the \refmeta{linearity} to prove that $\bbE_\eta[Y^2]$ is close to 
    \begin{equation}
    \sum_{i=1}^m \bbE_\eta[X_i^2] + \sum_{i,j\in[m],i\ne j}\bbE_\eta[X_iX_j], \label{equ: tech Y sum}
    \end{equation}
    which is subsequently close to 
    \[
    \sum_{i=1}^m \bbE_\eta[X_i^2] + \sum_{i,j\in[m],i\ne j}\mu_i\mu_j
    \]
    by the almost pairwise independence of $X_1,\dots,X_m$. Finally, we can apply the \refmeta{linearity} to show that 
    \[
    \mu_1^2 + \dots + \mu_m^2 + \sum_{i,j\in[m],i\ne j}\mu_i\mu_j = \left(\sum_{i=1}^m \mu_i\right)^2 \approx \mu^2.
    \]
    Putting the estimates together provides the upper bound for $\Delta$.

     \paragraph{Step 1: Approximation of $\bbE_\eta[Y^2]$.} We first show that $\bbE_\eta[Y^2]$ is close to \Cref{equ: tech Y sum}. Recall that $Y^2$, $X^2$, and $X_{i}X_j$ are the random variables over $V^2$ defined as follows: 
    \begin{compactitem}
    \item $Y^2$ is defined by the tuple $(V^2, n, S^2)$, where $S^2(x)\eqdef (S(x))^2$. 
    \item For $i\in[m]$, $X_i^2$ is defined by the tuple $(V^2, n, C_i^2)$, where $C_i^2(x)\eqdef (C_i(x))^2$.  
    \item For $i,j\in[m]$ such that $i\ne j$, $X_{i}X_j$ is defined by the tuple $(V^2, n, C_{ij})$, where $C_{ij}(x)\eqdef C_i(x)\cdot C_j(x)$.
    \end{compactitem}
    
    It is clear from the definition of the terms that for any $z\in\zo^n$, we have 
    \[
    S^2(z) = \sum_{i=1}^m C_i^2(z) + \sum_{i,j\in[m],i\ne j} C_{ij}(z). 
    \]
    Thus by the \refmeta{linearity}, we have 
    \begin{equation}
    \left|\bbE_\eta[Y^2] - \left(\sum_{i=1}^m \bbE_\eta[X_i^2] + \sum_{i,j\in[m],i\ne j}\bbE_\eta[X_{i}X_j]\right)\right| \le 3\eta \cdot \|V^2\|\cdot (m^2 + m + 1) \le 9\eta \cdot m^2\cdot \|V^2\|. \label{equ: tech step 1 finished}
    \end{equation}

    \paragraph{Step 2: Applying pairwise independence.} In this step we show that $\bbE_\eta[X_iX_j]$ is close to $\mu_i\mu_j$ for any $i,j\in[m]$. Recall that $\mu_i\eqdef \bbE_\eta[X_i]$. By the \refmeta{precision exp}, we have  
    \[
    \left|\bbE_{\eta}[X_iX_j] - \bbE_\delta[X_iX_j]\right| \le (\delta+2\eta)\cdot \|V^2\|
    \]
    for any $i,j\in[m]$, $i\ne j$. Moreover, since $X_1,\dots,X_m$ are $\eps$-almost $\delta$-approximately pairwise independent, we know that 
    \[
    \left|\bbE_\delta[X_iX_j] - \bbE_\delta[X_i]\cdot \bbE_\delta[X_j]\right| \le \eps. 
    \]
    For each $i\in[m]$, we have that $|\bbE_\delta[X_i] - \bbE_\eta[X_i]| \le (\delta + 2\eta)\cdot \|V\|$, which implies that 
    \begin{align*}
        & \left|\bbE_\delta[X_i]\cdot \bbE_\delta[X_j] - \bbE_\eta[X_i]\cdot \bbE_\eta[X_j]\right| \\ 
    \le & \left|\bbE_\delta[X_i]\cdot \bbE_\delta[X_j] - \bbE_\eta[X_i]\cdot \bbE_\delta[X_j]\right| + \left|\bbE_\eta[X_i]\cdot \bbE_\delta[X_j] - \bbE_\eta[X_i]\cdot \bbE_\eta[X_j]\right| \\ 
    \le & (\delta + 2\eta)\cdot \|V\| \cdot 
    (|\bbE_\eta[X_i]| + |\bbE_\delta[X_j]|)\le 2(\delta+2\eta)\cdot \|V\|^2. 
    \end{align*}
    Combining the equations above, we have that 
    \begin{equation}
    |\bbE_\eta[X_iX_j] - \mu_i\mu_j| \le \eps + (\delta+2\eta)\cdot \|V^2\| + 2(\delta+2\eta)\cdot \|V\|^2 \le \eps + 3(\delta+2\eta)\|V\|^2. \label{equ: tech step 2 finished}
    \end{equation}

    \paragraph{Step 3: Applying linearity of expectation.} The last step is to prove that $\mu\approx \mu_1 + \dots + \mu_m$. Recall that $\mu \eqdef \bbE_\eta[Y]$ and $\mu_i\eqdef \bbE_\eta[X_i]$, where $Y$ is defined by the tuple $(V,n,S)$ and $X_i$ is defined by the tuple $(V,n,C_i)$. It is clear from the definition of $S$ and $C_i$ that for any $z\in\zo^n$, 
    \[
    S(z) = C_1(z) + C_2(z) + \dots + C_m(z).  
    \]
    Therefore, by the \refmeta{linearity}, we have that 
    \begin{equation}
    |\mu - (\mu_1+\dots+\mu_m)| = |\bbE_\eta[Y] - (\bbE_\eta[X_1] + \dots + \bbE_\eta[X_m])| \le 3\eta\cdot \|V\|\cdot (m+1). \label{equ: tech step 3 middle}
    \end{equation}
    Subsequently, we can see that    
    \begin{align}
    \left|\mu^2 - (\mu_1+\dots+\mu_m)^2\right| =~& |\mu + (\mu_1+\dots+\mu_m)|\cdot |\mu - (\mu_1+\dots+\mu_m)| \nonumber\\ 
    \le~&(2|\mu| + 3\eta\cdot \|V\|\cdot (m+1))\cdot 3\eta\cdot \|V\|\cdot (m+1) \nonumber \nonumber \\ 
    \le~& 6|\mu| \cdot \eta \cdot 2m \cdot \|V\| + 9\eta^2 \cdot (2m)^2\cdot \|V\|^2 \nonumber \\ 
    \le~& 4\cdot (9\eta^2+6\eta)\cdot m^2\cdot \|V\|^2 \tag{since $|u| \leq \|V\|$} \nonumber \\ 
    \le~& 60\cdot \eta \cdot (m\|V\|)^2, \label{equ: tech step 3 finished}
    \end{align}
    where the last inequality holds as we will take $\eta \le 1$. 

    \paragraph{Wrapping things up.} Combining \Cref{equ: tech step 1 finished}, \eqref{equ: tech step 2 finished} and \eqref{equ: tech step 3 finished}, we can see that 
    \begin{align*}
    \Delta =~& |\bbE_\eta[Y^2] - (\bbE_\eta[X_1^2] + \dots + \bbE_\eta[X_m^2]) - \mu^2 + (\mu_1^2 + \dots \mu_m^2)| \\ 
    \le~&\left|\sum_{i,j\in[m],i\ne j}\bbE_\eta[X_iX_j] - \mu^2 + \sum_{i=1}^m \mu_i^2\right| + 9\eta\cdot m^2\cdot \|V\|^2 \tag{\Cref{equ: tech step 1 finished}} \\ 
    \le~&\left|\sum_{i,j\in[m],i\ne j}\mu_i\mu_j - \mu^2 + \sum_{i=1}^m \mu_i^2\right| + 9\eta\cdot m^2\cdot \|V\|^2 + m^2(\eps + 3(\delta+2 \eta)\cdot \|V\|^2) \tag{\Cref{equ: tech step 2 finished}} \\ 
    \le~&\left|-\mu^2 + \left(\sum_{i=1}^m \mu_i\right)^2\right| + (\eps +3\delta\cdot \|V\|^2) \cdot m^2 + 20\cdot \eta\cdot (m\|V\|)^2 \\ 
    \le~& 60\cdot \eta \cdot (m\|V\|)^2+ (\eps +3\delta\cdot \|V\|^2) \cdot m^2 + 20\cdot \eta\cdot (m\|V\|)^2 \tag{\Cref{equ: tech step 3 finished}} \\ 
    \le~& (\eps +3\delta\cdot \|V\|^2) \cdot m^2 + 80\cdot \eta \cdot (m\|V\|)^2. 
    \end{align*}

    Finally, we combine this with \Cref{equ: tech var Y} and \eqref{equ: tech var Xi}, which gives 
    \begin{align}
        \left|\Var_\delta[Y] - (\Var_\delta[X_1] + \dots + \Var_\delta[X_m])\right| \le \Delta + 3\eta \cdot (1+|\mu|)\cdot \|\hat V_Y\| + 3\eta \cdot \sum_{i=1}^m (1+|\mu_i|) \cdot \|\hat V_i\|.\label{equ: tech upper bound with above} 
    \end{align}
    Note that 
    \begin{align*}
    |\mu| \le\|\hat V_Y\| \le~&\|V\| + \|(V-\mu)^2\| + \|V^2\| + 1 \\ 
    \le~&\|V\| + \|V\|^2 + 1 + \sum_{v\in V} (v-\mu)^2 \\ 
    \le~&\|V\| + \|V\|^2 + 1 + \|V\|^2 + |V|\cdot \mu^2 + 2\mu\cdot \|V\|^2 \\ 
    \le~& 2\cdot \|V\|^3 + (|V| + 2)\cdot \|V\|^2 + \|V\| + 1 = 8\cdot |V|\cdot (\|V\|+1)^3.   
    \end{align*}
    Similarly, we have that $|\mu_i| \le \|\hat V_i\| \le 8\cdot |V|\cdot (\|V\|+1)^3$. 
    
    Let $\eta \eqdef \min \{\beta / (2000\cdot m^2\cdot |V|), 1/10\} \leq 1$. By combining \Cref{equ: tech upper bound with above} and the upper bound above, we have that 
    \begin{align*}
        & \left|\Var_\delta[Y] - (\Var_\delta[X_1] + \dots + \Var_\delta[X_m])\right| \\ 
    \le~& (\eps +3\delta\cdot \|V\|^2) \cdot m^2 + 80\cdot \eta \cdot (m\|V\|)^2 + 1000\cdot m\cdot \eta\cdot |V|\cdot (\|V\|+1)^3 \\ 
    \le~& (\eps +3\delta\cdot \|V\|^2) \cdot m^2 + \beta\cdot (\|V\|+1)^3.  
    \end{align*}
    This completes the proof. 
\end{proof}

\subsection{Independence, Error Reduction, and Concentration Bounds}

We now consider the provability in $\Apx_1$ of concentration bounds for independent and identically distributed (i.i.d.) random variables, which are important tools in combinatorics, probability, and the analysis of  randomized algorithms. 

\subsubsection{Explicit Independence}

Before stating and proving the concentration bounds, we formally define the way we formulate independent and identically distributed random variables, and prove a form of the multiplication principle for approximate counting. 

\pgr{Formalization of i.i.d.~RVs} We will only consider random variables that are ``explicitly'' i.i.d., in the sense that they are defined by the same sampling algorithm using disjoint parts of the random seed. This suffices for our applications and greatly simplifies the calculation of parameters for approximate counting. We first formally define explicit independence of random variables as follows: 

\begin{definition}[Explicit Independence]
Let $X$ and $X'$ be random variables defined by $(V,n,C)$ and $(V,n,C')$, respectively. We say that $X$ and $X'$ are \emph{explicitly independent} if $C$ and $C'$ read disjoint bits of the $n$-bit seed; that is, there is a partition $\pi_1\cup \pi_2$ of $[n]$ such that for any seed $x\in\zo^n$, $C$ only reads $x_{\pi_1}$ and $C'$ only reads $x_{\pi_2}$, where $x_{\pi}$ denotes the bits of $x$ with indices in $\pi$. 
\end{definition}

Similarly, we define explicitly i.i.d.~random variables as follows: 

\begin{definition}[Explicitly i.i.d.~RVs]
Let $n,m\in\Log$, $C:\zo^n\to\bbQ$ be a circuit, $V\subseteq\bbQ$ be an explicit set such that $\forall x\in\zo^n~C(x)\in V$. The explicitly i.i.d.~random variables $X_1,\dots,X_m$ defined by the tuple $(V,n,C)$ are obtained as follows. 

Let $C_i(\cdot):\zo^{nm}\to V$ be the circuit such that for any $\overline x = x_1\circ x_2\circ \dots\circ x_m\in\zo^n$, where $x_i\in\zo^n$ for each $i\in[m]$, $C_i(\overline x) \eqdef C(x_i)$. For each $i\in[m]$, the random variable $X_i$ is defined by the tuple $(V,nm,C_i)$. 
\end{definition}

\subsubsection{Multiplication Principle}

We will need a form of multiplication principle: For any explicitly independent random variables $X$ and $Y$, we have $\Ex[XY]\approx \Ex[X]\cdot \Ex[Y]$, or equivalently, $\Cov(X,Y)$ is small. In other words, explicitly independent random variables are approximately independent. Formally: 

\dummylabel{multiplication}{Multiplication Principle}
\begin{theorem}[Multiplication Principle]\label{thm: multiplication}
    $\Apx_1$ proves the following statement. Let $n,\delta^{-1}\in\Log$, $V_1,V_2$ be explicit sets, $C_1,C_2\in\zo^n\to\bbQ$ be circuits. Suppose that the random variables $X_1$ and $X_2$, defined by the tuples $(V_1,n,C_1)$ and $(V_2,n,C_2)$, respectively, are explicitly independent. Then for any $\beta^{-1}\in\Log$,
    \[
    \Cov_\delta(X,Y) = |\bbE_{\delta}[X_1X_2] - \bbE_{\delta}[X_1]\cdot \bbE_{\delta}[X_2]| \le (2\delta+\beta)\cdot \|\hat V\| + (4\delta +\beta)\cdot \|\hat V\|^2, 
    \]
    where $\hat V \eqdef V_1\cup V_2\cup V_1V_2$, and $V_1 V_2 \eqdef \{v_1 v_2 \mid v_1 \in V_1, v_2 \in V_2\}$. 
\end{theorem}

\begin{proof}
    We argue in $\Apx_1$. Fix $n,\delta^{-1}\in\Log$, $V_1,V_2,C_1,C_2$, and $\beta^{-1}\in\Log$. Let $\eta^{-1}\in\Log$ be determined later and $\hat V \eqdef V_1\cup V_2\cup V_1V_2$ be an explicit set.

    \pgr{Overview of the proof} At a high level, the proof goes as follows. Let $\pi_1\cup \pi_2$ be a partition of $[n]$ such that for any seed $x\in\zo^n$, $X_1$ only reads $x_{\pi_1}$ and $X_2$ only reads $x_{\pi_2}$. Suppose, towards a contradiction, that $|\bbE[X_1 X_2] - \bbE[X_1]\cdot \bbE[X_2]|$ is large. By the \refmeta{avg on exp}, we may find an assignment $\rho_1$ to the part $x_{\pi_1}$ of the seed such that 
    \[
    \left|\bbE[X_1X_2|_{\rho_1}] - \bbE[X_1|_{\rho_1}] \cdot \bbE[X_2]\right| 
    \]
    is large. More formally, we are applying \refmeta{avg on exp} by treating $X_1X_2$ and $X_1$ as random variables with coefficients $1$ and $-\bbE[X_2]$, respectively. Note that since $X_2$ does not read the part $x_{\pi_1}$ of the seed, we know that $\bbE[X_2]$ is close to $\bbE[X_2|_{\rho_1}]$, and subsequently 
    \[
    \left|\bbE[X_1X_2|_{\rho_1}] - \bbE[X_1|_{\rho_1}] \cdot \bbE[X_2|_{\rho_1}]\right| 
    \]
    is also large. 

    Next, we apply the \refmeta{avg on exp} again by treating $X_1X_2|_{\rho_1}$ and $X_2|_{\rho_1}$ as random variables with coefficients $1$ and $-\bbE[X_1|_{\rho_1}]$, respectively, using the seed $x_{\pi_2}$. This gives an assignment $\rho_2$ to $x_{\pi_2}$ such that 
    \[
    \left|\bbE[X_1X_2|_{\rho_1}|_{\rho_2}] - \bbE[X_1|_{\rho_1}] \cdot \bbE[X_2|_{\rho_1}|_{\rho_2}]\right|. 
    \]
    is large. Note that, again, since $X_1$ does not read the part $x_{\pi_2}$ of the seed, we know that $\bbE[X_1|_{\rho_1}]$ is close to $\bbE[X_1|_{\rho_1}|_{\rho_2}]$, and subsequently 
    \[
    \left|\bbE[X_1X_2|_{\rho_1}|_{\rho_2}] - \bbE[X_1|_{\rho_1}|_{\rho_2}] \cdot \bbE[X_2|_{\rho_1}|_{\rho_2}]\right|. 
    \]
    is also large. However, this is impossible as $X_1X_2|_{\rho_1}|_{\rho_2}$, $X_1|_{\rho_1}|_{\rho_2}$, and $X_2|_{\rho_1}|_{\rho_2}$ are random variables with seed length $0$ and are supposed to satisfy $X_1X_2|_{\rho_1}|_{\rho_2} = X_1|_{\rho_1}|_{\rho_2}\cdot X_2|_{\rho_1}|_{\rho_2}$ by definition. 

    We now prove the theorem in detail. Note that we will implement the proof idea above in backward direction for simplicity of calculation.  

    \pgr{Step 1: Averaging argument after fixing $\rho_1$} Let $\hat Y,\hat X_1,\hat X_2$ be the random variables over $\hat V\eqdef V_1\cup V_2\cup V_1V_2$ obtained from $X_1X_2$, $X_1$, $X_2$ via support extension. Let $\pi_1\cup\pi_2$ be a partition of $[n]$ such that for any seed $x\in\zo^n$, $C_1$ only reads $x_{\pi_1}$ and $C_2$ only reads $x_{\pi_2}$. Let $\eta^{-1}\in\Log$ be a parameter to be determined later. It is clear that for any fixed assignments $\rho_1,\rho_2$ to $x_{\pi_1},x_{\pi_2}$, respectively, 
    \begin{equation}
    \bbE_{\eta}[\hat Y|_{\rho_1}|_{\rho_2}] = \bbE_\eta[\hat X_1|_{\rho_1}|_{\rho_2}]\cdot \bbE_\eta[\hat X_2|_{\rho_1}|_{\rho_2}]\label{equ: eq after fixing both}
    \end{equation}
    by the definition of the random variables. 
    
    Note that $\hat X_1|_{\rho_1}$ is a random variable that does not read the its seed $x_{\rho_2}$. Therefore, for any fixed $\rho_2,\rho_2'$, $|\bbE_{\eta}[\hat X_1|_{\rho_1}|_{\rho_2'}] -\bbE_\eta[\hat X_1|_{\rho_1}|_{\rho_2}]| = 0$. By the \refmeta{avg on exp}, we have that for any $\rho_2$, 
    \begin{equation}
        |\bbE_\eta[\hat X_1|_{\rho_1}]-\bbE_\eta[\hat X_1|_{\rho_1}|_{\rho_2}]|\le 3\eta\cdot \|\hat V\|,\label{equ: close X1 assigned}
    \end{equation}
    and subsequently by \Cref{equ: eq after fixing both}, 
    \begin{equation}
    |\bbE_{\eta}[\hat Y|_{\rho_1}|_{\rho_2}]- \bbE_\eta[\hat X_1|_{\rho_1}]\cdot \bbE_\eta[\hat X_2|_{\rho_1}|_{\rho_2}]|\le 3\eta\cdot \|\hat V\|\cdot \bbE_\eta[\hat X_2|_{\rho_1}|_{\rho_2}] \le 3\eta \cdot \|\hat V\|^2.
    \end{equation}
    
    Therefore, by the \refmeta{avg on exp}, we have that for any fixed assignment $\rho_1$ to $x_{\pi_1}$, 
    \begin{equation}
    |\bbE_{\eta}[\hat Y|_{\rho_1}]- \bbE_\eta[\hat X_1|_{\rho_1}]\cdot \bbE_\eta[\hat X_2|_{\rho_1}]| \le 3\eta \cdot \|\hat V\|\cdot (1+|\bbE_\eta[\hat X_1|_{\rho_1}]|)\le 3\eta \cdot \|\hat V\|\cdot (1+\|\hat V\|). \label{equ: result step 1 avg}
    \end{equation}
    Note that here we treat $\hat Y|_{\rho_1}$, $\hat X_2|_{\rho_1}$ as random variables, and $\bbE_\eta[\hat X_1|_{\rho_1}]$ as a coefficient. 

    \pgr{Step 2: Averaging argument again} Similarly to \Cref{equ: close X1 assigned}, we will first show that $\bbE_\eta[\hat X_2|_{\rho_1}]$ is close to $\bbE_\eta[\hat X_2]$ for any $\rho_1$. We can see that for any assignments $\rho_1,\rho_1'$ to $x_{\pi_1}$ and $\rho_2$ to $x_{\pi_2}$, $\mathbb{E}_\eta[\hat X_2|_{\rho_1}|_{\rho_2}] = \bbE_\eta[\hat X_2|_{\rho_1'}|_{\rho_2}]$ as $\hat X_2$ does not read $x_{\pi_1}$. Therefore, by the \refmeta{avg on exp}, for every $\rho_1,\rho_1'$, 
    \[
    |\bbE_{\eta}[\hat X_2|_{\rho_1}] - \bbE_\eta[\hat X_2|_{\rho_1'}]| \le 6\eta \cdot \|\hat V\|\cdot 
    \]
    Again, by the \refmeta{avg on exp} applied to $\hat{X}_2$, we get $\rho_1'$ such that 
    \[
    |\bbE_{\eta}[\hat X_2|_{\rho_1}] - \bbE_\eta[\hat X_2]| \le 3\eta \cdot \|\hat V\| + |\bbE_{\eta}[\hat X_2|_{\rho_1}] - \bbE_\eta[\hat X_2|_{\rho_1'}]| \le 9\eta \cdot \|V\|. 
    \]
    Combining this with \Cref{equ: result step 1 avg}, we have 
    \begin{align}
    & |\bbE_{\eta}[\hat Y|_{\rho_1}]- \bbE_\eta[\hat X_1|_{\rho_1}]\cdot \bbE_\eta[\hat X_2]| \le |\bbE_{\eta}[\hat Y|_{\rho_1}]- \bbE_\eta[\hat X_1|_{\rho_1}]\cdot \bbE_\eta[\hat X_2|_{\rho_1}]| + 9\eta \cdot \|\hat V\|^2\nonumber  \\ 
    & \quad \le 12\eta\cdot \|\hat V\|^2 + 3\eta \cdot \|\hat V\|. \label{equ: step 2 avg before finish}
    \end{align}
    By applying the \refmeta{avg on exp} on \Cref{equ: step 2 avg before finish} with random variables $\hat Y$ and $\hat X_1$, we have that 
    \begin{align}
    & |\bbE_{\eta}[\hat Y]- \bbE_\eta[\hat X_1]\cdot \bbE_\eta[\hat X_2]| \le |\bbE_{\eta}[\hat Y|_{\rho_1}]- \bbE_\eta[\hat X_1|_{\rho_1}]\cdot \bbE_\eta[\hat X_2]| +  3\eta\cdot \|V\|\cdot (1+\|V\|)\nonumber  \\ 
    & \quad \le 15\eta\cdot \|\hat V\|^2 + 6\eta \cdot \|\hat V\|. \label{equ: result step 2 avg}
    \end{align}

    \pgr{Wrapping things up} Note that by the \refmeta{support ext}, we have that 
    \[
    |\bbE_\eta[\hat Y] - \bbE_\eta[X_1X_2]|, |\bbE_\eta[\hat X_1]-\bbE_\eta[X_1]|, |\bbE_\eta[\hat X_2]-\bbE_\eta[X_2]| \le 3\eta\cdot \|\hat V\|. 
    \]
    By the \refmeta{precision exp}, we have 
    \[ 
    |\bbE_\delta[X_1X_2] - \bbE_\eta[X_1X_2]|, |\bbE_\delta[ X_1]-\bbE_\eta[X_1]|, |\bbE_\delta[X_2]-\bbE_\eta[X_2]| \le (2\delta+\eta)\cdot \|\hat V\|.
    \]
    We can therefore deduce from \Cref{equ: result step 2 avg} that 
    \begin{align*}
        & |\bbE_{\delta}[X_1X_2] - \bbE_{\delta}[X_1]\cdot \bbE_{\delta}[X_2]| \\ 
    \le~& |\bbE_{\eta}[X_1X_2] - \bbE_{\eta}[X_1]\cdot \bbE_{\eta}[X_2]| + (2\delta + \eta)\cdot \|\hat V\|+(4\delta +2\eta)\cdot \|V\|^2 \\ 
    \le~& |\bbE_{\eta}[\hat Y] - \bbE_{\eta}[\hat X_1]\cdot \bbE_{\eta}[\hat X_2]| + (2\delta + \eta)\cdot \|\hat V\|+(4\delta +2\eta)\cdot \|V\|^2 + 3\eta\cdot \|\hat V\| + 6\eta\cdot \|\hat V\|^2 \\ 
    \le~& |\bbE_{\eta}[\hat Y] - \bbE_{\eta}[\hat X_1]\cdot \bbE_{\eta}[\hat X_2]| + (2\delta + 4\eta)\cdot\|\hat V\| + (4\delta + 8\eta)\cdot \|\hat V^2\| \\ 
    \le~& (2\delta + 10\eta)\cdot \|\hat V\| + (4\delta + 23\eta)\cdot \|\hat V\|^2. 
    \end{align*}
    This completes the proof by setting $\eta = \beta /100$. 
\end{proof}

Subsequently, we can obtain the following more convenient form of the  multiplication principle for explicitly independent Bernoulli random variables: 

\dummylabel{multiplication BRVs}{Multiplication Principle for Bernoulli RVs}
\begin{corollary}[Multiplication Principle for Bernoulli RVs]
    $\APX_1$ proves the following statement. Let $n,m,\delta^{-1}\in\Log$ and $X_1,\dots,X_m$ be explicitly independent random variables over $\zo$ with seed length $n$. Then
    \[
    \left|\bbE_{\delta}[X_1X_2\dots X_m] - \prod_{j\in [i]}\bbE_\delta[X_i]\right| \le 8\delta \cdot m,
    \]
    where the random variable $X_1 X_2 \dots X_m$ is defined in the natural way.
\end{corollary}

\begin{proof}
    We argue in $\APX_1$. Fix $n,m,\delta^{-1}\in\Log$. We prove by induction on $i$ such that for every $i\in[m]$, we have 
\begin{equation}
\left|\bbE_\delta[X_1X_2\dots X_i] - \prod_{j\in [i]}\bbE_\delta[X_i] \right| \le 8\delta \cdot i.\label{equ: ind hyp multiplication} 
\end{equation}
Note that this employs induction on a quantifier-free formula, which is available in $\Apx_1$. The base case $i=1$ is straightforward. Suppose that \Cref{equ: ind hyp multiplication} holds. By the \refmeta{multiplication}, 
\begin{align}
    & \left|\bbE_\delta[X_1X_2\dots X_iX_{i+1}] - \bbE_\delta[X_1X_2\dots X_i]\cdot \bbE_\delta[X_{i+1}]\right|\nonumber \\ 
    & \quad \le  3\delta\cdot \|V\| + 5\delta\cdot \|V\|^2\le 8\delta, \label{equ: apply multiplication principle}
\end{align}
and subsequently 
\begin{align*}
    & \left|\bbE_\delta[X_1X_2\dots X_iX_{i+1}] -  \prod_{j\in [i+1]}\bbE_\delta[X_i] \right| \\ 
\le~& \left|\bbE_\delta[X_1X_2\dots X_i]\cdot \bbE_\delta[X_{i+1}] - \prod_{j\in [i+1]}\bbE_\delta[X_i]  \right| + 8\delta \tag{\Cref{equ: apply multiplication principle}}\\ 
\le~&  \left|\bbE_\delta[X_1X_2\dots X_i] - \prod_{j\in [i]}\bbE_\delta[X_i] \right|\cdot |\bbE_\delta[X_{i+1}]| + 8\delta  \\ 
\le~& 8\delta\cdot i \cdot \|V\| + 8\delta \le 8\delta\cdot (i+1). \tag{Induction Hypothesis}
\end{align*}
This completes the proof. 
\end{proof}

\subsubsection{Error Reduction for One-Sided Error Statements}

The multiplication principle allows us to prove the correctness of error reduction via repetition for one-sided error algorithms. Specifically, for any circuit $C$ that accepts a $(1-\eps)$-fraction of its inputs, the circuit $C^{\lor k}(x_1,\dots,x_k)=\bigvee_{i\in[k]} C(x_i)$ accepts all but an $\eps^k$-fraction its inputs. 

This is formalized as the following theorem: 

\dummylabel{os error red}{One-sided Error Reduction Lemma}
\begin{theorem}[One-sided error reduction lemma]\label{thm:one-sided}
    Let $C^{\lor k}:\zo^{nk}\to\zo$ be the circuit defined as $C^{\lor k}(x_1,\dots,x_k)\eqdef \bigvee_{i\in[k]} C(x_i)$ for any circuit $C:\zo^n\to\zo$. The following statement is provable in $\Apx_1$. For any $n,k,\delta^{-1},\gamma^{-1},\beta^{-1}\in\Log$ and circuit $C:\zo^n\to\zo$, if $\P_\delta(\lnot C) \le \eps$ and $\gamma\ge (\delta+\beta+\eps)^k+\delta + \beta$, then $\P_{\delta}(\lnot C^{\lor k})\le \gamma$. 
\end{theorem}

\begin{proof}
We argue in $\Apx_1$. Fix $n,\delta^{-1},\gamma^{-1},\beta^{-1}\in\Log$ and a circuit $C:\zo^n\to\zo$. Let $V\eqdef \{0,1\}$, and let $\eta^{-1}\in\Log$ be a parameter to be determined later. Let $X_i$ be the random variable that takes an $nk$-bit seed $(x_1,\dots,x_k)\in(\zo^n)^k$ and outputs $\lnot C(x_i)$ for $i\in[k]$. It is clear that $X_1,\dots,X_m$ are explicit independent random variables. Let $Y$ be the random variable that takes an $nk$-bit seed $(x_1,\dots,x_k)\in(\zo^n)^k$ and outputs $\lnot C^{\lor k}(x_1,\dots,x_k)$. 

\paragraph{RVs and approximate counting.} It is clear that $Y$ is the indicator random variable for the circuit $\lnot C^{\lor k}$, and thus by \Cref{prop: indicator variable}, $|\bbE_\eta[Y]-\P_\eta(\lnot C^{\lor k})|\le 3\eta$. 

We will prove that for every $i\in[k]$, $|\bbE_\eta[X_i]-\P_\eta(\lnot C)|\le 6\eta$. Fix any $i\in[k]$ and let $\pi_i\cup\overline{\pi_i}=[nk]$, where $\pi_i$ denotes the set of indices corresponding to the $i$-th $n$-bit block that $X_i$ reads, and $\overline{\pi_i}$ denotes the other $(k-1)n$ indices. For every assignment $\rho$ to $\overline{\pi_i}$, we can see that $X_i|_{\rho}$ is the indicator random variable for $C$, and thus by \Cref{prop: indicator variable},
\[
|\bbE_\eta[X_i|_\rho] - \P_\eta(\lnot C)|\le 3\eta. 
\]
Subsequently, we can apply the \refmeta{avg on exp} to prove that 
\begin{equation}
|\bbE_\eta[X_i] - \P_\eta(\lnot C)|\le 6\eta.
\end{equation}

Note that by the assumption, we have that $\P_\delta(\lnot C)\le \eps$. By the \refaxiom{precision} for approximate counting, we have $\P_\eta(C) \le \delta + 2\eta+\eps$, and thus 
\begin{equation}
\bbE_\eta[X_i]\le \P_\eta(C) + 6\eta \le \delta +8\eta +\eps\label{equ: xi upper bound prC}
\end{equation}
for every $i\in[k]$. 

\paragraph{Wrapping things up.} We first prove that $|\bbE_\eta[Y]-\bbE_\eta[X_1X_2\dots X_k]|\le 6\eta$. Note that $Y$ and $X_1X_2\dots X_k$ are both random variables taking $nk$-bit seeds, and for every assignment $\rho$ to the seeds, we have $\bbE_\eta[Y|_\rho] = \bbE_\eta[X_1 X_2\dots X_k|_\rho]$. By the \refmeta{avg on exp}, we have 
\begin{equation}
|\bbE_\eta[Y]-\bbE_\eta[X_1X_2\dots X_k]|\le 3\eta\cdot \|V\|\cdot 2\le 6\eta. \label{equ: Y vs products}
\end{equation}

Additionally, by the \refaxiom{precision} of approximate counting, we have 
\[ 
\P_\delta(\lnot C^{\lor k})\le \P_\eta(\lnot C^{\lor k}) + \delta + 2\eta.
\]

Let $\eta \eqdef \beta / (20k)$. Recall that it suffices to prove $\bbE_\eta[Y]\le \gamma$. It follows that
\begin{align*} 
\P_\delta(\lnot C^{\lor k}) &\le \P_\eta(\lnot C^{\lor k}) + \delta + 2\eta \tag{\refaxiom{precision}}\\ 
&\le \bbE_\eta[Y] + \delta + 5\eta  \\ 
&\le \bbE_\eta[X_1X_2\dots X_k] + \delta + 12\eta \tag{\Cref{equ: Y vs products}} \\ 
&\le \prod_{i\in[k]}\bbE_\eta[X_i] + \delta + 8\eta k + 12\eta \tag{\refmeta{multiplication BRVs}}\\ 
&\le \prod_{i\in[k]}(\delta+8\eta+\eps) + \delta + 8\eta k + 12\eta \tag{\Cref{equ: xi upper bound prC}} \\ 
&\le (\delta+\beta+\eps)^k + \delta + \beta \le \gamma,
\end{align*}
which completes the proof. 
\end{proof}

\subsubsection{Chernoff Bound for \texorpdfstring{$O(\log n)$}{O(log n)} Random Variables}

We consider a form of the Chernoff bound where the number of random variables $m\in\Log\Log$. In a nutshell, we formalize a combinatorial proof using Binomial coefficients due to Chv\'{a}tal (see, e.g., \cite[Section 3.2]{Mulzer18}).

\begin{definition}[Binomial coefficient, in $\PV_1$]
Let $n,m\in\Log$ and $m\le n$. The binomial coefficient $\binom{n}{m}$ is defined recursively as: 
\begin{equation}
\binom{n}{0} \eqdef 1,\; \binom{n}{m} \eqdef \binom{n-1}{m-1} + \binom{n-1}{m}.\label{equ: binom def} 
\end{equation}
For $m > n$, we let $\binom{n}{m} \eqdef 0$.
Note that the function computing $(1^n,1^m)\mapsto \binom{n}{m}$ can be defined by a $\PV$ function that runs in $\poly(n)$ time, such that \Cref{equ: binom def} is provable in $\PV_1$. 
\end{definition}

\dummylabel{binomial}{Binomial Theorem}
\begin{lemma}[Binomial theorem]
The following statement is provable in $\PV_1$. For every $n\in\Log$ and $x,y\in\bbQ\setminus\{0\}$, 
\[
(x+y)^n = \sum_{i=0}^n\binom{n}{i}x^iy^{n-i}.  
\]
\end{lemma}
\begin{proof}[Proof Sketch]
Fix $n\in\Log$ and $x,y\in\bbQ$. We prove by induction on $n\in\Log$. In the base case, the equation trivially holds for $n=0$ as both sides are $1$. The induction step follows from \Cref{equ: binom def}.  
\end{proof}

Now we are ready to prove the Chernoff bound when $n\in\Log\Log$. 

\dummylabel{loglog Chernoff}{Chernoff Bound (LogLog Form)}
\begin{theorem}[Chernoff bound I, $\Log\Log$ form]
The following sentence is provable in $\APX_1$. Let $V=\zo$, $n,\delta^{-1},\beta^{-1}\in\Log$, and $m\in\Log\Log$. Let $X_1,\dots,X_m$ be a sequence of explicit i.i.d.~random variables over $V$ defined by a tuple $(V,n,C)$ and taking an $nm$-bit seed $(z_1,\dots,z_m)\in(\zo^{n})^m$. Let $p\eqdef \bbE_\delta[X_i]\in\bbQ$ for any $i\in[m]$, and $Y_{\ge k}$ be the indicator variable of $X_1+\dots+X_m\ge k$ that takes an $nm$-bit seed. Then for $t\in\bbQ$, $0\le t\le 1$, and $k\ge (1+t)p m$, 
\[
\bbE_\delta[Y_{\ge k}]\le \left(e^{-t^2p/4} + (4\delta+\beta)\cdot 2^{-pt(1+t)}\right)^m + \delta + \beta.
\]
\end{theorem}

\begin{proof}
    We argue in $\APX_1$. Fix $V=\zo$, $n,\delta^{-1},\beta^{-1}\in\Log$, $m\in\Log\Log$, and $X_1,\dots,X_m$. Let $p\eqdef \bbE_\delta[X_i]$, $t,k$, and $Y_{\ge k}$ be defined as above. Let $\eta^{-1}\in\Log$ be a parameter to be determined later. By the \refmeta{precision exp}, 
    \begin{equation}
    |\bbE_\eta[X_i] - p| = |\bbE_\delta[X_i] - \bbE_\eta[X_i]| \le \delta +2\eta. 
    \end{equation}

    \paragraph{Probability of each subset.} Let $\alpha\subseteq[m]$ encode a subset of variables. We define the random variable $Y_\alpha$ over $\zo$ as follows: Given any seed $(z_1,\dots,z_m)\in(\zo^n)^m$, $Y_\alpha = 1$ if and only if for every $i\in\alpha$, $X_i=1$. In the first step, we show that 
    \begin{equation}
    \bbE_\eta[Y_\alpha] \le (p+\delta + O(\eta m))^{|\alpha|} 
    \end{equation}
    for each $\alpha\subseteq[m]$. Note that as $m\in\Log\Log$, we can define all random variables $Y_\alpha$ for $\alpha\subseteq[m]$ by an explicit list of circuits. This will be useful later in the proof.

 Fix any $\alpha=\{i_1,\dots,i_t\}\subseteq[m]$. As $X_1,\dots,X_m$ are explicitly independent random variables, by the  \refmeta{multiplication BRVs}, we have that 
    \[
    \left|\bbE_\eta\left[\prod_{j\in[t]} X_{i_j}\right] - \prod_{j\in[t]}\bbE_\eta[X_{i_j}]\right| \le 8\eta m. 
    \]

    Subsequently, we have 
    \[
    \bbE_\eta\left[\prod_{j\in[t]} X_{i_j}\right] \le 8\eta m + (p+\delta+2\eta)^{|\alpha|}.  
    \]
    Note that for every assignment $\rho$ to the random seed, $\bbE_\eta[Y_\alpha|_\rho] = \bbE_\eta[\prod_{j\in[t]} X_{i_j}|_\rho]$. By the \refmeta{avg on exp}, we have 
    \begin{equation}
    \bbE_\eta[Y_\alpha] \le \bbE_\eta\left[\prod_{j\in[t]} X_{i_j}\right] + 6\eta \le 14\eta m+(p+\delta+2\eta)^{|\alpha|}. \label{equ: ub Y alpha}
    \end{equation}

    Let $\overline{Y}_\alpha$ be the random variable defined over $\zo$ as follows: Given any seed $(z_1,\dots,z_m)\in(\zo^n)^m$, $Y_\alpha=1$ if and only if for every $i\in[m]\setminus \alpha$, $X_i=0$. Similar to the proof above, we have that 
    \begin{equation}
    \bbE_\eta\left[\overline Y_\alpha\right] \le (1-p+\delta+10\eta m)^{m-|\alpha|} + 20\eta m.  \label{equ: ub Y alpha not}
    \end{equation}

    \paragraph{Combining all subsets.} Let $\hat Y_{\ge k}$ be the following random variable over $\zo$: 
    \[
    \hat Y_{\ge k} \eqdef \sum_{\alpha\subseteq[m], |\alpha|\ge k}Y_\alpha \overline Y_\alpha. 
    \]
    By applying the \refmeta{multiplication} and \refmeta{linearity}, we have 
    \begin{align*}
    &\quad \quad \bbE_\eta[\hat Y_{\ge k}] \\ 
    & \le 3\eta \cdot 2^m + \sum_{\alpha\subseteq[m],|\alpha|\ge k}\bbE_\eta[Y_{\alpha}\overline Y_\alpha] \tag{\refmeta{union bound}} \\ 
    & \le 3\eta \cdot 2^m + 8\eta\cdot 2^m + \sum_{\alpha\subseteq[m],|\alpha|\ge k}\bbE_\eta[Y_{\alpha}]\cdot\bbE[\overline Y_\alpha]    \tag{\refmeta{multiplication}}\\
    & \le 11\eta\cdot 2^m + \sum_{\alpha\subseteq[m],|\alpha|\ge k}\left((p+\delta+2\eta m)^{|\alpha|} + 14\eta m\right) \left((1-p+\delta+10\eta m)^{m-|\alpha|}+20\eta m\right) \tag{\Cref{equ: ub Y alpha,equ: ub Y alpha not}}\\
    & \le 90\eta m \cdot 2^m + \sum_{j\ge k}\binom{m}{j} (p+\delta+2\eta m)^{j} (1-p+\delta+10\eta m)^{m-j}. 
    \end{align*}
    Note that the binomial number in the last line is efficiently computable (even using a brute-force counting algorithm) as $m,j\in\Log\Log$. 
    
    Recall that $Y_{\ge k}$ is the random variable indicating that $X_1+\dots+X_m\ge k$. Note that for every assignment $\rho$ to the random seed, we have that $\bbE_\eta[Y_{\ge k}|_\rho] \le \bbE_\eta[\hat Y_{\ge k}|_\rho]$. Therefore, by the \refmeta{avg on exp}, we have
    \begin{equation}
    \bbE_\eta[Y_{\ge k}] \le \bbE_\eta[\hat Y_{\ge k}] + 6\eta \le \sum_{j\ge k}\binom{m}{j} (p+\delta+2\eta m)^{j} (1-p+\delta+10\eta m)^{m-j} + 100\eta m \cdot 2^m.\label{equ: Y ge k ub} 
    \end{equation}

    \paragraph{Binomial coefficient inequalities.} It remains to prove an upper bound for \Cref{equ: Y ge k ub}. Note that as $m\in\Log\Log$, we can easy formalize the standard proof in \cite[Section 3.2]{Mulzer18}, where all equalities about binomial coefficients can be proved in $\PV$. 
    
    Let $\tau=e^\lambda \ge 1$ be a parameter to be determined later. In more detail, we can perform the following calculation for any $\eps^{-1}\in\Log$:
    \begin{align*}
        & \sum_{j\ge k}\binom{m}{j} (p+\eps)^{j}(1-p+\eps)^{m-j}  \\ 
    \le & \sum_{j\ge k}\binom{m}{j} (p+\eps)^{j}(1-p+\eps)^{m-j} \tau^{j-k} + \sum_{0\le j < k} (p+\eps)^{j}(1-p+\eps)^{m-j} \tau^{j-k} \\ 
    =   & \tau^{-k} \sum_{j=0}^m \binom{m}{j} \left((p+\eps)\tau\right)^j (1-p+\eps)^{m-j} \\ 
    =   & \tau^{-k} \big(1+p(\tau-1)+\eps(\tau+1)\big)^m \tag{\refmeta{binomial}} \\ 
    \le & \left(\frac{1+p(e^\lambda -1)+\eps(e^\lambda+1)}{e^{\lambda p (1+t)}}\right)^m. 
    \end{align*}
    Note that for $\lambda \in (0,1)$, we have $1+\lambda \le e^\lambda \le 1+\lambda + 3\lambda^2/4$, and this is provable in $\PV$. Then 
    \begin{align*}
        \ln\left(1 + p(e^\lambda-1)\right) \le p(e^\lambda-1) \le p \lambda +3p\lambda^2/4. 
    \end{align*}
    This implies that $1 + p(e^\lambda-1) \le \exp(p \lambda +3p\lambda^2/4)$. We set $\eps \eqdef \delta + 30\eta m$, $\eta \le \beta/(120m)$, then $\eps \le \delta + \beta / 4$, and $\lambda\eqdef t$. Then we have 
    \[
    \left(\frac{1+p(e^\lambda -1)+\eps(e^\lambda+1)}{e^{\lambda p (1+t)}}\right)^m \le \left(e^{-t^2p/4} + 4(\delta+\beta/4)\cdot e^{-pt(1+t)}\right)^m. 
    \]
    Finally, we can obtain that 
    \begin{align*}
    \bbE_\delta[Y_{\ge k}] & \le \bbE_\eta[Y_{\ge k}] + \delta + 2\eta \tag{\refmeta{precision exp}} \\ 
    & \le \left(e^{-t^2p/4} + 4(\delta+\beta/4)\cdot e^{-pt(1+t)}\right)^m + 200\eta m \cdot 2^m + \delta + 2\eta \\ 
    & \le \left(e^{-t^2p/4} + (4\delta+\beta)\cdot e^{-pt(1+t)}\right)^m + \delta + \beta, 
    \end{align*}
    where the last inequality follows if we set $\eta \eqdef \beta / (1000 m\cdot 2^m)$. 
\end{proof}

Using essentially the same proof with $\tau = e^{-\lambda} < 1$, we can obtain a Chernoff bound for the other side of the tail probability: 

\begin{theorem}[Chernoff bound II, $\Log\Log$ form]
The following sentence is provable in $\APX_1$. Let $V=\zo$, $n,\delta^{-1},\beta^{-1}\in\Log$, and $m\in\Log\Log$. Let $X_1,\dots,X_m$ be a sequence of explicit i.i.d.~random variables over $V$ that takes an $nm$-bit seed $(z_1,\dots,z_m)\in(\zo^{n})^m$. Let $p\eqdef \bbE_\delta[X_i]\in\bbQ$, and $Y_{\le k}$ be the indicator variable of $X_1+\dots+X_m\le k$ that takes an $nm$-bit seed. Then for $t\in\bbQ$, $0\le t\le 1$, and $k\le (1-t)p m$, 
\[
\bbE_\delta[Y_{\le k}]\le \left(e^{-t^2p/4} + (4\delta+\beta)\cdot 2^{-pt(1-t)}\right)^m + \delta + \beta.
\]
\end{theorem}

\section{Theoretical Computer Science in \texorpdfstring{$\Apx_1$}{APX1}}\label{sec:formalizations}

In this section, we formalize in $\Apx_1$ several fundamental results from algorithms, complexity theory, and related areas. 

\subsection{Yao's Distinguisher-to-Predictor Transformation}\label{sec:Yao_transformation}

\begin{theorem}[Yao's transformation]\label{thm: Yao D2P}
$\Apx_1$ proves the following statement. Let $n,m,\delta^{-1},\beta^{-1}\in\Log$, $G:\zo^m\to\zo^n$ be a multi-output circuit, and $C\in B_n$ be a circuit such that 
\begin{equation}
\left|\P_\delta(C\circ G)-\P_\delta(C)\right|\ge 2\delta+\eps\label{equ: Yao condition}
\end{equation}
for some $\eps\in(0,1)\cap\bbQ$, where $C\circ G$ is the $m$-input circuit defined as $(C\circ G)(u)\eqdef C(G(u))$. 

Then there is an index $i\in[n]$ and a circuit $P:\zo^{i-1}\to\zo$ such that the following holds: Let $T(u)\in B_{m}$ be the circuit such that $T(u)=1$ if $P(G(u)_{< i}) = G(u)_{i}$ (i.e., $P$ successfully predicts the $i$-th bit of $G(u)$), then 
\[
\left|\P_\delta(T)-\frac{1}{2}\right|\ge \frac{\eps}{4n} - (\delta + \beta). 
\]
\end{theorem}

\begin{proof}
We formalize the standard proof of Yao's lemma in $\Apx_1$ (see, e.g., \cite[Chapter 9]{AB09}). Fix $n,m,\delta^{-1},\beta^{-1}\in\Log$, $G:\zo^m\to\zo^n$, and circuit $C\in B_n$. For every index $i\in\{0,1,\dots,n\}$, we define the circuit $C_i:\zo^m\times \zo^n\to\zo$ as follows: 
\begin{itemize}
    \item $C_i$ parses its input as $(u,x)\in \zo^m\times \zo^n$.
    \item Let $z\eqdef G(u)_{\le i}\circ x_{>i}$, i.e., the string where the first $i$ bits agree with the first $i$ bits of $G(u)$, and the remaining bits agree with the last $n-i$ bits of $x$. The circuit $C_i(u,x)$ then outputs $C(z)$.  
\end{itemize}
Let $X_0,X_1,\dots,X_n$ be the random variables over $\{0,1\}$ with seed length $m+n$, where $X_i$ is the indicator variable of $C_i(j,x)=1$. That is, $X_i$ is defined by the tuple $(\zo, m+n, C_i)$. 

\pgr{Step 1: Gap between $\bbE_\eta[X_0]$ and $\bbE_\eta[X_n]$} We first argue that $|\bbE_\eta[X_0]-\bbE_\eta[X_n]|$ is large. Note that by \Cref{equ: Yao condition} and the \refaxiom{precision}, we have 
\begin{equation}\label{equ: Cr and C far}
    \left|\P_\eta(C\circ G)-\P_\eta(C)\right|\ge \eps-4\eta.
\end{equation}

Note that $X_0$ is the random variable over $\zo$ with seed length $m+n$ that, on the seed $(u,x)\in\zo^m\times \zo^n$, ignores the first part of the seed and outputs $C(x)$. For any assignment $u\in\zo^m$ to the first part of the seed, it can be verified that $|\bbE_\eta[X_0|_u]-\P_\eta(C)| \le 3\eta$, and thus by the \refmeta{avg on exp}, we know that 
$$|\bbE_\eta[X_0]-\P_\eta(C)|\le 6\eta.$$
Similarly, we can prove that 
\begin{equation}\label{equ: Xn close to YC}
    \left|\bbE_\eta[X_n]-\P_\eta(C\circ G)\right|\le 6\eta. 
\end{equation}
This is because $X_n$ is the random variable that, on the seed $(u,x)\in\zo^m\times \zo^n$, ignores the second part of the seed and outputs $C\circ G(u)$. 

By combining \Cref{equ: Cr and C far} and \Cref{equ: Xn close to YC}, we have 
\begin{equation}
|\bbE_\eta[X_0]-\bbE_{\eta}[X_n]| \ge \eps - 16\eta 
\end{equation}

\pgr{Step 2: Gap between $\bbE_\eta[X_i]$ and $\bbE_{\eta}[\overline{X_{i+1}}]$} As $|\bbE_\eta[X_0]-\bbE_{\eta}[X_n]| \ge \beta - 16\eta+\eps$, we know that for some $1\le i\le n$, 
\begin{equation}
|\bbE_\eta[X_{i-1}]-\bbE_{\eta}[X_{i}]| \geq \frac{\eps}{n}-\frac{16\eta}{n}.\label{equ: diff xi xsi} 
\end{equation}
In more detail, suppose towards a contradiction that this is not true. We can prove by induction on $j$ that if $j\le n$, $|\bbE_\eta[X_0]-\bbE_{\eta}[X_j]| < (j/n)\cdot (\eps-16\eta)$. 

We will produce a predictor $P$ for this index $i$. Recall that both $X_{i-1}$ and $X_i$ are random variables that parse their seeds as $(u,x)\in\zo^{m}\times \zo^n$, and  
\begin{compactitem}
\item $X_{i-1}$ outputs $C(G(u)_{<i}\circ x_i\circ x_{>i})$; 
\item $X_i$ outputs $C(G(u)_{<i}\circ G(u)_i\circ x_{>i})$.
\end{compactitem}
Let $\overline{X_i}$ be the random variable that parses its input as $(u,x)\in\zo^m\times \zo^n$ and outputs $C(G(u)_{<i}\circ \overline{G(u)_i}\circ x_{>i})$. We will prove 
\begin{equation}
\left|\bbE_\eta[X_i] - \bbE_{\eta}[\overline{X_i}]\right| > \frac{\eps}{2n} - 100\eta \label{equ: diff xi not xi}
\end{equation} 
by rewriting both $\bbE_\eta[X_{i-1}]$ and $\bbE_\eta[X_i]$ in \Cref{equ: diff xi xsi}. 

Let $X_{i-1}^0, X_{i-1}^1$ be the random variables over $\zo$ that parse their seeds as $(u,x)\in\zo^m\times \zo^n$ satisfying that: 
\begin{compactitem}
\item $X_{i-1}^0$ outputs $1$ if and only if $C(G(u)_{<i}\circ x_i\circ x_{>i})=1$ and $G(u)_i=x_i$. 
\item $X_{i-1}^1$ outputs $1$ if and only if $C(G(u)_{<i}\circ \overline{G(u)_i}\circ x_{>i})=1$ and $G(u)_i\ne x_i$. 
\end{compactitem} 
One can observe that for any assignment $\rho$ to their seeds, $X_{i-1}|_\rho = X_{i-1}^0|_\rho+ X_{i-1}^1|_\rho$, and thus by the \refmeta{linearity}, 
\begin{equation}
\left|\bbE_\eta[X_{i-1}] - \left(\bbE_{\eta}[X_{i-1}^0] + \bbE_{\eta}[X_{i-1}^1]\right)\right| \le 6\eta. \label{equ: xii close 0 plus 1}
\end{equation}

Similarly, let $X_{i}^0, X_{i}^1$ be the random variables over $\zo$ that parse their seeds as $(u,x)\in\zo^m\times \zo^n$ satisfying that:
\begin{compactitem}
\item $X_{i}^0$ outputs $1$ if and only if $C(G(u)_{<i}\circ x_i\circ x_{>i})=1$ and $G(u)_i=x_i$. 
\item $X_{i}^1$ outputs $1$ if and only if $C(G(u)_{<i}\circ G(u)_i\circ x_{>i})=1$ and $G(u)_i\ne x_i$. 
\end{compactitem}
For any assignment $\rho$ to their seeds, $X_i|_\rho = X_{i}^0|_\rho + X_i^1|_\rho$, and thus by the \refmeta{avg on exp}, 
\begin{equation}
\left|\bbE_\eta[X_{i}] - \left(\bbE_{\eta}[X_{i}^0] + \bbE_{\eta}[X_{i}^1]\right)\right| \le 6\eta.  \label{equ: xi close 0 plus 1}   
\end{equation}

One can observe that $X_i^0$ and $X_{i-1}^0$ are exactly the same random variable. Moreover, we argue that 
\begin{equation}
\left|\bbE_\eta[X_i] - 2\cdot \bbE_\eta[X_i^1]\right| \le 20\eta. \label{equ: xi close to xi 1}
\end{equation}
(As a sanity check, $\bbE[X_i]=2\cdot \bbE[X^1_i]$ in exact expectation.) Assume for contradiction that this does not hold. By the \refmeta{avg on exp}, there is an assignment $\rho$ to all but $x_i$ such that 
\[
\left|\bbE_\eta[X_i|_\rho] - 2\cdot \bbE_\eta[X_i^1|_\rho]\right| > 10\eta. 
\]
Note that both $X_i|_\rho$ and $X_{i}^1|_\rho$ have seed length $1$, and we know that $\bbE[X_i|_\rho] = 2\cdot \bbE[X_i^1|_\rho]$. This leads to a contradiction by the \refmeta{brute force}. Similarly, we can prove that 
\begin{equation}
\left|\bbE_\eta[X_{i-1}] - 2\cdot \bbE_\eta[X_{i-1}^1]\right| \le 10\eta.\label{equ: xii close to xii 1}
\end{equation}

By \Cref{equ: diff xi xsi,equ: xii close 0 plus 1,equ: xi close 0 plus 1,equ: xi close to xi 1,equ: xii close to xii 1}, we have 
\begin{align*}
    \left|\bbE_\eta[X_i] - \bbE_{\eta}[\overline{X_i}]\right| & \ge \left|\left(\bbE_{\eta}[X_{i-1}^0] + \bbE_{\eta}[X_{i-1}^1]\right) - \left(\bbE_{\eta}[X_{i}^0] + \bbE_{\eta}[X_{i}^1]\right)\right| - 12\eta\tag{\Cref{equ: xii close 0 plus 1,equ: xi close 0 plus 1}} \\ 
    & \ge \left|\bbE_\eta[X_{i-1}^1] - \bbE_\eta[X_{i}^1]\right| - 12\eta \\ 
    & \ge \frac{\left|\bbE_\eta[X_{i-1}] - \bbE_\eta[X_{i}]\right|}{2} - 60\eta \tag{\Cref{equ: xi close to xi 1,equ: xii close to xii 1}}\\ 
    & > \frac{\eps}{2n} - 100\eta. \tag{\Cref{equ: diff xi xsi}}
\end{align*}

\pgr{Step 3: Producing the Predictor} We now prove that \Cref{equ: diff xi not xi} suffices to produce the predictor. Let $P:\zo^{i-1}\times \zo^n\to\zo$ and $T:\zo^m\times \zo^n\to \zo$ be the following circuits: $P(v,x)\eqdef C(v\circ x_i\circ x_{>i})\oplus x_i$, and $T(u,x)$ outputs $1$ if and only if $P(G(u)_{<i},x) = G(u)_i$. 

We also define two circuits $T_0,T_1:\zo^m\times \zo^n\to \zo$ such that 
\begin{compactitem}
\item $T_0(u,x)$ outputs $1$ if $C(G(u)_{<i}\circ G(u)_i\circ x_{>i}) = 0$ and $x_i = G(u)_i$. 
\item $T_1(u,x)$ outputs $1$ if $C(G(u)_{<i}\circ \overline{G(u)_i}\circ x_{>i})=1$ and $x_i \ne G(u)_i$.
\end{compactitem}
One can prove (in $\PV_1$) that $T_0(u,x) = 1$ if and only if $T(u,x)=1$ and $x_i = G(u)_i$, while $T_1(u,x) = 1$ if and only if $T(u,x)=1$ and $x_i \ne G(u)_i$. Therefore for every $(u,x)\in\zo^m\times \zo^n$, $T(u,x) = T_0(u,x) + T_1(u,x)$. By considering their indicator variables and applying the \refmeta{linearity}, one can prove that 
\begin{equation}
\left|\P_\eta(T) - (\P_\eta(T_0) + \P_\eta(T_1))\right| \le 50\eta.\label{equ: T close to T0 plus T1} 
\end{equation}

Next, we argue that 
\begin{equation}
\left|\bbE_\eta[X_i] - \left(1-2\cdot \P_\eta(T_0)\right)\right| \le 20\eta. \label{equ: Xi close to PT0}
\end{equation}
Suppose, towards a contradiction, that \Cref{equ: Xi close to PT0} does not hold. Let $I_{T_0}$ be the indicator random variable of $T_0$, and we have $|\bbE_\eta[I_{T_0}]-\P_\eta(T_0)|\le 3\eta$ (see \Cref{prop: indicator variable}). Both $X_i$ and $I_{T_0}$ have seed $(u,x)\in\zo^m\times \zo^n$. By the \refmeta{avg on exp}, there is an assignment $\rho$ to all variables but $x_i$ (i.e.~the $i$-th bit of $x$) such that 
\[
\left|\bbE_\eta[X_i|_\rho] - \left(1-2\cdot \bbE_\eta[I_{T_0}|_\rho]\right)\right| > 10\eta. 
\]
Note that both $X_i|_\rho$ and $I_{T_0}|_\rho$ have seed length $1$, and we know by the definition that\footnote{Indeed, since after fixing $\rho$ the only randomness is $x_i$ and $X_i\!\restriction_\rho$ is independent of $x_i$, we have $I_{T_0}\!\restriction_\rho
=\mathbf{1}\{x_i=G(u)_i\}\,\mathbf{1}\{X_i\!\restriction_\rho=0\}
=\mathbf{1}\{x_i=G(u)_i\}\,\bigl(1-X_i\!\restriction_\rho\bigr)$. Consequently, 
$\mathbb{E}\!\bigl[I_{T_0}\!\restriction_\rho\bigr]
= \tfrac12 \cdot \bigl(1-X_i\!\restriction_\rho\bigr)
=\tfrac12 \cdot \Bigl(1-\mathbb{E}\bigl[X_i\!\restriction_\rho\bigr]\Bigr)$, which implies that $\mathbb{E}\bigl[X_i\!\restriction_\rho\bigr]
=1-2 \cdot \,\mathbb{E}\bigl[I_{T_0}\!\restriction_\rho\bigr].$
} 
\[
\left|\bbE[X_i|_\rho] - \left(1-2\cdot \bbE[I_{T_0}|_\rho]\right)\right| = 0. 
\]
By the \refmeta{brute force}, we conclude a contradiction that thus proves \Cref{equ: Xi close to PT0}. Similarly, we can prove that 
\begin{equation}
\left|\bbE_\eta[\overline{X_i}] - 2\cdot \P_\eta(T_1)\right| \le 20\eta. \label{equ: not Xi close to PT1}
\end{equation}

By combining \Cref{equ: diff xi not xi,equ: T close to T0 plus T1,equ: Xi close to PT0,equ: not Xi close to PT1}, we have that 
\begin{align*}
\left|\P_\eta(T) - \frac{1}{2}\right|&\ge \left|\P_\eta(T_0) + \P_\eta(T_1)-\frac{1}{2}\right| - 50\eta \tag{\Cref{equ: T close to T0 plus T1}} \\ 
&\ge \left|\frac{1-\bbE_\eta[X_i]}{2} + \frac{\bbE_{\eta}[\overline{X_i}]}{2} - \frac{1}{2}\right| - 80\eta \tag{\Cref{equ: Xi close to PT0,equ: not Xi close to PT1}} \\ 
&\ge \frac{1}{2}\cdot \left|\bbE_\eta[X_i]-\bbE_\eta[\overline{X_{i}}]\right| - 80 \eta \\ 
&\ge \frac{\eps}{4n} - 180\eta \tag{\Cref{equ: diff xi not xi}}.
\end{align*}
Recall that $T$ takes an input $(u,x)\in\zo^m\times \zo^n$. Let $T_x$ be the circuit obtained by fixing the second part of its input to be $x\in\zo^n$. By considering its indicator variable (see \Cref{prop: indicator variable}) and applying the \refmeta{avg on exp}, there exists a string $x\in\zo^n$ such that 
\[
\left|\P_\eta[T_x]-\frac{1}{2}\right| \ge \frac{\eps}{4n} - 200\eta.  
\]
By the definition, one can see that $T_x$ evaluating to $1$ means the predictor $P_x(v)\eqdef C(v \circ x \circ x_{>i})\oplus x_i$ correctly predicts the $i$-th bit of $G(u)$. This concludes the proof by applying the \refmeta{global} of approximate counting, the \refaxiom{precision}, and setting $\eta\eqdef \beta/400$. 
\end{proof}

\subsection{Schwartz-Zippel Lemma}

Before stating the Schwartz-Zippel lemma, we need to clarify the formalization of finite fields and polynomials. A finite field $\F$ is said to be feasible if $|\F|\in\Log$, e.g., $\F=\F_p$ for some prime number $p\in\Log$. It is verified in \cite[Section 4.3.3]{Jerabek-phd} that for any feasible field, the field elements can be encoded such that (1) the field operations can be implemented by $\PV$ function symbols, and (2) field axioms can be proved in $\PV$. For simplicity, we identify an element in $\F$ and its encoding as a string. 

Fix a feasible field $\F$. A degree-$d$ univariate polynomial $p\in\F[x]$ can be defined by a list of coefficients $c_0,c_1,\dots,c_d\in\F$ such that $p(x)\eqdef c_0+c_1x+\dots+c_d x^d$. A polynomial is said to be nonzero if any of $c_0,c_1,\dots,c_d$ is nonzero.

\begin{proposition}[Implicit in {\cite[Lemma 4.3.6]{Jerabek-phd}}]\label{prop: root upper}
It is provable in $\PV$ that any nonzero degree-$d$ polynomial $p\in\F[x]$ has at most $d$ roots.  
\end{proposition}

\begin{proposition}[Implicit in \cite{Jerabek-phd}]
It is provable in $\PV$ that for any $d< |\F|$, any distinct $x_1,\dots,x_d\in \F$, and $y_1,\dots,y_d\in\F$, there is a degree-$d$ polynomial $p\in\F[x]$ such that $p(x_i)=y_i$ for every $i\in[d]$. 
\end{proposition}

Let $m,d\in\Log$. We say that a circuit $C$ \emph{computes an $m$-variate function in $\F$}, denoted by $C:\F^m\to\F$, if for every $x_1,\dots,x_m\in\F$, $C(x_1,\dots,x_m) \in\F$. Depending on the encoding of field elements, some circuit may not compute an $m$-variate function in $\F$ as it outputs a string that does not encode any field element. We say that $C:\F^m\to\F$ has \emph{individual degree at most $d$} if for every $i\in[m]$ and any assignment $\rho$ to all but the $i$-th variable, there is a polynomial $p_{i,\rho}\in\F[x]$ of degree at most $d$ such that $p_{i,\rho}(x)=C|_\rho(x)$ for $x\in\F$. 

\begin{theorem}[Schwartz-Zippel Lemma]\label{thm: SW}
    $\Apx_1$ proves the following statement. Let $\F$ be a feasible field such that each field element is encoded by a string of length $b\in\Log\Log$. Let $m,d\in\Log$, $d < |\F|$, and $C:\F^m\to\F$ be a circuit of individual degree at most $d$. Let $T_C:(\zo^{b})^m\to \zo$ be the circuit that given $(x_1,\dots,x_m)\in\zo^b$, it accepts if $x_i\in\F$ for each $i\in[m]$, and $C(x_1,\dots,x_m)=0$. 
    
    Suppose that for some $\vec z\in \F^m$, $C(\vec z)\ne 0$. Then for every $\delta^{-1},\beta^{-1}\in\Log$, $\P_\delta(T_C)\le md/|\F|+\delta + \beta$. 
\end{theorem}

\begin{proof}
    The key idea is to formalize the proof of Atserias and Tzameret \cite{AT25} in $\Apx_1$. We argue in $\Apx_1$. Fix any feasible field $\F$, $b\in\Log\Log$, $m,d\in\Log$, $C:\F^m\to\F$, $\vec z=(z_1,\dots,z_m)\in\F_m$, and $\delta^{-1},\beta^{-1}\in\Log$, as in the statement. 
    
    Let $\eta^{-1}\in\Log$ be a parameter to be determined later. For every $i\in\{0,1,\dots,m\}$, we define $T_C^i:(\zo^b)^m\to\zo$ as follows: Given $(x_1,\dots,x_m)\in\zo^b$, it accepts if $x_i\in\F$ for each $i\in[m]$ and $C(x_1,\dots,x_{i},z_{i+1},\dots,z_m) = 0$. It is clear that $T^0_C\equiv \Null$, and $T^m_C\equiv T_C$. 
    
    Let $X_0,X_1,\dots,X_m$ be the indicator random variables for $T_C^0,T_C^1,\dots,T_C^m$. We will prove that for every $i\ge 1$, 
    \begin{equation} 
    \bbE_\eta[X_{i}]-\bbE_\eta[X_{i-1}]\le d/|\F| + 10\eta.\label{equ: diff small Xi vs Xi-1}
    \end{equation}
    Assume for contradiction that it is not the case. By \refmeta{avg on exp}, there is an assignment $\rho$ to all but the $i$-th part $x_i$ of the seed such that 
    \begin{equation}
    \bbE_\eta[X_i|_\rho] - \bbE_{\eta}[X_{i-1}|_\rho] >  d/|\F|+10\eta-6\eta>0.\label{equ: lb from contradiction Xi vs Xi-1} 
    \end{equation}
    Note that $X_{i-1}|_\rho$ is the indicator random variable of $T_C^{i-1}|_\rho$, and by the definition, $T_C^{i-1}|_\rho$ is a constant circuit that does not read the seed. Therefore it falls into one of the two cases: 
    \begin{itemize}
    \item Suppose that $T_C^{i-1}|_\rho\equiv \True$. We also know that the circuit $\EQ(T_C^{i-1}|_\rho,1)$ that given $x\in\zo^b$, outputs $1$ if and only if $T_C^{i-1}(x)=1$ is also functionally equivalent to $\True$. One can prove by the definition that 
    \[
    \bbE_\eta[X_{i-1}|_\rho] = \P_\eta(\EQ(T^{i-1}_C|_\rho, 1)) = 1, 
    \]
    where the second inequality follows from the \refaxiom{boundary}. This leads to a contradiction to \Cref{equ: lb from contradiction Xi vs Xi-1}.
    \item Otherwise, $T_C^{i-1}|_\rho\equiv \Null$. Let $\rho=(y_1,\dots,y_{i-1},*,y_{i+1},\dots,y_m)$. Recall that as the individual degree of $C$ is at most $d$, there is a polynomial $p\in\F[x]$ such that $p\equiv C(y_1,\dots,y_{i-1},\cdot,z_{i+1},\dots,z_m)$. It follows from $T_C^{i-1}|_\rho\equiv \Null$ and \Cref{equ: lb from contradiction Xi vs Xi-1} that 
    \[
    p(z_i) = C(y_1,\dots,y_{i-1},z_i,z_{i+1},\dots,z_m) \ne 0,
    \]
    and thus $p$ is a nonzero polynomial. By \Cref{prop: root upper}, it has at most $d$ roots. Moreover, as the input length of $T_C^{i-1}|_\rho$ is $b\in\Log\Log$, we can prove by \Cref{prop: indicator variable} and \refmeta{brute force} that 
    \[
    \bbE_\eta[X_i|_\rho]\le \P_\eta(T_C^i)+2\eta  \le d/|\F|+4\eta.  
    \]
    This leads to a contradiction to \Cref{equ: lb from contradiction Xi vs Xi-1}.
    \end{itemize}

    Finally, as $T^0_C\equiv\Null$, we know that $\bbE_\eta[X_0] \leq 10 \eta$. By induction on $i$ (using the \refaxiom{ind}) and \Cref{equ: diff small Xi vs Xi-1}, we can prove that 
    \begin{equation}
    \bbE_\eta[X_m]\le md/|\F|+10\eta\cdot (m+1).\label{equ: Xm small hybrid result} 
    \end{equation}
    Subsequently, we have 
    \begin{align*}
    \P_\delta(T_C) &\le \P_\eta(T_C)+ \delta + 2\eta \tag{\refaxiom{precision}} \\ 
    &\le \P_\eta(T_C^m) + \delta + 5\eta \tag{\refmeta{global}}\\
    &\le \bbE_\eta(X_m) + \delta + 8\eta \tag{\Cref{prop: indicator variable}} \\ 
    &\le md/|\F|+\delta+ 10\eta \cdot (m+1) + 8\eta \\ 
    &\le md/|\F| + \delta + \beta,
    \end{align*}
    where the last inequality follows by taking $\eta\eqdef \beta / (10m+20)$. 
\end{proof}

\subsection{Linear Hashing}
\label{ss: linear hash}

The linear hash function $x\mapsto Ax\bmod 2$ is one of the simplest constructions of hash functions. The following theorem formalizes that linear hashing is an (almost) universal hash function. 

\begin{theorem}[Universality of linear hashing]\label{thm: linear hashing}
    $\Apx_1$ proves the following statement. Let $n,m,\delta^{-1},\beta^{-1}\in\Log$. For every $x,y\in\zo^n$, let $T_{x,y}:\zo^{nm}\to\zo$ be the circuit that parses its input as a Boolean matrix $A\in\zo^{m\times n}$ and outputs $1$ if and only if $Ax\equiv A y\pmod 2$. Then for all distinct $x,y\in\zo^n$, $\P_\delta(T_{x,y})\le \delta + \beta + (1/2 +\beta)^{m}$.  
\end{theorem}

\begin{proof}
    We argue in $\APX_1$. Fix $n,m,\delta^{-1},\beta^{-1}\in\Log$, $x,y\in\zo^n$, and let $T_{x,y}$ be the circuit defined above. Suppose that $x\ne y$. We will prove that $\P_\delta(T_{x,y})\le \delta + \beta + (1/2 +\beta)^{m}$. 

    We first upper bound the probability when $m=1$. Let $\eta^{-1}\in\Log$ be a parameter to be determined later, and $C:\zo^n\to\zo$ be the circuit that given $c\in\zo^n$, it outputs $\langle c,x-y\rangle\bmod 2$. We will prove that $\P_\eta(\lnot C)\le 1/2 + 6\eta$.

    Suppose, towards a contradiction, that $\P_\eta(\lnot C) > 1/2 + 6\eta$. As $x\ne y$, there is an index $i\in[n]$ such that $x_i\ne y_i$. Fix the index $i$. Let $X$ be the indicator random variable of $\lnot C$; by \Cref{prop: indicator variable}, we know that $\bbE_\eta[X]>1/2+6\eta-3\eta =  1/2+3\eta$. Subsequently, by \refmeta{avg on exp}, there is an assignment $\rho$ to all but the $i$-th bit of the random seed such that $\bbE_\eta[X|_\rho] > 1/2$. However, as the seed length of $X|_\rho$ is $1$ and the probability is $1/2$, this violates the \refmeta{brute force}. 

    Let $C^{\lor m}:\zo^{nm}\to\zo$ denote the circuit $C^{\lor m}(x_1,\dots,x_m)\eqdef \bigvee_{i\in[m]} C(x_i)$. By \refmeta{os error red}, we have
    \[
    \P_\eta(\lnot C^{\lor m}) \le (1/2 + 8\eta)^m + 2\eta.  
    \]
    It can be observed that $\lnot C^{\lor m}$ is functionally equivalent to $T_{x,y}$, and thus by the \refmeta{global} of approximate counting, 
    \[
    \P_\eta(T_{x,y}) \le \P_\eta(\lnot C^{\lor m}) + 3\eta \le (1/2 + 8\eta)^m + 5\eta.  
    \]
    By the \refaxiom{precision}, we have 
    \[
    \P_\delta(T_{x,y})\le \P_\eta(T_{x,y}) + \delta + 2\eta \le (1/2 + 8\eta)^m + \delta + 7\eta. 
    \]
    This completes the proof when we take $\eta\eqdef \beta/20$. 
\end{proof}

\subsection{Lower Bounds for Parity Against \texorpdfstring{$\AC^0$}{AC0} Circuits}\label{sec:AC0_LB}

The first result of this section is a formalization in $\APX_1$ of an average-case lower bound for the Parity function $\oplus_n$ against $\AC^0$. Our proof is based on a technique due to Furst, Saxe, and Sipser \cite{DBLP:journals/mst/FurstSS84}. Previous formalizations of the lower bound\footnote{Both results consider only worst-case lower bounds, but it can be verified that an average-case lower bound follows from a similar argument when formalized appropriately (in the style of \Jerabek{} \cite{Jerabek-phd,Jerabek07}).}  due to M\"uller and Pich \cite{DBLP:journals/apal/MullerP20} (following \cite{DBLP:journals/mst/FurstSS84}) and \Krajicek{} \cite[Theorem 15.2.3]{Krajicek-book} (following Razborov's proof of the switching lemma \cite{Razborov-switching-l}) require the theory $\APC_1=\PV_1+\dWPHP(\PV)$.\footnote{Note that both results automatically give a formalization of the worst-case lower bound in $\PV_1+\rWPHP(\PV)$, as $\APC_1$ is $\forall\Sigma_1^b$-conservative over $\PV_1+\rWPHP(\PV)$ \cite{Jerabek04,Jerabek07} and the worst-case lower bound can be formalized as a $\forall\Sigma_1^b$-sentence (see, e.g., \cite[Theorem 1.1]{DBLP:journals/apal/MullerP20}).}

\AvgParity*

The main technical challenge is to avoid ``encoding-based counting argument'' that rely on $\dWPHP(\PV)$, which is not available in $\APX_1$. The encoding-based counting argument is used in both \cite{DBLP:journals/apal/MullerP20} and Razborov's \cite{Razborov-switching-l} proof of the switching lemma. This was partially addressed by Agrawal et al.~\cite{AAIPR01} (see also \cite{DBLP:conf/coco/Agrawal01})\footnote{As mentioned in \cite{AAIPR01}, the result was known to Ajtai and Wigderson (unpublished).}, which presented a deterministic polynomial-time algorithm that outputs a suitable restriction given by the switching lemma. As one of our contributions, we show that the correctness of the algorithm in Agrawal et al.~\cite{AAIPR01} can be established in $\PV_1$. This, together with tools developed in \Cref{sec: basic properties}), allow us to formalize the \emph{average-case} lower bound in $\APX_1$. 

Interestingly, using the same technique, we further show that the \emph{worst-case} lower bound $\oplus_n\notin\AC^0$ can be formalized in $\PV_1$. This resolves an open problem from \cite{DBLP:journals/apal/MullerP20}. 

\WorstParity*

\paragraph{Notation.} We let $n$ denote the number of input variables. A $k$-CNF is a propositional formula of the form $C_1\land C_2\land \cdots \land C_m$, where each clause $C_i$ is a disjunction of at most $k$ literals, i.e., variables or their negations. Similarly, a $k$-DNF is a disjunction of terms, where each term is a conjunction of at most $k$ literals. We say that a formula is a $k$-NF if it is either a $k$-CNF or a $k$-DNF. A clause\footnote{In the context of $k$-NFs, we  use ``clause'' to refer to a subformula, regardless of the connective.} $C$ of a $k$-NF can be described by its type (i.e.~$\land$ or $\lor$) and two subsets $S_C^+,S_C^-\subseteq[n]$ of size $|S_C^+|+|S_C^-|\le k$, where $S_C^+$ denotes the ID of variables in the clause, and $S_C^-$ denotes the ID of negations of variables in the clause. Let $S_C\eqdef S_C^+\cup S_C^-$. 

Note that we may assume without loss of generality that $S_C^+\cap S_C^-=\varnothing$, as otherwise the clause will be either always $0$ or always $1$. 

For simplicity, we assume without loss of generality that $\AC^0$ circuits satisfy the following properties:
\begin{compactitem}
\item The circuit is \emph{layered}, i.e., gates in each layer are fed by gates only in previous layer. 
\item All negation gates are pushed to be directly above input variables. Equivalently, there is no negation gate inside the circuit, i.e., gates in the first layer can be fed by literals, i.e., input variables or their negations. 
\end{compactitem}
An $\AC^0$ circuit satisfying these properties is called a \emph{well-formed} circuit. Note that an arbitrary $\AC^0$ circuit can be transformed into a well-formed circuit with only a polynomial size overhead, and the correctness of the transformation can be proved in $\PV_1$. We start directly from well-formed circuits to simplify calculations.  

\subsubsection{Deterministic Selection of Subset for Restriction}

We start by stating two lemmas that formalize the core combinatorial property used in the proof of Furst, Saxe, and Sipser \cite{DBLP:journals/mst/FurstSS84}. Similar functions are used implicitly in \cite{AAIPR01}. Here we use potential functions in order to implement a derandomization via the method of conditional expectations. The approach provides both a feasible algorithm and a feasible proof.

\newcommand{\wide}{\mathsf{s}}

\begin{lemma}[Potential Function for Small Sets]\label{lmm: potential small sets}
The following sentence is provable in $\PV_1$ for every constant $c\ge 1$. Let $n,t\in\Log$ with $t \leq n$, $s,m\in\Log\Log$, $s\le m$, $p\eqdef t/n\in\bbQ$, and $S_1,\dots,S_m\subseteq [n]$ be disjoint sets of size at most $c$ such that $pc<1$. There is a circuit $\Phi_\wide:\{*,\circ\}^{\le n}\to\bbQ\cap [0,1]$ such that the following holds. 
\begin{compactitem}
\item \emph{(Initial Condition)}. $\Phi_\wide(\eps) \le m^s(pc)^{m-s}$. 
\item \emph{(Recursion Condition)}. For every $x\in\{*,\circ\}^i$, $i<n$, we have 
\[
\Phi_\wide(x) = p\cdot \Phi_\wide(x *) + (1-p)\cdot \Phi_\wide(x\circ). 
\]
\item \emph{(Termination Condition)}. For every $x\in\{*,\circ\}^n$, $\Phi_\wide(x)\in\zo$. Moreover, let $T_x\eqdef \{i\in[n]\mid x_i=\circ\}$. Then $\Phi_\wide(x)=0$ if and only if there are at least $s$ subsets $S\in\{S_1,\dots,S_m\}$ such that $S\subseteq T_x$. 
\end{compactitem}
\end{lemma}

\begin{proof}[Proof Sketch]
We argue in $\PV_1$. Let $n,t,s\in\Log$, $m\in\Log\Log$, $p\eqdef t/n$, and $S_1,\dots,S_m\subseteq [n]$. The circuit $\Phi_\wide$ is defined as follows: Given $x\in\{*,\circ\}^{i}$ for some $0\le i\le n$, let $T_x\eqdef \{j\le i\mid x_j=\circ\}$. It outputs 
\[
\Phi_\wide(x)\eqdef \sum_{\alpha\subseteq[m],|\alpha|< s} \phi(x,\alpha),
\]
where 
\begin{align*}
& \phi(x,\alpha)\eqdef \begin{cases}
    \displaystyle0 & \exists j\in \alpha~\exists k\in[i]~(k\in S_j\land x_k=*)\\ 
    \displaystyle\prod_{j\in[\alpha]}(1-p)^{|S_j\setminus T_x|}\prod_{j'\in[m]\setminus \alpha}\psi(x,j') & \text{otherwise}
\end{cases}, \\ 
& \psi(x,j')\eqdef \begin{cases}
    1 & \exists k\in [i]~(k\in S_{j'}\land x_k = *) \\ 
    1 - (1-p)^{|S_{j'}\setminus T_x|} &\text{otherwise}
\end{cases}. 
\end{align*}

For instructive purposes, we mention the combinatorial interpretation of the functions (which is not a part of the $\PV_1$ proof): Given any $x\in\{*,\circ\}^i$, we randomly assign $x_{i+1},\dots,x_{n}$ independently to $*$ with probability $p$ and to $\circ$ with probability $1-p$. Let $T_x'\eqdef \{i\in[n]\mid x_i=\circ\}$. Then 
\begin{compactitem}
\item $\psi(x,j')$ is the probability that $S_{j'}\nsubseteq T_x'$. 
\item $\phi(x,\alpha)$ is the probability that $S_{j}\subseteq T_x'$ if and only if $j\in \alpha$. 
\item $\Phi_\wide(x)$ is the probability that at most $s-1$ subsets $S\in\{S_1,\dots,S_m\}$ satisfy $S\subseteq T_x$. 
\end{compactitem}

We come back to the $\PV_1$ proof. The recursion condition and termination condition can be verified by a tedious but straightforward calculation, which we omit here. To prove the initial condition, notice that 
\begin{align*}
& \psi(\eps,j') = 1 - (1-p)^{|S_j'|} \le 1 - (1-p)^c \le pc\\ 
& \phi(\eps,\alpha) \le \prod_{j'\in[m]\setminus \alpha}\psi(\eps,j') \le (pc)^{m-|\alpha|} \\ 
& \Phi_\wide(\eps) \le \sum_{0\le j<s}\binom{m}{j}(pc)^{m-j} \le m^s(pc)^{m-s}
\end{align*}
The argument can be implemented in $\PV_1$ using that $s, m \in \Log\Log$ and $c$ is constant. This completes the proof. 
\end{proof}

\begin{lemma}[Potential for General Set Systems]\label{lmm: potential general}
For every choice of constants $k,c\ge 1$, there are constants $b,n_0\ge 1$ such that the following sentence is provable in $\PV_1$. Let $n,t,m\in\Log$, $n>n_0$, $p\eqdef t/n\in\bbQ$, and $S_1,S_2,\dots,S_m\subseteq[n]$ be nonempty subsets of size at most $c$. Then there is a circuit $\Phi:\{*,\circ\}^{\le n}\to\bbQ\cap[0,1]$ such that the following holds. 
\begin{compactitem}
\item \emph{(Initial Condition)}. If $t\le \sqrt n$, then $\Phi(\eps) \le n^{-k}$. 
\item \emph{(Recursion Condition)}. For every $x\in\{*,\circ\}^i$, $i < n$, we have 
\[
\Phi(x) = p\cdot \Phi(x *) + (1-p)\cdot \Phi(x\circ). 
\]
\item \emph{(Termination Condition)}. For every $x\in\{*,\circ\}^n$, $\Phi(x)\in\zo$. Moreover, for $T_x\eqdef \{i\in[n]\mid x_i=\circ\}$, if $\Phi(x)=0$, then one of the following conditions holds: 
\begin{compactenum}
\item $|S_1\cup S_2\cup \dots\cup S_m\setminus T_x|\le b$.
\item There are disjoint nonempty sets $V_1,\dots,V_\ell\subseteq[n]$ that are subsets of $\ell$ distinct sets among $S_1,\dots,S_m$, such that $V_i\subseteq T_x$ for every $i\in[\ell]$, where $\ell \ge k\ln n$. 
\end{compactenum}
Moreover, the disjoint sets $V_1,\dots,V_\ell$ can be obtained by a $\PV$ function given $S_1,\dots,S_m$ and $x\in\{*,\circ\}^n$. 
\end{compactitem}
\end{lemma}

\newcommand{\calU}{\mathcal{U}}
\newcommand{\calV}{\mathcal{V}}

\begin{proof}
    Fix any constants $k,c\ge 1$ and let $b,n_0\ge 1$ be constants to be determined later. We argue in $\PV_1$. Fix $n,t,m\in\Log$, $p\eqdef t/n$, and $S_1,\dots,S_m\subseteq [n]$.

    \paragraph{Disjoint Set Decomposition.} Consider the iterative algorithm: Let $\calU_0\gets \{S_1,\dots,S_m\}$. In the $i$-th step, where $i \geq 1$, we choose a maximal set $\calV_i\subseteq \calU_{i-1}$ such that the sets in $\calV_i$ are disjoint, and compute $\calU_{i}$ as follows: 
    \begin{compactitem} 
    \item Let $\calU_i\gets \varnothing$. For every $S\in\calU_{i-1}$, we include $S\setminus (\bigcup_{S\in \calV_{i}} S)$ in $\calU_i$ if this set is nonempty.   
    \end{compactitem}
    Note that each set remaining in $\calU_i$ must be a subset of some $S_1,\dots,S_m$. Moreover, it must be of size at most $c-i$, as the maximal set $\calV_i$ intersects with each set in $\calU_{i-1}$ (otherwise it is not maximal). The algorithm terminates if no set is added to $\calU_i$ during an iteration. 
    
    Let $d\le c$ and $\calV_1,\dots,\calV_d$ be the sets obtained by the algorithm. One can prove that the union of the sets in $\calV_1,\dots,\calV_d$ covers $S_1\cup\dots\cup S_m$. Furthermore, it can be verified that for each $i\in[d]$ and $V_1,\dots,V_\ell\in \calV_i$, there are distinct $i_1,\dots,i_\ell\in [m]$ such that $V_1\subseteq S_{i_1},\dots,V_\ell\subseteq S_{i_\ell}$.  

    \paragraph{Wide Case.} We first consider the case that for some $i\in [d]$, $\calV_i$ contains at least $\ell \eqdef 2k\ln n$ sets. Fix $i$ to be the smallest number satisfying this. Let $V_1,\dots,V_\ell$ be the first $\ell$ sets in $\calV_i$. Let $\Phi_\wide(x)$ be the potential function in \Cref{lmm: potential small sets} for $V_1,\dots,V_\ell$ and $s\eqdef \ell/2$ (note that $\ell\in\Log\Log$). We define $\Phi(x)\eqdef \Phi_\wide(x)$. 
    
    It then suffices to verify that the three required properties hold.  
    \begin{compactitem}
    \item (\emph{Initial Condition}). By \Cref{lmm: potential small sets}, we know that 
    \[ 
    \Phi(\varepsilon)\le \ell^{\ell/2}\cdot (pc)^{\ell /2} \le \left(\frac{bkc\ln n}{\sqrt n}\right)^{\ell / 2} \le n^{-k}, 
    \] 
    where we assume a choice of $b \geq 2$ and 
    a sufficiently large $n$, which can be ensured by setting $n_0$ as a large constant. 
    \item (\emph{Recursion Condition}). It follows from the recursion condition in \Cref{lmm: potential small sets}. 
    \item (\emph{Termination Condition}). For $x\in\{*,\circ\}^n,\Phi(x)\in\zo$ by \Cref{lmm: potential small sets}. Let $T_x\eqdef \{i\in [n]\mid x_i=\circ\}$. Suppose that $\Phi(x)=0$. We know that there are at least $\ell/2\ge k\ln n$ sets $V$ among $V_1,\dots,V_\ell$ such that $V\subseteq T_x$. This satisfies the second termination condition. Moreover, the $k\ln n$ sets can be obtained by a $\PV$ function given $S_1,\dots,S_m$ and $x\in\{*,\circ\}^n$, using the algorithm that constructs $\calV_1,\dots,\calV_d$. 
    \end{compactitem}

    \paragraph{Narrow Case.} Now we consider the case that $\calV_i\le 2k\ln n$ sets for every $i\in[d]$. Note that the union of the sets in $\calV_1,\dots,\calV_d$ covers $S_1\cup \dots\cup S_m$. We know that 
    \[
    |S_1\cup \dots\cup S_m| = \left|\bigcup_{i\in [d]}\bigcup_{V\in \calV_i}V\right| \le d\cdot (2k  \ln n) \cdot c \le 2kc^2\ln n. 
    \]
    
    Recall that $b\ge1 $ is a constant to be determined. Let $S\eqdef S_1\cup \dots\cup S_m$. We define the potential function $\Phi(x)$ as follows. Given $x\in \{*,\circ\}^{i}$ for $i < n$, let $q^*\eqdef |\{j\in [i]\cap S\mid  x_j=*\}|$, $q^{\circ}\eqdef |\{j\in[i]\cap S\mid x_j=\circ\}|$, and $q^{\diamond}\eqdef |S| - q^*-q^\circ$. Then 
    \begin{equation}
    \Phi(x)\eqdef \sum_{j=\max\{b+1,q^*\}}^{|S|-q^\circ} \binom{q^\diamond}{j-q^*} p^{j-q^*}(1-p)^{q^\diamond-j+q^*}. 
    \end{equation}
    For instructive purposes, we mention that the combinatorial interpretation of the function is as follows. We randomly assign $x_{i+1},\dots,x_n$ independently to $*$ with probability $p$ and to $\circ$ with probability $1-p$. Then $\Phi(x)$ is the probability that the number of indices $j\in[n]\cap S$ such that $x_j=*$ is at least $b+1$. The combinatorial interpretation is \emph{not} a part of the $\PV_1$ proof. 

    Note that $\Phi(x)\ge 0$ as each term is non-negative, and $\Phi(x)\le 1$ by the \refmeta{binomial}. It suffices to verify the three required properties. 
    \begin{compactitem}
    \item (\emph{Initial Condition}). Note that 
    \[
    \Phi(\eps) \le \sum_{j=b+1}^{|S|} \binom{|S|}{j} p^j \le \sum_{j=b+1}^{|S|}|S|^j p^j\le \sum_{j=b+1}^{|S|} \left(\frac{2kc^2\ln n}{\sqrt n}\right)^{j} \le n^{-k},
    \]
    where the last inequality holds if we set $b\ge 10k$ and $n_0$ to be sufficiently large. 
    \item (\emph{Recursion Condition}). Fix any $x\in\{*,\circ\}^i$ for $i< n$, we consider whether $i+1\in S$. If not, we have that $\Phi(x) = \Phi(x*) = \Phi(x\circ)$ and the equation holds. Otherwise, let $q^*\eqdef |\{j\in [i]\cap S\mid  x_j=*\}$, $q^{\circ}\eqdef \{j\in[i]\cap S\mid x_j=\circ\}$, and $q^{\diamond}\eqdef |S| - q^*-q^\circ$. We can see that 
    \begin{align*}
    & \Phi(x*)\eqdef \sum_{j=\max\{b+1,q^*+1\}}^{|S|-q^\circ} \binom{q^\diamond-1}{j-q^*-1} p^{j-q^*-1}(1-p)^{q^\diamond-j+q^*}; \\ 
    & \Phi(x\circ)\eqdef \sum_{j=\max\{b+1,q^*\}}^{|S|-q^\circ-1} \binom{q^\diamond-1}{j-q^*} p^{j-q^*}(1-p)^{q^\diamond-1-j+q^*}.
    \end{align*}
    It follows from \Cref{equ: binom def} that $\Phi(x)=p\cdot \Phi(x*)+(1-p)\cdot \Phi(x\circ)$. 
    \item (\emph{Termination Condition}). For $x\in\{*,\circ\}^n$, let $q^*,q^\circ,q^\diamond$ be defined as above. We have $q^\diamond=0$ and $q^*+q^\circ = |S|$. In that case, one can observe that if $q^*>b$, then $\Phi(x)=1$, and otherwise $\Phi(x)=0$. Moreover, we know by the definition of $q^*$ that if $\Phi(x)=0$, for $T_x\eqdef \{i\in[n]\mid x_i=\circ\}$, we have $|S\setminus T_x| = q^*\le b$. This satisfies the first termination condition.  
    \end{compactitem}
    This completes the proof. 
\end{proof}

We are now ready to state the \emph{Subset Selection Lemma}. 

\dummylabel{subset selection}{Subset Selection Lemma}
\begin{lemma}[Subset Selection Lemma]\label{lem:subset_selection}
For every $k,c\ge 1$, there are constants $b,n_0\ge 1$ such that the following sentence is provable in $\PV_1$. Let $n,t,\ell\in\Log$, $n>n_0$, $\ell\le n^k$, and $F_1,F_2,\dots,F_\ell$ be $c$-NFs over $n$ input variables. If $t\le \sqrt n$, there exists a subset $T\subseteq[n]$ of size at most $n-t$ such that for every $i\in [\ell]$, at least one of the following conditions hold. 
\begin{compactitem}
\item If we fix the $j$-th variable for every $j\in T$, $F_i$ is fed by at most $b$ literals.  
\item There are $m'\ge k\ln n$ disjoint non-empty clauses $C_1',\dots,C'_{m'}$ such that (1) each $C'_j$ is a sub-clause of a different clause in $F_i$; (2) $S_{C_j'}\subseteq T$.
\end{compactitem}
Moreover, the subset $T$ and the clauses $C_1',\dots,C_m'$ for each $i$ are computed by a $\PV$ function given $F_1,\dots,F_\ell$ and $1^n,1^t,1^\ell$.  
\end{lemma}

\begin{proof}
Fix any $k,c\ge 1$ and let $b,n_0\ge 1$ be determined later. We argue in $\PV_1$. Fix $n,t \in\Log$, $\ell\le n^k$, $p\eqdef t/n$, and $c$-NFs $F_1,\dots,F_\ell$ over $n$ variables. 

Fix any $i\in[\ell]$. Suppose that $F_i$ has $m_i$ clauses, and let $C_{ij}$ be the $j$-th clause of $F_i$. Let $S_{ij}\eqdef S_{C_{ij}}$ be the subset of variables which or whose negation feeds $C_{ij}$. For each $i\in[\ell]$, consider the sets $S_{i1},\dots,S_{im_i}$; by \Cref{lmm: potential general} (using $3k$ instead of $k$), there are constants $b',n_0'$ and a potential function $\Phi_i:\{*,\circ\}^{\le n}\to\bbQ\cap[0,1]$ such that the following conditions hold. 
\begin{compactitem}
\item (\emph{Initial Condition}). $\Phi_i(\eps)\le n^{-3k}$; 
\item (\emph{Recursion Condition}). $\Phi_i(x)=p\cdot \Phi_i(x*)+(1-p)\cdot \Phi_i(x\circ)$. 
\item (\emph{Termination Condition}). For $x\in\{*,\circ\}^n$, $\Phi_i(x)\in\zo$. Moreover, let $T_x\eqdef \{i\in[n]\mid x_i=\circ\}$. If $\Phi_i(x)=0$, then one of the following two conditions holds: 
\begin{compactenum}
\item $|S_{i1}\cup \dots\cup S_{im_i}\setminus T_x|\le b$. This effectively means that if we fix $j$-th variable for every $j\in T_x$, then $F_i$ is fed by at most $b$ literals. 
\item There are disjoint nonempty sets $V_1,\dots,V_{\ell_i}\subseteq[n]$ that are subsets of $\ell_i$ distinct sets among $S_{i1},\dots,S_{im_i}$, such that $V_i\subseteq T_x$ for every $i\in[\ell]$, where $\ell_i \ge 3k \ln n \ge k \ln n$. This means that there are $m'=\ell_i\ge k\ln n$ disjoint clauses $C_1',\dots,C_{m'}'$, each of which is a sub-clause of a clause in $F_i$, such that $S_{C_{j}'}\subseteq T$.  
\end{compactenum}
Therefore, if $\Phi_i(x) = 0$, then the clause satisfies the required property if we choose $T\eqdef T_x$.  
\end{compactitem}

Consider the potential function $\Phi(x)$ as follows. Given any $x\in\{*,\circ\}^i$, $i\le n$, let $q^\circ$ be the number of $\circ$'s in $x$. We define 
\[
\Phi(x)\eqdef q^\circ + (1-p)(n-i) + n\cdot \sum_{i\in[\ell]} \Phi_i(x). 
\]
We can show that $\Phi(x)$ satisfies the following conditions: 
\begin{compactitem}
\item (\emph{Initial Condition}). $\Phi(\eps)\le (1-p)n + n\cdot \ell \cdot n^{-3k} < (1-p)n + 1$. 
\item (\emph{Recursion Condition}). $\Phi(x)=p\cdot \Phi(x*) + (1-p)\cdot \Phi(x\circ)$. 
\item (\emph{Termination Condition}). For $x\in\{*,\circ\}^n$, $\Phi(x)$ is an integer. Moreover, $\Phi(x)$ is at most the sum of (1) the number of $*$'s in $x$ and (2) $n$ times the number of $c$-NFs $F_i$ that violates the required properties if we choose $T=T_x=\{i\in[n]\mid x_i=\circ\}$. In particular, if $\Phi(x)\le (1-p)n$, $T=T_x$ is a desired subset.   
\end{compactitem}

It remains to construct $x\in\{*,\circ\}^n$ such that $\Phi(x)\le (1-p)n$. Indeed, it can be obtained by the greedy algorithm that, starting from $x\gets \eps$,  appends either $*$ or $\circ$ to $x$ to minimize $\Phi(x)$. By induction on $i$, we can prove that the string $x\in\{*,\circ\}^{\le i}$ after the $i$-th round of the algorithm satisfies that $\Phi(x)\le \Phi(\eps) < (1-p)n+1$. This is available in $\PV_1$ as the property can be verified by a straightforward polynomial-time algorithm. Finally, the string $x\in\{*,\circ\}^n$ obtained after $n$ rounds satisfies that $\Phi(x)\le (1-p)n$, as it must be an integer smaller than $(1-p)n+1$. This completes the proof. 
\end{proof}

\subsubsection{Average-Case Lower Bound in \texorpdfstring{$\APX_1$}{APX1}}

We say that a partial assignment $\rho$ \emph{trivializes} an NF $F$ if $F$ is a constant function after applying $\rho$. Note that for every NF $F$ that contains a non-constant clause $C$, there is an assignment to variables in $C$ that trivializes $F$. 

\dummylabel{random restriction}{Random Restriction Lemma}
\begin{lemma}[Random Restriction Lemma] \label{lem:random_restriction_lemma}
    For every $k, c, b\in\bbN$, there exists an $n_0\ge 1$ such that the following is provable in $\APX_1$. Let $n,t,\ell\in\Log$, $n>n_0$. Let $F_1,\dots,F_\ell$ be  $c$-NFs over $n$ input variables, and $T\subseteq[n]$ be a subset of size $n-t$ such that for every $i\in[\ell]$, at least one of the following conditions hold. 
    \begin{compactitem}
    \item \emph{(Narrow)}. If we simultaneously fix all variables whose indices are in $T$, the $c$-NF $F_i$ will be fed by at most $b$ literals. 
    \item \emph{(Wide)}. There are $m'_i\ge k\ln n$ (explicitly given) non-empty disjoint clauses $C_{i,1}',\dots,C'_{i,m'_i}$ such that (1) each $C_{i,j}'$ is a sub-clause of a different clause in $F_i$; (2) $S_{C'_{i,j}}\subseteq T$ for every $j\in[m'_i]$. 
    \end{compactitem}

    Let $Y$ be the random variable over $\zo$ that takes a seed $x$ of length $n$, parses it as an assignment $\rho$ to variables in $T$ (i.e.~it fixes the $i$-th variable to $x_i$ for every $i\in T$), and outputs $1$ if and only if at least one of $F_1,\dots,F_\ell$ is neither trivialized nor depends on at most $b$ literals after applying $\rho$. Then for every $\delta^{-1},\beta^{-1}\in\Log$ 
    \[
    \bbE_\delta[Y] \le \ell \cdot (1-2^{-c}+\beta)^{k\ln n} + \delta+\beta. 
    \]
\end{lemma}

\begin{proof}
    We argue in $\APX_1$. Fix $k,c,b\in\bbN$ and let $n_0\ge 1$ be a constant to be determined. Fix $n,t,\ell\in\Log$, $F_1,\dots,F_\ell$, and $T\subseteq[n]$. Let $\eta^{-1}\in\Log$ be a parameter to be determined later.  
    
    For simplicity, we assume that $F_1,\dots,F_{\ell'}$ are the gates that satisfy the second condition, and let $m'_i\ge k\cdot \ln n$ be the number of disjoint clauses $C'_{i,1},\dots,C'_{i,m'_i}$. We will define random variables over $\zo$ with $n$ bit seeds as follows. 
    \begin{compactitem}
    \item For every $i\in[\ell']$ and $j\in[m_i']$, $X_{ij}$ is defined as the following random variable: Let $x\in\zo^n$ be the seed. Then $X_{ij} = 1$ if and only if fixing the $r$-th variable to $x_r$ for every $r \in T$ does not trivialize $C'_{i,j}$.
    \item For every $i\in[\ell']$, let $X_{i}=\prod_{j=1}^{m'_i} X_{ij}$.
    \item Notice that $Y \leq \bigvee_{i\in[\ell']} X_i$.  
    \end{compactitem}
    Note that each $X_{ij}$ reads at most $c\in\bbN$ bits of its seed. Thus by the \refmeta{brute force}, we know that $\bbE_\eta[X_{ij}] \le 1-2^{-c} + 2\eta$. By the \refmeta{multiplication}, we have 
    \[
    \bbE_\eta[X_{i1}X_{i2}\dots X_{im'_i}] \le (1-2^{-c} + 2\eta)^{m'_i} + 8\eta\cdot m'_i \le (1-2^{-c} + 2\eta)^{k\cdot \ln n} + 8\eta\cdot |F_i| \le (1-2^{-c}+\beta)^{k\ln n} + 8\eta \cdot |F_i|,
    \]
    where the last inequality holds when $\eta\le \beta / 2$. By the \refmeta{union bound}, 
    \[
    \bbE_\eta[Y]\le \sum_{i=1}^{\ell'} \bbE_\eta[X_i] + 3\eta\cdot \ell' \le \ell\cdot (1-2^{-c}+\beta)^{k\ln n} + 8\eta \cdot \ell \cdot (2^c n^c)+ 3\eta \cdot \ell. 
    \]
    The lemma then follows from the \refmeta{precision exp} by setting $\eta \eqdef \beta / (20 \cdot \ell \cdot 2^c n^c)$.
\end{proof}

Now we are ready to prove the average-case lower bound for $\oplus_n$ against $\AC^0$.

\AvgParity*

\begin{proof}
    We prove by induction on $d$ in the meta-theory that the statement holds for every $k$. The constant $n_0^{k,d}\ge 1$ will be determined later in the proof. In both the base case and induction case, we argue in $\APX_1$. Fix $n,\delta^{-1},\beta^{-1}\in\Log$ and the circuit $C:\zo^n\to\zo$. Let $\eta^{-1}\in\Log$ be a parameter to be determined later. 

    \paragraph{Base Case.} Suppose that $d=1$. Towards a contradiction, assume that \Cref{equ: prob T small AC0} does not hold. Then by the \refaxiom{precision}, we have 
    \begin{equation}
    \P_\eta(T_C) > \frac{1}{2} + \frac{1}{n^{k}} + (\beta - 2\eta). \label{equ: contra T large AC0}
    \end{equation}
    Note that $\Fix_b(T_C)$ is the circuit $T_{\Fix_b(C)\oplus b}$. 

    We first show that for $D:\zo^n\to\zo$, if $D$ depends on at most $n-1$ of its input bits, then $\P_\eta(T_D)\le 1/2 + O(\eta)$. Suppose that $D$ does not depend on the $i$-th input bit. Let $X_D$ be the indicator random variable of $T_D$ and $\rho$ be any assignment to all but the $i$-th input bit of $D$. By the \refmeta{brute force}, we have that $\bbE_\eta[X_D|_\rho]\le 1/2 + 2\eta$. Subsequently, by the \refmeta{avg on exp}, we can conclude that 
    \begin{equation}
    \bbE_\eta[X_D] \le \frac{1}{2} + 2\eta+ 3\eta \le \frac{1}{2} + 5\eta.  \label{equ: ub constant circuit}
    \end{equation}
    This implies that $\P_\eta(T_D)\le 1/2 + 8\eta$ by \Cref{prop: indicator variable}. 

    Now we assume that $C$ depends on all of its input bits and is of depth at most $d = 1$. Consider the following iterative $\P$-oracle algorithm. Let $C_0\eqdef C$ and $s_0\eqdef 0$. The algorithm maintains the invariant that after the $i$-th round, $C_i$ depends on all of its input bits. Given that the invariant holds after the $(i-1)$-th round, there exists $b_i\in\zo$ such that 
    \begin{compactitem}
    \item $\Fix_{b_i}(C_{i-1})$ is a constant circuit;
    \item $\Fix_{1-b_i}(C_{i-1})$ depends on all of its input bits.   
    \end{compactitem} 
    Fix that $b_i\in\zo$. The algorithm then defines $C_i\eqdef \Fix_{1-b_i}(C_{i-1})\oplus (1-b_i)$ and $s_i\eqdef s_{i-1}\oplus (1-b_i)$ in the $i$-th round, and the invariant is maintained. 
    
    We will prove by induction on $i\le n$ that 
    \[
    \P_\eta(T_{C_i\oplus s_i}) > \frac{1}{2}+\frac{2^i}{n^k}+(\beta - 8(i+1)\cdot\eta). 
    \]
    (Note that $n \in \Log$ and $i \leq n$, so the induction hypothesis can be expressed by an open formula in $\APX_1$.)
    The base case follows from \Cref{equ: contra T large AC0}. Suppose that the inequality holds for $i<n$. Notice that 
    \begin{align*}
    \P_\eta(T_{C_{i}\oplus s_i}) &\le \frac{\P_\eta(\Fix_{b_i}(T_{C_{i}})\oplus s_i) + \P_\eta(\Fix_{1-b_i}(T_{C_i})\oplus s_i)}{2} \tag{\refaxiom{precision}} \\ 
    & = \frac{\P_\eta(T_{\Fix_{b_i}(C_{i})\oplus s_i\oplus b_i}) + \P_\eta(T_{\Fix_{1-b_i}(C_i)\oplus s_i\oplus (1-b_i)})}{2} \\ 
    & = \frac{\P_\eta(T_{\Fix_{b_i}(C_{i})\oplus s_i\oplus b_i}) + \P_\eta(T_{C_{i+1}\oplus s_{i+1}})}{2} \\ 
    &\le \frac{1}{2}\left(\frac{1}{2}+8\eta+ \P_\eta(T_{C_{i+1}\oplus s_{i+1}})\right). 
    \end{align*}
    (The last inequality follows as $\Fix_{b_i}(C_i)$ is a constant circuit and thus does not depend on all of its input bits.) It then follows by the induction hypothesis that 
    \[ 
    \P_\eta(T_{C_{i+1}\oplus s_{i+1}}) > \frac{1}{2}+\frac{2^{i+1}}{n^k}+(\beta - 8(i+2)\cdot\eta).
    \]  
    
    We set $n_0\in\bbN$ to be sufficiently large such that $2^{n-10}/n^k\ge 1$ for every $n> n_0$. Therefore, we have that $\P_\eta(T_{C_{n-10}}) > 3/2+(\beta - 8(n+1)\cdot \eta)$, where $C_{n-10}$ has input length exactly $10$ and is an $\AC^0$ circuit of depth $1$. This is provably impossible if we set $\eta \eqdef \beta / (20n)$ by the \refmeta{brute force}. 

    \paragraph{Induction Case.} Suppose that the theorem holds for $d\in\bbN$. Our goal is to prove the theorem for $d+1$. Let $n_0^{k,d}$ be the constant $n_0\ge 1$ corresponding to the theorem for $d$ and $k$. Fix any $k\ge 1$ and let $n_0^{k,d+1}$ be a constant to be determined. Towards a contradiction, assume that \Cref{equ: prob T small AC0} does not hold. Then by the \refaxiom{precision}, we have 
    \begin{equation}
    \P_\eta(T_C) > \frac{1}{2} + \frac{1}{n^{k}} + (\beta - 2\eta). \label{equ: contra T large AC0 ind}
    \end{equation}
    
    At a high level, we will apply random restrictions twice to convert $C$ to an $\AC^0_{d}$ circuit that computes the parity function w.h.p.~on a smaller input length; after that, we can apply the induction hypothesis to conclude the proof.  

    \subparagraph{Restriction 1.} Let $G_1,G_2,\dots,G_\ell$ be the gates in the first layer, i.e., directly fed by literals. We may view them as $1$-NFs, as each literal can be viewed as a clause with one literal. Let $t\eqdef \sqrt n$. By the \refmeta{subset selection}, there exists a subset $T_1\subseteq[n]$ of size at most $n-t$ such that for every gate $G_i$, one of the conditions hold. 
    \begin{compactitem}
    \item If we fix all variables in $T_1$, $G_i$ will be fed by at most $b_1$ variables, where $b_1\in\bbN$ is a constant. 
    \item At least $100k\cdot \ln n$ literals of $G_i$ are using variables in $T_1$. 
    \end{compactitem}
    Fix the subset $T_1\subseteq[n]$. We assume that $|T_1|=n-t$; if not, we add $n-t-|T_1|$ arbitrary elements to it. 
    
    We define random variables over $\zo$ with $n$ bit seeds as follows. 
    \begin{compactitem}
    \item Let $Y$ be the random variable in the \refmeta{random restriction}. That is, given $x\in\zo^n$, it parses $x$ as a partial assignment $\rho$ to variables in $T_1$, and outputs $1$ if and only if at least one of the gates is neither trivialized nor depends on at most $b$ literals after applying $\rho$. 
    \item Let $Y_T$ be the indicator random variable of $T_C$. 
    \end{compactitem}
    By the \refmeta{random restriction} and using $c = 1$, we have that 
    \[
    \bbE_\eta[Y]\le \ell \cdot (1/2 +\eta)^{100 k\ln n} + 2\eta \le n^{-8k}, 
    \]
    where the last inequality holds when $\eta$ is sufficiently small and $n$ is sufficiently large (by setting $n_0^{k,d+1}\in\bbN$). By the \refmeta{avg on exp}, there exists an assignment $\rho_1$ to variables in $T_1$ such that 
    \begin{equation}
    \bbE_\eta[Y_T|_{\rho_1}]-\bbE_\eta[Y|_{\rho_1}] \ge \frac{1}{2} + \frac{1}{n^{k}} + (\beta - 30\eta) - \frac{1}{n^{8k}} > 5\eta,
    \end{equation}
    where the last inequality holds if $\eta$ is sufficiently small. 

    Fix the assignment $\rho_1$. Note that $Y|_{\rho_1}\in\zo$ as $Y$ only reads its input variables in $T_1$. Therefore, we must have $Y|_{\rho_1} = 0$. In this case, all gates are either trivialized or fed by at most $b_1$ variables after applying $\rho_1$, and thus can be replaced by a gate of fan-in at most $b_1$. 
    
    Let $\sigma_1\eqdef \oplus_{n-t}(\rho_1)$, $n_1\eqdef t$, and $C_1:\zo^{n_1}\to\zo$ be the circuit obtained from $C$ by applying the assignment $\rho_1$, replacing each gate in the first layer with an equivalent gate of fan-in at most $b_1$, and XORing the output of the circuit with the bit $\sigma_1$. Note that $C|_{\rho_1}(x)=C_1(x)\oplus \sigma_1$ for every $x\in\zo^{n_1}$. Moreover, $C_1$ is of size at most $n^k\le n_1^{2k}$. 
    
    Let $T_{C_1}:\zo^{n_1}\to\zo$ be the circuit that, given $x$, it outputs $1$ if and only if $C_1(x)=\oplus_{n_1}(x)$. It turns out that $Y_T|_{\rho_1}$ is the indicator random variable of $T_{C_1}$; to see this, notice that for every assignment $x$ to all variables but $T_1$, $C(x\cup \rho_1) = \oplus_n(x\cup \rho_1)$ if and only if $C|_{\rho_1}(x)=(\oplus_{n_1}(x))\oplus \sigma_1$, where $C|_{\rho_1}(x)\oplus \sigma_1 = C_1(x)$. Therefore, by \Cref{prop: indicator variable}, we have that 
    \[
    \P_\eta(T_{C_1}) \ge \bbE_{\eta}[Y_T|_\rho] - 3\eta \ge \frac{1}{2} + \frac{1}{n^{k}} + (\beta -33\eta)-\frac{1}{n^{8k}} \ge \frac{1}{2}+\frac{1}{n_1^{4k}} + (\beta - 33\eta), 
    \]
    where the last inequality holds if $n$ is sufficiently large (by setting $n_0^{k,d+1}\in\bbN$).

    \subparagraph{Restriction 2.} As mentioned above, each gate in the first layer of $C_1$ has fan-in at most $b_1$, and thus the gates in the second layer of $C_1$ computes $b_1$-NFs. Let $F_1,F_2,\dots,F_{\ell_1}$ be the $b_1$-NFs in the second layer of $C_1$. Let $t_1\eqdef \sqrt n_1$. By the \refmeta{subset selection} with appropriate choice of parameters, there exists a subset $T_2\subseteq [n]$ of size at most $n_1-t_1$ such that for every $i\in[\ell_1]$, one of the conditions hold. 
    \begin{compactitem}
    \item If we fix all variables in $T_2$, $F_i$ will depend on at most $b_2$ variables, where $b_2\in\bbN$ is a constant. 
    \item There are $m_1'\ge 100k\cdot 4^{b_1}\cdot  \ln n$ disjoint sub-clauses of $F_i$ that only use literals from variables in $T_2$. 
    \end{compactitem}
    Let $n_2\eqdef t_1$ and fix the set $T_2\subseteq[n]$. We assume that $|T_2|=n_1-t_1$; if not, we add $n_1-t_1-|T_2|$ arbitrary elements to it. We define  random variables over $\zo$ with $n$ bit seeds as follows. 
    \begin{compactitem}
    \item Let $Y'$ be the random variable in the \refmeta{random restriction}. That is, given $x\in\zo^n$, it parses $x$ as a partial assignment $\rho$ to variables in $T_2$, and outputs $1$ if and only if each of the $b$-NFs is either trivialized or depends on at most $b_2$ literals after applying $\rho$. 
    \item Let $Y_T'$ be the indicator random variable of $T_{C_1}$. 
    \end{compactitem}
    By the \refmeta{random restriction}, we have that 
    \[
    \bbE_\eta[Y'] \le \ell_1\cdot (1 - 2^{-b_1}+\eta)^{100k\cdot 4^{b_1}\cdot \ln n} + 2\eta \le n_1^{-8k},  
    \]
    where the last inequality holds when $\eta$ is sufficiently small and $n$ is sufficiently large (by setting $n_0^{k,d+1}\in\bbN$). By the \refmeta{avg on exp}, there exists an assignment $\rho_2$ to variables in $T_2$ such that 
    \begin{equation}
    \bbE_\eta[Y'_T|_{\rho_2}]-\bbE_\eta[Y'|_{\rho_2}] \ge \frac{1}{2} + \frac{1}{n_1^{4k}} + (\beta - 36\eta) - \frac{1}{n_1^{8k}} > 5\eta,\label{equ: YT prime rho large}
    \end{equation}
    where the last inequality holds if $\eta$ is sufficiently small. 

    Fix the assignment $\rho_2$. Note that $Y'|_{\rho_2}\in\zo$ as it only reads its input variables in $T_2$. Therefore, we must have $Y_{\rho_2}' = 0$. In such case, all $b_1$-NFs (i.e.~gates in the second layer of $C_1$) are either trivialized or fed by at most $b_2$ variables. In such case, we can transform $C_1$ into an equivalent circuit of depth at most $d$ as follows. Suppose that $d\ge 2$ (the case for $d=1$ is left as an exercise). For each gate $G$ in the second layer, if it is not trivialized, we remove $G$ and consider each gate $G'$ in the third layer originally fed by $G$:
    \begin{compactitem}
    \item If $G'$ is an AND gate, we rewrite $G$ as an equivalent CNF of size at most $b_2\cdot 2^{b_2}$ and connect all clauses of it to $G'$. 
    \item If $G'$ is an OR gate, we rewrite $G$ as an equivalent DNF of size at most $b_2\cdot 2^{b_2}$ and connect all clauses of it to $G'$.  
    \end{compactitem}
    In either case, the circuit remains functionally equivalent. 
    
    Let $\sigma_2\eqdef \oplus_{n_1-t_1}(\rho_2)$ and $C_2:\zo^{n_2}\to\zo$ be the depth-$d$ circuit that computes $C_1|_{\rho_2}(x)\oplus \sigma_2$. The size of $C_2$ blows up by a linear factor, which is at most $O(n^k)\le n_2^{6k}$, when $n$ is sufficiently large (by setting $n_0^{k,d+1}\in\bbN$). Let $T_{C_2}:\zo^{n_2}\to\zo$ be the circuit that, given $x$, it outputs $1$ if and only if $C_2(x)=\oplus_{n_2}(x_2)$. As before, $Y_T'|_\rho$ is the indicator random variable of $T_{C_2}$. Therefore, we have 
    \begin{align}
    \P_\eta(T_{C_2}) & \ge \bbE_\eta[Y_T'|_\rho] - 3\eta \tag{\Cref{prop: indicator variable}} \\
    & \ge \frac{1}{2} + \frac{1}{n_1^{4k}} + (\beta - 39\eta) - \frac{1}{n_1^{8k}} \tag{\Cref{equ: YT prime rho large}} \\ 
    & \ge \frac{1}{2} + \frac{1}{n_2^{6k}} + (\beta - 39\eta) \nonumber\\ 
    & > \frac{1}{2} + \frac{1}{n_2^{6k}} + 2\eta, \label{equ: TC2 reach contradiction}
    \end{align}
    where the last two lines hold when $n$ is sufficiently large (by setting $n_0^{k,d+1}\in\bbN$) and $\eta$ is sufficiently small.  

    Now we arrive at a contradiction: $C_2:\zo^{n_2}\to\zo$ is a depth-$d$ circuit of size at most $n_2^{6k}$, and it computes parity with advantage $1/n_2^{6k}$. This violates \Cref{equ: prob T small AC0}. The theorem then follows from the induction hypothesis for depth $d$ and size $n_2^{6k}$ if we set $\eta^{-1}\in\Log$ and $n_0^{k,d+1}\in\bbN$ appropriately based on $b_1,b_2,n_0^{6k,d}$ and the requirements of inequalities used in the proofs. 
\end{proof}

\subsubsection{Worst-Case Lower Bound in \texorpdfstring{$\PV_1$}{PV1}}

First, we derandomize the \refmeta{random restriction} via an explicit implementation of the method of conditional expectations in $\PV_1$.

\dummylabel{derandomized restriction}{Derandomized Restriction Lemma}
\begin{lemma}[Derandomized Restriction Lemma]\label{lem:derand_restriction}
    For every $k, c, b\in\bbN$, there exists an $n_0\ge 1$ such that the following is provable in $\PV_1$. Let $n,t,\ell\in\Log$, $n>n_0$. Let $F_1,\dots,F_\ell$ be $c$-NFs over $n$ input variables, and $T\subseteq[n]$ be a subset of size $n-t$ such that for every $i\in[\ell]$, at least one of the following conditions hold. 
    \begin{compactitem}
    \item \emph{(Narrow)}. If we fix the $j$-th variable for every $j\in T$, $F_i$ is fed by at most $b$ literals. 
    \item \emph{(Wide)}. There are $m'_i\ge k\ln n$ (explicitly given) disjoint clauses $C_{i,1}',\dots,C'_{i,m'_i}$ such that (1) each $C'_{i,j}$ is a sub-clause of a different clause in $F_i$; (2) $S_{C'_{i,j}}\subseteq T$ for every $j\in[m'_i]$. 
    \end{compactitem}

    Suppose that $\ell\cdot (1-2^{-c})^{k\ln n} < 1$. Then there exists an assignment $\rho$ to the variables in $T$ such that each of $F_1,\dots,F_\ell$ is either trivialized or depends on at most $b$ literals after applying $\rho$. 
\end{lemma}

\begin{proof}
    We argue in $\PV_1$. Let $k,c,b\in\bbN$ and $n_0\ge 1$ be a constant to be determined later. Fix $n,t,\ell\in \Log$, $c$-NFs $F_1,\dots,F_\ell$, and $T\subseteq[n]$. We say that a $c$-NF is good after applying a restriction $\rho$ to the variables in $T$ if it is either trivialized or depends on at most $b$ literals. As the $c$-NFs satisfying the first bullet are good regardless of the assignment $\rho$, we assume without loss of generality that all $c$-NFs satisfy the second bullet. We will construct an assignment $\rho$ such that all such $c$-NFs are trivialized. 

    For simplicity of presentation, we assume that $T=\{1,2,\dots,n-t\}$. Fix any $i\in[\ell]$, and let $C'_{i,1},\dots,C'_{i,m'_i}$ be the disjoint sub-clauses such that $S_{C'_{i,j}}\subseteq T$. For a partial assignment $x\in\zo^{\le n-t}$ to the first $|x|$ variables, we say that: 
    \begin{compactitem}
    \item $C'_{i,j}$ is positively determined if it is an AND gate and all literals of it are fixed to $1$, or it is an OR gate and all literals of it are fixed to $0$. 
    \item $C'_{i,j}$ is negatively determined if it is an AND gate and one of its literals is fixed to $0$, or it is an OR gate and one of its literals is fixed to $1$. 
    \item $C'_{i,j}$ is $d$-far from positively determined if it is not negatively determined, and there are exactly $d$ of its literals that remain unfixed.  
    \end{compactitem}
    We define $\phi_{ij},\Phi_i,\Phi:\zo^{\le n-t}\to\zo$ as follows. Given any $x\in\zo^{\le n-t}$ parsed as a partial assignment to the first $|x|$ literals, 
    \begin{align}
    & \Phi(x)\eqdef \sum_{i=1}^\ell \Phi_i(x), \quad \Phi_i(x)\eqdef \prod_{j=1}^{m_i'}(1-\phi_{ij}(x))\\ 
    & \phi_{ij}(x)\eqdef \begin{cases}
        0 & C'_{ij} \text{ is negatively determined} \\ 
        2^{-d} & C'_{ij} \text{ is $d$-far from positively determined}
    \end{cases}
    \end{align}

    For instructive purposes, we mention that the combinatorial interpretation of $\Phi(x)$ is the expected number of $c$-NFs that are not trivialized if we extend $x$ to an assignment to variables in $T$ by fixing each unfixed bit uniformly at random. Note that this is not a part of the $\PV_1$ proof. Instead, we prove that: 
    \begin{compactitem}
    \item (\emph{Initial Condition}). Note that $\phi_{ij}(\eps) \ge 2^{-c}$ and thus 
    \[ 
    \Phi_i(\eps)\le (1-2^{-c})^{m_i'} \le (1-2^{-c})^{k\ln n} < \frac{1}{\ell}, \quad \Phi(\eps) < 1. 
    \] 
    \item (\emph{Recursion Condition}). For every $x\in\zo^{<n-t}$, we can prove that $\Phi_i(x) = (\Phi_i(x0) + \Phi_i(x1))/2$. To see this, notice that: 
    \begin{compactitem}
    \item When the $(|x|+1)$-th variable does not appear in $C_{i,1},\dots,C_{i,m'_i}$, $\Phi_i(x)=\Phi(x0)=\Phi(x1)$. 
    \item Otherwise, it appears in exactly one of $C_{i,1},\dots,C_{i,m'_i}$ as the clauses are disjoint. Assume for simplicity that it appears in $C_{i,1}$ and it is an OR gate. Then 
    \begin{align*}
    & \Phi_i(x) = (1-\phi_{i1}(x))\cdot \prod_{j=2}^{m_i'}(1-\phi_{ij}(x)),\\  
    & \Phi_i(x0) = (1-2\cdot \phi_{i1}(x))\prod_{j=2}^{m_i'}(1- \phi_{ij}(x)),\quad 
    \Phi_i(x1) = \prod_{j=2}^{m_i'}(1-\phi_{ij}(x)).
    \end{align*}
    \end{compactitem}
    Thus $\Phi_i(x)=(\Phi_i(x0)+\Phi_i(x1))/2$. Subsequently, $\Phi(x)=(\Phi(x0)+\Phi(x1))/2$. 
    \item (\emph{Termination Condition}). $\Phi_i(x)\in\{0,1\}$ for $x\in\zo^{n-t}$. Moreover, if $\Phi(x)=0$, the partial assignment $x$ will trivialize all $c$-NFs. 
    \end{compactitem}
    The lemma then follows from a greedy algorithm as in the \refmeta{subset selection}. 
\end{proof}

\WorstParity*

\begin{proof}[Proof Sketch]
The proof closely follows the proof of \refmeta{avg parity}, so we will only sketch the argument. We prove it by induction on $d$ in the meta-theory, and the constant $n_0=n_0^{k,d}$ depends on both $k$ and $d$. The case when $d=1$ is easy and left as an exercise. 

For $d\ge 2$, we assume towards a contradiction that $C$ computes $\oplus_n$. We first apply the \refmeta{subset selection} to find a subset $T$ of size $n-\sqrt n$ by viewing the gates in the first layer as $1$-NFs, and then apply the \refmeta{derandomized restriction} to find an assignment $\rho$ to variables in $T$ such that the gates in the first layer are either trivialized or of fan-in at most $b=O(1)$ after applying $\rho$. Let $n_1=\sqrt n$. We can construct (from $C$ and $\rho$) a circuit $C_1$ that computes $\oplus_{n_1}$ on the unfixed bits such that all gates in the first layer are of fan-in $b$. 

We then apply the \refmeta{subset selection} again to find a subset $T_1$ of size $n_1-\sqrt{n_1}$ by viewing the gates in the second layer as $b$-NFs, and then apply the \refmeta{derandomized restriction} to find an assignment $\rho_1$ to variables in $T_1$ such that the gates in the second layer are either trivialized or of fan-in at most $b_1=O(1)$ after applying $\rho_1$. Let $n_2=\sqrt n_1$. We can then construct (from $C_1$ and $\rho_1$) a circuit $C_2$ of depth at most $d-1$ that computes $\oplus_{n_2}$ on the unfixed bits. The size of the circuit is at most $n^{k+1} \leq n_2^{6k}$. This leads to a contradiction to the induction hypothesis by setting $n_0^{k,d}$ to be sufficiently large based on $n_0^{6k,d-1}$. 
\end{proof}

\subsection{Blum-Luby-Rubinfeld Linearity Testing}
\label{sec: BLR}

We now formalize the linearity testing algorithm due to Blum, Luby, and Rubinfeld \cite{BLR93}. Recall that a function $g:\zo^n\to\zo$ is said to be linear if $g(x\oplus y)=g(x) \oplus g(y)$, where $\oplus$ denotes bit-wise XOR; equivalently, $g(x)=\langle x,z\rangle \bmod 2$ for some $z\in\zo^n$. Let $g:\zo^n\to\zo$ be a function. Blum, Luby, and Rubinfeld \cite{BLR93} proved that for any sufficiently small constant $\eps>0$: 
\begin{itemize}
\item (\emph{Linearity Testing}): If $g$ is $\eps$-far from any linear function, then the BLR linearity testing algorithm fails with probability at least $\Omega(\eps)$. Conversely, if $g$ is $\eps$-close to a linear function, the BLR linearity testing algorithm fails with probability at most $O(\eps)$. 
\item (\emph{Self Correction}): The key idea behind linearity testing is a random self correctness algorithm: If $g$ is $\eps$-close to a linear function $\hat g$, then the function $f(x,r)\eqdef g(x\oplus r)\oplus g(r)$ is a \emph{randomized} algorithm that computes $\hat g$ with error $O(\eps)$, where $x$ is the input and $r$ is the random seed.
\end{itemize}

Linearity testing is the key component of the exponential length PCP theorem $\NP\subseteq\PCP[\poly,1]$, which is further used to reduce the number of queries in the proof of the PCP theorem $\NP=\PCP[\log n,1]$ (see, e.g., \cite{Harsha-thesis}).

\newcommand{\BLR}{\mathsf{BLR}}

We first state the main theorems, namely the completeness and soundness of the BLR linearity testing. The completeness states that a function that is close to a linear function is likely to be accepted. Formally: 

\dummylabel{completeness BLR testing}{Completeness of BLR Linearity Testing}
\begin{restatable}[Completeness of BLR linearity testing]{theorem}{BLRComplete}
    $\APX_1$ proves the following. Let $n,\delta^{-1},\beta^{-1}\in\Log$,  $C:\zo^n\to\zo$ be a circuit, and $z\in\zo^n$ be a string. Let $\eps\in\bbQ$ such that $\eps < 1/2$. Define the following circuits: 
    \begin{compactitem}
    \item Let $T_C(x)$ be the circuit that outputs $1$ if and only if $C(x)\ne \langle x,z\rangle\bmod 2$.
    \item Let $T_{C,\BLR}(x,y):\zo^n\times \zo^n\to\zo$ be the circuit that outputs $1$ when $C(x)\oplus C(y)\ne C(x\oplus y)$. 
    \end{compactitem}
    Suppose that $\P_\delta(T_C)\le \eps$. Then  $\P_\delta(T_{C,\BLR})\le 3\eps+4\delta+\beta$. 
\end{restatable}

The soundness states that if a function is likely to be accepted by the BLR linearity testing algorithm, then it is close to a linear function. 

\begin{restatable}[Soundness of BLR linearity testing]{theorem}{BLRSound}
    $\APX_1$ proves the following statement. Let $n,\delta^{-1},\beta^{-1}\in\Log$ and $C:\zo^n\to\zo$ be a circuit. Let $\eps\in\bbQ$. Assume that $\eps,\delta,\beta < 0.01$. We define the following circuits: 
    \begin{compactitem}
    \item For every $z\in\zo^n$, let $T_{C,z}(x)$ be the circuit that outputs $1$ if and only if $C(x)\ne \langle x,z\rangle\bmod 2$. 
    \item Let $T_{C,\BLR}(x,y)$ be the circuit that outputs $1$ if and only if $C(x)\oplus C(y)\ne C(x\oplus y)$.  
    \end{compactitem}
    Suppose that $\P_\delta(T_{C,\BLR})\le \eps$. Then there exists a string $z\in\zo^n$ such that $\P_\delta(T_{C,z})\le 5\eps + 6\delta + \beta$. 
\end{restatable}
\dummylabel{soundness BLR testing}{Soundness of BLR Linearity Testing}

We formalize the combinatorial proof \cite{BLR93} via majority correction (see \citep{DBLP:journals/tit/BellareCHKS96} for an alternate proof). Note that the same proof is also formalized by Pich \cite{DBLP:journals/apal/Pich15} in $\APC_1$ to prove the exponential PCP theorem $\NP\subseteq\PCP[\poly,1]$, and our main contribution is to show that it can be formalized in the (possibly weaker) theory $\APX_1$.\footnote{Our formalization is slightly different: We formalize linear functions $x\mapsto \langle x,z\rangle \bmod 2$ by explicitly giving $z$, while Pich \cite{thesis_Pich} formalizes linear functions $f$ using the sentence that for every $x,y$, $f(x\oplus y)=f(x)\oplus f(y)$; nevertheless, the difference in formalization does not matter in most cases.}

\subsubsection{Two Useful Lemmas}

Before formalizing the BLR linearity testing algorithm, we prove two useful lemmas. The first lemma shows that the acceptance probability of a circuit does not change significantly if the input is XORed with a \emph{fixed} string. Formally: 

\dummylabel{rerand}{Re-randomization Lemma}
\begin{lemma}[Re-randomization]\label{lmm: rerand}
$\Apx_1$ proves the following statement. For every $n,\delta^{-1},\beta^{-1}\in\Log$, circuit $T:\zo^n\to\zo$, and $x\in\zo^n$, let $T_x^{\oplus}:\zo^n\to\zo$ be the circuit defined as $T_x^{\oplus}(r)\eqdef T(x\oplus r)$. Then $\left|\P_\delta(T)-\P_\delta(T_x^\oplus)\right| \le 2\delta + \beta$.
\end{lemma}

\begin{proof}
    We argue in $\Apx_1$. Fix $n,\delta^{-1},\beta^{-1}\in\Log$, $x\in\zo^n$, the circuit $T$ and $T_x^\oplus$. Let $\eta^{-1}\in\Log$ be a parameter to be determined later. Suppose, towards a contradiction, that $|\P_\delta(T)-\P_\delta(T_x^\oplus)| > 2\delta + \beta$. Then by the \refaxiom{precision}, we have that 
    \[
    |\P_\eta(T)-\P_\eta(T_x^\oplus)| > \beta - 4\eta. 
    \]
    
    Recall that for a circuit $C$, $C^{k,z}$ denotes the circuit obtained by fixing the last $k$ input bits of $C$ to be $z$. We will design an $n$-round iterative $\PV(\P)$ algorithm that, in the $i$-th round, outputs a string $z_i\in\zo^i$ such that 
    \begin{equation}\label{equ: rerand i round}
    |\P_\eta(T^{i,z_i})-\P_\eta(T^{\oplus,i,z_i'}_x)| > \beta - 10\cdot (i+1)\cdot \eta.  
    \end{equation}
    where $z_i'\eqdef z_i\oplus x_{>n-i}$. The algorithm initializes by setting $z_0\eqdef \eps$. In the $i$-th round, it works as follows: 
    \begin{compactitem}
    \item Recall that by the invariant that \Cref{equ: rerand i round} holds in the $(i-1)$-th round, we have 
    \[ 
    |\P_\eta(T^{i-1,z_{i-1}})-\P_\eta(T^{\oplus,i-1,z_{i-1}'}_x)| > \beta - 10\cdot i\cdot \eta. 
    \] 
    \item By the \refaxiom{local}, we know that 
    \begin{align*}
    & |\P_\eta(T^{i-1,z_{i-1}})-\P_\eta(T^{\oplus,i-1,z_{i-1}'}_x)| \\ 
    \le~& \frac{1}{2}\sum_{b\in\zo}\left|\P_\eta(\Fix_b(T^{i-1,z_{i-1}}))-\P_\eta(\Fix_{b\oplus x_{n-i+1}}(T^{\oplus,i-1,z_{i-1}'}_x))\right| + 3\eta. 
    \end{align*}
    Subsequently, there is a constant $b\in\zo$ such that 
    \begin{equation}\label{equ: rerand result round i raw}
    \left|\P_\eta(\Fix_b(T^{i-1,z_{i-1}}))-\P_\eta(\Fix_{b\oplus x_{n-i+1}}(T^{\oplus,i-1,z_{i-1}'}_x))\right| > \beta - 10\cdot i \cdot \eta - 3\eta.  
    \end{equation}
    The algorithm finds such $b\in\zo$ by querying the $\P$-oracle, and outputs $z_{i}\eqdef b\circ z_{i-1}$. 
    \end{compactitem}
    
    To see that the algorithm is correct, notice that the circuit $\Fix_b(T^{i-1,z_{i-1}})$ is functionally equivalent to $T^{i,z_i}$, and $\Fix_{b\oplus x_{n-i+1}}(T^{\oplus i-1,z_{i-1}'}_x)$ is functionally equivalent to $T^{\oplus,i,z_i}_x$. Therefore, by \Cref{equ: rerand result round i raw} and the \refmeta{global}, we have 
    \[
    |\P_\eta(T^{i,z_i})-\P_\eta(T^{\oplus,i,z_i'}_x)| \ge \left|\P_\eta(\Fix_b(T^{i-1,z_{i-1}}))-\P_\eta(\Fix_{b\oplus x_{n-i+1}}(T^{\oplus,i-1,z_{i-1}'}_x))\right| - 6\eta \ge \beta -10\cdot (i+1)\cdot \eta. 
    \]
    The correctness of the algorithm can thus be proved by induction on a $\PV(\P)$ term, which is available by \Cref{thm: ind via binary search}. 

    Finally, in the $n$-th round, the algorithm outputs a string $z_n\in\zo^n$ such that 
    \[
    |\P_\eta(T^{n,z_n})-\P_\eta(T^{\oplus,n,z_n'}_x)| > \beta - 10\cdot (n+1)\cdot \eta, 
    \]
    where $z_n'\eqdef z_n\oplus x$. Note that both circuits above have input length $0$ and, by the definition, must output the same value. This violates the \refaxiom{boundary} by setting $\eta\eqdef \beta/(100(n+1))$.  
\end{proof}

The second lemma is as follows. Let $X_1$ and $X_2$ be two explicitly i.i.d.~RVs over $\zo$. If $\Pr[X_1= X_2]$ is larger than $1/2$, then $\bbE[X_i]$ must be biased. Formally: 

\begin{lemma}\label{lmm: collision biased}
$\APX_1$ proves the following statement. Let $n,\delta^{-1},\beta^{-1}\in\Log$, $V=\zo$, and $X_1,X_2$ be explicitly i.i.d.~RVs over $V$ defined by the circuit $C:\zo^n\to\zo$. Let $Y_c$ be the indicator random variable of $X_1= X_2$, where $X_1$ and $X_2$ takes disjoint random seeds. Then for $i\in\{1,2\}$,  
\[
\left|\bbE_\delta[X_i]-\frac{1}{2}\right| \ge \sqrt{\frac{\bbE_\delta(Y_c)}{2} - \frac{1}{4} - 5\delta -\beta}.
\]
\end{lemma}

\begin{proof}
We argue in $\APX_1$. Fix $n,\delta^{-1},\beta^{-1}\in\Log$, circuits $C,T$, and random variables $X_1,X_2,Y_c$. Let $\eta^{-1}\in\Log$ be a parameter to be determined later. Note that as $X_1$ and $X_2$ are both the indicator random variable of $C$, we can prove by \Cref{prop: indicator variable} that $\left|\bbE_\eta[X_1] - \P_\eta(C)\right|, \left|\bbE_\eta[X_2] - \P_\eta(C)\right| \le 3\eta$, 
and subsequently 
\begin{equation}\label{equ: X1 X2 are close}
\left|\bbE_\eta[X_1]-\bbE_\eta[X_2]\right| \le 6\eta. 
\end{equation}

Let $Y_0,Y_1$ be random variables over $\zo$ such that $Y_i$ takes $(x_1,x_2)\in\zo^n\times \zo^n$ and output $1$ if and only if $C(x_1)=C(x_2)=i$. It is easy to see that for every assignment $\rho=(x_1,x_2)\in\zo^{2n}$ to the seed, $Y_c|_\rho = Y_0|_\rho + Y_1|_\rho$. Therefore, by the \refmeta{avg on exp}, 
\begin{equation}\label{equ: collision close to 0 plus 1}
\left|\bbE_\eta[Y_c] - \bbE_\eta[Y_0] - \bbE_\eta[Y_1]\right| \le 6\eta. 
\end{equation}

Let $\overline{X}_1,\overline{X}_2$ be the random variables defined by $1-X_1$ and $1-X_2$, respectively. Using \refmeta{complementation}, 
\begin{equation} \label{equ: Xi complement}
\left|\bbE_\eta[\overline{X}_i] + \bbE_\eta[X_i] - 1\right| \le 6\eta 
\end{equation}
for $i\in\{1,2\}$. We can further observe that for every assignment $\rho$ to the random seed, $Y_0|_\rho = \overline{X}_1\overline X_2|_\rho$ and $Y_1|_\rho = X_1X_2|_\rho$, and subsequently by the \refmeta{avg on exp}, 
\begin{equation*}
\left|\bbE_\eta[Y_0] - \bbE_\eta[\overline{X}_1\overline X_2]\right| \le 6\eta, \quad \left|\bbE_\eta[Y_1] - \bbE_\eta[X_1X_2]\right| \le 6\eta.  
\end{equation*}
Subsequently, by the \refmeta{multiplication}, 
\begin{equation}\label{equ: Y0 Y1 close to multiplication}
\left|\bbE_\eta[Y_0] - \bbE_\eta[\overline{X}_1]\cdot \bbE_\eta[\overline X_2]\right| \le 14\eta, \quad \left|\bbE_\eta[Y_1] - \bbE_\eta[X_1]\cdot \bbE_\eta[X_2]\right| \le 14\eta. 
\end{equation}

Fix any $i\in\{1,2\}$ and let $p\eqdef \bbE_\delta[X_i]$, $q\eqdef \bbE_\delta[Y_C]$. For simplicity, we assume that $0<\bbE_\eta[X_1]<1$ and $0<\bbE_\eta[X_2]<1$. We can perform the following calculation: 
\begin{align*}
\bbE_\delta[Y_C] &\le \bbE_\eta[Y_C] + (\delta + 2\eta) \tag{\refmeta{precision exp}} \\ 
&\le \bbE_\eta[Y_0] + \bbE_\eta[Y_1] + (\delta + 8\eta) \tag{\Cref{equ: collision close to 0 plus 1}} \\ 
&\le \bbE_\eta[\overline{X}_1]\cdot \bbE_\eta[\overline X_2] + \bbE_\eta[X_1]\cdot \bbE_\eta[X_2] + (\delta + 36\eta) \tag{\Cref{equ: Y0 Y1 close to multiplication}} \\ 
&\le (1-\bbE_\eta[X_1]+6\eta)(1-\bbE_\eta[X_2]+6\eta) + \bbE_\eta[X_1]\cdot \bbE_\eta[X_2] + (\delta + 36\eta) \tag{\Cref{equ: Xi complement}}\\ 
&\le (1-\bbE_\eta[X_i]+6\eta)(1+\bbE_\eta[X_{i}]+12\eta) + \bbE_\eta[X_i]\cdot (\bbE_\eta[X_{i}]-6\eta) + (\delta + 36\eta) \tag{\Cref{equ: X1 X2 are close}} \\ 
&\le (1-p+\delta +14\eta)^2 + (p+\delta + 14\eta)^2 + (\delta + 36\eta) \tag{\refmeta{precision exp}} \\ 
&\le (1-p)^2 + (\delta + 14\eta)^2 + 2(1-p)(\delta+14\eta) + p^2 + (\delta+14\eta)^2 + 2p(\delta+14\eta) + (\delta + 36\eta) \\ 
&\le (1-p)^2 + p^2 + 5\delta + 78\eta \\ 
&\le 1 - 2(p-p^2) + 5\delta + \beta, 
\end{align*}
where the last inequality holds if we set $\eta\eqdef \beta/100$. Thus we have $p-p^2 \le (1-q)/2 + 5\delta + \beta$, and subsequently 
\[
\left|\frac{1}{2}-p\right| = \sqrt{\left(\frac{1}{2}-p\right)^2} = \sqrt{\frac{1}{4}-(p-p^2)} \ge \sqrt{\frac{1}{4}-\left(\frac{1-q}{2}+5\delta + \beta\right)} \ge \sqrt{\frac{q}{2}-\frac{1}{4}-5\delta - \beta}. 
\]
This completes the proof. 
\end{proof}

\subsubsection{Completeness of BLR Linearity Testing}

We first formalize the completeness of the linearity testing algorithm. That is, if a circuit $C:\zo^n\to\zo$ computes a function that is indeed close to a linear function $x\mapsto \langle z,x\rangle\bmod 2$, then the self-correction algorithm works. Formally: 

\dummylabel{completeness BLR correction}{Completeness of BLR Self Correction}
\begin{lemma}[Completeness of BLR self-correction]\label{lmm: completeness self-correction}
    $\APX_1$ proves the following statement. Let $n,\delta^{-1},\beta^{-1}\in\Log$, $C:\zo^n\to\zo$ be a circuit, and $z\in\zo^n$ be a string. Let $\eps\in\bbQ$ such that $\eps < 1/2$. Define the following circuits: 
    \begin{compactitem}
    \item Let $T_C(x)$ be the circuit that outputs $1$ if and only if $C(x)\ne \langle x,z\rangle\bmod 2$.
    \item Let $D(x,r):\zo^n\times \zo^n\to\zo$ be the circuit that outputs $C(x\oplus r)\oplus C(r)$. 
    \item For $x\in\zo^n$, let $T_{D,x}(r)$ be the circuit that outputs $1$ if and only if $D(x,r)\ne \langle x,z\rangle\bmod 2$. 
    \end{compactitem}
    Suppose that $\P_\delta(T_C)\le \eps$. Then for every $x\in\zo^n$, $\P_\delta(T_{D,x})\le 2(\delta+\eps)+\beta$. 
\end{lemma}

\begin{proof}
We argue in $\APX_1$. Fix $n,\delta^{-1},\beta^{-1}\in\Log$, the circuit $C$, $z\in\zo^n$, and $\eps\in\bbQ$. Let $T_C,D,T_{D,x}$ be the circuit as defined above, and $\eta^{-1}\in\Log$ be a parameter to be determined later. Suppose that $\P_\delta(T_C)\le \eps$, we know by the \refaxiom{precision} that 
\begin{equation}
\P_\eta(T_C)\le \eps + \delta + 2\eta. \label{equ: closeness} 
\end{equation}

Fix any $x\in\zo^n$. Let $X_1,X_2$ be random variables over $\zo$ that takes a seed $r\in\zo^n$, where $X_1=1$ if $C(r)\ne \langle r,z\rangle \bmod 2$, and $X_2=1$ if $C(x\oplus r) \ne \langle x\oplus r,z\rangle \bmod 2$. It is clear that $X_1$ is the indicator random variable of $T_C$, and thus by \Cref{equ: closeness} and \Cref{prop: indicator variable}, $\bbE_\eta[X_1]\le \eps + \delta + 5\eta$. By \Cref{prop: indicator variable} and the \refmeta{rerand}, we can further show that 
\[
\bbE_\eta[X_2]\le \P_\eta(T_C) + 6\eta \le \eps + \delta+8\eta. 
\]

Let $T_\lor(r)$ be the circuit that outputs $1$ if and only if $C(r)\ne \langle r,z\rangle \bmod 2$ or $C(x\oplus r) \ne \langle x\oplus r,z\rangle \bmod 2$, and $Y$ be the indicator random variable of $T_\lor$. By \Cref{prop: indicator variable} and the \refmeta{union bound}, we have 
\begin{equation}
\P_\eta(T_\lor) \le \bbE_{\eta}[Y] + 3\eta \le \bbE_\eta[X_1] + \bbE_\eta[X_2] + 6\eta \le 2\eps + 2\delta + 19\eta.\label{equ: ub T or} 
\end{equation}

Finally, we observe that if $T_{D,x}(r) = 1$, then $T_\lor(r)=1$. To see this, assume that $T_\lor(r)=0$, we have 
\[
D(x,r) = C(x\oplus r)\oplus C(r) = \langle x\oplus r,z\rangle + \langle r,z\rangle \bmod 2 = \langle x,z\rangle \bmod 2,  
\]
which implies that $T_{D,x}(r)=0$. Therefore, we have that 
\begin{align*}
\P_\delta(T_{D,x}) &\le \P_\eta(T_{D,x})+\delta + 2\eta \tag{\refaxiom{precision}} \\ 
&\le \P_\eta (T_\lor) + \delta + 5\eta \tag{\refmeta{monotone}} \\ 
&\le 2\eps + 2\delta + 24\eta.  \tag{\Cref{equ: ub T or}}
\end{align*}
This completes the proof by setting $\eta \eqdef \beta/30$. 
\end{proof}

It can be observed that this immediately gives the completeness of the BLR identity testing algorithm. Namely, if $C$ is close to a linear function, then it passes the linearity testing with high probability.  

\BLRComplete*

\begin{proof}
    We argue in $\APX$. Fix $n,\delta^{-1},\beta^{-1}\in\Log$, $C:\zo^n\to\zo$, $z\in\zo^n$, $\eps\in\bbQ$, and $T_C$, $T_{C,\BLR}$ be the circuits as described above. Let $T'_{C,\BLR}:\zo^{n}\times\zo^n\to\zo$ be the circuit that given $(x,y)$, outputs $1$ when $C(x\oplus y)\oplus C(y) \ne \langle x,z\rangle \bmod 2$. Let $I_C,I_{C,\BLR},I_{C,\BLR}'$ be the indicator random variables of $T_C,T_{C,\BLR},T_{C,\BLR}'$, respectively. 

    Let $\eta^{-1}\in\Log$ be a parameter to be determined later. Note that by the \refmeta{completeness BLR correction} and \Cref{prop: indicator variable}, we can prove that for any assignment $\rho$ to be first part of the seed of $T'_{C,\BLR}$, we have $\bbE_\eta[I_{C,\BLR}|_\rho]\le 2(\delta+\eps)+4\eta$. Subsequently, by the \refmeta{avg on exp}, we have  
    \begin{equation}
        \bbE_\eta[I_{C,\BLR}] \le 2(\delta+\eps)+4\eta + 6\eta \le 2(\delta+\eps) + 10\eta. 
    \end{equation} 
    By \Cref{prop: indicator variable} and \refaxiom{precision}, we also have $\bbE_{\eta}[I_C] \le \eps + \delta + 5\eta$.  
    
    It can be observed that if $T_{C,\BLR}(x,y)=1$, then either $T_{C}(x,y)=1$ or $T_{C,\BLR}'(x,y)=1$. Therefore, by the \refmeta{union bound}, we can prove that
    \[
    \bbE_\eta[I_{C,\BLR}] \le \bbE_\eta[I_C] + \bbE_\eta[I_{C,\BLR}'] + 6\eta \le 3(\delta+\eps)+15\eta. 
    \]
    Subsequently, by \Cref{prop: indicator variable}, $\P_\eta(T_{C,\BLR})\le 3(\delta+\eps)+\delta+18\eta$. The result then follows from the \refaxiom{precision} by setting $\eta\eqdef \beta/30$. 
\end{proof}

\subsubsection{Correctness of Majority Correction}

We move on to prove the soundness of the BLR linearity testing. As a first step, we prove that if $C$ passes the linearity testing, then the BLR self correction algorithm is \emph{single-valued}. Formally:  

\begin{lemma}[Single-valuedness of BLR correction]\label{lmm: single valued}
    $\APX_1$ proves the following statement. Let $n,\delta^{-1},\beta^{-1}\in\Log$ and $C:\zo^n\to\zo$ be a circuit. Let $\eps\in\bbQ$. Assume that $\eps,\delta,\beta\le 0.01$. Define the following circuits: 
    \begin{compactitem}
    \item Let $T_{C,\BLR}(x,y)$ be the circuit that outputs $1$ if and only if $C(x)\oplus C(y)\ne C(x\oplus y)$.  
    \item Let $D_{x,b}(r):\zo^n\to\zo$ be the circuit that outputs $1$ if and only if $C(x\oplus r)\oplus C(r)=b$. 
    \end{compactitem}
    Suppose that $\P_\delta(T_{C,\BLR})\le \eps$. For every $x\in\zo^n$, $\P_\delta(D_{x,b})\ge 1-4\eps-(4\delta+\beta)$ for some $b\in\zo$.
\end{lemma}

\begin{proof}
    We argue in $\APX_1$. Fix $n,\delta^{-1},\beta^{-1}\in\Log$, the circuit $C$, and $\eps\in\bbQ$. Let $T_{C,\BLR}$ and $D_{x,b}$ be the circuits as defined above, and $\eta^{-1}\in\Log$ be a parameter to be determined later. Suppose that $\P_\delta(T_{C,\BLR})\le \eps$, we know by the \refaxiom{precision} that 
    \begin{equation}\label{equ: BLR ub}
    \P_\eta(T_{C,\BLR})\le \eps + \delta + 2\eta. 
    \end{equation}

    Fix any $x\in\zo^n$. Let $D_{x}':\zo^{2n}\to\zo$ be circuit that takes $(r_1,r_2)\in\zo^n$ are the input, and outputs $1$ if $C(x\oplus r_1)\oplus C(r_1) = C(x\oplus r_2)\oplus C(r_2)$. Let $Y$ be the indicator random variable of $D_x'$, and $\overline Y$ be the indicator random variable of $1-D_x'$. It follows from \refmeta{avg on exp} that $\bbE_\eta [Y] \ge 1-\bbE_\eta[\overline Y] - 6\eta$. 
    
    Consider the following two circuits $T,T':\zo^{2n}\to\zo$: 
    \begin{compactitem}
    \item $T(r_1,r_2)\eqdef 1$ if and only if $C(r_1)\oplus C(r_2)\ne C(r_1\oplus r_2)$. 
    \item $T'(r_1,r_2)\eqdef 1$ if and only if $C(x\oplus r_1)\oplus C(x\oplus r_2) \ne  C((x\oplus r_1)\oplus (x\oplus r_2))$. 
    \end{compactitem}
    Let $X,X'$ be the indicator random variable of $T$ and $T'$, respectively. It is clear that $T$ is exactly $T_{C,\BLR}$, and thus by \Cref{equ: BLR ub}, $\P_\eta(T)\le \eps + \delta + 2\eta$. Similarly, as $T_i'$ is obtained $T_{C,\BLR}$ by taking bitwise-XOR to the input string with the fixed string $(x,x)$, by \Cref{equ: BLR ub} and the \refmeta{rerand}, $\P_\eta(T')\le \eta + \delta + 5\eta$. 

    Moreover, one can observe that $D_x'(x,y)=0$ implies that either $T(x,y)$ or $T'(x,y)$ outputs $1$: This is because if $T(x,y)=T'(x,y)=0$, we can conclude that 
    \[
    C(x\oplus r_1)\oplus C(x\oplus r_2)\oplus C(r_1)\oplus C(r_2) = C(r_1\oplus r_2)\oplus C((x\oplus r_1)\oplus (x\oplus r_2)) = 0, 
    \]
    which implies that $D_x'(x,y)=1$. Subsequently, by the \refmeta{union bound}, we have 
    \[
    \bbE_\eta[\overline{Y}] \le \bbE_\eta[X_i] + \bbE_\eta[X_i'] + 3\eta \le 2(\eps+\delta) + 10\eta. 
    \]
    and thus $\bbE_\eta[Y]\ge 1-\bbE_\eta[\overline Y]-6\eta \ge 1 - 2(\eps+\delta) - 16\eta$. 

    Let $I_x,I_x'$ be explicitly i.i.d.~RVs over $\zo$ defined by the circuit $D_x(r)\eqdef C(x\oplus r)\oplus C(r)$, and $Y_c$ is the indicator random variable of $I_x=I_x'$. By definitions, we can see that $Y_c|_\rho=Y|_\rho$ for any assignment $\rho$, and thus by the \refmeta{avg on exp}, 
    \[
    \bbE_\eta[Y_c] \ge \bbE_\eta[Y] - 6\eta \ge 1-2(\eps+\delta)-22\eta. 
    \]
    Subsequently, by \Cref{lmm: collision biased}, we have 
    \begin{equation}
        \left|\bbE_\eta[I_x]-\frac{1}{2}\right| \ge \sqrt{\frac{1-2(\eps+\delta)-22\eta}{2}-\frac{1}{4}-6\eta} \ge \sqrt{\frac{1}{4}-(\eps+\delta+17\eta)} \ge \frac{1}{2}-4(\eps+\delta+17\eta). 
    \end{equation}

    Recall that $I_x$ is the random variable that takes $(r_1,r_2)\in\zo^{2n}$ as random seed and outputs $C(x\oplus r_1)\oplus C(r_1)$. Suppose that $\bbE_\eta[I_x] \ge 1 - 4(\eps+\delta+17\eta)$. By \refmeta{avg on exp}, there is an assignment $\rho$ of the second part $r_2$ of the seed (which was for $I_{x}'$) such that 
    \[
    \bbE_\eta[I_x|_\rho] \ge \bbE_\eta[I_x] - 6\eta \ge 1 - 4\eta - (4\delta + 74\eta). 
    \]
    As $I_x|_\rho$ is the indicator random variable of $D_{x,1}$, it follows from \Cref{prop: indicator variable} that $\P_\eta(D_{x,1}) \ge \bbE_\eta[I_x] - 3\eta \ge 1 - 4\eps - (4\delta + 74\eta)$. It suffices if we set $\eta\le \beta/100$. The other case $\bbE_\eta[I_x] \le 4(\eps+\delta+17\eta)$ can be resolved by considering $\overline{I_x}\eqdef 1-I_x$. 
\end{proof}

\Cref{lmm: single valued} shows that the BLR self correction algorithm is single-valued assuming that the circuit $C$ passes the linearity testing. Second, we show that the ``corrected'' function $g(\cdot )$ satisfies that $g(x)\oplus g(y)=g(x\oplus y)$ for \emph{every} $x,y\in\zo^n$. Formally: 

\begin{lemma}[Linearity of BLR correction]\label{lmm: linearity BLR}
    $\APX_1$ proves the following statement. Let $n,\delta^{-1},\beta^{-1}\in\Log$ and $C:\zo^n\to\zo$ be a circuit. Let $\eps\in\bbQ$. Assume that $\eps,\delta,\beta\le 0.01$. Let $D_{x,b}$ be the circuit in \Cref{lmm: single valued}, and $g(x)$ be the $\P$-oracle circuit that works as follows: Given $x\in\zo^n$, it outputs $b\in\zo$ if $\P_{\delta}(D_{x,b})\ge 1-4\eps -(4\delta+\beta)$, and $\bot$ otherwise.
    
     Let $T_{C,\BLR}(x,y)$ be the circuit that outputs $1$ if and only if $C(x)\oplus C(y)\ne C(x\oplus y)$. Suppose that $\P_\delta(T_{C,\BLR})\le \eps$. Then for every $x_1,x_2\in\zo^n$, $g(x_1)\oplus g(x_2)=g(x_1\oplus x_2)$. 
\end{lemma}

\begin{proof}
    We argue in $\APX_1$. Fix $n,\delta^{-1},\beta^{-1}\in\Log$, the circuit $C$, and $\eps\in\bbQ$. Assume that $\P_\delta(T_{C,\BLR})\le \eps$. Note that by \Cref{lmm: single valued} and $\eps,\delta,\beta<0.01$, $g(x)\in\zo$ for every $x\in\zo^n$. Fix $x_1,x_2\in\zo^n$ and let $b_j\eqdef g(x_j)$ for $j\in\{1,2\}$, and $b_\oplus\eqdef g(x_1\oplus x_2)$. 
    
    Let $\eta^{-1}\in\Log$ be a parameter to be determined later, and $X_1,X_2,X_\oplus$ be the random variables over $\zo$ with seed $r\in\zo^n$ defined as follows: 
    \begin{compactitem}
    \item $X_1$ outputs $1$ if and only if $b_1 = C(x_1\oplus r)\oplus C(r)$. 
    \item $X_2$ outputs $1$ if and only if $b_2 = C(x_2\oplus r)\oplus C(r)$. 
    \item $X_\oplus$ outputs $1$ if and only if $b_\oplus = C(x_1\oplus x_2\oplus (x_2\oplus r)) \oplus C(x_2\oplus r)$.
    \end{compactitem}
    Note that as $b_j\eqdef g(x_j)$, we know by the definition of $g$ that $\P_{\delta}(D_{x_j,b_j})\ge 1-4\eps-(4\delta+\beta)$. It can be observed that $X_1,X_2$ are the indicator random variables of $D_{x_1,b_1}, D_{x_2,b_2}$, thus by \Cref{prop: indicator variable} and the \refmeta{precision exp},
    \[
    \bbE_{\eta}[X_1],\bbE_{\eta}[X_2] \ge 1-4\eps-(5\delta+\beta-2\eta).  
    \]
    Moreover, $X_\oplus$ is the indicator variable of the circuit that outputs $1$ if $b_\oplus = C(x_1\oplus x_2\oplus (x_2\oplus r)) \oplus C(x_2\oplus r)$, and the circuit is obtained from $D_{x_1\oplus x_2,b_\oplus}$ by taking XOR to the input with a fixed string $x_2$. Therefore, by \Cref{prop: indicator variable} and the \refmeta{rerand}, we have $\bbE_{0.01}[X_\oplus]\ge 1-4\eps-(5\delta+\beta-2\eta)$. 

    Let $\overline{X}_1\eqdef 1-X_1$, $\overline{X}_2\eqdef 1-X_2$, and $\overline{X}_\oplus\eqdef 1-X_\oplus$. Let $Y\eqdef \overline{X}_1\lor \overline{X}_2\lor \overline{X}_\oplus$ and $\overline{Y} \eqdef 1-Y$. Then using \refmeta{complementation} and \refmeta{union bound}, we have 
    \begin{align*}
    \bbE_{\eta}[Y]& \le \bbE_\eta[\overline X_1] + \bbE_\eta[\overline X_2] + \bbE_\eta[\overline X_\oplus] + 9\eta \tag{\refmeta{union bound}}\\ 
       &  \le (1-\bbE_\eta[X_1]) + (1 - \bbE_\eta[X_2]) + (1 - \bbE_\eta[X_\oplus]) + 18\eta \tag{\refmeta{complementation}} \\ 
       & \le 12\eps + 15\delta + \beta + 24\eta.  
    \end{align*}
    Again, using \refmeta{complementation}, we have $\bbE_{\eta}[\overline {Y}] \ge 1 - (12\eps + 15\delta + \beta + 24\eta)$. By setting $\eta \eqdef \beta / 100$, we have $\bbE_{\eta}[\overline {Y}] > 3\eta$. 
    
    By \refmeta{avg on exp}, there exists an assignment $\rho$ such that $\bbE_{0.01}[\overline{Y}|_\rho] \ge 0.05$, or in other words, $\overline{Y}|_\rho=1$ as its seed length is $0$ after applying the restriction $\rho$. By the definition of the random variables, this indicates that     
    \begin{align*}
        b_1 &= C(x_1\oplus r)\oplus C(r),\\ 
        b_2 &= C(x_2\oplus r)\oplus C(r), \\ 
        b_\oplus &= C(x_1\oplus x_2\oplus (x_2\oplus r)) \oplus C(x_2\oplus r) = C(x_1\oplus r)\oplus C(x_2\oplus r).  
    \end{align*}
    It immediately follows that $b_\oplus =b_1\oplus b_2$.  
\end{proof}

\subsubsection{Soundness of BLR Linearity Testing}

Now we are ready to prove the soundness of the BLR linearity testing. At a high level, we will recover the string $z\in\zo^n$ that defines the linear function using the oracle circuit $g(\cdot)$. It is worth noting that the correctness proof of the string $z$ is quite non-trivial: It crucially builds on the tools for random variables developed in \Cref{sec: basic properties}, especially \refmeta{avg on exp}. 

\BLRSound*

\begin{proof}
    We argue in $\APX_1$. Fix $n,\delta^{-1},\beta^{-1}\in\Log$, the circuit $C$, and $\eps\in\bbQ$. Let $T_{C,z}$ and $T_{C,\BLR}$ be the circuits as defined above, and $\eta^{-1}\in\Log$ be a parameter to be determined later. Suppose that $\P_\delta(T_{C,\BLR})\le \eps$, we know by the \refaxiom{precision} that 
    \begin{equation}\label{equ: BLR ub 2}
    \P_\eta(T_{C,\BLR})\le \eps + \delta + 2\eta. 
    \end{equation}

    Let $D_{x,b}(r)\eqdef C(x\oplus r)\oplus C(r)\oplus (1-b)$ be the circuit in \Cref{lmm: single valued}, and $g(x)$ be the oracle circuit in \Cref{lmm: linearity BLR}. By \Cref{lmm: single valued}, we have that for every $x\in\zo^n$, there is a $b\in\zo$ such that
    \begin{equation}
        \P_\eta(D_{x,b}) \ge 1 - 4(\eps + \delta + 2\eta) - 5\eta \ge  1-4(\eps+\delta) - 13\eta. \label{equ: single valued eta}
    \end{equation}
    Note that we will choose $\eta$ such that $13\eta < 0.01$. Therefore, by \Cref{lmm: linearity BLR} that for every $x\in\zo^n$, the bit $b$ satisfying \Cref{equ: single valued eta} is given by $g(x)$. Moreover, for every $x_1,x_2\in\zo^n$, we have that $g(x_1)\oplus g(x_2)=g(x_1\oplus x_2)$. 

    Let $e_i$ be the string that is $0$ on all but the $i$-th bit, and $z\eqdef g(e_1)\circ g(e_2)\circ \dots\circ g(e_n)$. That is, $z_i=g(e_i)$ for every $i\in[n]$. Let $X$ and $Y$ be random variables over $\zo$ that take a seed $(x,r)\in\zo^n\times \zo^n$ of length $2n$ and are defined as follows.
    \begin{compactitem}
    \item $X$ outputs $1$ if $C(x\oplus r)\oplus C(x)\ne C(r)$. That is, $X$ is the indicator random variable of $T_{C,\BLR}$. By \Cref{prop: indicator variable},  
    \begin{equation}
        \bbE_\eta[X] \le \P_\eta(T_{C,\BLR}) + 3\eta \le \eps + \delta + 5\eta. 
    \end{equation} 
    \item $Y$ outputs $1$ if $C(x) \ne \langle x,z\rangle \bmod 2$. Note that for every assignment $r$ to the second part of its seed, $Y|_r$ is the indicator random variable of $T_{C,z}$. 
    \end{compactitem} 

    Next, we will prove that for every $x\in\zo^n$, 
    \begin{equation}
    \left|\bbE_\eta[X|_x] - \bbE_\eta[Y|_x]\right| \le 4(\eps+\delta) + 13\eta. \label{equ: close X and Y after assignment}
    \end{equation} 
    Recall that by the definition of $X$ and $D_{x,b}$, we have that for every assignment $x$ to the \emph{first} part of their seeds, $X|_x$ is the indicator random variable of $D_{x,1}$, and for $\overline X\eqdef 1-X$, $\overline X|_x$ is the indicator random variable of $D_{x,0}$. Therefore, by \Cref{lmm: single valued}, 
    \[
    \max \{ \bbE_\eta[X|_x], 1-\bbE_\eta[X|_x] \} \ge 1- 4(\eps+\delta+2\eta) - 5\eta \ge 1-4(\eps+\delta)-13\eta.   
    \]
    We consider the case that $\bbE_\eta[X|_x] \ge 1-4(\eps+\delta)-13\eta$, and the other case is similar. By the definition of $g$, we know that $g(x) = 1$, and subsequently 
    \[
    \langle x,z\rangle \bmod 2 = \sum_{i\in[n],x_i=1} g(e_i)\bmod 2 = g(x),
    \]
    where the last equality follows from \Cref{lmm: linearity BLR} and the \refaxiom{ind 1} (note that the induction axiom suffices as $n\in\Log$). For any assignment $r$ to the second part of the seed, we have $\bbE_\eta[Y|_x|_y] = 1$, which subsequently implies that $\bbE_\eta[Y|_x] = 1$. Therefore, for any assignment $x\in\zo^n$ to the first part of the seed, 
    \[
    \left|\bbE_\eta[X|_x] - \bbE_\eta[Y|_x]\right| \le 4(\eps+\delta) + 13\eta. 
    \]
    By the \refmeta{avg on exp}, we have $\left|\bbE_\eta[X] - \bbE_\eta[Y]\right| \le 4(\eps + \delta) + 19\eta$, and thus 
    \[
        \bbE_\eta[Y] \le \bbE_\eta[X] + 4(\eps + \delta) + 19\eta \le  e. 
    \]

    Again, by the \refmeta{avg on exp}, there is an assignment $r\in\zo^n$ to the second part of its seed such that 
    \begin{equation}
        \bbE_\eta[Y |_r] \le \bbE_\eta[Y] + 3\eta \le 5(\eps + \delta) + 27\eta. 
    \end{equation}
    Fix the assignment $r$. As mentioned above, $\bbE_\eta[Y|_r]$ is the indicator random variable of $T_{C,z}$, and thus by \Cref{prop: indicator variable} and \refaxiom{precision}, 
    \[ 
    \P_\delta(T_{C,z}) \le \P_\eta(T_{C,z}) + \delta + 2\eta \le \bbE_\eta[Y|_r] + \delta  + 5\eta \le  5\eps + 6\delta + 32\eta. 
    \] 
    It completes the proof by taking $\eta\eqdef \beta / 50$. 
\end{proof}

\section{Witnessing Theorems and Relative Strength of \texorpdfstring{$\Apx_1$}{Apx1}}\label{sec: witnessing}

In this section, we prove a witnessing theorem for $\Apx_1$ and consider its relation to other theories of bounded arithmetic, including $\PV_1$ and $\APC_1$. 

\subsection{Provably Total \texorpdfstring{$\TFNP$}{TFNP} Problems in \texorpdfstring{$\APX_1$}{APX1}}

In this subsection, we will introduce a witnessing theorem for the $\forall\Sigma_1^b$-consequences of $\APX_1$ (i.e.~provably total $\TFNP$ problems in $\APX_1$). 

\subsubsection{A \texorpdfstring{$\TFZPP$}{TFZPP} Problem: \texorpdfstring{$\RefYao$}{Refuter(Yao)}} 

We will first introduce a $\TFZPP$ problem\footnote{A search problem is said to be in $\TFZPP$ if it is a $\TFNP$ problem solvable by randomized polynomial-time algorithms.} called \emph{Refutation of Yao-Predictor Generators}; we denote it by $\RefYao$. Recall that Yao's distinguisher-to-predictor transformation \cite{Yao82} (see  \Cref{sec:Yao_transformation}) shows that if a distribution $\disD$ over $\zo^n$ is not $\eps$-pseudorandom, i.e., there is a circuit $C:\zo^n\to\zo$ (called \emph{distinguisher}) such that 
\[
\Big |\Pr[C(\disD)] - \Pr[C(\disU_n)]\Big | > \eps, 
\]
then there exists $i\in[n]$ and a \emph{predictor} $P_i:\zo^{i-1}\to\zo$ such that 
\[
\Pr_{x\gets \disD}[P(x_{<i}) = x_i] \ge \frac{1}{2} + \frac{\eps}{4n},
\]
i.e., $P$ predicts the $i$-th bit of $\mathcal{D}$ with advantage at least $\varepsilon/4n$.
This transformation serves as a key step in the construction and analysis of pseudorandom generators (see, e.g., \cite{NW,IW97}): it shows that an unpredictable distribution is  necessarily pseudorandom. 

In the statement below, we say that a discrete probability distribution $\mathcal{D}$ is \emph{flat} if it is uniform over its support, i.e., over the set of elements with non-zero probability over $\mathcal{D}$. The \emph{size} of the distribution is the size of its support. We will represent flat distributions explicitly as a list of strings. In the subsequent discussions, we might tacitly assume that the relevant distribution is flat and explicitly represented.

\begin{definition}
The search problem $\RefYao$ is defined as follows. 
\begin{compactitem}
\item (\emph{Parameters}). Length of strings $n$, distribution size $m$, predictor size $s$, and advantage $\delta\in[0,1]$.
\item (\emph{Input}). A circuit $G:\zo^{nm}\to[n]\times \zo^s$ (called \emph{predictor generator}). 
\item (\emph{Solution}). Any explicit flat distribution $\disD\in(\zo^{n})^m$ of size $m$ such that the following holds:

\vspace{0.2cm}

Let $(i,P)\eqdef G(\disD)$, where $P:\zo^{i-1}\to\zo$ is parsed as a circuit of description length $\leq s$. Then 
\[
\Pr_{x\gets \disD}[P(x_{<i}) = x_i] < \frac{1}{2} + \delta. 
\]
In other words, $P$ is not a predictor of the $i$-th bit of $\disD$ with advantage $\delta$.  
\end{compactitem}
\end{definition}

For concreteness, one may think of the parameter regime $m=n^{10}$, $s = n^2$, and $\delta = 0.1$. In this case, a random distribution of $m$ strings of length $n$ is likely $o(1)$-pseudorandom against any circuit of size $s$, and thus must be a solution of $\RefYao$ no matter the input circuit $G$. 

At a high level, $\RefYao$ asks to generate a distribution $\disD$ that is unpredictable against a \emph{given} predictor generator $G$ --- a deterministic algorithm that aims to output a predictor $P$ for $\disD$. The distribution $\disD$ is not necessarily an unpredictable (or equivalently, pseudorandom) distribution against \emph{small circuits}; it suffices to fool the given \emph{deterministic} predictor generator $G$. This makes it a special case of constructing targeted PRGs, which is known to be $\pr\BPP$-complete (see, e.g., \cite{Goldreich11g,CT21,LPT24}).

\subsubsection{Connection to \texorpdfstring{$\LossyCode$}{Lossy Code}} 

A closely related $\TFZPP$ relation is the \emph{Lossy Code Problem}; we denote it by $\LossyCode$. Inspired by the literature in bounded arithmetic (see \citep[Section 3.1]{Jerabek07} and the discussion below), the problem is defined in \cite{Korten22} as a more feasible variant of the Range Avoidance Problem; see \cite{Korten-survey} and references therein for an introduction to this line of work. 

\begin{definition}
    The search problem $\LossyCode$ is defined as follows. 
    \begin{compactitem}
    \item (\emph{Input}). Circuits $C:\zo^{n}\to\zo^{n-1}$ and $D:\zo^{n-1}\to\zo^{n}$. These two circuits are called \emph{compressor} and \emph{decompressor}, respectively.
    \item (\emph{Output}). A string $x\in\zo^n$ such that $D(C(x))\ne x$.
    \end{compactitem}
\end{definition}

It is clear that $\LossyCode\in\TFZPP$. Indeed, Wilkie (unpublished) and Thapen \cite{thapen2002weak} proved that the problem captures the $\forall\Sigma_1^b$-fragment of the theory $\APC_1$. 

\begin{theorem}[{\cite[Proposition 1.14]{Jerabek04}, also see \cite[Theorem D.1]{LPT24}}]\label{thm: APC1 witnessing}
Let $\phi(x,y)$ be a quantifier-free formula in the language of $\APC_1$ that only has $x$ and $y$ as open variables. If $\APC_1\vdash \forall x~\exists y~\phi(x,y)$, then there is a deterministic polynomial-time reduction from the following problem to $\LossyCode$: Given $n\in\bbN$, output  $m\in\bbN$ such that $\phi(n,m)$ is true in the standard model. 
\end{theorem}

Moreover, it has been recently discovered that some natural $\TFZPP$-search problems admit deterministic reductions to $\LossyCode$ or its variants: constructing large prime numbers with factoring oracles \cite{Korten22} and the simulation of catalytic logspace machines \cite{CLMP25}. Variants of $\LossyCode$ are relevant to both full and partial derandomizations of $\pr\BPP$; see \cite{LPT24} for a comprehensive introduction. 

Note that assuming $\pr\BPP=\pr\P$, both $\LossyCode$ and $\RefYao$ are in $\FP$. Nevertheless, it is interesting to discover the relative hardness of their derandomization. By adapting an idea from \cite{Korten22}, we show that $\RefYao$ admits a deterministic polynomial-time reduction to $\LossyCode$. Therefore, showing that  $\RefYao\in\FP$  is necessary before proving that  $\LossyCode \in \FP$. 

\begin{theorem}[Implicit in the proof of {\cite[Corollary 41]{Korten22}}]\label{thm: ref yao to lossy}
    There is a deterministic polynomial-time mapping reduction from $\RefYao$ with parameters $(\delta^2/10)\cdot m \ge s + \lceil \log n\rceil + 1$ to $\LossyCode$.
\end{theorem}

\begin{proof}
    Note that we can encode $m$-bit strings with Hamming weight at most $k$ by $\log_2 \binom{m}{k}+O(\log k)$ bits, where the encoding and decoding algorithms run in polynomial time (see, e.g., \citep[Lemma 5.4]{CLO24}). In particular, when $k\eqdef (1/2 - \delta)\cdot m$ and $m$ is sufficiently large, the encoding length is
    \begin{align*}
        &~\log_2\binom{m}{(1/2-\delta)\cdot m} + O(\log (1/2-\delta) + \log m) \\ 
    \le &~m + \log_2(e^{-(2\delta)^2m/4}) + O(\log (1/2-\delta) + \log m) \\ 
    \le &~m - (\delta^2/10)\cdot  m,
    \end{align*}
    where the first inequality follows from the Chernoff bound. 
    
    Now we describe the reduction. Given any predictor generator $G:\zo^{nm}\to[n]\times \zo^s$, consider the following compressor $C:\zo^{nm}\to\zo^{nm-1}$ and decompressor $D:\zo^{nm-1}\to\zo^{nm}$: 
    \begin{itemize}
    \item (\emph{Compressor}). Given any $\disD\in\zo^{nm}$, the compressor parses it as a distribution over $n$-bit strings of size $m$. It computes $(i,P)\eqdef G(\disD)$. If $P$ fails to predict the $i$-th bit of $\disD$ with advantage $\delta$, it fails and aborts. Otherwise,  
    \[
    \Pr_{x\gets \disD}[P(x_{<i}) = x_i] \ge  \frac{1}{2} + \delta.
    \]
    Let $y$ be the $m$-bit string defined as $y_j\eqdef P(x_{<i}^{(j)})\oplus x_i^{(j)}$, where $x^{(j)}$ is the $j$-th string in $\disD$. Then $y$ is a string of Hamming weight at most $(1/2-\delta)\cdot m$, and thus can be efficiently encoded using $m-(\delta^2/10)\cdot m$ bits. 
    
    Let $\hat y$ be its encoding, and $\disD_{-i}\in\zo^{m(n-1)}$ be the distribution $\disD$ after removing the $i$-th bit from all strings. The compressor outputs the tuple $(i,P,\hat y,\disD_{-i})$, which is of length at most 
    \[ 
    \lceil \log n\rceil + s +(m-(\delta^2/10)\cdot m) + (nm- m) < nm 
    \]
    due to the assumption on parameters.
    \item (\emph{Decompressor}). When the compressor does not fail, the decompressor can recover $\disD$ from $(i,P,\hat y,\disD_{-i})$ by first recovering $y$ then computing the missing bits
    \[ 
    x^{(j)}_i\eqdef P(x_{<i}^{(j)})\oplus y_j. 
    \] 
    \end{itemize}
    Given a predictor generator $G$ and parameters as above, the mapping reduction from $\RefYao$ to $\LossyCode$ outputs $(C,D)$ as an instance of $\LossyCode$. 

    It suffices to prove that the reduction is correct. Given any string $\disD$ such that $D(C(\disD))\ne \disD$, we know by the discussion above that the compressor must fail. In other words, $G(\disD)$ fails to produce a predictor with advantage $\delta$. This means that $\disD$ is a solution to the $\RefYao$ instance and thus concludes the proof.  
\end{proof}

\subsubsection{The Witnessing Theorem} 

We are now ready to show the following  witnessing theorem for $\APX_1$: any provably total $\TFNP$ problem in $\APX_1$ is deterministically reducible to $\RefYao$. 

\WitnessingAPX*

Note that the inequality $(\delta^2/10)\cdot m \ge s + \lceil \log n\rceil +1$ implies that the $\RefYao$ instance reduces to $\LossyCode$ and, in particular, it is a total search problem. A more refined analysis of our proof may lead to an improved trade-off between the parameters, which we leave for future work.

To prove this witnessing theorem, we will need the standard Herbrand's theorem for universal first-order theories and a lemma that extracts a predictor from an  $\APX$ proof. The latter requires a proof-theoretic analysis and is deferred to the end of the section (see \Cref{sec:predictor_extraction_proof}).

\newcommand{\calT}{\mathcal{T}}
\newcommand{\calO}{\mathcal{O}}

\begin{theorem}[Herbrand's Theorem; see, e.g.,  \cite{Buss94}]
    Let $\calT$ be a universal first-order theory and $\varphi(x,y)$ be a quantifier-free formula with only $x$ and $y$ as open variables. If $\calT\vdash\forall x~\exists y~ \varphi(x,y)$, there exists a constant $c\in\bbN$ and terms $t_1,t_2,\dots,t_c$ such that 
    \[
    \calT \vdash  \forall x~\bigvee_{i=1}^c \varphi(x,t_i(y))\,.
    \]
\end{theorem}

\begin{restatable}[Predictor Extraction Lemma]{lemma}{Extraction}\label{lmm: predictor extraction}
    Let $t_1(\vec x)=t_2(\vec x)$ be an equation provable in $\APX$. Then there are polynomials $k(n)$, $m(n)$, and a deterministic  polynomial-time algorithm $E$ that satisfies the following conditions when $n$ is sufficiently large: 
    \begin{itemize}
    \item \emph{(}Input\emph{)}. A string $\vec x\in\zo^n$ and a flat distribution $\disD\in(\zo^{k})^m$ of size $m$ over $k$-bit strings.
    \item \emph{(}Simulation of Terms\emph{)}. Recall that $t_1,t_2$ are interpreted as polynomial-time $\P$-oracle algorithms in standard models. We will simulate the algorithms on input $\vec x$ as follows: For every oracle query $\P(C,\Delta)$, where $C:\zo^{t}\to\zo$, we will ensure that $t\le k(n)$ and answer the query by 
    \begin{equation}\label{equ: simulation of P by dist}
    \P(C,\Delta)\eqdef \Pr_{u\gets\disD}[C(u_{\le t})]. 
    \end{equation}
    We denote the output of $t_1$ in the simulation as $t_1^{\disD}(\vec x)$, and the output of $t_2$ as $t_2^{\disD}(\vec x)$. 
    \item \emph{(}Output\emph{)}. Suppose that $t_1^\disD(\vec x)\ne t_2^{\disD}(\vec x)$. Then  $E(\vec x,\disD)$ outputs $i\in[k(n)]$ and a circuit $P \colon \{0,1\}^{i-1} \to \{0,1\}$ of size at most $s$ such that $P$ predicts the $i$-th bit of $\disD$ with advantage $\delta$ such that
    $$(\delta^2/10) \cdot m(n)\ge s+\lceil \log k(n)\rceil +1.
    $$ 
    \end{itemize}
\end{restatable}

\begin{proof}[Proof of \Cref{thm: witness TFNP}]
    Recall that $\APX_1$ admits a universal axiomatization (see \Cref{prop: apx1 universal}). Suppose that $\APX_1\vdash \forall x~\exists y~\varphi(x,y)$. By Herbrand's  theorem, there are terms $t_1,\dots,t_c$ in the language of $\APX_1$ such that 
    \[
    \APX_1\vdash \forall x~\bigvee_{i=1}^c \varphi(x,t_i(y))\,,
    \]
    for some $c\in\bbN$. 
    
    Note that the language of $\APX_1$ is the language of $\PV$ extended by the approximate counting oracle $\P$; therefore, $t_1,\dots,t_c$ are polynomial time $\P$-oracle algorithms in the standard model. Let $t_\varphi$ be a term in $\APX$ such that 
    \[
    \APX_1\vdash \bigvee_{i=1}^{c}\varphi(x,t_i(x))\leftrightarrow t_\varphi(x,t_1(x),\dots,t_c(x))=1.  
    \]
    This can be done as $\varphi$ is a quantifier-free formula; see, e.g., \cite[Chapter 3]{DBLP:journals/eccc/Li25}. Then we know that $\APX_1$ proves that $\forall x~t_\varphi(x,t_1(x),\dots,t_c(x))=1$. As $\APX_1$ is conservative over $\APX$ (see \Cref{prop: APX1 conservative}), we know that $\APX\vdash t_\varphi(x,t_1(x),\dots,t_c(x))=1$. 

    By \Cref{lmm: predictor extraction} (instantiated with $t_1\eqdef t_\varphi(x,t_1(x),\dots,t_c(x))$ and $t_2\eqdef 1$), there are $k=k(n)$, $m=m(n)$, $s=s(n)$, $\delta=\delta(n)$ and a polynomial time $E$ such that the following holds. Given $x\in\zo^n$ and a distribution $\disD\in(\zo^{k})^m$ of size $m$, 
    \begin{compactitem}
    \item either $t_\varphi^{\disD}(x,t_1^{\disD}(x),\dots,t_k^{\disD}(x))=1$; or 
    \item $E(x,\disD)$ outputs $i\in[k(n)]$ and $P$ that predicts the $i$-th bit of $\disD$ with advantage $\delta$, where $(\delta(n)^2/10) \cdot m(n)\ge s(n)+\lceil \log k(n)\rceil +1$. 
    \end{compactitem}

    The reduction produces the circuit $E(x,\cdot)$ as an instance of $\RefYao$, where the size-$m(n)$ distribution is supported over $k(n)$-bit strings, the predictor size is $s(n)$, and the advantage is $\delta(n)$. Given $x$ of length $n$, for any solution $\disD$ to the resulting instance $E(x, \cdot)$ of $\RefYao$, we know by definition that $E(x,\disD)$ cannot output a predictor with advantage $\delta$. As a result, the first bullet above must hold: 
    \[
    t_\varphi^{\disD}(x,t_1^{\disD}(x),\dots,t_k^{\disD}(x))=1. 
    \]
    Subsequently, given any solution $\disD$ to the instance of $\RefYao$, one of $t_1^{\disD}(x),t_2^{\disD}(x),\dots,t_k^{\disD}(x)$ must output $y$ such that $\varphi(x,y)$ holds. This gives a correct reduction, as simulations of $t_1^{\disD}(x),t_2^{\disD}(x),\dots,t_k^{\disD}(x)$ can be implemented in deterministic polynomial time given $x$ and the explicit description of $\mathcal{D}$.  
\end{proof}

\subsection{Relationship to \texorpdfstring{$\PV_1$}{PV1}: Is \texorpdfstring{$\mathsf{prBPP} =  \mathsf{prP}$}{prBPP=prP} Feasibly Provable?}

In this subsection, we introduce a few questions regarding the relative strength of $\APX_1$ and $\PV_1$. We will discuss their importance and connection to the program of proving $\pr\BPP=\pr\P$. No meaningful progress is reported in the paper; we believe the resolution of the questions, even conditionally, would advance our understanding of feasible mathematics and derandomization. 

\paragraph{Feasible proof of $\pr\BPP=\pr\P$.} One major open problem in complexity theory is whether derandomization is possible in general with polynomial runtime overhead. The seminal work of Nisan, Wigderson, and Impagliazzo \cite{NW,IW97} shows that $\pr\BPP=\pr\P$ follows from exponential circuit lower bounds for $\E=\DTIME[2^n]$; hence many researchers expect a positive answer. However, despite enormous efforts, both $\pr\BPP=\pr\P$ and the circuit lower bounds for $\E$ remain open.

From the perspective of meta-mathematics, an interesting question is to study whether $\pr\BPP=\pr\P$ is \emph{\emph{(}un\emph{)}provable} in a weak arithmetic theory such as $\PV_1$. A technical challenge is that, as the language of $\PV_1$ is designed to capture \emph{deterministic} polynomial time computable functions, it is a priori not obvious how to formalize the statement $\pr\BPP=\pr\P$, which involves the acceptance probability of circuits over inputs from a set of exponential size. 

We propose the investigation of the following related question. 

\begin{open}\label{open: BPP = P}
    Is there a $\PV$ function symbol $\P(C,\Delta)$ such that the \refaxiom{basic}, \refaxiom{boundary}, \refaxiom{precision}, and \refaxiom{local} are provable in $\PV_1$?  
\end{open}

We note that an unconditional positive answer is unlikely to be obtained in the near future, as it immediately implies $\pr\BPP=\pr\P$ by the soundness of $\PV_1$ and \Cref{thm: standard equiv admissible}. Indeed, a positive answer shows, intuitively, that $\pr\BPP=\pr\P$ admits a \emph{deterministic polynomial-time proof}. To our knowledge, it is unclear whether a positive or negative answer is more plausible.

\paragraph{Feasibly provable derandomization for deterministic statements.} On the other hand, we may also consider a weaker collapse: it is in principle possible that, despite that there may not be a $\PV$ function symbol $\P$ such that the relevant axioms of $\APX_1$ are provable in $\PV_1$, the introduction of the oracle $\P$ does not help in proving any sentence that \emph{does not} involve the oracle $\P$. Formally:  

\begin{open}\label{open: conservative PV}
    Is $\APX_1$ \emph{conservative} over $\PV_1$? In other words, is it the case that every first-order sentence in the language of $\PV$ that is provable in $\APX_1$ is also provable in $\PV_1$? 
\end{open}

It is clear that a positive answer to \Cref{open: BPP = P} implies a positive answer to \Cref{open: conservative PV}. Moreover, a positive answer of this open problem would immediately imply a witnessing theorem that improves \Cref{thm: witness TFNP}: any $\APX_1$ provably total $\TFNP$ relation (expressed by a quantifier-free formula in the language of $\PV_1$) is in $\FP$. This is because $\PV_1$ provably total $\TFNP$ problems are in $\FP$ (see, e.g., \cite[Section 3.1]{Oliveira25}). 

An interesting characteristic of \Cref{open: conservative PV} is that it appears to be \emph{incomparable} with $\pr\BPP=\pr\P$. If $\pr\BPP=\pr\P$ but the proof is not feasible, $\APX_1$ may not be conservative over $\PV_1$. More interestingly, if the answer to \Cref{open: conservative PV} is positive, it is still unclear to us whether $\pr\BPP=\pr\P$ or any other nontrivial derandomization follows. Formally: 

\begin{open}
    Suppose that $\APX_1$ is conservative over $\PV_1$.  Does it follow that $\pr\BPP=\pr\P$, $\pr\BPP=\pr\ZPP$, or any other unknown general derandomization  result hold? 
\end{open}

At a high level, this is to ask whether it is necessary to \emph{derandomize computations in general} if we want to \emph{derandomize proofs in general}. We contend that these problems are fundamental and merit deeper investigation.  

\subsection{Relationship to \texorpdfstring{$\APC_1$}{APC1}}

\newcommand{\NW}{\mathsf{NW}}
\newcommand{\HardA}{\Hard^{\mathsf{A}}}
\newcommand{\UAPC}{\mathsf{UAPC}}
\newcommand{\iO}{i\mathcal{O}}
\newcommand{\HARDA}{\HARD^\mathsf{A}}

We now study the relative strength of $\APX_1$ and $\APC_1$. We will show that, in a formal sense, $\APC_1$ can be viewed as an extension of $\APX_1$. We will then show that $\APC_1$ is likely a \emph{strict} extension of $\APX_1$. Finally, we introduce a few open problems related to the relative strength of $\APC_1$ and $\APX_1$. 

\subsubsection{An Upper Bound: \texorpdfstring{$\APC_1$}{APC1} Extends \texorpdfstring{$\APX_1$}{APX1}}

We first prove an upper bound for $\APX_1$ that is implicit in \Jerabek{}'s results on approximate counting \cite{Jerabek07}. In particular, this shows that a provable first-order sentence in the language of $\PV_1$ is also provable in $\APC_1$.

We start by defining a sentence $\HardA_\eps(\alpha)$ in the language of the relativized theory $\PV_1(\alpha)$. This sentence formalizes that $\alpha(1^{2^n})$ outputs a truth table of length $2^n$ that is $2^{-\eps n}$-hard on average. Formally: 

\begin{definition}[{\cite[Definition 2.1]{Jerabek07}}]
$\HardA_\eps(\alpha)$ is the following sentence in the language of $\PV_1(\alpha)$: For every $n\in\Log\Log$ and $x$ such that $||x||=n$, $\alpha(x)\in\zo^{2^n}$ is a truth-table of a Boolean function $f$ in $n$ variables such that for every circuit $C:\zo^n\to\zo$ of size at most $2^{\eps n}$, 
\[
\Pr_{x\gets\zo^n}[C(x)=f(x)] \le \frac{1}{2} + \frac{1}{2^{\eps n}}.
\]
Note that this probability is defined by a brute-force exact counting algorithm as $n\in\Log\Log$. 
\end{definition}

\begin{definition}[{\cite[Definition 2.13]{Jerabek07}}]
The theory $\HARDA$ is defined as $\PV_1(\alpha)+\dWPHP(\PV(\alpha))+\HardA_{1/4}(\alpha)$, where $\dWPHP(\PV(\alpha))$ denotes the dual Weak Pigeonhole Principle for $\PV(\alpha)$ functions. 
\end{definition}

The following theorem can be proved using tools from \cite{Jerabek07}, where $\NW(\cdot,\cdot)$ is an instantiation of the Nisan-Wigderson PRG \cite{NW} with the hard truth table provided by $\alpha(\cdot)$. The proof is straightforward but requires familiarity with the theory $\APC_1$; for completeness, we provide a proof of the theorem in \Cref{apd: ub APC}.

\begin{restatable}[Simulating $\P(C,\Delta)$ with $\NW(C,\Delta)$]{theorem}{UBAPX}\label{thm: ub APC}
    For every $\eps<1/3$, there is a term $\NW(C,\Delta)$ in the language of $\PV_1(\alpha)$ such that  \refaxiom{basic}, \refaxiom{boundary}, \refaxiom{local}, and \refaxiom{precision} are provable in $\HARDA$ when the oracle $\P(C,\Delta)$ is replaced by $\NW(C,\Delta)$.  
\end{restatable}

We note that by \cite[Theorem 2.13]{Jerabek07}, the theory $\HARDA$ is a conservative extension of $\APC_1$. It then immediately follows that: 

\begin{corollary}\label{cor: APX subset APC1}
    Any first-order sentence in the language of $\PV_1$ provable in $\APX_1$ is also provable in $\APC_1$. 
\end{corollary}

\subsubsection{A Conditional Separation: \texorpdfstring{$\APC_1$}{APC1} is Likely Stronger Than \texorpdfstring{$\APX_1$}{APX1}} 

As $\APX_1$ is an alternative theory for polynomial-time approximate counting and probabilistic reasoning, an interesting question is whether it is \emph{strictly} weaker than $\APC_1$. We provide a positive answer under plausible assumptions, by adapting a technique from \cite{ILW23}.      

A main technical tool is a KPT witnessing theorem (see \cite{KrajicekPT91,Oliveira25}) for the theory $\APX_1$, where the ``student'' is implemented by polynomial-size circuits. Formally: 

\begin{restatable}[KPT Witnessing with Circuits]{definition}{KPTdef}\label{def: KPT}
    Let $\calT$ be an extension of $\PV_1$. We say that $\calT$ satisfies the \emph{KPT witnessing property with circuits} if the following holds. Let $\varphi(\vec x,y,z)$ be any quantifier-free formula in the language of $\PV_1$ such that $\calT\vdash \forall\vec x~\exists y~\forall z~\varphi(\vec x,y,z)$. Then there is a constant $k\in\bbN$ and functions $f_1,f_2,\dots,f_k$ (in the standard model) such that the following holds. 

    For every vector $\vec x$ of strings  and every $z_1,z_2,\dots,z_k\in\zo^*$, it holds in the standard model that: 
    \begin{compactitem}
    \item either $\varphi(\vec x,f_1(\vec x),z_1)$ is true; 
    \item or $\varphi(\vec x,f_2(\vec x,z_1), z_2)$ is true; 
    \item or $\varphi(\vec x,f_3(\vec x,z_1,z_2), z_3)$ is true;
    \item $\dots$; 
    \item or $\varphi(\vec x,f_k(\vec x,z_1,\dots,z_{k-1}),z_k)$ is true. 
    \end{compactitem}
    Moreover, over any fixed input length for $\vec x$, $f_1,\dots,f_k$ are computable by \emph{polynomial-size deterministic circuits}.  
\end{restatable}

\begin{restatable}[KPT Witnessing for $\APX_1$]{theorem}{KPTforAPX}\label{thm: KPT APX1}
    $\APX_1$ admits the KPT witnessing property with circuits. 
\end{restatable}

The theorem can be proved using the standard KPT witnessing theorem (see, e.g., \cite[Theorem 3.2]{Oliveira25}) and the fact that the circuit acceptance probability problem is solvable by (non-uniform) polynomial-size circuits.\footnote{A similar KPT witnessing theorem is proved in \cite[Theorem 4]{DBLP:conf/stoc/PichS21} (see also \cite[Theorem 25]{ILW23}) for the theory $\PV_1+$ ``uniform $\dWPHP(\PV)$'', which  might be incomparable with $\APX_1$.} Put another way, \Cref{thm: KPT APX1} holds as we can hard-wire a sequence of explicit pseudorandom distributions to implement the approximate counting oracle $\P$. Since the argument is standard, we defer the proof of the theorem to \Cref{apd: KPT APX}. 

We will use the following result that is implicit in the proof of \cite[Theorem 24]{ILW23}; we refer readers to \cite{ILW23} for precise statements of the assumptions. 

\begin{theorem}[Implicit in {\cite[Theorem 24]{ILW23}}]\label{thm: ILW}
Assume the existence of JLS-secure $\iO$ and that $\coNP$ is not contained infinitely often in $\NP_{/\poly}$. For any theory $\calT$ extending $\PV_1$ that satisfies the KPT witnessing property with circuits, there is a $\forall\Sigma_2^b$ sentence in the language of $\PV_1$ that is provable in $\APC_1$, but is unprovable in $\calT$.  
\end{theorem}

\begin{proof}[Proof Sketch]
\cite[Theorem 24]{ILW23} only proves the theorem for a theory called $\UAPC_1$. Nevertheless, by a closer inspection, the only property of the theory used in the proof is that $\UAPC_1$ satisfies the KPT witnessing property with circuits (see \cite[Theorem 25]{ILW23}). 
\end{proof}

By combining \Cref{thm: KPT APX1} and \ref{thm: ILW}, it immediately follows that: 

\begin{corollary}\label{cor: APC strict extension}
    Assume the existence of JLS-secure $\iO$ and that $\coNP$ is not contained infinitely often in $\NP_{/\poly}$. There is a $\forall\Sigma_2^b$ sentence in the language of $\PV_1$ that is provable in $\APC_1$, but is unprovable in $\APX_1$. 
\end{corollary}

\subsubsection{An Open Problem: Further Separations?}

\Cref{cor: APC strict extension} shows that $\APC_1$ is likely strictly stronger than $\APX_1$. In other words, under computational assumptions, there are $\forall\Sigma_2^b(\PV)$ sentences provable in $\APC_1$ that are not provable in $\APX_1$. An intriguing open problem is whether $\APC_1$ is strictly stronger than $\APX_1$ with respect to $\forall\Sigma_1^b(\PV)$ sentences: 

\begin{open}\label{open: TFNP separation}
    Is there a $\forall\Sigma_1^b$ sentence in the language of $\PV_1$ that is provable in $\APC_1$, but unprovable in $\APX_1$? In other words, is there an $\APC_1$ provably total $\TFNP$ problem (in the language of $\PV$) that is not provably total in $\APX_1$?  
\end{open}

As $\LossyCode$ captures the $\forall\Sigma_1^b$-fragment of $\APC_1$ (see \Cref{thm: APC1 witnessing}), and $\RefYao$ captures the $\forall\Sigma_1^b$-fragment of $\APX_1$ (see \Cref{thm: witness TFNP}), a related question in the theory of pseudorandomness is whether derandomizing $\LossyCode$ is harder than derandomizing $\RefYao$. Formally: 

\begin{open}\label{open: lossy to ref yao}
    Is there a deterministic polynomial-time reduction from $\LossyCode$ to $\RefYao$? In other words, is there a converse to \Cref{thm: ref yao to lossy}? 
\end{open}

We note that these two open problems are technically incomparable. For instance, a positive answer to \Cref{open: lossy to ref yao} may not give a negative answer to \Cref{open: TFNP separation} if it does not have a feasible correctness proof. Nevertheless, it is conceivable that techniques developed for one of them are likely useful for the other.

\subsection{Predictor Extraction Lemma: Proof of \Cref{lmm: predictor extraction}}\label{sec:predictor_extraction_proof}

\Extraction*

Before proving this lemma, we briefly explain  the  intuition. Recall that $\APX$ is defined as the extension of $\PV(\P)$ by additional axioms: \refaxiom{basic}, \refaxiom{boundary}, \refaxiom{local}, and \refaxiom{precision}. An $\APX$ proof of the equation $t_1(x)=t_2(x)$ is, at a high level, a proof of the following statement: For every interpretation of the oracle $\P$, either $t_1(x)=t_2(x)$, or the oracle $\P$ does not satisfy one of the axioms. 

For our specific implementation of the oracle $\P$ in \Cref{lmm: predictor extraction}, \refaxiom{basic}, \refaxiom{boundary} and \refaxiom{precision} are always satisfied, therefore only \refaxiom{local} can be violated. In such cases, for a circuit $C:\zo^t\to\zo$ and strings $\Delta, B$ (constructed in the $\APX$ proof) such that
\[
\left|\Pr_{u\gets\disD}[C(u_{\le t})]-\frac{1}{2}\left(\Pr_{u\gets\disD}[C(u_{<t} 0)]+\Pr_{u\gets \disD}[C(u_{<t}1)]\right)\right| > 2 \cdot  \frac{1}{|\Delta|}+\frac{1}{|B|},
\]
following a similar argument as in the proof of Yao's lemma (see, e.g., \cite[Chapter 9]{AB09} or \Cref{sec:Yao_transformation}), we can construct a predictor from $C$ via a deterministic polynomial-time algorithm. 

From a conceptual point of view, the argument crucially explores that a predictor can be constructed not only from the ability to distinguish a distribution from a random string, as in the standard formulation of Yao's lemma, but also from the ability to detect a \emph{local inconsistency} when using the distribution as a random source for approximate counting. This is a perspective that might be of independent interest.

In order to implement this intuition, we prove the lemma using a careful proof-theoretic analysis. Formally, we will prove \Cref{lmm: predictor extraction} by induction on the $\APX$ proof $\pi$ of the equation $t_1(x)=t_2(x)$. The functions $\ell,k,m,s,\delta$ and the algorithm $E$ will be determined based on the last rule or axiom of $\pi$ and the functions and algorithms obtained from the induction hypothesis. 

\begin{proof}[Proof of \Cref{lmm: predictor extraction}]
    We will prove a stronger statement: For any provable equation $e(x):t_1(\vec x)=t_2(\vec x)$, there are non-decreasing polynomials $k_0(n),m_0(n)$ such that the lemma holds for every polynomials $k(n)\ge k_0(n)$ and $m(n)\ge m_0(n)$, when the function $m(n)$ is non-decreasing. We prove this by induction on the length of the proof. 
    
    Consider the axiom or rule used in the last line of the proof that concludes $e(x)$. 

    \paragraph{\refaxiom{basic}.} Suppose that $e$ is a provable equation in $\PV(\P)$ and it is introduced via the \refaxiom{basic}.\footnote{Note that our proof does not look into the $\PV(\P)$ proof of $e$; it works provided that $e$ is a $\PV(\P)$ provable equation.} Then, for any interpretation of $\P$ and any string $x$, $e(x)$ must be true. We set $k_0(n)$ to be sufficiently large such that $t\le k_0(n)$ for every $C:\zo^t\to\zo$ queried in simulations $t_1^\calD(\vec x),t_2^\calD(\vec x)$ on $x\in\zo^n$; this is possible as $t_1,t_2$ are polynomial-time oracle algorithms. We set other functions and $E$ arbitrarily as $t_1^\calD(\vec x)=t_2^\calD(x)$ always holds. 

    Suppose that $e$ is one of the equations encoding $\P(C,\Delta)\in\bbQ$, $\P(C,\Delta)\le 1$, or $\P(C,\Delta)\ge 0$, where $C,\Delta$ are open variables. Let $k_0(n)\eqdef n$ and $m_0(n)=1$. Similar to the previous case, one can see that for any $k(n)\ge k_0(n)$ and $m(n)\ge m_0(n)$, $e(\vec x)$ must be true if we interpret $\P$ following \Cref{equ: simulation of P by dist}. We can set other functions and $E$ arbitrarily as $t_1^\calD(\vec x)=t_2^\calD(\vec x)$ always holds. 

    \paragraph{\refaxiom{boundary}.} Suppose that $e(\vec x)$ is an equation encoding that for any circuit $C\in B_n$, $\IsConst(C)\to \P(C,\Delta)=\Bool(C)$. Note that $C$ and $\Delta$ are the only open variables of the equation. We can set $k_0(n)=n$ and $m_0(n)=1$. For every $k(n)\ge k_0(n)$ and $m(n)\ge m_0(n)$, we can set $E$ arbitrarily. 
    
    This is correct as for every $(C,\Delta) \in\zo^n$ and every distribution $\disD$ of size $m(n)$ over $k(n)\ge n$ bits, when we interpret $\P$ following \Cref{equ: simulation of P by dist}, if $C:\zo^t\to\zo$ is a constant circuit, we have $t\le n\le k(n)$ and 
    \[
    \P(C,\Delta) = \Pr_{u\gets\disD}[C(u_{\le t})] = \Bool(C). 
    \]
    In other words, $t_1^\disD(C,\Delta) = t_2^\disD(C,\Delta)$ is always true. 
    
    \paragraph{\refaxiom{precision}.} This is similar to the case for the \refaxiom{boundary}. Indeed, when the oracle $\P$ is interpreted as \Cref{equ: simulation of P by dist}, $\P(C,\Delta_1)=\P(C,\Delta_2)$ for \emph{any} $\Delta_1,\Delta_2$. 

    \paragraph{\refaxiom{local}.} In this case, $e(\vec x)$ is an equation encoding the following sentence: For every circuit $C:\zo^t\to\zo$ and strings $\Delta,B$, we have 
    \[
    \left|\P(C,\Delta) - \frac{\P(\Fix_0(C),\Delta)+\P(\Fix_1(C),\Delta)}{2}\right| \le \frac{2}{|\Delta|} + \frac{1}{|B|},
    \]
    where $\Fix_\sigma(C)$ is a $\PV$-term that outputs the circuit obtained from $C$ by fixing the rightmost input bit to $\sigma\in\zo$. In this equation, $C,\Delta,B$ are the only open variables. Let $k_0(n)$ and $m_0(n)$ be polynomials to be determined later. 

    Let $k(n)\ge k_0(n)$ and $m(n)\ge m_0(n)$. Let $\vec x=(C,\Delta,B)\in\zo^n$ and $\disD\in(\zo^{k(n)})^{m(n)}$. Suppose that $t_1^\disD (\vec x)\ne t_2^\disD(\vec x)$. Since the oracle $\P$ is interpreted following \Cref{equ: simulation of P by dist}, we have 
    \[
    \left|\Pr_{u\gets\disD}[C(u_{\le t})] - \frac{\Pr_{u\gets\disD}[C(u_{< t}0)] + \Pr_{u\gets\disD}[C(u_{< t}1)]}{2}\right| > \frac{2}{|\Delta|} + \frac{1}{|B|}\ge \frac{1}{n}. 
    \]
    This implies that 
    \begin{align}
      & \frac{1}{2}\left|\Pr_{u\gets\disD}[C(u_{<t} u_t)] - \Pr_{u\gets\disD}[C(u_{<t}\overline{u_t})]\right| \nonumber\\ 
    = & \left|\Pr_{u\gets\disD}[C(u_{\le t})] - \frac{\Pr_{u\gets\disD}[C(u_{< t} u_t)] + \Pr_{u\gets \disD}[C(u_{<t} \overline{u_t})]}{2}\right| \nonumber\\ 
    = & \left|\Pr_{u\gets\disD}[C(u_{\le t})] - \frac{\Pr_{u\gets\disD}[C(u_{< t}0)] + \Pr_{u\gets\disD}[C(u_{< t}1)]}{2}\right| \ge \frac{1}{n}, \label{equ: diff Cu Cnot}
    \end{align}
    where in the second equality we used linearity of expectation and that for every fixed $u$, $C(u_{<t} u_t) + C(u_{<t} \overline{u_{t}}) = C(u_{<t}0) + C(u_{<t} 1)$.
    
    For simplicity, we only consider the case that 
    \begin{equation}
    \Pr_{u\gets\disD}[C(u_{<t} u_t)] - \Pr_{u\gets\disD}[C(u_{<t}\overline{u_t})]\ge \frac{2}{n},
    \end{equation}
    and the other case can be resolved accordingly. We can rewrite the equation above as follows: 
    \begin{align*}
        & \left(\Pr_{u\gets\disD}[C(u_{<t} 0)\oplus 1=u_t\land u_t=0] - \Pr_{u\gets\disD}[C(u_{<t}0)\oplus 1\ne u_t\land u_t=1]\right)  \\
    +   & \left(\Pr_{u\gets\disD}[C(u_{<t} 1)\oplus 1\ne u_t\land u_t=1] - \Pr_{u\gets\disD}[C(u_{<t}1)\oplus 1 =u_t\land u_t=0]\right) \ge \frac{2}{n}. 
    \end{align*}
    Therefore, one of the terms in the LHS must be at least $1/n$. Again, we will only consider the case that the first term is at least $1/n$, and the other term can be resolved accordingly. 
    
    Notice that 
    \begin{align*}
      &  \Pr_{u\gets\disD}[C(u_{<t} 0)\oplus 1=u_t\land u_t=0] - \Pr_{u\gets\disD}[C(u_{<t}0)\oplus 1\ne u_t\land u_t=1] \\
    = &  \Pr_{u\gets\disD}[C(u_{<t} 0)\oplus 1=u_t\land u_t=0] + \Pr_{u\gets\disD}[C(u_{<t}0)\oplus 1= u_t\land u_t=1] - \Pr_{u\gets\disD}[u_t=1] \\ 
    = & \Pr_{u\gets\disD}[C(u_{<t} 0)\oplus 1 = u_t] - \Pr_{u\gets\disD}[u_t=1] \ge \frac{1}{n}. 
    \end{align*}
    Subsequently, either $\Pr_{u\gets\disD}[u_t=1]\le 1/2 - 1/(2n)$ or $\Pr_{u\gets\disD}[C(u_{<t} 0)\oplus 1 = u_t] \ge 1/2 + 1/(2n)$. In either case -- and we can efficiently determine which case holds since $\mathcal{D}$ is explicitly given -- we can construct a circuit of size at most $|C|\le s\eqdef n$ that predicts the $t$-th bit of $\disD$ from the first $(t-1)$ bits with advantage at least $\delta\eqdef 1/(2n)$. By setting $m_0(n)\eqdef n^4$ and $k_0(n)\eqdef n$, we can ensure that 
    \[
    (\delta^2/10)\cdot m(n) \ge s + \lceil \log k(n)\rceil + 1
    \]
    when $n$ is sufficiently large. 

    \paragraph{\refaxiom{logical}.} We will only consider the substitution rule $t_1=t_2\vdash t_1(x/t)=t_2(x/t)$; other logical rules can be resolved accordingly. In this case, $e(\vec x)$ is of form $t_1(x/t)=t_2(x/t)$, where $t$ is a term and $x$ is one of the open variables of $t_1$ and $t_2$, and there is a shorter proof of the premise $t_1=t_2$. Without loss of generality, we assume that the variable $x$ does not occur in the term $t$. 

    Let $\vec y$ be the open variables in $t$ but not in $t_1,t_2$, $\vec z$ be the open variables in $t_1,t_2$ (excluding $x$) but not in $t$, and $\vec w$ be the open variables in both $t$ and $t_1,t_2$ (excluding $x$). By the induction hypothesis, there are polynomials $k_0'(n),m_0'(n)$ such that for every polynomials $k'(n)\ge k_0'(n),m'(n)\ge m_0'(n)$, there are $\delta'(n)$, $s'(n)$, and an  algorithm $E'(\vec x,\disD)$ that satisfies the lemma for the equation
    \begin{equation}
    t_1(x,\vec z,\vec w) = t_2(x,\vec z,\vec w). 
    \end{equation}

    Let $\ell(n)$ be an upper bound on the output length of $t$ when the input length is at most $n$ (this is called the \emph{bounding value} of the term, see \cite{Coo75}). We define $k_0(n)\eqdef k_0'(n+\ell(n))$ and $m_0(n)\eqdef m_0'(n+\ell(n))$. 

    To show that this is correct, fix any $k(n)\ge k_0(n)$ and $m(n)\ge m_0(n)$, and let $s(n),\delta(n)$ be determined later. The algorithm $E$ works as follows. Given any $(\vec y,\vec z,\vec w)\in\zo^n$ and $\disD\in(\zo^{k(n)})^{m(n)}$ such that 
    \begin{equation}
    t_1^\disD(t(\vec y,\vec w),\vec z,\vec w) \ne t_2^\disD(t(\vec y,\vec w),\vec z,\vec w),\label{equ: t1 noteq t2 with disD simulation}
    \end{equation}
    our goal is to output a predictor of a bit of $\disD$ with size $s(n)$ and advantage $\delta(n)$. 
    
    The algorithm first computes $x_\disD\eqdef t^\disD(\vec y,\vec w)$, which is a string of length at most $\ell(n)$. By \Cref{equ: t1 noteq t2 with disD simulation} and the definition of the simulation, 
    \[
    t_1^\disD(x_\disD,\vec z,\vec w) \ne t_2^\disD(x_{\disD},\vec z,\vec w). 
    \]
    Subsequently, by the induction hypothesis (with functions $k'(n+\ell(n))\ge k(n)$ and $m'(n+\ell(n))\ge m(n)$), $E'((x_\disD,\vec z,\vec w),\disD)$ outputs $(i,P)$ such that $P$ is a circuit of size $s$ that predicts the $i$-th bit of $\disD$ with advantage $\delta$ such that
    \[
    ((\delta^2/10)\cdot m'(n+\ell(n)) \ge s+\lceil \log k'(n+\ell(n))\rceil + 1. 
    \]
    It suffices to define $E((\vec y,\vec z,\vec w),\disD)\eqdef E'((x_\disD,\vec z,\vec w),\disD)$.
    
    \paragraph{\refaxiom{induction}.} In this case, $e$ is of form $f_1(x,\vec y)=f_2(x,\vec y)$ for $\PV(\P)$ functions $f_1,f_2$, and there are $\PV(\P)$ functions $g,h_0,h_1$ such that there are shorter proofs of equations  
    \begin{align}
    & f_j(\eps,\vec y) = g(\vec y) \label{equ: ind base fj eq g} \\ 
    & f_j(s_i(x),\vec y) = h_i(x, \vec y, f_j(x,\vec y))\label{equ: ind ind fj hi}
    \end{align}
    for $j\in\{1,2\}$ and $i\in\zo$. By the induction hypothesis, the lemma holds for each of the $6$ equations above. 

    Let $k_0(n)$ and $m_0(n)$ are polynomials to be determined. For any polynomials $k(n)\ge k_0(n)$ and $m(n)\ge m_0(n)$, we will design an algorithm $E$ that, given $(x,\vec y)\in\zo^n$ and $\disD\in(\zo^{k(n)})^{m(n)}$ satisfying 
    \begin{align*}
    f^\disD_1(x,\vec y) \ne f_2^\disD(x,\vec y),
    \end{align*}
    it outputs $(i,P)$ such that $P$ is a circuit of size $s$ that predicts the $i$-th bit of $\disD$ with advantage $\delta$ such that $(\delta^2/10)\cdot m(n) \ge s + \lceil \log k(n)\rceil + 1$. 

    \subparagraph{Case 1.} Suppose that $f_1^\disD(\eps,\vec y)\ne f_2^\disD(\eps,\vec y)$. Then there exists $j\in\{1,2\}$ such that $f_j^\disD(\eps,\vec y)\ne g^\disD(\vec y)$. Note that $\vec y$ is of length at most $n$. As \Cref{equ: ind base fj eq g} admits a shorter proof, by the induction hypothesis, there are polynomials $k_0^{(j)}(n'),m_0^{(j)}(n')$ such that when 
    \begin{equation} 
    k(n')\ge k_0^{(j)}(n'),\quad m(n')\ge m_0^{(j)}(n'), \label{equ: case 1 requirement k and m}
    \end{equation}
    then $E_j(\vec y,\disD)$ outputs a size-$s$ predictor with advantage $\delta$ such that $(\delta^2/10)\cdot m(n') \ge s + \lceil \log k(n')\rceil + 1$, where $n'$ denotes the input length of $\vec y$. It then suffices to define 
    \begin{align*}
    & E((x,\vec y),\disD)\eqdef E_j(\vec y,\disD) \\ 
    & k_0(n)\eqdef k_0^{(j)}(n) \in \poly(n), \quad 
    m_0(n)\eqdef m_0^{(j)}(n)\ge m_0^{(j)}(n').
    \end{align*}

    \subparagraph{Case 2.} Let $t\eqdef |x|$. The algorithm $E$ first finds the first index $i\le t$ such that $f_1^\disD(x_{\le i},\vec y) =f_2^\disD(\vec x_{\le i},\vec y)$ but $f_1^\disD(x_{\le i+1},\vec y) \ne f_2^\disD(\vec x_{\le i+1},\vec y)$; such an index must exist as $f_1^\disD(\eps,\vec y)= f_2^\disD(\eps,\vec y)$ and $f_1^\disD(x,\vec y)\ne f_2^\disD(x,\vec y)$. 
    
    Let $\sigma\eqdef x_{i+1}$. Then there exists $j \in\{1,2\}$ such that 
    \[
    f_j^{\disD}(s_{\sigma}(x_{\le i}),\vec y) \ne h_{\sigma}^\disD(x_{\le j},\vec y,f_j^\disD(x_{\le i},\vec y)). 
    \]
    That is, the string $(x_{\le j},\vec y)$ of length at most $n$ violates \Cref{equ: ind ind fj hi} when the approximate counting oracle $\P$ is implemented using $\disD$ following \Cref{equ: simulation of P by dist}. By the induction hypothesis applied to \Cref{equ: ind ind fj hi}, there are polynomials $k_0'(n)$ and $m_0'(n)$ such that when
    \begin{equation} 
    k(n)\ge k_0'(n),\quad m(n)\ge m_0'(n), \label{equ: case 2 requirement k and m}
    \end{equation}
    there is an algorithm $E'$ such that $E'((x_{\le i},\vec y),\disD)$ outputs a size-$s$ predictor with advantage $\delta$ such that $(\delta^2/10)\cdot m(n) \ge s + \lceil \log k(n)\rceil + 1$. It then suffices to define $E((x,\vec y),\disD)\eqdef E'((x_{\le i},\vec y),\disD)$. 

    \subparagraph{Wrapping up.} Finally, we set $k_0(n)$ and $m_0(n)$ to be sufficiently large polynomials such that when $k(n)\ge k_0(n)$ and $m(n)\ge m_0(n)$, both \Cref{equ: case 1 requirement k and m} and \eqref{equ: case 2 requirement k and m} hold. Therefore, in either case, the algorithm $E$ satisfies the requirement of the lemma. 
\end{proof}

\begin{remark}
    By looking into the proof, we note that the polynomials $k_0(n),m_0(n)$ (which define the minimum size of the distribution $\disD$) and the running time of $E$ depend on the $\APX$ proof; in particular, they may far exceed the running time of the terms $t_1,t_2$ in the equation $e$.  
    
    For instance, $k_0(n)$ is defined as $k_0'(n+\ell(n))$ in the case for logical rules, where $\ell(n)$ is the output length (i.e.~bounding value) of a term $t$ in the proof, and the term $t$ does not necessarily appear in the final equation. At a high level, this is because we need to set the distribution $\disD$ to be large enough to accommodate \emph{all oracle queries} in the $\APX$ proof; we can then look through the proof and find a violation of the \refaxiom{local}, which produces a predictor. 
\end{remark}

\subsection{Simulating \texorpdfstring{$\P(C,\Delta)$}{P(C,Delta)} with \texorpdfstring{$\NW(C,\Delta)$}{NW(C, Delta)}: Proof of \texorpdfstring{\Cref{thm: ub APC}}{Theorem \ref{thm: ub APC}}}
\label{apd: ub APC}

\newcommand{\surjecto}{\twoheadrightarrow}
\newcommand{\lsim}{\lesssim}

We follow the notation in \cite{Jerabek07}. A set $X$ is said to be a \emph{bounded set} defined by a circuit $C$ if $X=\{x < a\mid C(x)=1\}$. We use $x\in X$ to denote the formula $x < a \land C(x)=1$, and $X\subseteq b$ to denote the formula $\forall x\in X~x < b$. Note that bounded definable sets are \emph{not} objects in the theory $\APC_1$, but an abbreviation in the meta-theory. For two bounded definable sets $X\subseteq a$ and $Y\subseteq b$, we define 
\begin{align*}
& X\times Y\eqdef \{bx+a\mid x\in X,y\in Y\}\subseteq ab,\\ 
& X\cupdot Y\eqdef X\cup\{y+a\mid y\in Y\}\subseteq a + b. 
\end{align*}

We say $C:X\to Y$ if $C$ is a circuit from $X$ to $Y$, i.e., for every $x\in X$, $C(x)\in Y$. We say $C:X\injecto Y$ if the circuit $C:X\to Y$ is injective, i.e., for $x_1,x_2\in X$, $x_1\ne x_2$, $C(x_1)\ne C(x_2)$. We use $C:X\surjecto Y$ to denote that $C$ is onto, i.e., for all $y\in Y$, there exists an $x\in X$ such that $C(x)=y$. 

\begin{definition}[in $\PV_1$]\label{def: size comparison}
Let $X,Y\subseteq 2^n$ be definable sets, and $\eps \le 1$. We say that $X$ is $\eps$-approximately smaller than $Y$, denoted by $X\lesssim_\eps Y$, if there exists a circuit $G$ and $v\ne 0$ such that 
\[ 
G:v\times (Y\cupdot \eps 2^n)\surjecto v\times X.
\] 
\end{definition}

\begin{definition}[in $\PV_1$]\label{def: equal size}
We say that $X$ and $Y$ are $\eps$-approximately of equal size, denoted by $X\approx_\eps Y$, if $X\lesssim_\eps Y$ and $Y\lesssim_\eps X$. In particular, we say that $X$ is $\eps$-approximately of size $s$ if $X\approx_\eps s$. 
\end{definition}

\begin{lemma}[{\cite[Lemma 2.10]{Jerabek07}}]\label{lmm: basic counting}
Let $X,Y,X',Y',Z\subseteq 2^n$ and $W,W'\subseteq 2^n$ be bounded definable sets, and $\eps,\delta\le 1$. The following statements are provable in $\PV_1$.
\begin{compactenum}[(1)]
\item \label{item: mono} If $X\lesssim_\eps Y, \eps\le \delta$, then $X\lesssim_\delta Y$. 
\item If $X\lesssim_0 Y$, then $X\lsim_\eps Y$. 
\item \label{item: trans} If $X\lsim_\eps Y$, $Y\lsim_\delta Z$, then $X\lsim_{\delta+\eps} Z$. 
\item \label{item: union} If $X\lsim_\eps X'$, $Y\lsim_\delta Y'$, and $X',Y'$ are separable by the set $W$ (i.e., $X'\subseteq W$ and $Y'\subseteq 2^n\setminus W$), then $X\cup Y\lsim_{\eps+\delta} X'\cup Y'$.
\item \label{item: product} If $X\lsim_\eps X',W\lsim_\delta W'$, then $X\times W\lsim_{\eps+\delta+\eps\delta} X'\times W'$. 
\end{compactenum}
\end{lemma}

\begin{lemma}[{\cite[Lemma 2.11]{Jerabek07}}]\label{lmm: counting}
Let $X,Y\subseteq 2^n$ be bounded definable sets, $s,t,u\le 2^n$, $\eps,\delta,\eta,\xi\le 1$, and $\xi^{-1}\in\Log$. The following statements are provable in $\APC_1$.  
\begin{compactenum}[(1)]
\item \label{item: size exists} There exists $s\le 2^n$ such that $X\approx_\xi s$.
\item \label{item: size cons} $s\lsim_\eps X\lsim_\delta t$ implies $s\le t+(\eps+\delta+\xi)\cdot 2^n$. 
\item \label{item: dichotomy} $X\lsim_\xi Y$ or $Y\lsim_\xi X$. 
\item \label{item: negation} $X\lsim_\eps Y$ implies $2^n\setminus Y\lsim_{\eps+\xi} 2^n\setminus X$. 
\item \label{item: inclusion exclusion} $X\approx_\eps s$, $Y\approx_\delta t$, $X\cap Y\approx_\eta u$ imply $X\cup Y\approx_{\eps+\delta+\eta+\xi} s + t - u$. 
\end{compactenum}
\end{lemma}

\begin{lemma}[{Implicit in \cite[Lemma 2.14]{Jerabek07}}]\label{lmm: size function}
Let $\eps<1/3$. There is a $\PV(\alpha)$ function $\Size(\cdot,\cdot)$ such that the following sentence is provable in $\HARDA$: For every $n,\delta^{-1}\in\Log$ and set $X\subseteq 2^n$ defined by a circuit $C:\zo^n\to\zo$, the following holds: 
\begin{compactitem}
\item $0\le \Size(C,1^{\delta^{-1}})\le 2^n$;
\item $X\approx_\delta\Size(C,1^{\delta^{-1}})$;
\item If $C$ is a constant circuit that always outputs $0$ \emph{(}resp.~$1$\emph{)}, then $\Size(C,\Delta)=0$ \emph{(}resp.~$\Size(C,\Delta)=2^n$\emph{)}. 
\end{compactitem}
\end{lemma}

\begin{proof}[Proof Sketch] We assume some familiarity with \cite{Jerabek07}. 
The first two bullets of the lemma hold for a $\PV(\alpha)$ function symbol  as the only non-uniformity in \cite[Theorem 2.7]{Jerabek07} is the choice of the hard function, which is given by $\alpha$. The last bullet holds as $\Size(C,\Delta)$ is obtained by computing the acceptance probability of $C$ on a pseudorandom distribution produced via the Nisan-Wigderson PRG, and for a circuit that always accepts (resp.~rejects), its acceptance probability on any distribution is always $1$ (resp.~$0$). 

We refer interested readers to \cite[Appendix D]{ILW23} for a self-contained presentation of the proof of \cite[Theorem 2.7]{Jerabek07}.  
\end{proof}

\UBAPX* 

\begin{proof}
    Let $\Size(\cdot,\cdot)$ be the function in \Cref{lmm: size function}. We define $\NW(C,\Delta)\eqdef \Size(C,1^{|\Delta|}) / 2^n$, where $C:\zo^n\to\zo$ is a circuit. Note that we can encode the rational number $\Size(C,1^{|\Delta|}) / 2^n$ \emph{precisely} using $n$ binary digits so that there is no rounding issue; recall that in the definition of $\APX$, the output length of $\P$ could be as large as $|C|\cdot |\Delta|^2 > n$ (see \Cref{sec: def: notation}). 
    
    It suffices to verify that the axioms are provable in $\HARDA$. In the rest of the proof, we argue in $\HARDA$.  
    \begin{compactitem}
        \item (\refaxiom{basic}). It follows immediately from the first bullet of \Cref{lmm: size function}. 
        \item (\refaxiom{boundary}). It follows immediately from the third bullet of \Cref{lmm: size function}.
        \item (\refaxiom{local}). Let $C:\zo^n\to\zo$ be a circuit, $\Delta,B$ be strings with $\delta=1/|\Delta|$, $\beta=1/|B|$, $C_0\eqdef\Fix_0(C),C_1\eqdef \Fix_1(C)$. We need to prove that 
        \begin{equation}
        \left|\NW(C,\Delta) - \frac{\NW(C_0,\Delta)+\NW(C_1,\Delta)}{2}\right| \le \delta + \beta. \label{equ: target NW local consistent}
        \end{equation}
        Let $X,X_0,X_1\subseteq\zo^n$ be the bounded sets defined by $C,C_0,C_1$, respectively. Let $s\eqdef \Size(C,1^{|\Delta|})$, $s_0=\Size(C_0,1^{|\Delta|})$, $s_1\eqdef \Size(C_1,1^{|\Delta|})$. By the second bullet of \Cref{lmm: size function}, we know that 
        \[
        X\approx_\delta s,\quad X_0\approx_\delta s_0, \quad X_1\approx_\delta s_2.
        \]
        Note that as $X_0\cap X_1=\varnothing$, we have that $X_0\cap X_1\approx_0 0$. By \Cref{lmm: counting} \eqref{item: inclusion exclusion}, $X_0\cup X_1\approx_{2\delta} s_0 + s_1$. 

        Let $f:v_1 (s +\delta\cdot 2^n)\surjecto v_1\times X$ be the witness of $X\lsim_\delta s$, and $g:v_2\times ((X_0\cup X_1)\cupdot 2\delta \cdot 2^{n-1})\surjecto v_2(s_0+s_1)$ be the witness of $s_0+s_1\lsim_{2\delta} X_0\cup X_1$. We define a function $h$  
        \begin{align*}
        & h:v_1v_2(s+2\delta\cdot 2^n)\surjecto v_1v_2(s_0+s_1) 
        \end{align*}
        as follows: 
        \begin{compactenum}
        \item[(\emph{i})] Let $i_1<v_1,i_2<v_2,j<s+\delta\cdot 2^n$, the tuple $(i_1,i_2,j)\in v_1v_2 (s+2\delta\cdot 2^n)$. We compute $(i'_1,x)\eqdef f(i_1,j)\in v_1\times X$, where $x_{\le n-1}\in X_0\cup X_1$. The algorithm then computes $(i_2',j')\eqdef g(i_2,x_{\le n-1})\in v_2(s_0+s_1)$, and outputs $h(i_1,i_2,j)\eqdef (i_1',i_2',j')$. 
        \item[(\emph{ii})] Let $i_1<v_1,i_2<v_2,s+\delta\cdot 2^n\le j<s+2\delta\cdot 2^n$, the tuple $(i_1,i_2,j)\in v_1v_2 (s+2\delta\cdot 2^n)$. Note that $j-(s+\delta\cdot 2^n)\in \delta\cdot 2^n=2\delta \cdot 2^{n-1}$. The algorithm then computes $(i_2',j')\eqdef g(i_2, j-(s+\delta\cdot 2^n))\in v_2 (s_0+s_1)$, and outputs $h(i_1,i_2,j)\eqdef (i_1,i_2',j')$. 
        \end{compactenum}
        
        It can be verified that the function is indeed onto,  and thus by definition, $s_0+s_1\lsim_{2\delta} s$. Similarly, we can prove that $s\lsim_{2\delta} s_0+s_1$. Then we have 
        \[
        s \lsim_0 s \lsim_{2\delta } s_0+s_1\quad \text{and}\quad s_0+s_1\lsim_{0}s_0+s_1\lsim_{2\delta} s,
        \]
        by \Cref{lmm: counting} \eqref{item: size cons}, we have that $s\in (s_0+s_1)\pm (2\delta + \beta)\cdot 2^n$.
        This immediately implies \Cref{equ: target NW local consistent} as $\NW(C,\Delta)\eqdef \Size(C,1^{|\Delta|})/2^n$.  
        \item (\refaxiom{precision}). Let $C:\zo^n\to\zo$ be a circuit, $\Delta_1,\Delta_2,B$ be strings and $\delta_1\eqdef 1/|\Delta_1|,\delta_2\eqdef 1/|\Delta_2|$, $\beta\eqdef 1/|B|$. We need to prove that 
        \begin{equation}
        |\NW(C,\Delta_1)-\NW(C,\Delta_2)|\le \delta_1+\delta_2+\beta. \label{equ: target NW precision}
        \end{equation}
        Let $X\subseteq\zo^n$ be the bounded set defined by $C$, and let $s_i\eqdef \Size(C,1^{|\Delta_i|})$, $i\in\{1,2\}$. By the second bullet of \Cref{lmm: size function}, 
        \[
        X\approx_{\delta_1} s_1,\quad X\approx_{\delta_2} s_2.  
        \]
        Therefore, we have $s_1\lsim_{\delta_1} X\lsim_{\delta_2} s_2$, and by \Cref{lmm: counting} \eqref{item: size cons}, $s_1\le s_2+(\delta_1+\delta_2+\beta)\cdot 2^n$. Similarly, we can prove that $s_2\le s_1+(\delta_1+\delta_2+\beta)\cdot 2^n$. This immediately implies \Cref{equ: target NW local consistent} as $\NW(C,\Delta)\eqdef \Size(C,1^{|\Delta|})/2^n$. 
    \end{compactitem}
    This completes the proof. 
\end{proof}

\subsection{A KPT Witnessing Theorem for \texorpdfstring{$\APX_1$}{APX1}\texorpdfstring{: Proof of \Cref{thm: KPT APX1}}{}}
\label{apd: KPT APX}

Recall the definition of the KPT witnessing property with circuits: 

\KPTdef*

\KPTforAPX*

To prove \Cref{thm: KPT APX1}, we will need the standard KPT witnessing theorem for universal first-order theories \cite{KrajicekPT91}; interested readers are referred to \cite[Theorem 3.2]{Oliveira25} for detailed discussions. 

\begin{theorem}[KPT witnessing theorem]\label{thm: KPT general}
Let $\calT$ be a universal theory. Let $\varphi(\vec x,y,z)$ be a quantifier-free formula in the language of $\calT$ such that $\calT\vdash \forall\vec x~\exists y~\forall z~\varphi(\vec x,y,z)$. Then there is a constant $k\in\bbN$ and terms $t_1,t_2,\dots,t_k$ (in the language of $\calT$) such that the following statement is provable in $\calT$: 

For every $\vec x$ and every $z_1,z_2,\dots,z_k$, either $\varphi(\vec x,t_1(\vec x),z_1)$, or $\varphi(\vec x,t_2(\vec x,z_1), z_2)$, or $\varphi(\vec x,t_3(\vec x,z_1,z_2), z_3)$, $\dots$, or $\varphi(\vec x,t_k(\vec x,z_1,\dots,z_{k-1}),z_k)$.   
\end{theorem}

\begin{proof}[Proof of \Cref{thm: KPT APX1}]
    Recall that $\APX_1$ admits a universal axiomatization (see \Cref{prop: apx1 universal}). By \Cref{thm: KPT general} with $\calT\eqdef \APX_1$, if $\APX_1\vdash \forall \vec x~\exists y~\forall z~\varphi(\vec x,y,z)$, there is a constant $k\in\bbN$ and terms $t_1,\dots,t_k$ such that $\APX_1$ proves the following sentence: 
    \begin{compactitem}
    \item[] ($\Phi$): For every $\vec x$ and $z_1,z_2,\dots,z_k$, either $\varphi(\vec x,t_1(\vec x),z_1)$, or $\varphi(\vec x,t_2(\vec x,z_1), z_2)$, or $\varphi(\vec x,t_3(\vec x,z_1,z_2), z_3)$, $\dots$, or $\varphi(\vec x,t_k(\vec x,z_1,\dots,z_{k-1}),z_k)$.
    \end{compactitem}
    Note that $t_1,\dots,t_k$ are $\PV(\P)$ terms, which can be  interpreted as polynomial time $\P$-oracle algorithms.

    Similarly to the proof of \Cref{thm: witness TFNP}, we can rewrite the universal sentence ($\Phi$) as an equation $e_\Phi$ in $\APX$, such that it is provable in $\APX$. Therefore, we know by the soundness of $\APX$ (see \Cref{prop: soundness APX}) that $e_\Phi$ is true in any standard model $\bbM(\hat\P)$.

    Let $\CAPP$ be the search problem that, given any circuit $C$ and a string $\Delta$, outputs a number $p\in\Pr[C(x)]\pm 1/|\Delta|$. It is well-known that the problem is in $\pr\BPP$ (see, e.g., \cite{Goldreich11g}), and thus can be computable by a family of deterministic polynomial-size circuits. Fix any family of circuits $F(C,\Delta)$ that solves $\CAPP$. By definition, $\bbM(F)$ is a standard model of $\APX$. Subsequently, $e_\Phi$ is true in the model $\bbM(F)$. 
    
    The theorem follows by setting $f_1,f_2,\dots,f_k$ as $t_1^{\bbM(F)},t_2^{\bbM(F)},\dots,t_k^{\bbM(F)}$. Since for each fixed input length, there is a polynomial upper bound on the size of $C$ and on the length of $\Delta$ in the oracle calls to $F(C,\Delta)$ during the computation of $t_1$, $\ldots$, $t_k$, this allows us to fix a family of polynomial-size circuits for $f_1$, $\ldots$, $f_k$.
\end{proof}

\section{Reverse Mathematics of Randomized and Average-Case Lower Bounds}\label{sec:rev_math}

The retraction weak pigeonhole principle for polynomial-time functions is one of the most important combinatorial principles that is known to be provable in $\APC_1$, but unknown to be provable in $\Apx_1$. In this section, we explore counting variants of the retraction pigeonhole principle and characterize their equivalence class (with respect to provability in $\Apx_1$). We show that this class encompasses certain average-case and randomized communication complexity lower bounds, establishing that these results are all equivalent to appropriate variants of the retraction pigeonhole principle.

\subsection{Variants of the Retraction Pigeonhole Principle}

We start with the definition of the Retraction Weak Pigeonhole Principle $\rWPHP(\PV)$. For simplicity, we introduce the following notation. We use $m:\Log\to\Log$ to denote a $\PV$ function symbol $m(n)$ whose input and output are encoded in unary. We use $\eps:\Log\to\Log^{-1}$ to denote a $\PV$ function symbol $e(n)$ whose input and output are encoded in unary, and $\eps(n)$ is an abbreviation of the rational number $1/e(n)$. 

\begin{definition}[retraction weak pigeonhole principle]
Let $m:\Log\to\Log$. The retraction weak pigeonhole principle $\rWPHP[m](\PV)$ with stretch $m$ denotes the following statement in the language of $\Apx_1$: 

For every $n\in\Log$ and circuits $C:\zo^n\to\zo^{m(n)}$, $D:\zo^{m(n)}\to\zo^n$, if $m(n)<n$, then there exists a string $x\in\zo^n$ such that $D(C(x))\ne x$.     
\end{definition}

We will define two variants of $\rWPHP(\PV)$: an approximate counting version $\rcWPHP(\PV)$, and a randomized compression version $\rrWPHP(\PV)$. 

\begin{definition}[approximate counting $\rWPHP$] 
Let $m:\Log\to\Log,\eps:\Log\to\Log^{-1}$. The approximate counting retraction weak pigeonhole principle $\rcWPHP[m,\eps](\PV)$ with stretch $m$ and error $\eps$ denotes the following statement in the language of $\Apx_1$:

For every $n\in\Log$ and circuits $C:\zo^n\to\zo^{m(n)}$, $D:\zo^{m(n)}\to\zo^n$, let $T$ be the circuit such that $T(x)=1$ if $D(C(x))\ne x$. Then, if $m(n)<n$ and $\eps(n) < 1-2^{n-m(n)}$, there exists $\delta^{-1},\beta^{-1}\in\Log$ such that $\P_\delta(T) > \eps(n)+\delta +\beta$. 
\end{definition}

\newcommand{\sd}{\mathsf{sd}}
\newcommand{\msg}{\mathsf{msg}}

\begin{definition}[randomized compression $\rWPHP$]
Let $m:\Log\to\Log$, $\eps:\Log\times\Log\to\Log^{-1}$. The randomized compression retraction weak pigeonhole principle $\rrWPHP[m,\eps](\PV)$ with stretch $m$ and error $\eps$ denotes the following statement in the language of $\Apx_1$: 

For every $n,r\in\Log$ and circuits $C:\zo^n\times \zo^r\to\zo^{m(n)}$, $D:\zo^{m(n)}\to\zo^n$, if $m(n)<n$ and $\eps(n+r)<1-2^{n-m(n)}$, then there exists an $x\in\zo^n$ and $\delta^{-1},\beta^{-1}\in\Log$ such that the following holds: Let $T_{x}:\zo^r\to\zo$ be the circuit such that $T_x(\sd)=1$ if and only if $D(C(x,\sd))\ne x$. Then $\P_\delta(T_x) > \eps(n,r)+\delta + \beta$.  
\end{definition}

It follows immediately from the definition that both variants of $\rWPHP$ are true statements in any standard model of $\APX_1$, which is left as an exercise. 

\begin{proposition}
    For every $m:\Log\to\Log$, $\eps:\Log\to\Log^{-1}$, $\rcWPHP[m,\eps](\PV)$ and $\rrWPHP[m,\eps](\PV)$ are true statements in any standard model $\bbM(\hat \P)$ of $\APX_1$.
\end{proposition}

These principles can be viewed as the worst-case and (weak) average-case hardness of compression-decompression algorithms. Specifically: 
\begin{compactitem}
\item $\rWPHP[m](\PV)$ says that for any deterministic compression-decompression pair $(C, D)$ with compression rate $m$, there is an incompressible string;
\item $\rcWPHP[m,\eps](\PV)$ says that for any deterministic compression-decompression pair $(C,D)$ with compression rate $m$, there is an $\eps$-fraction of incompressible strings;
\item $\rrWPHP_m(\PV)$ says that for any $(C,D)$ where $C$ is a \emph{randomized} compression algorithm and $D$ is a \emph{deterministic} decompression algorithm, there must be an input string over which the compression-decompression pair has \emph{error probability} $\eps$.
\end{compactitem}

A classical result in bounded arithmetic is that the retraction weak pigeonhole principle admits a stretch reduction in $\PV$. Concretely: 

\begin{theorem}[\cite{thapen2002weak,Jerabek-phd}]
For any constant $\eps\in(0,1)$, $\PV+\rWPHP[{n^\eps}](\PV)\vdash \rWPHP[{n-1}](\PV)$.  
\end{theorem}

\subsection{One-Way Communication Lower Bounds}

We prove an equivalence result involving  $\rcWPHP(\PV)$ and communication complexity lower bounds for \emph{Set Disjointness} against \emph{one-way protocols} with either \emph{public randomness} or \emph{private randomness}.

\paragraph{Formalization of One-Way Communication Protocols.} We start with the formalization of a communication protocol. Let $n,m,r\in\Log$. A pair of circuits $g_\Alice:\zo^n\times \zo^r\to\zo^m$ and $d_\Bob:\zo^n\times \zo^m\times \zo^r\to\zo$ defines a one-way randomized communication protocol as follows: 
\begin{itemize}
    \item (\emph{Public Coin Model}). On any pair of inputs $(x,y)\in\zo^n\times \zo^n$ and a uniformly generated public random seed $\sd\in\zo^r$, Alice sends the message $\msg\eqdef g_\Alice(x,\sd)$ to Bob, and Bob decides to accept if and only if $d_\Bob(y,\msg,\sd) = 1$. 
    \item (\emph{Private Coin Model}). On any pair of inputs $(x,y)\in\zo^n\times \zo^n$ and uniformly generated private random seeds $\sd\in\zo^r$, Alice sends the message $\msg\eqdef g_\Alice(x,\sd)$ to Bob, and Bob decides to accept if and only if $d_\Bob(y,\msg,0^r) = 1$.
\end{itemize}

Let $f:\zo^n\times \zo^n\to\zo$ be a function specified by a circuit. For every $x,y\in\zo^n$, let $T^\pub_{f,x,y}:\zo^r\to\zo$ be the circuit such that $T^\pub_{f,x,y}(\sd)=1$ if and only if the public-coin protocol outputs $1-f(x,y)$ on the input $(x,y)$ with seed $\sd$, i.e., 
\begin{equation}\label{equ: def T pub}
T^\pub_{f,x,y}(\sd)\eqdef \Id[d_\Bob(y,g_\Alice(x,\sd),\sd) \ne  f(x,y)]\in\zo.
\end{equation}
Let $\eps\in\bbQ$.  We say that a public-coin protocol $(g_\Alice,d_\Bob)$ computes the function $f$ with error $\eps$ if for $x,y\in\zo^n$, $\delta^{-1},\beta^{-1}\in\Log$, 
\begin{equation}
\P_\delta(T^\pub_{f,x,y})\le \eps + \delta + \beta.\label{equ: pub correctness}
\end{equation}
Note that here we consider the two-sided error setting, while one can also naturally define the correctness in terms of one-sided error.

Accordingly, one may define $T^\prv_{f,x,y}:\zo^r\to\zo$ to be the circuit such that $T^\pub_{f,x,y}(\sd)=1$ if and only if the private-coin protocol outputs $1-f(x,y)$ on the input $(x,y)$ with seed $\sd$, i.e., 
\begin{equation}\label{equ: def T prv}
T^\prv_{f,x,y}(\sd) \eqdef \Id[d_\Bob(y,g_\Alice(x,\sd),0^r)\ne f(x,y)]\in\zo. 
\end{equation}
We say that a private-coin protocol $(g_\Alice, d_\Bob)$ computes the function $f$ with error $\eps$ if for every $x,y\in\zo^n$, $\delta^{-1},\beta^{-1}\in\Log$, 
\begin{equation}
    \P_\delta(T_{f,x,y}^\prv)\le \eps + \delta + \beta.
\end{equation}

\paragraph{Communication Complexity Lower Bounds.} Fix any function $f:\zo^n\times\zo^n\to\zo$, $n,m,r\in\Log$, and $\eps\in\bbQ$. We define the sentence $\pub\de\rsLB^f[n,m,r,\eps]$ as follows: For every \emph{public-coin} protocol $(g_\Alice,d_\Bob)$ as defined above, $(g_\Alice,d_\Bob)$ fails to compute $f$ with error $\eps$. In other words, there are $x,y\in\zo^n$ and $\delta^{-1},\beta^{-1}\in\Log$ such that $\P_\delta(T^\pub_{f,x,y}) > \eps+\delta + \beta$. 

Accordingly, we define the sentence $\prv\de\rsLB^f[n,m,r,\eps]$ as follows: For every \emph{private-coin} protocol $(g_\Alice,d_\Bob)$ as defined above, $(g_\Alice,d_\Bob)$ fails to compute $f$ with error $\eps$. 

Recall that the Set Disjointness function $\SetDisj(x,y)$ outputs $1$ if and only if for every index $i\in[n]$, either $x_i=0$ or $y_i=0$, i.e., $x$ and $y$ have no common $1$-index. Let $m:\Log\to\Log$, $\eps:\Log\times\Log\to\Log^{-1}$ be functions. We define $\pub\de\rsLB^\SetDisj[m,\eps]$ as the following sentence: 
\begin{quote}
    For $n,r\in\Log$, $\pub\de\rsLB^{\SetDisj}[n,m(n),r,\eps(n,r)]$. 
\end{quote}
In other words, every public-coin one-way protocol computing $\SetDisj$ with communication complexity $m(n)$ must have error probability at least $\eps(n)$. As we will prove in \Cref{sec: equivalence class comm}, the lower bound is correct even for $m(n)=n-n^{\Omega(1)}$. Accordingly, we define  $\prv\de\rsLB^{\SetDisj}[m,\eps]$ as the following sentence: 
\begin{quote}
    For $n,r\in\Log$, $\prv\de\rsLB^{\SetDisj}[n,m(n),r,\eps(n,r)]$. 
\end{quote}

We also consider a weaker statement that, instead of formalizing the lower bound for  a specific function, formalizes the \emph{existence} of a function  for which the lower bound holds. Let $m:\Log\to\Log,\eps:\Log\times\Log\to\Log^{-1}$. We define $\pub\de\rLB^{\some}[m,\eps]$ as the following sentence: 
\begin{quote}
    For $n,r\in\Log$, there exists a circuit $f:\zo^n\times \zo^n\to\zo$ such that $\pub\de\rsLB^{f}[n,m(n),r,\eps(n,r)]$ holds. 
\end{quote}
In other words, there exists a function $f$ such that every public-coin one-way protocol computing $f$ with communication complexity $m(n)$ must have a non-negligible error probability. This is implied by $\pub\de\rsLB^\SetDisj[m,\eps]$ by fixing $f$ to be $\SetDisj$. Accordingly, we can define $\prv\de\rsLB^{\some}[m,\eps]$ as the following sentence:
\begin{quote}
    For $n,r\in\Log$, there exists a circuit $f:\zo^n\times \zo^n\to\zo$ such that $\prv\de\rsLB^{f}[n,m(n),r,\eps(n,r)]$ holds. 
\end{quote}

\paragraph{Upper bound for Equality.} As a sanity check, we note that as there is a communication complexity \emph{upper bound} for \emph{Equality} with public randomness using linear hashing (see \Cref{thm: linear hashing}), the corresponding lower bound is unprovable in $\Apx_1$.

\begin{theorem}[Upper Bound for Equality]
    There are $\PV$ functions $m:\Log\to\Log$, $\eps:\Log\to\Log^{-1}$ satisfying that $m(n)=\Theta(\log n)$, $\eps(n)=1-1/n^{\Theta(1)}$ such that 
    \begin{equation}
    \APX_1\vdash \forall n~\lnot \pub\de\rsLB^\EQ[n,m(n),n\cdot m(n),\eps(n)].\label{equ: equality random ub}
    \end{equation}
    In particular, $\APX_1\vdash \lnot \pub\de\rLB[m,\eps]$. 
\end{theorem}

\begin{proof}[Proof Sketch]
    We argue in $\APX_1$ that \Cref{equ: equality random ub} holds, where $m,r,\eps$ will be determined later. As the cases when $n$ is small can be proved in brute-force, it suffices to consider $n>n_0$, where $n_0\in\bbN$ is a constant to be determined later. 
    
    Fix any $n>n_0$. The one-way communication works as follows. Let $x\in\zo^n$ be the input for Alice and $y\in\zo^n$ be that for Bob. They parse the public randomness as a matrix $A\in\zo^{m(n)\times n}$. Alice sends $Ax\in\zo^{m(n)}$ as the message, and Bob accepts if and only if $Ax=Ay$. It remains to prove that the protocol works with error at most $\eps(n)$. 

    Fix any input $x,y\in\zo^n$. If $x=y$, the protocol always accepts. In other words, the circuit $T_{\EQ,x,y}^\pub$ in \Cref{equ: def T pub} is a constant circuit that always rejects. The correctness, i.e.~\Cref{equ: pub correctness}, follows immediately from the \refaxiom{boundary}. For the case that $x\ne y$, the circuit $T_{\EQ,x,y}^\pub$ is functionally equivalent to the negation of the circuit $T_{x,y}$ in \Cref{thm: linear hashing}. Therefore, by \refmeta{complementation}, the theorem holds as long as we set $m(n)=n^{O(1)}$, $\eps(n)>1-0.51^{m(n)}$, and $n_0\in\bbN$ be sufficiently large. This completes the proof.  
\end{proof}

\subsection{The Main Equivalence Result  for Communication Complexity}
\label{sec: equivalence class comm}

We establish an equivalence between several statements with respect to their provability in $\Apx_1$. 

\begin{theorem}\label{thm:rev_equivalences}
    The following statements are equivalent over $\Apx_1$:
    \vspace{0.2cm}
    \begin{compactenum}[\rm(1)]
    \item\label{rcWPHP small} $\rcWPHP[n-1,n^{-k}](\PV)$, where $k\in\bbN$ is some constant; 
    \item\label{rcWPHP large} $\rcWPHP[n^{\eps},n^{-k}](\PV)$, where $\eps\in(0,1)$ and $k\in\bbN$ are some constants; 
    \item\label{rrWPHP small} $\rrWPHP[n-1,(n+r)^{-k}](\PV)$, where $k\in\bbN$ is some constant; 
    \item\label{rrWPHP large} $\rrWPHP[n^{\eps},(n+r)^{-k}](\PV)$, where $\eps\in(0,1)$ and $k\in\bbN$ are some constants;
    \item\label{pub setdisj max} $\pub\de\rsLB^\SetDisj[n-1,(n+r)^{-k}]$, where $k\in\bbN$ is some constant;  
    \item\label{pub setdisj weak} $\pub\de\rsLB^\SetDisj[n^{\eps},(n+r)^{-k}]$, where $\eps\in(0,1)$ and $k\in\bbN$ are some constants;
    \item\label{prv setdisj max} $\prv\de\rsLB^\SetDisj[n-1,(n+r)^{-k}]$, where $k\in\bbN$ is some constant;
    \item\label{prv setdisj weak} $\prv\de\rsLB^\SetDisj[n^{\eps},(n+r)^{-k}]$, where $\eps\in(0,1)$ and $k\in\bbN$ are some constants;
    \item\label{pub some max}  $\pub\de\rsLB^\some[n-1,(n+r)^{-k}]$, where $k\in\bbN$ is some constant;
    \item\label{pub some weak} $\pub\de\rsLB^\some[n^{\eps},(n+r)^{-k}]$, where $\eps\in(0,1)$ and $k\in\bbN$ are some constants;
    \item\label{prv some max} $\prv\de\rsLB^\some[n-1,(n+r)^{-k}]$, where $k\in\bbN$ is some constant;
    \item\label{prv some weak} $\prv\de\rsLB^\some[n^{\eps},(n+r)^{-k}]$, where $\eps\in(0,1)$ and $k\in\bbN$ are some constants.
    \end{compactenum}
\end{theorem}

\begin{remark}
    In the statements above, the quantification over $k\in\bbN$ and $\eps\in(0,1)$ takes place outside the theory. For instance, \eqref{rcWPHP small} $\Rightarrow$ \eqref{rcWPHP large} means that for every $k_1\in\bbN$, there exists a $k_2\in\bbN$ and $\eps_2\in(0,1)$ such that the sentence
    \[
    \rcWPHP[n-1,n^{-k_1}](\PV) \rightarrow \rcWPHP[n^{\eps_2},n^{-k_2}](\PV). 
    \]
    is provable in $\APX_1$.  
\end{remark}

\pgr{Trivial directions} Both \eqref{rcWPHP small} $\Rightarrow$ \eqref{rcWPHP large} and \eqref{rrWPHP small} $\Rightarrow$ \eqref{rrWPHP large} are straightforward. Indeed, a  compression-decompression pair with small stretch can be converted into one with larger stretch by padding dummy bits.  It is also easy to observe that statements \eqref{pub setdisj max} to \eqref{prv some weak} form a lattice isomorphic to a three-dimensional Boolean cube with respect to implication over $\Apx_1$, where \eqref{pub setdisj max} is the maximal element (i.e., the strongest lower bound) and \eqref{prv some weak} is the minimal element (i.e., the weakest lower bound). This is because lower bounds against public-coin protocols imply lower bounds against private-coin protocols; $n-1$ communication lower bounds imply $n^{\Omega(1)}$ communication lower bounds; and lower bounds for $\SetDisj$ imply lower bounds for some function $f$ (by fixing $f$ to be $\SetDisj$).  

\pgr{Non-trivial directions} Observe that, in order to complete the proof of \Cref{thm:rev_equivalences}, it suffices to establish the following implications: \eqref{prv some weak} $\Rightarrow$ \eqref{rrWPHP large}, \eqref{rcWPHP small} $\Rightarrow$ \eqref{pub setdisj max}, \eqref{rcWPHP large} $\Rightarrow$ \eqref{rrWPHP small}, and \eqref{rrWPHP large} $\Rightarrow$ \eqref{rcWPHP small}. The proof of these implications is provided in the subsequent sections. 

\subsubsection{Compression Implies Communication Upper Bound: \texorpdfstring{\eqref{prv some weak} $\Rightarrow$ \eqref{rrWPHP large}}{(12) => (4)}}

\begin{lemma}
    For every $\eps_{12}\in(0,1)$ and $k_{12}\in\bbN$, there are $\eps_{4}\in(0,1)$ and $k_4\in\bbN$ such that 
    \[ 
    \Apx_1+\prv\de\rsLB^{\some}[n^{\eps_{12}},(n+r)^{-k_{12}}]\vdash\rrWPHP[n^{\eps_4},(n+r)^{-k_4}](\PV).
    \] 
\end{lemma}

\begin{proof}
    Fix any constant $\eps_6\in(0,1)$, $k_{12}\in\bbN$, let $\eps_2\eqdef \eps_6$ and $k_4\eqdef k_{12}$. Let $m_4(n)\eqdef n^{\eps_4}$ and $m_{12}(n)\eqdef n^{\eps_{12}}$. We will prove in $\Apx_1$ that $\lnot\rrWPHP[n^{\eps_4},n^{-k_4}](\PV)$ implies $\lnot\prv\de\rsLB^{\some}[n^{\eps_{12}},n^{-k_{12}}]$. 

    Suppose that $\rrWPHP[m_4,n^{-k_4}](\PV)$ does not hold. Then there are $n,r\in\Log$ and circuits $C:\zo^n\times \zo^r\to\zo^{m_2(n)},D:\zo^{m_2(n)}\to\zo^n$ such that the following holds. Let $T_x$ be the circuit that $T_x(\sd)=1$ if $D(C(x,\sd))\ne x$. Then for every $\delta^{-1},\beta^{-1}\in\Log$ and $x\in\zo^n$, 
    \begin{equation} 
    \P_\delta(T_x)\le  (n+r)^{-k_{4}}+\delta + \beta .\label{equ: Tx is large} 
    \end{equation}
    Fix $n,r,C,D$ as described above. 

    We will now prove that $\lnot\prv\de\rsLB^{\some}[n^{\eps_{12}},n^{-k_{12}}]$. In particular, we will prove that for every $f:\zo^{n}\times \zo^n\to\zo$, $\prv\de\rsLB^{f}[n^{\eps_{12}},n^{-k_{12}}]$ does not hold. Fix any circuit $f:\zo^n\times \zo^n\to\zo$. Our goal is to construct a private-coin communication protocol with communication complexity $n^{\eps_{12}}$ that computes $f$ with error $n^{-k_{12}}$. The protocol works as follows: 
    \begin{compactitem}
    \item Given $x\in\zo^n$ and uniformly random seed $\sd$, Alice sends the message $g_\Alice(x,\sd)\eqdef C(x,\sd)$.  
    \item Given $y\in\zo^n$ and the message $\msg$, Bob accepts if and only if $d_\Bob(y,\msg,\sd)\eqdef f(D(\msg),y)=1$. 
    \end{compactitem}

    To prove that the protocol computes $f$ with error $n^{-k_{12}}$, fix any $\delta^{-1},\beta^{-1}\in\Log$. Let $\eta^{-1}\in\Log$ be a parameter to be determined later, and $T_{f,x,y}^\prv(\sd)$ be the circuit as defined in \Cref{equ: def T prv}. It can be verified that for every $x,y\in\zo^n$ and $\sd\in\zo^r$, if $T_{f,x,y}^\prv(\sd)=1$, then $T_x(\sd)=1$. By the \refmeta{monotone}, 
    \[
    \P_\eta(T_{f,x,y}^\prv) \le \P_\eta(T_x) + 3\eta \le (n+r)^{-k_4}+5\eta,  
    \]
    where the last inequality follows from \Cref{equ: Tx is large}. It then follows from the \refaxiom{precision} that 
    \[
    \P_\delta(T_{f,x,y}^\prv) \le \delta + \P_\eta(T_{f,x,y}^\prv) \le (n+r)^{-k_4}+\delta + 5\eta \le (n+r)^{-k_4}+\delta + \beta
    \]
    by setting $\eta\eqdef \beta/5$. This completes the proof. 
\end{proof}

\subsubsection{Compression from Communication Upper Bound: \texorpdfstring{\eqref{rcWPHP small} $\Rightarrow$ \eqref{pub setdisj max}}{(1) => (5)}} 

\begin{lemma}
    For every constant $k_1\in \bbN$, there exists a $k_5\in\bbN$ such that $\Apx_1+\rcWPHP[n-1,n^{-k_1}](\PV)\vdash \pub\de\rsLB^\SetDisj[n-1,(n+r)^{-k_5}]$. 
\end{lemma}

\begin{proof}
    Fix any $k_1\in\bbN$ and let $k_5\in\bbN$ be determined later. We argue in $\Apx_1$ that $\lnot \pub\de\rsLB^\SetDisj[n-1,n^{-k_5}]$ implies $\lnot \rcWPHP[n-1,n^{-k_1}](\PV)$. 
    
    Suppose that $\pub\de\rsLB^\SetDisj[n-1,n^{-k_5}]$ does not hold. Then there are $n,r\in\Log$ and a one-way public-coin protocol $g_\Alice:\zo^n\times \zo^r\to\zo^{n-1}$, $d_\Bob:\zo^n\times \zo^{n-1}\times \zo^r\to\zo$ such that the following holds: For every $x,y\in\zo^n$ and $\delta^{-1},\beta^{-1}\in\Log$, let $T_{\SetDisj,x,y}^\pub(\sd)$ be the circuit defined as \Cref{equ: def T pub}, then 
    \begin{equation}\label{equ: error set disj is small}
        \P_\delta(T_{\SetDisj,x,y}^\pub) \le (n+r)^{-k_5} + \delta + \beta. 
    \end{equation}
    Our goal is to construct a compression-decompression scheme that violates $\rcWPHP[n-1,n^{-k_1}](\PV)$. 

    \pgr{Construction of the compression scheme} We construct a pair of circuits $C:\zo^{n+r}\to\zo^{n+r-1},D:\zo^{n+r-1}\to\zo^{n+r}$  as follows.
    \begin{compactitem}
    \item (\emph{Compression}): The circuit $C$ parses the input as $(x,\sd)\in\zo^n\times \zo^r$ and computes $\sigma\in\zo^n$ defined as  
    \begin{equation}
    \sigma_i\eqdef d_\Bob(e_i,g_\Alice(x,\sd),\sd)\oplus x_i\oplus 1, \label{equ: def error in compression}
    \end{equation}
    where $e_i$ denotes the string with the $i$-th bit being its only $1$-index. If $e\ne 0^n$, the compression fails and it outputs $0^{n}$. Otherwise, it outputs the concatenation of $\msg\eqdef g_\Alice(x,\sd)$ and $\sd$.
    \item (\emph{Decompression}): The circuit $D$ parses the input as the concatenation of $\msg$ and $\sd$ as mentioned above, computes $y\in\zo^{n}$ as 
    \[
    y_i\eqdef d_\Bob(e_i, \msg, \sd)\oplus1, 
    \]
    and outputs the concatenation of $y$ and $\sd$. 
    \end{compactitem}
    It is clear that when $\sigma=0^n$, the compression-decompression scheme is correct. 

    \pgr{Analysis of the error probability} We will prove that $(C,D)$ is a compression-decompression scheme that violates $\rcWPHP[n-1,n^{-k_1}](\PV)$. Fix any $\delta^{-1},\beta^{-1}\in\Log$ and let $T:\zo^{n+r}\to\zo$ be the circuit that $T(z)$ outputs $1$ if $D(C(z))\ne z$. Our goal is to prove that $\P_\delta(T)\le n^{-k_1}+\delta + \beta$.  

    Let $\eta^{-1}\in\Log$ be a parameter to be determined later, and $T':\zo^{n+r}\to\zo$ be the following circuit: Given $(x,\sd)\in\zo^n\times\zo^r$, it computes $\sigma$ via \Cref{equ: def error in compression}, and outputs $1$ if and only if $\sigma\ne 0^{n}$. As mentioned above, for every $x$ and $\sd$, $T(x\circ \sd)=1$ implies that $T'(x\circ \sd)=1$. Therefore, by the \refmeta{monotone}, we have 
    \begin{equation}\label{equ: monotone T and T'}
    \P_\eta(T)\le \P_\eta(T') + 3\eta.  
    \end{equation}

    Let $V\eqdef \{0,1\}$ and $X_1,X_2,\dots,X_n$ be the random variables supported over $\zo^{n+r}$ such that $X_i=1$ if and only if $e_i\ne 0$. Let $F_i(x,\sd)$ be the circuit that defines $X_i$ for every $i\in[n]$. It is clear that $T'(x,\sd)$ is the circuit that outputs $1$ if and only if $X_i=1$ for some $i\in[n]$. Let $Y$ be the random variable defined by $(V,n+r,T')$. By the \refmeta{union bound}, we have 
    \begin{equation}\label{equ: union bound bit by bit}
        \P_\eta(T')\le \bbE_\eta[Y] + 3\eta \le \bbE_\eta[X_1] + \dots + \bbE_\eta[X_n] + 3\eta\cdot (n+1), 
    \end{equation}
    where the first inequality follows from \Cref{prop: indicator variable}. 
    
    In addition, for every $i\in[n]$ and every $x\in\zo^n$, we can see that $F_i(x,\cdot)$ is functionally equivalent to $T_{\SetDisj,x,e_i}^\pub(\cdot)$. Let $X_i|_x$ be the random variable obtained by fixing the first part of the seed to be $x$. Then for $x\in\zo^n$, 
    \[ 
    \bbE_\eta[X_i|_x] \le \P_\eta(F_i(x,\cdot)) + 3\eta \le \P_\eta(T_{\SetDisj,x,e_i}^\pub) + 6\eta \le (n+r)^{-k_5} + 8\eta, 
    \] 
    where the last inequality follows from \Cref{equ: error set disj is small}. By \refmeta{avg on exp}, we have 
    \begin{equation}\label{equ: each bit small prob}
    \bbE_\eta[X_i] \le (n+r)^{-k_5} + 8\eta + 3\eta \le n^{-k_5} + 11\eta. 
    \end{equation}

    Combining the results above, we have: 
    \begin{align*}
        \P_\delta(T) & \le \P_\eta(T) + \delta + 2\eta \tag{\refaxiom{precision}} \\ 
        &\le \P_\eta(T') + \delta + 5\eta \tag{\Cref{equ: monotone T and T'}} \\ 
        &\le \bbE_\eta[X_1] + \dots + \bbE_\eta[X_n] + \delta + 3\eta\cdot (n+1) + 5\eta \tag{\Cref{equ: union bound bit by bit}} \\ 
        &\le \delta + ((n+r)^{-k_5}+11\eta)\cdot n + 3\eta\cdot (n+1) + 5\eta \tag{\Cref{equ: each bit small prob}} \\ 
        &\le (n+r)^{-k_1} + \delta + \beta, 
    \end{align*}
    where the last inequality follows by setting $k_5\eqdef k_1+1$ and $\eta \eqdef \beta / (50n)$. This violates $\rcWPHP[n-1,n^{-k_1}](\PV)$ and thus completes the proof. 
\end{proof}

\subsubsection{Stretch Reduction for Compression: \texorpdfstring{\eqref{rcWPHP large} $\Rightarrow$ \eqref{rrWPHP small}}{(2) => (3)}} 

\begin{lemma}\label{lmm: stretch reduction for compression}
    For any $\eps_2\in(0,1)$ and $k_2\in\bbN$, there exists $k_3\in\bbN$ such that $\Apx_1+\rcWPHP[n^{\eps_2},n^{-k_2}](\PV)\vdash \rrWPHP[n-1,(n+r)^{-k_3}](\PV)$.
\end{lemma}
\begin{proof}
    Fix any constant $\eps_2\in(0,1)$ and $k_2\in\bbN$, and let $k_3\in\bbN$ be determined later. We argue in $\Apx_1$ that if $\rrWPHP[n-1,(n+r)^{-k_3}](\PV)$ does not hold, then $\rcWPHP[n^{\eps_2},n^{-k_2}](\PV)$ does not hold. 

    Suppose that $\rrWPHP[n-1,n^{-k_3}](\PV)$ does not hold. Then there are $n,r\in\Log$ and circuits $C:\zo^n\times \zo^r\to\zo^{n-1}$, $D:\zo^{n-1}\to\zo^n$ such that for every $x\in\zo^n$ and every $\delta^{-1},\beta^{-1}\in\Log$, let $T_x:\zo^r\to\zo$ be the circuit such that $T_x(\sd)=1$ if and only if $D(C(x,\sd))\ne x$, then 
    \begin{equation}\label{equ: Tx small from assumption}
        \P_\delta(T_x) \le (n+r)^{-k_3}+\delta + \beta. 
    \end{equation}
    In other words, there is a one-bit randomized compression scheme that is worst-case correct with error $n^{-k_3}$. Our goal is to construct a deterministic and average-case compression-decompression algorithm that violates $\rcWPHP[n^{\eps_2},n^{-k_2}](\PV)$. 

    \pgr{Compression and decompression circuits} Let $\ell\in\Log$ and $d\in\Log\Log$ be parameters to be determined later. The compression circuit takes an $(\ell+dr)$-bit string as input, parses it as $z\in\zo^\ell$ and $(\sd_1,\dots,\sd_d)\in\zo^r$, and runs a $d$-round iterative compression algorithm. 
    
    Initialize $z_0\gets z$. In the $i$-th round, the iterative algorithm works as follows:  
    \begin{compactenum}
    \item Parse $z_{i-1}$ as $x_1\circ x_2\circ \dots\circ x_k\circ y_i$, where $k\eqdef \lfloor |z_{i-1}|/n\rfloor$ and $x_1,\dots,x_k\in\zo^n$. 
    \item For every $j\in[k]$, compute $x'_j\eqdef C(x_j,\sd_i)$. 
    \item Set $z_i\gets x'_1\circ x'_2\circ \dots \circ x'_k$. 
    \end{compactenum}
    Finally, the compression circuit outputs the encoding of the tuple $(z_d, y_1,\dots,y_d,\sd_1,\dots,\sd_d)$. 

    The decompression circuit takes $(z_d, y_1,\dots,y_d,\sd_1,\dots,\sd_d)$ and works reversely via a $d$-round iterative algorithm. In the $i$-th iteration, the algorithm works as follows:  
    \begin{compactenum}
    \item Parse $z_{d+1-i}$ as $x_1'\circ x_2'\circ\dots\circ x_k'$, where $k\eqdef \lfloor|z_{d+1-i}|/(n-1)\rfloor$ and $x'_1,\dots,x'_k\in\zo^{n-1}$. 
    \item For every $j\in[k]$, compute $x_j\eqdef D(x_j')$. 
    \item Set $z_{d-i}\gets x_1\circ x_2\circ \dots\circ x_k\circ y_{d+1-i}$. 
    \end{compactenum}

    We now set the parameters $\ell$ and $d$ such that the compression scheme above has stretch at least $(\ell+dr)^{\eps_2}$. Let $z_0,z_1,\dots,z_d,y_1,\dots,y_d$ be the strings obtained by the compression algorithm, it is clear that 
    \[
    |z_i| \le |z_{i-1}|\cdot \left(1-\frac{1}{n}\right), |y_i|\le n,
    \]
    and thus the output length of the compression circuit is at most 
    \[
    \ell\cdot \left(1-\frac{1}{n}\right)^d + d\cdot (n+r). 
    \]
    We can set $\ell=(n+r)^{10/\eps_2}$ and $d=10\cdot n\log \ell$ such that the output length is at most $O(n(n+r)\log\ell)\ll \ell^{\eps_2}\le (\ell+dr)^{\eps_2}$. Therefore, the compression stretch is $(\ell+dr)^{\eps_2}$ for sufficiently large $n$ and $r$; the cases when $n,r$ are small can be proved by a brute-force case study. 

    \pgr{Analysis of the error probability} Fix $\ell\in\Log$ and $d\in\Log\Log$ as above. Let $C':\zo^{\ell+dr}\to\zo^{(\ell+dr)^{\eps_2}}$, $D':\zo^{(\ell+dr)^{\eps_2}}\to\zo^{\ell+dr}$ be the compression and decompression algorithms mentioned above. Let $T:\zo^{\ell+dr}$ be the circuit that $T(z)=1$ if and only if $D'(C'(z))\ne z$, i.e., the compression scheme fails. Our goal is to prove that for every $\delta^{-1},\beta^{-1}\in\Log$, $\P_\delta(T)\le (\ell+dr)^{-k_2}+\delta + \beta$. 

    Fix any $\delta^{-1},\beta^{-1}\in\Log$ and let $\eta^{-1}\in\Log$ be a parameter to be determined later. Let $V\eqdef \{0,1\}$. For every $i\in[d]$ and $j\le \lfloor\ell/n\rfloor$, we define $F_{ij}(z,\sd_1,\dots,\sd_d)$ be the circuit that outputs $1$ if and only if the following holds: 
    \begin{compactitem}
    \item In the $i$-th round of the compression algorithm, let $k\eqdef \lfloor |z_{i-1}|/n\rfloor$, then $j\le k$ and $D(C(x_j,\sd_i))\ne x_j$.  
    \end{compactitem}
    Let $X_{ij}$ be the random variable defined by $(V,\ell+dr,F_{ij})$. Let $F(z,\sd_1,\dots,\sd_d)$ be the circuit that outputs $1$ if and only if $F_{ij}(z,\sd_1,\dots,\sd_d)=1$ for some $i\in[d]$ and $j\le \lfloor \ell/n\rfloor$, and $Y$ be the random variable defined by $(V,\ell+dr,F_{ij})$. By the \refmeta{union bound}, we have 
    \begin{equation}
    \P_\eta(F)\le \bbE_\eta[Y]+3\eta \le \sum_{ij}\bbE_\eta[X_{ij}] + 3\eta\cdot (d\ell+1),\label{equ: F smaller than sum Xij via union bound}
    \end{equation}
    where the first inequality follows from \Cref{prop: indicator variable}. 
    
    It is clear that $\PV$ proves that for every $z\in\zo^\ell,\sd_1,\dots,\sd_d\in\zo^r$, if $T(z,\sd_1,\dots,\sd_d)=1$, then $F(z,\sd_1,\dots,\sd_d)=1$. To see this, notice that if $F(z,\sd_1,\dots,\sd_d)=0$, we can prove by induction on $i$ that if we run the iterative compression algorithm on the input $z\circ \sd_1\circ\dots\circ \sd_d$ for $i$ rounds, and run the iterative decompression algorithm starting from the $d-i$ round, it will be correctly decompressed. This can be implemented by induction on a $\PV$ term, which is available in $\PV$. Subsequently, by the \refmeta{monotone}, 
    \begin{equation}\label{equ: T smaller than F via monotone}
    \P_\eta(T)\le \P_\eta(F)+3\eta. 
    \end{equation}

    Next, we prove an upper bound on $\bbE_\eta[X_{ij}]$. Fix any $i\in[d]$ and $j\le \lfloor \ell/n\rfloor$. Let 
    \[ 
    \rho=(z,\sd_1,\dots,\sd_{i-1},\sd_{i+1},\dots,\sd_d)
    \] 
    be an arbitrary assignment to all but the interval $\sd_i$ in the seed of $X_{ij}$. Let $x_j$ be the string in the $i$-th round of the compression algorithm on the input $z$ and using $\sd_1,\dots,\sd_{i-1}$ in the first $i-1$ rounds. Note that $x_j$ can be computed by a $\PV$ term given $\rho$. Recall that $T_{x_j}(\sd)$ is the circuit that outputs $1$ if and only if $D(C(x_j,\sd))\ne x_j$. It can be proved that $T_{x_j}(\sd_i)=1$ if and only if $F_{ij}(\rho\cup\sd_i)=1$, i.e., $X_{ij}|_\rho$ is the indicator variable of $T_{x_j}(\sd)=1$. Subsequently, 
    \begin{equation}
    \bbE_\eta[X_{ij}|_{\rho}] \le \P_\eta(T_{x_j}) + 6\eta \le (n+r)^{-k_3}+8\eta, \label{equ: Xij small after assignment}
    \end{equation}
    where the first inequality follows from \Cref{prop: indicator variable} and \refmeta{global}, and the second inequality follows from \Cref{equ: Tx small from assumption}. 
    
    Note that \Cref{equ: Xij small after assignment} holds for any assignment $\rho$. By \refmeta{avg on exp}, we can further deduce that 
    \begin{equation}\label{equ: Xij small via averaging on exp}
    \bbE_\eta[X_{ij}] \le n^{-k_3}+8\eta + 3\eta \le n^{-k_3}+11\eta. 
    \end{equation}

    Combining the results above, we can now calculate 
    \begin{align*}
        \P_\delta(T) & \le \P_\eta(T) + \delta + 2\eta \tag{\refaxiom{precision}} \\ 
        & \le \P_\eta(F) + \delta + 5\eta \tag{\Cref{equ: T smaller than F via monotone}} \\ 
        & \le \sum_{ij}\bbE_\eta[X_{ij}] + 3\eta\cdot (d\ell+1) + \delta + 5\eta \tag{\Cref{equ: F smaller than sum Xij via union bound}}\\
        & \le ((n+r)^{-k_3}+11\eta) \cdot d\cdot \ell + 3\eta \cdot (d\ell+1) + \delta + 5\eta \tag{\Cref{equ: Xij small via averaging on exp}} \\
        & \le (\ell+dr)^{-k_2}+\delta + \beta,
    \end{align*}
    where the last inequality follows by setting $\eta\eqdef \beta/(100(d\ell + 1))$ and $k_3\eqdef 100 k_2/\eps_2 + 10k_2 + 10$. This shows that the pair of circuits $C',D'$ violates $\rcWPHP[n^{\eps_2},n^{-k_2}](\PV)$ and thus completes the proof. 
\end{proof}

\subsubsection{Worst-Case to Average-Case Reduction: \texorpdfstring{\eqref{rrWPHP large} $\Rightarrow$ \eqref{rcWPHP small}}{(4) => (1)}} 

\begin{lemma}\label{lmm: w2a compression}
    For any $\eps_4\in(0,1)$ and $k_4\in\bbN$, there exists $k_1\in\bbN$ such that $\Apx_1+\rrWPHP[n^{\eps_4},(n+r)^{-k_4}](\PV)\vdash \rcWPHP[n-1,n^{-k_1}](\PV)$. 
\end{lemma}

We will use the iterative compression algorithm in \Cref{lmm: stretch reduction for compression} to boost the stretch to $m_2$, while a new trick is required to construct worst-case compression from average-case compression algorithm. At a high level, we observe that the compression-decompression problem with large stretch admits \emph{random self-reducibility} that is provably correct via the \refmeta{rerand}.

\begin{proof}[Proof of \Cref{lmm: w2a compression}.]
    Fix any constant $\eps_4\in(0,1)$, $k_4\in\bbN$, and let $k_1\in\bbN$ be determined later. We argue in $\Apx_1$ that assuming $\rcWPHP[n-1,n^{-k_4}](\PV)$ does not hold, $\rrWPHP[n^{\eps_4},(n+r)^{-k_4}](\PV)$ also does not hold. In other words, we will construct a polynomial-stretch randomized worst-case compression scheme from a one-bit deterministic average-case compression scheme. 

    Assume for contradiction that $\rcWPHP[n-1,n^{-k_1}](\PV)$ does not hold. Then there is an $n\in\Log$ and circuits $C:\zo^n\to\zo^{n-1}$, $D:\zo^{n-1}\to\zo^n$ such that the following holds. Let $T$ be the circuit that $T(x)=1$ if $D(C(x))\ne x$. Then for every $\delta^{-1},\beta^{-1}\in\Log$, $\P_\delta(T)\le n^{-k_1}+\delta + \beta$. By the \refmeta{rerand}, we know that for every $x\in\zo^n$, let $T_x^\oplus$ be the circuit $T_x^\oplus(\sd)\eqdef T(x\oplus \sd)$, then 
    \begin{equation}\label{equ: bound after rerand}
    \P_\delta(T_x^\oplus) \le \P_\delta(T)+2\delta+\beta\le n^{-k_1}+3\delta + 2\beta. 
    \end{equation}
    Note that we can assume without loss of generality that $n$ is larger than any fixed standard integer $n_0\in\bbN$, as the cases when $n\le n_0$ can be resolved in brute force.  

    \paragraph{Compression and decompression circuits.} Let $\ell\in\Log$ and $d\in\Log\Log$ be parameters to be determined later. The compression circuit takes an $\ell$-bit string as input $z$, an $nd$-bit random seed $(\sd_1,\dots,\sd_d)\in\zo^{r}$, and performs the following $d$-round iterative algorithm. It initializes $z_0\gets z$. In the $i$-th round, the algorithm works as follows: 
    \begin{compactenum}
    \item Parse $z_{i-1}$ as $x_1\circ x_2\circ \dots\circ x_k\circ y_i$, where $k\eqdef \lfloor|z_{i-1}|/n\rfloor$ and $x_1,\dots,x_k\in\zo^n$; 
    \item For every $j\in[k]$, compute $x_j'\eqdef C(x_j\oplus \sd_i)$. 
    \item Set $z_i\gets x_1'\oplus x_2'\oplus\dots \oplus x_k'$. 
    \end{compactenum}
    Finally, the compression circuit outputs the encoding of the tuple $(z_d,y_1,\dots,y_d,\sd_1,\dots,\sd_d)$. 

    The decompression circuit takes $(z_d, y_1,\dots,y_d,\sd_1,\dots,\sd_d)$ and works reversely via a $d$-round iterative algorithm. In the $i$-th iteration, the algorithm works as follows:  
    \begin{compactenum}
    \item Parse $z_{d+1-i}$ as $x_1'\circ x_2'\circ\dots\circ x_k'$, where $k\eqdef \lfloor|z_{d+1-i}|/(n-1)\rfloor$ and $x'_1,\dots,x'_k\in\zo^{n-1}$. 
    \item For every $j\in[k]$, compute $x_j\eqdef D(x_j')\oplus \sd_i$. 
    \item Set $z_{d-i}\gets x_1\circ x_2\circ \dots\circ x_k\circ y_{d+1-i}$. 
    \end{compactenum}

    Similar to the proof of \Cref{lmm: stretch reduction for compression}, we can set the parameters $\ell\eqdef n^{10/\eps_4}$ and $d\eqdef 10\cdot n\log \ell$ such that the compression scheme above has stretch at least $\ell^{\eps_4}$. The length of random string of the compression scheme is $r\eqdef dn$. 

    \paragraph{Analysis of the error probability.} Fix $\ell\in\Log$ and $d\in\Log\Log$ as above. Let $C':\zo^\ell\times \zo^{dn}\to\zo^{\ell^{\eps_4}}$, $D':\zo^{\ell^{\eps_4}}\to\zo^\ell$ be the compression and decompression algorithms mentioned above. Let $T_z:\zo^{dn}\to\zo$ be the circuit that parses the input as $\sd\eqdef (\sd_1,\dots,\sd_d)\in\zo^{dn}$ and outputs $1$ if and only if $D'(C'(z,\sd))\ne z$, i.e., the compression scheme fails on the input $z$. Our goal is to prove that for every $\delta^{-1},\beta^{-1}\in\Log$ and $z\in\zo^\ell$, $\P_\delta(T_z)\le (\ell+dn)^{-k_4}+\delta+\beta$. 

    Fix any $\delta^{-1},\beta^{-1}\in\Log,z\in\zo^\ell$ and let $\eta^{-1}\in\Log$ be a parameter to be determined later. Let $V\eqdef\zo$. For every $i\in[d]$ and $j\le \lfloor\ell/n\rfloor$, we define $F_{ij}(\sd_1,\dots,\sd_d)$ be the circuit that outputs $1$ if and only if the following holds: 
    \begin{compactitem}
    \item In the $i$-th round of the compression algorithm, let $k\eqdef \lfloor |z_{i-1}|/n\rfloor$, then $j\le k$ and $D(C(x_j\oplus \sd_i))\ne x_j\oplus \sd_i$.  
    \end{compactitem}
    Let $X_{ij}$ be the random variable defined by $(V,dn,F_{ij})$. Let $F(\sd_1,\dots,\sd_d)$ be the circuit that outputs $1$ if and only if $F_{ij}(\sd_1,\dots,\sd_d)=1$ for some $i\in[d]$ and $j\le \lfloor \ell/n\rfloor$, and $Y$ be the random variable defined by $(V,dn,F_{ij})$. By the \refmeta{union bound}, we have 
    \begin{equation}
    \P_\eta(F)\le \bbE_\eta[Y]+3\eta \le \sum_{ij}\bbE_\eta[X_{ij}] + 3\eta\cdot (d\ell+1),\label{equ: w2a F smaller than sum Xij via union bound}
    \end{equation}
    where the first inequality follows from \Cref{prop: indicator variable}. 

    It is clear that $\PV$ proves that for every $\sd_1,\dots,\sd_d\in\zo^r$, if $T_z(\sd_1,\dots,\sd_d)=1$, it follows that $F(\sd_1,\dots,\sd_d)=1$. To see this, notice that if $F(\sd_1,\dots,\sd_d)=0$, we can prove by induction on $i$ that if we run the iterative compression algorithm on the input $z$ for $i$ rounds, and run the iterative decompression algorithm starting from the $d-i$ round, it will be correctly decompressed. This can be implemented by induction on a $\PV$ term, which is available in $\PV$. Subsequently, by the \refmeta{monotone}, 
    \begin{equation}\label{equ: w2a T smaller than F via monotone} 
    \P_\eta(T_z)\le \P_\eta(F)+3\eta. 
    \end{equation}
    
    Next, we prove an upper bound on $\bbE_\eta[X_{ij}]$. Fix any $i\in[d]$ and $j\le \lfloor \ell/n\rfloor$. Let 
    \[ 
    \rho=(\sd_1,\dots,\sd_{i-1},\sd_{i+1},\dots,\sd_d)
    \] 
    be an arbitrary assignment to all but the interval $\sd_i$ in the seed of $X_{ij}$. Let $x_j$ be the string in the $i$-th round of the compression algorithm on the input $z$ and using $\sd_1,\dots,\sd_{i-1}$ in the first $i-1$ rounds. Note that $x_j$ can be computed by a $\PV$ term given $\rho$. Recall that $T^\oplus_x(\sd)$ is the circuit that outputs $1$ if and only if $D(C(x\oplus \sd))\ne x\oplus \sd$. It can be proved that $T^\oplus_{x_j}(\sd_i)=1$ if and only if $F_{ij}(\rho\cup\sd_i)=1$, i.e., $X_{ij}|_\rho$ is the indicator variable of $T_{x_j}^\oplus(\sd)=1$. Subsequently, 
    \begin{equation}
    \bbE_\eta[X_{ij}|_{\rho}] \le \P_\eta(T_{x_j}^{\oplus}) + 6\eta \le n^{-k_1}+11\eta, \label{equ: w2a Xij small after assignment}
    \end{equation}
    where the first inequality follows from \Cref{prop: indicator variable} and \refmeta{global}, and the second inequality follows from \Cref{equ: bound after rerand}. 
    
    Note that \Cref{equ: w2a Xij small after assignment} holds for any assignment $\rho$. By \refmeta{avg on exp}, we can further deduce that 
    \begin{equation}\label{equ: w2a Xij small via averaging on exp}
    \bbE_\eta[X_{ij}] \le n^{-k_1}+11\eta + 3\eta \le n^{-k_1}+14\eta. 
    \end{equation}

    Combining the results above, we can now calculate 
    \begin{align*}
        \P_\delta(T) & \le \P_\eta(T) + \delta + 2\eta \tag{\refaxiom{precision}} \\ 
        & \le \P_\eta(F) + \delta + 5\eta \tag{\Cref{equ: w2a T smaller than F via monotone}} \\ 
        & \le \sum_{ij}\bbE_\eta[X_{ij}] + 3\eta\cdot (d\ell+1) + \delta + 5\eta \tag{\Cref{equ: w2a F smaller than sum Xij via union bound}}\\
        & \le \left(n^{-k_1}+14\eta\right) \cdot d\cdot \ell + 3\eta \cdot  (d\ell+1) + \delta + 5\eta \tag{\Cref{equ: w2a Xij small via averaging on exp}} \\
        & \le n^{-k_1}\cdot d\ell+\delta + \beta,
    \end{align*}
    where the last inequality follows by setting $\eta\eqdef \beta/(100(d\ell + 1))$. Recall that $\ell=n^{10/\eps_4}$ and $d=10\cdot n\log\ell$, we have 
    \[
    n^{-k_1}\cdot d\ell = n^{-k_1}\cdot n^{10/\eps_4}\cdot \frac{100n}{\eps_4}\cdot \log n\le n^{-k_4} \le (\ell+dn)^{-k_4}. 
    \]
    by setting $k_1\eqdef 100/\eps_4 + 10 k_4 + 10$ when $n$ is sufficiently large. This shows that $C',D'$ violates $\rrWPHP_{m_2}(\PV)$ and thus completes the proof. 
\end{proof}

\bibliographystyle{alpha}
\bibliography{ref}

\end{document}